\newif\iffull \fulltrue
\newif\ifshownotes \shownotesfalse
\newif\ifluxuryspace \luxuryspacefalse

\documentclass[a4paper,USenglish,thm-restate]{lipics-v2019}

\title{Space-Efficient Gradual Typing in Coercion-Passing Style}

\author{Yuya Tsuda}{Graduate School of Informatics, Kyoto University, Japan}{tsuda@fos.kuis.kyoto-u.ac.jp}{https://orcid.org/0000-0002-7420-2575}{}

\author{Atsushi Igarashi}{Graduate School of Informatics, Kyoto University, Japan}{igarashi@kuis.kyoto-u.ac.jp}{https://orcid.org/0000-0002-5143-9764}{}

\author{Tomoya Tabuchi}{Graduate School of Informatics, Kyoto University, Japan}{tabuchi@fos.kuis.kyoto-u.ac.jp}{}{}

\authorrunning{Y. Tsuda, A. Igarashi, and T. Tabuchi}

\Copyright{Yuya Tsuda, Atsushi Igarashi, and Tomoya Tabuchi} 



\relatedversion{A full version of the paper is available at \url{https://arxiv.org/abs/1908.02414}.} 

\funding{This work was partially supported by JSPS KAKENHI Grant Number JP17H01723.}

\acknowledgements{We thank anonymous reviewers for valuable comments and John Toman for proofreading.}

\nolinenumbers 

\EventEditors{Robert Hirschfeld and Tobias Pape}
\EventNoEds{2}
\EventLongTitle{34th European Conference on Object-Oriented Programming (ECOOP 2020)}
\EventShortTitle{ECOOP 2020}
\EventAcronym{ECOOP}
\EventYear{2020}
\EventDate{July 13--17, 2020}
\EventLocation{Berlin, Germany}
\EventLogo{}
\SeriesVolume{166}
\ArticleNo{8}

\newcommand\citep\cite

\hypersetup{final}  

\usepackage{amsmath,amssymb,amsfonts,mathrsfs,stmaryrd}
\usepackage{centernot,accents}
\usepackage{xypic}
\usepackage{mathtools} 

\usepackage{xcolor}
\usepackage{./lib/mymacros}

\usepackage{./lib/bcprules}
\smallrulenames  
\usepackage{./lib/mybcprules}

\newcommand{\ottnt}[1]{\mathit{#1}}
\newcommand{\ottmv}[1]{\mathit{#1}}
\newcommand{\ottkw}[1]{\mathbf{#1}}
\newcommand{\ottsym}[1]{#1}

\renewcommand{\ottkw}[1]{\mathsf{#1} }
\newcommand\metafun[1]{\mathit{#1} }
\newcommand\reduces{\longrightarrow}
\newcommand\evalto{\longmapsto}

\iffull
\usepackage{apptools}
\AtAppendix{%
  \numberwithin{theorem}{section}
  \numberwithin{lemma}{section}
  \numberwithin{proposition}{section}
  \numberwithin{corollary}{section}
  \numberwithin{definition}{section}
  \numberwithin{example}{section}
}
\fi

\usepackage{listings}    

\lstdefinelanguage{ocaml}{      
  language=[Objective]{Caml},
  basicstyle=\ttfamily,         
  keywordstyle=\color{blue},
  extendedchars=true,
  breaklines=true,
  deletekeywords={true,false},  
  tabsize=8,
  escapeinside={<@}{@>},
  mathescape=true,              
  columns=fullflexible,         
  xleftmargin=2em,
}

\lstdefinelanguage{Grift}{
  language=ocaml,  
  morekeywords={letrec,lambda,if},
  basicstyle=\ttfamily,
}

\begin{document}

\maketitle

\begin{abstract}
  Herman et al.\ pointed out that the insertion of
  run-time checks into a gradually typed program could hamper
  tail-call optimization and, as a result, worsen the space complexity
  of the program.  To address the problem, they proposed a
  space-efficient coercion calculus, which was subsequently improved
  by Siek et al.  The semantics of
  these calculi involves eager composition of run-time checks expressed
  by coercions to prevent the size of a term from growing.  However,
  it relies also on a nonstandard reduction rule, which does not seem
  easy to implement.  In fact, no compiler implementation of gradually
  typed languages fully supports the space-efficient semantics
  faithfully.

  In this paper, we study \emph{coercion-passing style}, which Herman
  et al.\ have already mentioned, as a technique for straightforward
  space-efficient implementation of gradually typed languages.  A
  program in coercion-passing style passes ``the rest of the run-time
  checks'' around---just like continuation-passing style (CPS), in
  which ``the rest of the computation'' is passed around---and
  (unlike CPS) composes coercions eagerly.  We give a formal
  coercion-passing translation from \lamS by Siek et al.\ to \lamSx,
  which is a new calculus of \emph{first-class coercions} tailored for
  coercion-passing style, and prove correctness of the translation.
  We also implement our coercion-passing style transformation for the
  Grift compiler developed by Kuhlenschmidt et al.  An experimental
  result shows stack overflow can be prevented properly at the cost of
  up to 3 times slower execution for most partially typed practical programs.
\end{abstract}

\section{Introduction}
\label{sec:introduction}

\subsection{Space-Efficiency Problem in Gradual Typing}

Gradual
typing~\citep{conf/scheme/SiekT06,DBLP:conf/oopsla/Tobin-HochstadtF06}
is one of the linguistic approaches to integrating static and dynamic
typing.  Allowing programmers to mix statically typed
and dynamically typed fragments in a single program, it
advocates the ``script to program'' evolution~\cite{DBLP:conf/oopsla/Tobin-HochstadtF06}.
Namely, software development
starts with simple, often dynamically typed scripts, which evolve to
more robust, fully statically typed programs through intermediate
stages of partially typed programs.  To make this evolution work in
practice, it is important that the performance of partially typed
programs at intermediate stages is comparable to that of (the slower of) the two ends,
that is, dynamically typed scripts and statically typed programs.

However, it has been pointed out that gradual typing suffers from
serious efficiency problems from both theoretical and practical viewpoints~\citep{DBLP:conf/sfp/HermanTF07, DBLP:journals/lisp/HermanTF10,
  DBLP:conf/popl/TakikawaFGNVF16}.  In particular,
Takikawa et al.~\cite{DBLP:conf/popl/TakikawaFGNVF16} showed that even a
state-of-the-art gradual typing implementation could show catastrophic
slowdown for partially typed programs due to run-time checking to
ensure safety.  Worse, such slowdown is not easy to predict because it
depends on implicit run-time checks inserted by the language
implementation and it requires fairly deep knowledge about the
underlying gradual type system to understand when and where run-time checks are
inserted and how they behave.  Since then, several pieces of work have
investigated the performance
issues~\cite{DBLP:journals/pacmpl/BaumanBST17,
  DBLP:journals/pacmpl/MuehlboeckT17,
  DBLP:journals/pacmpl/RichardsAT17, DBLP:conf/popl/RastogiSFBV15,
  DBLP:conf/pldi/KuhlenschmidtAS19, DBLP:journals/pacmpl/FelteyGSFS18}.

Earlier work by Herman et al.~\cite{DBLP:conf/sfp/HermanTF07,
  DBLP:journals/lisp/HermanTF10} pointed out a related problem.  They
showed that, when values are passed between a statically typed part
and a dynamically typed part many times, delayed run-time checks may
accumulate and make space complexity of a program worse than an unchecked
semantics.

To make the discussion more concrete, consider the following mutually
recursive functions (written in ML-like syntax):
\begin{lstlisting}
let rec even (x : int) : $ \mathord{\star} $ =
  if x = 0 then true$ \graytext{ \langle   \ottkw{bool} \texttt{!}   \rangle } $ else (odd (x - 1))$ \graytext{ \langle   \ottkw{bool} \texttt{!}   \rangle } $
and odd (x : int) : bool =
  if x = 0 then false else (even (x - 1))$ \graytext{ \langle   \ottkw{bool} \texttt{?}^{ \ottnt{p} }   \rangle } $
\end{lstlisting}
Ignoring the \graytext{gray} part (in angle brackets), which will be explained shortly, this is a
tail-recursive definition of functions to decide whether a given integer is
even or odd, except that the return type of one of the functions is
written $ \mathord{\star} $, which is the dynamic type, which can be any tagged
value.  This definition expresses a situation where a statically typed
and a dynamically typed function call each
other.\footnote{In this sense, the argument of \texttt{even} should
  have been $ \mathord{\star} $, too, but it would clutter the code after
  inserting run-time checks.}  The \graytext{gray} part represents inserted
run-time checks, written using Henglein's coercion syntax~\citep{DBLP:journals/scp/Henglein94}:
$ \ottkw{bool} \texttt{!} $ is a coercion from $\ottkw{bool}$ to $ \mathord{\star} $ and
$\texttt{true}\langle   \ottkw{bool} \texttt{!}   \rangle$ means that (untagged) Boolean value
$\ottkw{true}$ will be tagged with $\ottkw{bool}$ to make a value of the
dynamic type;
$ \ottkw{bool} \texttt{?}^{ \ottnt{p} } $ is a coercion from $ \mathord{\star} $ to $\ottkw{bool}$ and
$\texttt{(even (x - 1))}\langle   \ottkw{bool} \texttt{?}^{ \ottnt{p} }   \rangle$ means that
the value returned from recursive call \texttt{even (x - 1)} will be
tested whether it is tagged with $\ottkw{bool}$---if so, the
run-time check removes the tag and returns the untagged Boolean value,
and, otherwise, it results in \emph{blame}, which is an uncatchable
exception (with label $\ottnt{p}$ to indicate where the check has failed).

The crux of this example is that the insertion of run-time checks has
broken tail recursion: due to \ifluxuryspace the presence of \fi $\langle   \ottkw{bool} \texttt{!}   \rangle$ and
$\langle   \ottkw{bool} \texttt{?}^{ \ottnt{p} }   \rangle$, the recursive calls are not in tail positions any
longer.  So, according to the original semantics of
coercions~\citep{DBLP:journals/scp/Henglein94}, evaluation of
$\ottkw{odd} \, \ottsym{4}$ is as follows:
\begin{align*}
  \ottkw{odd} \, \ottsym{4}
  & \mathbin{ \evalto ^*}  \ottsym{(}  \ottkw{even} \, \ottsym{3}  \ottsym{)}  \langle   \ottkw{bool} \texttt{?}^{ \ottnt{p} }   \rangle  \mathbin{ \evalto ^*}  \ottsym{(}  \ottkw{odd} \, \ottsym{2}  \ottsym{)}  \langle   \ottkw{bool} \texttt{!}   \rangle  \langle   \ottkw{bool} \texttt{?}^{ \ottnt{p} }   \rangle \\
  & \mathbin{ \evalto ^*}  \ottsym{(}  \ottkw{even} \, \ottsym{1}  \ottsym{)}  \langle   \ottkw{bool} \texttt{?}^{ \ottnt{p} }   \rangle  \langle   \ottkw{bool} \texttt{!}   \rangle  \langle   \ottkw{bool} \texttt{?}^{ \ottnt{p} }   \rangle 
    \mathbin{ \evalto ^*}  \ottsym{(}  \ottkw{odd} \, \ottsym{0}  \ottsym{)}  \langle   \ottkw{bool} \texttt{!}   \rangle  \langle   \ottkw{bool} \texttt{?}^{ \ottnt{p} }   \rangle  \langle   \ottkw{bool} \texttt{!}   \rangle  \langle   \ottkw{bool} \texttt{?}^{ \ottnt{p} }   \rangle \\
  & \mathbin{ \evalto ^*}  \ottkw{false}  \langle   \ottkw{bool} \texttt{!}   \rangle  \langle   \ottkw{bool} \texttt{?}^{ \ottnt{p} }   \rangle  \langle   \ottkw{bool} \texttt{!}   \rangle  \langle   \ottkw{bool} \texttt{?}^{ \ottnt{p} }   \rangle  \mathbin{ \evalto ^*}  \ottkw{false}
\end{align*}
Thus, the size of a term being evaluated is proportional to the
argument $n$ at its longest, whereas unchecked semantics
(without coercions) allows for tail-call optimization and constant-space
execution.  This is the space-efficiency problem of gradual typing.

\subsection{Space-Efficient Gradual Typing}
Herman et al.~\cite{DBLP:conf/sfp/HermanTF07,DBLP:journals/lisp/HermanTF10} also
presented a solution to this problem.  In the evaluation sequence of
$\ottkw{odd} \, n$ above, we could immediately ``compress'' nested
coercion applications $\ottnt{M}  \langle   \ottkw{bool} \texttt{!}   \rangle  \langle   \ottkw{bool} \texttt{?}^{ \ottnt{p} }   \rangle$ before computation of the
target term $\ottnt{M}$ ends, because $\langle   \ottkw{bool} \texttt{!}   \rangle\langle   \ottkw{bool} \texttt{?}^{ \ottnt{p} }   \rangle$---tagging
immediately followed by untagging---is equivalent to the identity function.  By doing
so, we can maintain that the order of the size of a term in the middle
of evaluation is constant.  This idea is formalized in terms of a
``space-efficient'' extension of the coercion
calculus~\citep{DBLP:journals/scp/Henglein94}.  Since then, a few
space-efficient coercion/cast calculi have been
proposed~\cite{DBLP:conf/pldi/SiekTW15, DBLP:conf/popl/SiekW10,
  DBLP:conf/esop/SiekGT09}.

Among them,
Siek et al.~\cite{DBLP:conf/pldi/SiekTW15} have proposed a space-efficient
coercion calculus \lamS.  \lamS is equipped with a composition
function that compresses consecutive coercions in certain canonical forms.
The coercion composition is achieved as a simple recursive function thanks to
the canonical forms.  We show evaluation of
$\ottkw{odd} \, \ottsym{4}$ according to the \lamS semantics in the left of Figure~\ref{fig:odd}.\footnote{%
  Strictly speaking, $ \ottkw{bool} \texttt{!} $ and $ \ottkw{bool} \texttt{?}^{ \ottnt{p} } $ are abbreviations of
  $ \ottkw{id} _{ \ottkw{bool} }   \ottsym{;}   \ottkw{bool} \texttt{!} $ and $ \ottkw{bool} \texttt{?}^{ \ottnt{p} }   \ottsym{;}   \ottkw{id} _{ \ottkw{bool} } $, respectively, in \lamS.
}
Here, $\ottnt{s}  \fatsemi  \ottnt{t}$ is a meta-level operation that composes
two coercions $\ottnt{s},\ottnt{t}$ (in canonical forms) and yields another canonical coercion
that semantically corresponds to their sequential composition.
This composition function enables us to prevent the size of a term from growing.

\begin{figure}
\[
\begin{array}{lll}
  \multicolumn{3}{l}{\ottkw{odd} \, \ottsym{4}} \\
  & \mathbin{ \evalto ^*} & \ottsym{(}  \ottkw{even} \, \ottsym{3}  \ottsym{)}  \langle   \ottkw{bool} \texttt{?}^{ \ottnt{p} }   \rangle \\ 
  & \evalto & \ottsym{(}  \ottkw{odd} \, \ottsym{(}  \ottsym{3}  \ottsym{-}  \ottsym{1}  \ottsym{)}  \ottsym{)}  \langle   \ottkw{bool} \texttt{!}   \rangle  \langle   \ottkw{bool} \texttt{?}^{ \ottnt{p} }   \rangle \\ 
  & \evalto & \ottsym{(}  \ottkw{odd} \, \ottsym{(}  \ottsym{3}  \ottsym{-}  \ottsym{1}  \ottsym{)}  \ottsym{)}  \langle   \ottkw{bool} \texttt{!}   \fatsemi   \ottkw{bool} \texttt{?}^{ \ottnt{p} }   \rangle \\ 
  &=& \ottsym{(}  \ottkw{odd} \, \ottsym{(}  \ottsym{3}  \ottsym{-}  \ottsym{1}  \ottsym{)}  \ottsym{)}  \langle   \ottkw{id} _{ \ottkw{bool} }   \rangle \\
  & \evalto & \ottsym{(}  \ottkw{odd} \, \ottsym{2}  \ottsym{)}  \langle   \ottkw{id} _{ \ottkw{bool} }   \rangle \\ 
  & \evalto & \ottsym{(}  \ottkw{even} \, \ottsym{(}  \ottsym{2}  \ottsym{-}  \ottsym{1}  \ottsym{)}  \ottsym{)}  \langle   \ottkw{bool} \texttt{?}^{ \ottnt{p} }   \rangle  \langle   \ottkw{id} _{ \ottkw{bool} }   \rangle \\ 
  & \evalto & \ottsym{(}  \ottkw{even} \, \ottsym{(}  \ottsym{2}  \ottsym{-}  \ottsym{1}  \ottsym{)}  \ottsym{)}  \langle   \ottkw{bool} \texttt{?}^{ \ottnt{p} }   \fatsemi   \ottkw{id} _{ \ottkw{bool} }   \rangle \\ 
  &=& \ottsym{(}  \ottkw{even} \, \ottsym{(}  \ottsym{2}  \ottsym{-}  \ottsym{1}  \ottsym{)}  \ottsym{)}  \langle   \ottkw{bool} \texttt{?}^{ \ottnt{p} }   \rangle \\
  & \evalto & \ottsym{(}  \ottkw{even} \, \ottsym{1}  \ottsym{)}  \langle   \ottkw{bool} \texttt{?}^{ \ottnt{p} }   \rangle \\ 
  & \evalto &  \dots
\end{array}\quad
\begin{array}{lll}
\multicolumn{3}{l}{ \ottkw{oddk} \, ( \ottsym{4} ,  \ottkw{id} _{ \ottkw{bool} }  ) }\\
  & \evalto &  \ottkw{evenk} \, ( \ottsym{4}  \ottsym{-}  \ottsym{1} ,  \ottkw{bool} \texttt{?}^{ \ottnt{p} }   \mathbin{;\!;}   \ottkw{id} _{ \ottkw{bool} }  )  \\ 
  & \evalto &  \ottkw{evenk} \, ( \ottsym{4}  \ottsym{-}  \ottsym{1} ,  \ottkw{bool} \texttt{?}^{ \ottnt{p} }  )  \\ 
  & \evalto &  \ottkw{evenk} \, ( \ottsym{3} ,  \ottkw{bool} \texttt{?}^{ \ottnt{p} }  )  \\ 
  & \evalto &  \ottkw{oddk} \, ( \ottsym{3}  \ottsym{-}  \ottsym{1} ,  \ottkw{bool} \texttt{!}   \mathbin{;\!;}   \ottkw{bool} \texttt{?}^{ \ottnt{p} }  )  \\ 
  & \evalto &  \ottkw{oddk} \, ( \ottsym{3}  \ottsym{-}  \ottsym{1} ,  \ottkw{id} _{ \ottkw{bool} }  )  \\ 
  & \evalto &  \ottkw{oddk} \, ( \ottsym{2} ,  \ottkw{id} _{ \ottkw{bool} }  )  \\ 
  & \evalto &  \ottkw{evenk} \, ( \ottsym{2}  \ottsym{-}  \ottsym{1} ,  \ottkw{bool} \texttt{?}^{ \ottnt{p} }   \mathbin{;\!;}   \ottkw{id} _{ \ottkw{bool} }  )  \\ 
  & \evalto &  \ottkw{evenk} \, ( \ottsym{2}  \ottsym{-}  \ottsym{1} ,  \ottkw{bool} \texttt{?}^{ \ottnt{p} }  )  \\ 
  & \evalto &  \ottkw{evenk} \, ( \ottsym{1} ,  \ottkw{bool} \texttt{?}^{ \ottnt{p} }  )  \\
  & \evalto & \ldots
\end{array}
\]
    \caption{Reduction from $\ottkw{odd} \, \ottsym{4}$ in \lamS (left) and reduction from $ \ottkw{odd} \, ( \ottsym{4} ,  \ottkw{id} _{ \ottkw{bool} }  ) $ in \lamSx (right).}
    \label{fig:odd}
  \end{figure}

However, in order to ensure that nested coercion applications are always merged,
the operational semantics of \lamS relies on a nonstandard reduction rule
and nonstandard evaluation contexts.
Although it does not cause any theoretical problems,
it does not seem easy to implement---in particular, its compilation method seems nontrivial.
In fact, none of the existing compiler implementations that address the space-efficiency
problem~\cite{DBLP:conf/pldi/KuhlenschmidtAS19,DBLP:journals/pacmpl/FelteyGSFS18}
solves the problem of growing coercions at tail positions (an exception
is recent work by Castagna et al.~\cite{conf/ifl/CastagnaDLS19}---See
Section~\ref{sec:related} for more comparison).

\subsection{Our Work: Coercion-Passing Style}

In this paper, we study coercion-passing style for space-efficient
gradual typing.  Just as con\-tin\-u\-a\-tion-passing style, in which ``the
rest of the computation'' is passed around as first-class functions and
every function call is at a tail position, a program in
coercion-passing style passes ``the rest of the run-time checks'' around.
Actually, the idea of
coercion-passing style has already been listed as one of the
possible implementation techniques by
Herman et al.~\cite{DBLP:conf/sfp/HermanTF07,DBLP:journals/lisp/HermanTF10} but it
has been neither well studied nor formalized.

We use the even/odd example above to describe our approach to the problem.
Here are the even/odd functions in coercion-passing style.  (We omit
type declarations for simplicity.)
\begin{lstlisting}
let rec evenk (x, $\kappa$)  =
  if x = 0 then true$\langle   \ottkw{bool} \texttt{!}   \mathbin{;\!;}  \kappa  \rangle$ else oddk (x - 1, $ \ottkw{bool} \texttt{!}   \mathbin{;\!;}  \kappa$)
and oddk (x, $\kappa$) =
  if x = 0 then false$\langle  \kappa  \rangle$ else evenk (x - 1, $ \ottkw{bool} \texttt{?}^{ \ottnt{p} }   \mathbin{;\!;}  \kappa$)
\end{lstlisting}
Additional parameters named $\kappa$ are for \emph{first-class
  coercions}, which are supposed to be applied---as in
\texttt{false}$\langle  \kappa  \rangle$---to values that are returned in the original function definition.
We often call these coercions
\emph{continuation coercions}.  Coercion applications such as
\texttt{true}$\langle   \ottkw{bool} \texttt{!}   \rangle$ and \texttt{(oddk (x - 1))$\langle   \ottkw{bool} \texttt{!}   \rangle$} at tail positions in the original program are
translated to coercion compositions such as
\texttt{true$\langle   \ottkw{bool} \texttt{!}   \mathbin{;\!;}  \kappa  \rangle$} and \texttt{oddk (x - 1,
  $ \ottkw{bool} \texttt{!}   \mathbin{;\!;}  \kappa$)}, respectively.  When $\kappa$ is bound to a
concrete coercion, it will be composed with $ \ottkw{bool} \texttt{!} $ \emph{before
  it is applied}.  Similarly to programs in CPS, function calls pass
(composed) coercions.

With these functions in coercion-passing style, the evaluation of
$ \ottkw{oddk} \, ( \ottsym{4} ,  \ottkw{id} _{ \ottkw{bool} }  ) $ (where $ \ottkw{id} _{ \ottkw{bool} } $ is an identity coercion,
which does nothing) proceeds as in the right of Figure~\ref{fig:odd}.
Since tagging followed by untagging (with the same tag) actually does nothing,
$ \ottkw{bool} \texttt{!}   \mathbin{;\!;}   \ottkw{bool} \texttt{?}^{ \ottnt{p} } $ composes to $ \ottkw{id} _{ \ottkw{bool} } $
by the (meta-level) coercion composition $ \ottkw{bool} \texttt{!}   \fatsemi   \ottkw{bool} \texttt{?}^{ \ottnt{p} } $.

Similarly to the \lamS semantics described above, coercion composition
in the argument takes place before a recursive call, thus the size of
coercions stays bounded by the constant order, overcoming the space efficiency
problem.  A nice property of our solution is that the evaluation
is standard call-by-value.

One can view the extra parameter $\kappa$ as an accumulating
parameter and continuation coercions as (delimited) continuations in
defunctionalized forms~\cite{Reynolds98HOSC}.  Unlike simple
defunctionalization, however, special composition of two defunctionalized
coercions is provided, preventing the sizes of composed coercions from
growing.

\paragraph*{Contributions}
Since the operational semantics of \lamS seems nontrivial to implement
due to a nonstandard reduction rule, we investigate implementation of
the space-efficient semantics via a translation into coercion-passing style.
Our contributions in this paper are summarized as follows:
\begin{itemize}
\item In the context of the space-efficiency problem of gradual typing,
  we develop a new calculus \lamSx of space-efficient first-class coercions.
\item We formalize a coercion-passing style translation from (a slight variant
  of) space-efficient coercion calculus \lamS~\cite{DBLP:conf/pldi/SiekTW15}
  to the new calculus \lamSx.
\item We prove correctness of the coercion-passing style translation
  via a simulation property.
\item We implement the coercion-passing style translation on top of the Grift
  compiler~\cite{DBLP:conf/pldi/KuhlenschmidtAS19}, and conduct some experiments to show that
  stack overflow is indeed avoided.
\end{itemize}

\paragraph*{Outline}
The rest of this paper is organized as follows.
We review the space-efficient coercion calculus \lamS~\cite{DBLP:conf/pldi/SiekTW15} in Section~\ref{sec:source}.
We introduce a new space-efficient coercion calculus with first-class coercions \lamSx in Section~\ref{sec:target}, 
formalize a translation into coercion-passing style as a translation from \lamS to \lamSx,
and prove correctness of the translation in Section~\ref{sec:translation}.
We discuss our implementation of coercion-passing translation on top of the Grift compiler~\cite{DBLP:conf/pldi/KuhlenschmidtAS19} and show an experimental result in Section~\ref{sec:implementation}.
Finally, we discuss related work in Section~\ref{sec:related} and conclude in Section~\ref{sec:conclusion}.
\iffull
Proofs of the stated properties can be found in Appendix~\ref{sec:appendix}.
\fi

\section{Space-Efficient Coercion Calculus}
\label{sec:source}

In this section, we review the space-efficient coercion calculus
\lamS~\citep{DBLP:conf/pldi/SiekTW15}, which is the source calculus of
our translation.  Our definition differs from the original in a few
respects, as we will explain later.  For simplicity, we do not include
(mutually) recursive functions and conditional expressions in the
formalization but it is straightforward to add them; in fact, our
implementation includes them.

Main novelties of \lamS over the original coercion calculus
\lamC~\citep{DBLP:journals/scp/Henglein94} are (1) space-efficient
coercions, which are canonical forms of coercions, whose composition
can be defined by a straightforward recursive function, and (2)
operational semantics in which a sequence of coercion applications is
collapsed eagerly---even before they are applied to a
value~\cite{DBLP:conf/sfp/HermanTF07, DBLP:journals/lisp/HermanTF10,
  DBLP:conf/esop/SiekGT09}.

Basic forms of coercions are inherited from
\lamC~\citep{DBLP:journals/scp/Henglein94}, which provides (1)
identity coercions $ \ottkw{id} _{ \ottnt{A} } $ (where $\ottnt{A}$ is a type), which do nothing; (2) injections
$ \ottnt{G} \texttt{!} $, which add a type tag $\ottnt{G}$ to a value to make a value of
the dynamic type; (3) projections $ \ottnt{G} \texttt{?}^{ \ottnt{p} } $, which test whether a
value of the dynamic type is tagged with $\ottnt{G}$, remove the tag if the test succeeds,
or raise blame labeled $\ottnt{p}$ if it fails; (4) function
coercions $\ottnt{c_{{\mathrm{1}}}}  \rightarrow  \ottnt{c_{{\mathrm{2}}}}$, which, when they are applied to a function,
coerce an argument to the function by $\ottnt{c_{{\mathrm{1}}}}$ and a value returned
from the function by $\ottnt{c_{{\mathrm{2}}}}$; and (5) sequential compositions
$\ottnt{c_{{\mathrm{1}}}}  \ottsym{;}  \ottnt{c_{{\mathrm{2}}}}$, which apply $\ottnt{c_{{\mathrm{1}}}}$ and $\ottnt{c_{{\mathrm{2}}}}$ in this order.
Space-efficient coercions restrict the way basic coercions are
combined by sequential composition; they can be roughly expressed by the
following regular expression:
\[
( \ottnt{G} \texttt{?}^{ \ottnt{p} } ;)^? (  \ottkw{id} _{ \iota }  + (\ottnt{s_{{\mathrm{1}}}}  \rightarrow  \ottnt{s_{{\mathrm{2}}}})) (; \ottnt{G'} \texttt{!} )^?
\]
(where $\iota$ is a base type, $\ottnt{s_{{\mathrm{1}}}}$ and $\ottnt{s_{{\mathrm{2}}}}$ stand for
space efficient coercions, $( \cdots )^?$ stands for an optional
element, and $+$ for alternatives).  As already mentioned, an
advantage of this form is that (meta-level) sequential composition
(denoted by $\ottnt{s_{{\mathrm{1}}}}  \fatsemi  \ottnt{s_{{\mathrm{2}}}}$) of two space-efficient coercions results in
another space-efficient coercion (if the composition is well typed), in other words,
space-efficient coercions are closed under $\ottnt{s_{{\mathrm{1}}}}  \fatsemi  \ottnt{s_{{\mathrm{2}}}}$.
For example, the composition
\(
(( \ottnt{G_{{\mathrm{1}}}} \texttt{?}^{ \ottnt{p} } ;)^? (  \ottkw{id} _{ \iota }  + (\ottnt{s_{{\mathrm{1}}}}  \rightarrow  \ottnt{s_{{\mathrm{2}}}})) ; \ottnt{G_{{\mathrm{2}}}} \texttt{!} )  \fatsemi 
( \ottnt{G_{{\mathrm{3}}}} \texttt{?}^{ \ottnt{p'} } ; (  \ottkw{id} _{ \iota }  + (\ottnt{s_{{\mathrm{3}}}}  \rightarrow  \ottnt{s_{{\mathrm{4}}}})) (; \ottnt{G_{{\mathrm{4}}}} \texttt{!} )^?)
\)
will be
\(
(( \ottnt{G_{{\mathrm{1}}}} \texttt{?}^{ \ottnt{p} } ;)^? (  \ottkw{id} _{ \iota }  + (\ottsym{(}  \ottnt{s_{{\mathrm{3}}}}  \fatsemi  \ottnt{s_{{\mathrm{1}}}}  \ottsym{)}  \rightarrow  \ottsym{(}  \ottnt{s_{{\mathrm{2}}}}  \fatsemi  \ottnt{s_{{\mathrm{4}}}}  \ottsym{)})) (; \ottnt{G_{{\mathrm{4}}}} \texttt{!} )^?)
\)
if $\ottnt{G_{{\mathrm{2}}}} = \ottnt{G_{{\mathrm{3}}}}$---that is, tagging with $\ottnt{G_{{\mathrm{2}}}}$ is immediately
followed by inspection whether $\ottnt{G_{{\mathrm{2}}}}$ is present.\footnote{%
  Here, we exclude ill-typed coercion compositions
  such as $\ottsym{(}  \ottnt{s_{{\mathrm{1}}}}  \rightarrow  \ottnt{s_{{\mathrm{2}}}}  \ottsym{)}  \fatsemi   \ottkw{id} _{ \iota } $.}
Notice that the
resulting coercion conforms to the regular expression again.  (The
other case where $\ottnt{G_{{\mathrm{2}}}} \neq \ottnt{G_{{\mathrm{3}}}}$ means that the projection $ \ottnt{G_{{\mathrm{3}}}} \texttt{?}^{ \ottnt{p'} } $
will fail; we will explain such failures later.)

The operational semantics includes the reduction rule
\(
 \mathcal{F}  [  \ottnt{M}  \langle  \ottnt{s}  \rangle  \langle  \ottnt{t}  \rangle  ]   \reduces   \mathcal{F}  [  \ottnt{M}  \langle  \ottnt{s}  \fatsemi  \ottnt{t}  \rangle  ] 
\)
where $\mathcal{F}$ is an evaluation context that does not include nested
coercion applications and whose innermost frame is not a coercion
application.  This rule intuitively means that two consecutive
coercions at the outermost position will be composed \emph{even before
  $\ottnt{M}$ is evaluated to a value.}  This eager composition avoids a
long chain of coercion applications in an evaluation context.

\subsection{Syntax}

\begin{figure}[tb]\small
  \begin{align*}
    \omit\rlap{\(
      \text{Variables} \hgap \ottmv{x}, \ottmv{y} \hgap
      \text{Constants} \hgap \ottnt{a},\ottnt{b} \hgap
      \text{Operators} \hgap \ottnt{op} \hgap
      \text{Blame labels} \hgap \ottnt{p}
    \)} \\
    \text{Base types} &&
    \iota &\grmeq
    \graytext{
      \ottkw{int} \grmor
      \ottkw{bool} \grmor
      \dots
    }
    \\
    \text{Types} &&
    \ottnt{A},\ottnt{B},\ottnt{C} &\grmeq
     \mathord{\star}  \grmor
    \graytext{
      \iota \grmor
      \ottnt{A}  \rightarrow  \ottnt{B}
    }
    \\
    \text{Ground types} &&
    \ottnt{G},\ottnt{H} &\grmeq
    \iota \grmor
    \mathord{\star}  \rightarrow  \mathord{\star}
    \\
    \text{Space-efficient coercions} &&
    \ottnt{s}, \ottnt{t} &\grmeq
     \ottkw{id} _{ \mathord{\star} }  \grmor
     \ottnt{G} \texttt{?}^{ \ottnt{p} }   \ottsym{;}  \ottnt{i} \grmor
    \ottnt{i}
    \\
    \text{Intermediate coercions} &&
    \ottnt{i} &\grmeq
    \ottnt{g}  \ottsym{;}   \ottnt{G} \texttt{!}  \grmor
    \ottnt{g} \grmor
     \bot^{ \ottnt{G}   \ottnt{p}   \ottnt{H} } 
    \\
    \text{Ground coercions} &&
    \ottnt{g}, \ottnt{h} &\grmeq
     \ottkw{id} _{ \ottnt{A} }  \graytext{\text{ (if $\ottnt{A}  \ne  \mathord{\star}$)}} \grmor
    \ottnt{s}  \rightarrow  \ottnt{t} \graytext{\text{ (if $\ottnt{s}  \ne  \ottkw{id}$ or $\ottnt{t}  \ne  \ottkw{id}$)}}
    \\
    \text{Delayed coercions} &&
    \ottnt{d} &\grmeq
    \ottnt{g}  \ottsym{;}   \ottnt{G} \texttt{!}  \grmor
    \ottnt{s}  \rightarrow  \ottnt{t} \graytext{\text{ (if $\ottnt{s}  \ne  \ottkw{id}$ or $\ottnt{t}  \ne  \ottkw{id}$)}}
    \\
    \text{Terms}&&
    \ottnt{L},\ottnt{M},\ottnt{N} &\grmeq
    \graytext{
      \ottnt{V} \grmor
      \ottnt{op}  \ottsym{(}  \ottnt{M}  \ottsym{,}  \ottnt{N}  \ottsym{)} \grmor
      \ottnt{M} \, \ottnt{N}
    } \grmor
    \ottnt{M}  \langle  \ottnt{s}  \rangle \grmor
    \ottkw{blame} \, \ottnt{p}
    \\
    \text{Values}&&
    \ottnt{V}, \ottnt{W} &\grmeq
    \graytext{
      \ottmv{x} \grmor
      \ottnt{U}
    } \grmor
    \ottnt{U}  \langle\!\langle  \ottnt{d}  \rangle\!\rangle
    \\
    \text{Uncoerced values} &&
    \ottnt{U} &\grmeq
    \graytext{
      \ottnt{a} \grmor
       \lambda   \ottmv{x} .\,  \ottnt{M} 
    }
    \\
    \text{Type environments} &&
    \Gamma &\grmeq
    \graytext{
       \emptyset  \grmor
      \Gamma  \ottsym{,}  \ottmv{x}  \ottsym{:}  \ottnt{A}
    }
  \end{align*}
  \caption{Syntax of \lamS.}
  \label{fig:syntaxS}
\end{figure}

We show the syntax of \lamS in Figure~\ref{fig:syntaxS}.
The syntax of \lamS extends that of the simply typed lambda calculus
(written in \graytext{gray}) with the dynamic type and (space-efficient) coercions.

\emph{Types}, ranged over by $\ottnt{A},\ottnt{B},\ottnt{C}$, include the dynamic type $ \mathord{\star} $,
base types $\iota$, and function types $\ottnt{A}  \rightarrow  \ottnt{B}$.
Base types $\iota$ include $\ottkw{int}$ (integer type)
and $\ottkw{bool}$ (Boolean type) and so on.
\emph{Ground types}, ranged over by $\ottnt{G},\ottnt{H}$, include base types $\iota$
and the function type $\mathord{\star}  \rightarrow  \mathord{\star}$. They are used for type tags put on values
of the dynamic type~\citep{DBLP:conf/esop/WadlerF09}.
Here, the ground type for functions is always $\mathord{\star}  \rightarrow  \mathord{\star}$, reflecting the fact
that many dynamically typed languages do not include information on
the argument and return types of the function in its type tag.

As we have already discussed, \lamS restricts coercions to only
canonical ones, namely space-efficient coercions $\ottnt{s}$, whose
grammar is defined via ground coercions \(\ottnt{g}\) and intermediate
coercions \(\ottnt{i}\).  Ground coercions correspond to the middle part
of space-efficient coercions; unlike the original \lamS, ground
coercions include identity coercions for any function types---such as
$ \ottkw{id} _{ \iota  \rightarrow  \iota } $---and exclude ``virtually identity'' coercions
such as \( \ottkw{id} _{ \iota }   \rightarrow   \ottkw{id} _{ \iota } \).  Although these two coercions are
extensionally the same, they reduce in slightly different ways:
applying $ \ottkw{id} _{ \iota  \rightarrow  \iota } $ to a function immediately returns the
function, whereas applying $ \ottkw{id} _{ \iota }   \rightarrow   \ottkw{id} _{ \iota } $ results in a
wrapped function whose argument and return values are monitored by
$ \ottkw{id} _{ \iota } $, which does nothing.  Adopting $ \ottkw{id} _{ \ottnt{A} } $ for any $\ottnt{A}$
simplifies our proof that the coercion-passing translation preserves the
semantics.  An intermediate coercion adds an optional injection to
a ground coercion.  Coercions of the form $ \bot^{ \ottnt{G}   \ottnt{p}   \ottnt{H} } $
trigger blame (labeled $\ottnt{p}$) if applied to a value.  They emerge
from coercion composition
\[
(( \ottnt{G_{{\mathrm{1}}}} \texttt{?}^{ \ottnt{p} } ;)^? (  \ottkw{id} _{ \ottnt{A} }  + (\ottnt{s_{{\mathrm{1}}}}  \rightarrow  \ottnt{s_{{\mathrm{2}}}})) ; \ottnt{G_{{\mathrm{2}}}} \texttt{!} )  \fatsemi 
( \ottnt{G_{{\mathrm{3}}}} \texttt{?}^{ \ottnt{p'} } ; (  \ottkw{id} _{ \ottnt{A} }  + (\ottnt{s_{{\mathrm{3}}}}  \rightarrow  \ottnt{s_{{\mathrm{4}}}})) (; \ottnt{G_{{\mathrm{4}}}} \texttt{!} )^?)
\]
where $\ottnt{A}  \ne  \mathord{\star}$ and $\ottnt{G_{{\mathrm{2}}}} \neq \ottnt{G_{{\mathrm{3}}}}$, which means that the projection $ \ottnt{G_{{\mathrm{3}}}} \texttt{?}^{ \ottnt{p'} } $ is
bound to fail.  The composition results in
$( \ottnt{G_{{\mathrm{1}}}} \texttt{?}^{ \ottnt{p} } ;)^?  \bot^{ \ottnt{G_{{\mathrm{1}}}}   \ottnt{p'}   \ottnt{G_{{\mathrm{3}}}} } $, which means that, unless the
optional projection fails---blaming $\ottnt{p}$---it fails with $\ottnt{p'}$.
Finally, space-efficient coercions are obtained by adding optional
projection to intermediate coercions.  $ \ottkw{id} _{ \mathord{\star} } $ is a special
coercion that does not conform to the regular expression above.
Strictly speaking, an injection, say $ \ottkw{int} \texttt{!} $, has to be written
$ \ottkw{id} _{ \ottkw{int} }   \ottsym{;}   \ottkw{int} \texttt{!} $ and a projection, say $ \ottkw{int} \texttt{?}^{ \ottnt{p} } $, has to be
written $ \ottkw{int} \texttt{?}^{ \ottnt{p} }   \ottsym{;}   \ottkw{id} _{ \ottkw{int} } $.  We often omit these
identity coercions in examples.

\emph{Terms}, ranged over by $\ottnt{L},\ottnt{M},\ottnt{N}$, include values $\ottnt{V}$,
primitive binary operations $\ottnt{op}  \ottsym{(}  \ottnt{M}  \ottsym{,}  \ottnt{N}  \ottsym{)}$, function applications $\ottnt{M} \, \ottnt{N}$,
coercion applications $\ottnt{M}  \langle  \ottnt{s}  \rangle$, and coercion failure $\ottkw{blame} \, \ottnt{p}$.
The term $\ottnt{M}  \langle  \ottnt{s}  \rangle$ coerces the value of $\ottnt{M}$ with coercion $\ottnt{s}$ at run time.
The term $\ottkw{blame} \, \ottnt{p}$ denotes a run-time type error caused
by the failure of a coercion (projection) with blame label $\ottnt{p}$.

\emph{Values}, ranged over by $\ottnt{V},\ottnt{W}$, include variables
$\ottmv{x}$, uncoerced values $\ottnt{U}$, and coerced values $\ottnt{U}  \langle\!\langle  \ottnt{d}  \rangle\!\rangle$.
Uncoerced values, ranged over by $\ottnt{U}$, include constants $\ottnt{a}$ of base types
and lambda abstractions $ \lambda   \ottmv{x} .\,  \ottnt{M} $.  Unlike \lamC, where values can
involve nested coercion applications, there is at most one coercion in
a value---nested coercions will be composed.
Coerced values $\ottnt{U}  \langle\!\langle  \ottnt{d}  \rangle\!\rangle$ have two forms: injected values
$\ottnt{U}  \langle\!\langle  \ottnt{g}  \ottsym{;}   \ottnt{G} \texttt{!}   \rangle\!\rangle$ and wrapped functions $\ottnt{U}  \langle\!\langle  \ottnt{s}  \rightarrow  \ottnt{t}  \rangle\!\rangle$.  The check of
function coercion is delayed until wrapped functions are applied to a
value~\citep{DBLP:journals/scp/Henglein94, DBLP:conf/icfp/FindlerF02,
  conf/scheme/SiekT06}.
We include variables as values for technical convenience in defining translations;
for operational semantics, though, it is not necessary to do so because we consider
evaluation of closed terms.

Unlike many other studies on coercion and blame calculi, we
syntactically distinguish coerced values $\ottnt{U}  \langle\!\langle  \ottnt{d}  \rangle\!\rangle$ from $\ottnt{U}  \langle  \ottnt{d}  \rangle$
(similarly to Wadler and Findler~\cite{DBLP:conf/esop/WadlerF09}).
This distinction plays an important role in our correctness proof;
roughly speaking, without the distinction, $\ottnt{U}  \langle  \ottnt{d}  \rangle  \langle  \ottnt{t}  \rangle$ would allow
two different interpretations: an application of $\ottnt{t}$ to a value $\ottnt{U}  \langle  \ottnt{d}  \rangle$
or two applications of $\ottnt{d}$ and $\ottnt{t}$ to a value $\ottnt{U}$,
which would result in different translation results.
We also note that variables $\ottmv{x}$ are considered values, rather than
uncoerced values, since they can
be bound to coerced values at function calls.  In other words,
we ensure that values are closed under value substitution.

As usual, applications are left-associative and \(\lambda\) extends
as far to the right as possible.  We do not commit to
a particular choice of precedence between function applications and
coercion applications; we will always use parentheses to
disambiguate terms like $\ottnt{M} \, \ottnt{N}  \langle  \ottnt{t}  \rangle$.
The term $ \lambda   \ottmv{x} .\,  \ottnt{M} $ binds $\ottmv{x}$ in $\ottnt{M}$ as usual.
The definitions of free variables and $\alpha$-equivalence of terms
are standard, and thus we omit them.
We identify $\alpha$-equivalent terms.

The metavariable $\Gamma$ ranges over \emph{type environments}.
A type environment is a sequence of pairs of a variable and its type.

\subsection{Type System}

\begin{figure}[tb]\small
  \textbf{Well-formed coercions} \hfill\fbox{$\ottnt{c}  \ottsym{:}  \ottnt{A}  \rightsquigarrow  \ottnt{B}$}
  \begin{center}
    \infrule[CT-Inj]{}{
       \ottnt{G} \texttt{!}   \ottsym{:}  \ottnt{G}  \rightsquigarrow  \mathord{\star}
    } \hfill
    \infrule[CT-Proj]{}{
       \ottnt{G} \texttt{?}^{ \ottnt{p} }   \ottsym{:}  \mathord{\star}  \rightsquigarrow  \ottnt{G}
    } \hfill
    \infrule[CT-Fun]{
      \ottnt{c_{{\mathrm{1}}}}  \ottsym{:}  \ottnt{A'}  \rightsquigarrow  \ottnt{A} \andalso
      \ottnt{c_{{\mathrm{2}}}}  \ottsym{:}  \ottnt{B}  \rightsquigarrow  \ottnt{B'}
    }{
      \ottnt{c_{{\mathrm{1}}}}  \rightarrow  \ottnt{c_{{\mathrm{2}}}}  \ottsym{:}  \ottnt{A}  \rightarrow  \ottnt{B}  \rightsquigarrow  \ottnt{A'}  \rightarrow  \ottnt{B'}
    } \\\vgap
    \infrule[CT-Id]{}{
       \ottkw{id} _{ \ottnt{A} }   \ottsym{:}  \ottnt{A}  \rightsquigarrow  \ottnt{A}
    } \hfill
    \infrule[CT-Seq]{
      \ottnt{c_{{\mathrm{1}}}}  \ottsym{:}  \ottnt{A}  \rightsquigarrow  \ottnt{B} \andalso
      \ottnt{c_{{\mathrm{2}}}}  \ottsym{:}  \ottnt{B}  \rightsquigarrow  \ottnt{C}
    }{
      \ottsym{(}  \ottnt{c_{{\mathrm{1}}}}  \ottsym{;}  \ottnt{c_{{\mathrm{2}}}}  \ottsym{)}  \ottsym{:}  \ottnt{A}  \rightsquigarrow  \ottnt{C}
    } \hfill
    \infrule[CT-Fail]{
      \ottnt{A}  \ne  \mathord{\star} \andalso
      \ottnt{A}  \sim  \ottnt{G} \andalso
      \ottnt{G}  \ne  \ottnt{H}
    }{
       \bot^{ \ottnt{G}   \ottnt{p}   \ottnt{H} }   \ottsym{:}  \ottnt{A}  \rightsquigarrow  \ottnt{B}
    } 
  \end{center}
  \textbf{Term typing} \hfill\fbox{$ \Gamma    \vdash_{\mathsf{S} }    \ottnt{M}  :  \ottnt{A} $}
  \begin{center}
    \infrule[T-Const]{}{
       \Gamma   \vdash   \ottnt{a}  :   \metafun{ty} ( \ottnt{a} )  
    } \hgap
    \infrule[T-Op]{
       \metafun{ty} ( \ottnt{op} )   \ottsym{=}  \iota_{{\mathrm{1}}}  \rightarrow  \iota_{{\mathrm{2}}}  \rightarrow  \iota \andalso
       \Gamma   \vdash   \ottnt{M}  :  \iota_{{\mathrm{1}}}  \andalso
       \Gamma   \vdash   \ottnt{N}  :  \iota_{{\mathrm{2}}} 
    }{
       \Gamma   \vdash   \ottnt{op}  \ottsym{(}  \ottnt{M}  \ottsym{,}  \ottnt{N}  \ottsym{)}  :  \iota 
    } \\\vgap
    \infrule[T-Var]{
       ( \ottmv{x}  :  \ottnt{A} ) \in  \Gamma 
    }{
       \Gamma   \vdash   \ottmv{x}  :  \ottnt{A} 
    } \hfill
    \infrule[T-Abs]{
       \Gamma  \ottsym{,}  \ottmv{x}  \ottsym{:}  \ottnt{A}   \vdash   \ottnt{M}  :  \ottnt{B} 
    }{
       \Gamma   \vdash    \lambda   \ottmv{x} .\,  \ottnt{M}   :  \ottnt{A}  \rightarrow  \ottnt{B} 
    } \hfill
    \infrule[T-App]{
       \Gamma   \vdash   \ottnt{M}  :  \ottnt{A}  \rightarrow  \ottnt{B}  \andalso
       \Gamma   \vdash   \ottnt{N}  :  \ottnt{A} 
    }{
       \Gamma   \vdash   \ottnt{M} \, \ottnt{N}  :  \ottnt{B} 
    } \\\vgap
    \infrule[T-Crc]{
       \Gamma   \vdash   \ottnt{M}  :  \ottnt{A}  \andalso
      \ottnt{s}  \ottsym{:}  \ottnt{A}  \rightsquigarrow  \ottnt{B}
    }{
       \Gamma   \vdash   \ottnt{M}  \langle  \ottnt{s}  \rangle  :  \ottnt{B} 
    } \hfill
    \infrule[T-CrcV]{
        \emptyset    \vdash   \ottnt{U}  :  \ottnt{A}  \andalso
      \ottnt{d}  \ottsym{:}  \ottnt{A}  \rightsquigarrow  \ottnt{B}
    }{
        \emptyset    \vdash   \ottnt{U}  \langle\!\langle  \ottnt{d}  \rangle\!\rangle  :  \ottnt{B} 
    } \hfill
    \infrule[T-Blame]{}{
        \emptyset    \vdash   \ottkw{blame} \, \ottnt{p}  :  \ottnt{A} 
    }
  \end{center}
  \caption{Typing rules of \lamS.}
  \label{fig:typingS}
\end{figure}

We give the type system of \lamS, which consists of three judgments
for \emph{type consistency} $\ottnt{A}  \sim  \ottnt{B}$, \emph{well-formed coercions}
$\ottnt{c}  \ottsym{:}  \ottnt{A}  \rightsquigarrow  \ottnt{B}$, and \emph{typing} $ \Gamma    \vdash_{\mathsf{S} }    \ottnt{M}  :  \ottnt{A} $.
We use $c$ to denote any kind of coercions.  The inference
rules (except for $\ottnt{A}  \sim  \ottnt{B}$) are shown in Figure~\ref{fig:typingS}.
(We omit the subscript $\mathsf{S}$ on $\vdash$ in rules,
as some of them are reused for \lamSx.)

The type consistency relation $\ottnt{A}  \sim  \ottnt{B}$ is the least reflexive and
symmetric and compatible relation that contains $\ottnt{A}  \sim  \mathord{\star}$.  As this
is standard~\citep{conf/scheme/SiekT06}, we omit inference rules here.
\iffull
(We put them in Appendix~\ref{sec:appendix}.)
\fi

The relation $\ottnt{c}  \ottsym{:}  \ottnt{A}  \rightsquigarrow  \ottnt{B}$ means that coercion $\ottnt{c}$, which ranges
over all kinds of coercions, converts a value from type $\ottnt{A}$ to
type $\ottnt{B}$.  We often call $\ottnt{A}$ and $\ottnt{B}$ the source and target
types of $\ottnt{c}$, respectively.  The rule \rnp{CT-Id} is for identity
coercion $ \ottkw{id} _{ \ottnt{A} } $.  The rule \rnp{CT-Inj} is for injection $ \ottnt{G} \texttt{!} $,
which converts type $\ottnt{G}$ to type $ \mathord{\star} $.  The rule \rnp{CT-Proj}
is for projection $ \ottnt{G} \texttt{?}^{ \ottnt{p} } $, which converts type $ \mathord{\star} $ to type
$\ottnt{G}$.  The rule \rnp{CT-Fun} is for function coercion $\ottnt{c_{{\mathrm{1}}}}  \rightarrow  \ottnt{c_{{\mathrm{2}}}}$.
If its argument coercion $\ottnt{c_{{\mathrm{1}}}}$ converts type $\ottnt{A'}$ to type
$\ottnt{A}$ and its return-value coercion $\ottnt{c_{{\mathrm{2}}}}$ converts type $\ottnt{B}$
to type $\ottnt{B'}$, then function coercion $\ottnt{c_{{\mathrm{1}}}}  \rightarrow  \ottnt{c_{{\mathrm{2}}}}$ converts type
$\ottnt{A}  \rightarrow  \ottnt{B}$ to type $\ottnt{A'}  \rightarrow  \ottnt{B'}$.  In other words, function coercions
are contravariant in their argument coercions and covariant in
return-value coercions.  The rule \rnp{CT-Fail} is for failure
coercion $ \bot^{ \ottnt{G}   \ottnt{p}   \ottnt{H} } $.  Here, the source type is not necessarily
$\ottnt{G}$ but can be any nondynamic type $\ottnt{A}$ consistent with $\ottnt{G}$
because the source type of a failure coercion may change during
coercion composition.  For example, the following judgments are derivable:
\[
\begin{array}{cll}
  \ottsym{(}   \ottkw{id} _{ \ottkw{int} }   \ottsym{;}   \ottkw{int} \texttt{!}   \ottsym{)}  \rightarrow  \ottsym{(}   \ottkw{int} \texttt{?}^{ \ottnt{p} }   \ottsym{;}   \ottkw{id} _{ \ottkw{int} }   \ottsym{)} &: \mathord{\star}  \rightarrow  \mathord{\star} & \rightsquigarrow  \ottkw{int}  \rightarrow  \ottkw{int} \\
   \bot^{ \mathord{\star}  \rightarrow  \mathord{\star}   \ottnt{p}   \ottkw{int} }  &: \ottkw{int}  \rightarrow  \ottkw{bool} & \rightsquigarrow  \ottkw{int}
\end{array}
\]

Proposition~\ref{prop:src-tgt} below, which is about the source and
target types of intermediate coercions and ground coercions, is useful
to understand the syntactic structure of space-efficient coercions.
In particular, it states that neither the source nor target type of
ground coercions $\ottnt{g}$ is the type $ \mathord{\star} $.

\begin{proposition}[name=Source and Target Types,restate=propSrcTgt] \label{prop:src-tgt}\leavevmode
  \begin{enumerate}
  \item If $\ottnt{i}  \ottsym{:}  \ottnt{A}  \rightsquigarrow  \ottnt{B}$ then $\ottnt{A}  \ne  \mathord{\star}$.
  \item If $\ottnt{g}  \ottsym{:}  \ottnt{A}  \rightsquigarrow  \ottnt{B}$, then $\ottnt{A}  \ne  \mathord{\star}$ and $\ottnt{B}  \ne  \mathord{\star}$
    \ifluxuryspace
    and there exists a unique $\ottnt{G}$ such that $\ottnt{A}  \sim  \ottnt{G}$ and $\ottnt{G}  \sim  \ottnt{B}$.
    \else
    and $\ottnt{A}  \sim  \ottnt{G}$ and $\ottnt{G}  \sim  \ottnt{B}$ for some unique $\ottnt{G}$.
    \fi
  \end{enumerate}
\end{proposition}

The judgment $ \Gamma    \vdash_{\mathsf{S} }    \ottnt{M}  :  \ottnt{A} $ means that
the \lamS-term $\ottnt{M}$ is given type $\ottnt{A}$ under type environment $\Gamma$.
When clear from the context, we sometimes write $ \vdash $ for $ \vdash_{\mathsf{S} } $
with the subscript $\mathsf{S}$ omitted.
We adopt similar conventions for other relations (such as $ \evalto_{\mathsf{S} } $)
introduced later.

The rules \rnp{T-Const}, \rnp{T-Op}, \rnp{T-Var}, \rnp{T-Abs}, and \rnp{T-App} are standard.
Here, $ \metafun{ty} ( \ottnt{a} ) $ maps constant $\ottnt{a}$ to a base type $\iota$, and
$ \metafun{ty} ( \ottnt{op} ) $ maps binary operator $\ottnt{op}$ to a (first-order) function type
$\iota_{{\mathrm{1}}}  \rightarrow  \iota_{{\mathrm{2}}}  \rightarrow  \iota$.
The rule \rnp{T-Crc} states that if $\ottnt{M}$ is given type $\ottnt{A}$ and
space-efficient coercion $\ottnt{s}$ converts type $\ottnt{A}$ to $\ottnt{B}$, then
coercion application $\ottnt{M}  \langle  \ottnt{s}  \rangle$ is given type $\ottnt{B}$.
The rule \rnp{T-CrcV} is similar to \rnp{T-Crc}, but for coerced values $\ottnt{U}  \langle\!\langle  \ottnt{d}  \rangle\!\rangle$.
The rule \rnp{T-Blame} allows $\ottkw{blame} \, \ottnt{p}$ to have an arbitrary type $\ottnt{A}$.
Here, type environments are always empty $ \emptyset $ in \rnp{T-CrcV}
and \rnp{T-Blame}.  It is valid because the terms $\ottnt{U}  \langle\!\langle  \ottnt{d}  \rangle\!\rangle$ and
$\ottkw{blame} \, \ottnt{p}$ arise only during evaluation, which runs a closed term.
In other words, these terms are not written by programmers in the
surface language, and also they do not appear as the result of
coercion insertion.

\subsection{Operational Semantics}

\begin{figure}[tb]\small
  \textbf{Coercion composition} \hfill\fbox{$\ottnt{s}  \fatsemi  \ottnt{t}  \ottsym{=}  \ottnt{s'}$}
  \begin{align*}
     \ottkw{id} _{ \mathord{\star} }   \fatsemi  \ottnt{t} &= \ottnt{t} & \quad \rn{CC-IdDynL} &&
    \ottsym{(}   \ottnt{G} \texttt{?}^{ \ottnt{p} }   \ottsym{;}  \ottnt{i}  \ottsym{)}  \fatsemi  \ottnt{t} &=  \ottnt{G} \texttt{?}^{ \ottnt{p} }   \ottsym{;}  \ottsym{(}  \ottnt{i}  \fatsemi  \ottnt{t}  \ottsym{)} &\rn{CC-ProjL}\\[0.5em]
    \ottsym{(}  \ottnt{g}  \ottsym{;}   \ottnt{G} \texttt{!}   \ottsym{)}  \fatsemi   \ottkw{id} _{ \mathord{\star} }  &= \ottnt{g}  \ottsym{;}   \ottnt{G} \texttt{!}  &\rn{CC-InjId} &&
    \ottsym{(}  \ottnt{g}  \ottsym{;}   \ottnt{G} \texttt{!}   \ottsym{)}  \fatsemi  \ottsym{(}   \ottnt{G} \texttt{?}^{ \ottnt{p} }   \ottsym{;}  \ottnt{i}  \ottsym{)} &= \ottnt{g}  \fatsemi  \ottnt{i} &\rn{CC-Collapse}\\
     \bot^{ \ottnt{G}   \ottnt{p}   \ottnt{H} }   \fatsemi  \ottnt{s} &=  \bot^{ \ottnt{G}   \ottnt{p}   \ottnt{H} }  &\rn{CC-FailL}&&
    \ottsym{(}  \ottnt{g}  \ottsym{;}   \ottnt{G} \texttt{!}   \ottsym{)}  \fatsemi  \ottsym{(}   \ottnt{H} \texttt{?}^{ \ottnt{p} }   \ottsym{;}  \ottnt{i}  \ottsym{)} &=  \bot^{ \ottnt{G}   \ottnt{p}   \ottnt{H} }  &\rn{CC-Conflict} \\
    &&&&&(\text{if $\ottnt{G}  \ne  \ottnt{H}$})\\[0.5em]
    \ottnt{g}  \fatsemi   \bot^{ \ottnt{G}   \ottnt{p}   \ottnt{H} }  &=  \bot^{ \ottnt{G}   \ottnt{p}   \ottnt{H} }  &\rn{CC-FailR}&&
    \ottnt{g}  \fatsemi  \ottsym{(}  \ottnt{h}  \ottsym{;}   \ottnt{H} \texttt{!}   \ottsym{)} &= \ottsym{(}  \ottnt{g}  \fatsemi  \ottnt{h}  \ottsym{)}  \ottsym{;}   \ottnt{H} \texttt{!}  &\rn{CC-InjR}\\[0.5em]
     \ottkw{id} _{ \ottnt{A} }   \fatsemi  \ottnt{g} & = \mathrlap{\ottnt{g} \quad(\text{if $\ottnt{A}  \ne  \mathord{\star}$})} & \rn{CC-IdL}&&
    \ottnt{g}  \fatsemi   \ottkw{id} _{ \ottnt{A} }  = \ottnt{g}& \mathrlap{\quad (\text{if $\ottnt{A}  \ne  \mathord{\star}$, $\ottnt{g}  \ne   \ottkw{id} _{ \ottnt{A} } $})}& \rn{CC-IdR}\\
    \ottsym{(}  \ottnt{s}  \rightarrow  \ottnt{t}  \ottsym{)}  \fatsemi  \ottsym{(}  \ottnt{s'}  \rightarrow  \ottnt{t'}  \ottsym{)} & \mathrlap{= \begin{cases}
       \ottkw{id} _{ \ottnt{A}  \rightarrow  \ottnt{B} }  &\text{if $\ottnt{s'}  \fatsemi  \ottnt{s}  \ottsym{=}   \ottkw{id} _{ \ottnt{A} } $ and $\ottnt{t}  \fatsemi  \ottnt{t'}  \ottsym{=}   \ottkw{id} _{ \ottnt{B} } $} \\
      \ottsym{(}  \ottnt{s'}  \fatsemi  \ottnt{s}  \ottsym{)}  \rightarrow  \ottsym{(}  \ottnt{t}  \fatsemi  \ottnt{t'}  \ottsym{)} &\text{otherwise}
    \end{cases}} &&&&& \rn{CC-Fun}
  \end{align*}
  \caption{Coercion composition rules of \lamS.}
  \label{fig:semS-cmp}
\end{figure}

\begin{figure}[tb]\small
  \textbf{Evaluation contexts}
  \begin{align*}
    \mathcal{E} &\grmeq
    \mathcal{F} \grmor
    \mathcal{F}  [  \square \, \langle  \ottnt{s}  \rangle  ]
    &
    \mathcal{F} &\grmeq
    \graytext{
       \square  \grmor
      \mathcal{E}  [  \ottnt{op}  \ottsym{(}  \square  \ottsym{,}  \ottnt{M}  \ottsym{)}  ] \grmor
      \mathcal{E}  [  \ottnt{op}  \ottsym{(}  \ottnt{V}  \ottsym{,} \, \square \, \ottsym{)}  ] \grmor
      \mathcal{E}  [  \square \, \ottnt{M}  ] \grmor
      \mathcal{E}  [  \ottnt{V} \, \square  ]
    }
  \end{align*}
  \textbf{Reduction} \hfill\fbox{$ \ottnt{M}    \mathbin{\accentset{\mathsf{e} }{\reduces}_{\mathsf{S} } }    \ottnt{N} $}\quad\fbox{$ \ottnt{M}    \mathbin{\accentset{\mathsf{c} }{\reduces}_{\mathsf{S} } }    \ottnt{N} $}
  \begin{align*}
    \ottnt{op}  \ottsym{(}  \ottnt{a}  \ottsym{,}  \ottnt{b}  \ottsym{)} & \mathbin{\accentset{\mathsf{e} }{\reduces} }  \delta \, \ottsym{(}  \ottnt{op}  \ottsym{,}  \ottnt{a}  \ottsym{,}  \ottnt{b}  \ottsym{)} &\rn{R-Op} &&
    \ottnt{U}  \langle   \ottkw{id} _{ \ottnt{A} }   \rangle & \mathbin{\accentset{\mathsf{c} }{\reduces} }  \ottnt{U} &\rn{R-Id} \\
    \ottsym{(}   \lambda   \ottmv{x} .\,  \ottnt{M}   \ottsym{)} \, \ottnt{V} & \mathbin{\accentset{\mathsf{e} }{\reduces} }  \ottnt{M}  [  \ottmv{x}  \ottsym{:=}  \ottnt{V}  ] &\rn{R-Beta}&&
    \ottnt{U}  \langle   \bot^{ \ottnt{G}   \ottnt{p}   \ottnt{H} }   \rangle & \mathbin{\accentset{\mathsf{c} }{\reduces} }  \ottkw{blame} \, \ottnt{p} &\rn{R-Fail}\\
    \ottsym{(}  \ottnt{U}  \langle\!\langle  \ottnt{s}  \rightarrow  \ottnt{t}  \rangle\!\rangle  \ottsym{)} \, \ottnt{V} & \mathbin{\accentset{\mathsf{e} }{\reduces} }  \ottsym{(}  \ottnt{U} \, \ottsym{(}  \ottnt{V}  \langle  \ottnt{s}  \rangle  \ottsym{)}  \ottsym{)}  \langle  \ottnt{t}  \rangle &\rn{R-Wrap} &&
    \ottnt{U}  \langle  \ottnt{d}  \rangle & \mathbin{\accentset{\mathsf{c} }{\reduces} }  \ottnt{U}  \langle\!\langle  \ottnt{d}  \rangle\!\rangle &\rn{R-Crc} \\ &&&&
    \ottnt{M}  \langle  \ottnt{s}  \rangle  \langle  \ottnt{t}  \rangle & \mathbin{\accentset{\mathsf{c} }{\reduces} }  \ottnt{M}  \langle  \ottnt{s}  \fatsemi  \ottnt{t}  \rangle &\rn{R-MergeC}\\ &&&&
    \ottnt{U}  \langle\!\langle  \ottnt{d}  \rangle\!\rangle  \langle  \ottnt{t}  \rangle & \mathbin{\accentset{\mathsf{c} }{\reduces} }  \ottnt{U}  \langle  \ottnt{d}  \fatsemi  \ottnt{t}  \rangle &\rn{R-MergeV}
  \end{align*}
  \textbf{Evaluation} \hfill\fbox{$ \ottnt{M}    \mathbin{  \accentset{\mathsf{e} }{\evalto}_{\mathsf{S_1} }  }    \ottnt{N} $}\quad\fbox{$ \ottnt{M}    \mathbin{  \accentset{\mathsf{c} }{\evalto}_{\mathsf{S_1} }  }    \ottnt{N} $}
  \begin{center}
    \infrule[E-CtxE]{
       \ottnt{M}    \mathbin{\accentset{\mathsf{e} }{\reduces} }    \ottnt{N} 
    }{
       \mathcal{E}  [  \ottnt{M}  ]    \mathbin{  \accentset{\mathsf{e} }{\evalto}  }    \mathcal{E}  [  \ottnt{N}  ] 
    } \hgap
    \infrule[E-CtxC]{
       \ottnt{M}    \mathbin{\accentset{\mathsf{c} }{\reduces} }    \ottnt{N} 
    }{
       \mathcal{F}  [  \ottnt{M}  ]    \mathbin{  \accentset{\mathsf{c} }{\evalto}  }    \mathcal{F}  [  \ottnt{N}  ] 
    } \hgap
    \infrule[E-Abort]{
      \mathcal{E}  \ne  \square
    }{
       \mathcal{E}  [  \ottkw{blame} \, \ottnt{p}  ]    \mathbin{  \accentset{\mathsf{e} }{\evalto}  }    \ottkw{blame} \, \ottnt{p} 
    }
  \end{center}
  \caption{Reduction/evaluation rules of \lamS.}
  \label{fig:semS}
\end{figure}

\subsubsection{Coercion Composition}

The coercion composition $\ottnt{s}  \fatsemi  \ottnt{t}$ is a recursive function that takes
two space-efficient coercions and computes another space-efficient
coercion corresponding to their sequential composition.  We show the
coercion composition rules in Figure~\ref{fig:semS-cmp}.  The function is
defined in such a way that the form of the first coercion determines
which rule to apply.

The rules \rnp{CC-IdDynL} and \rnp{CC-ProjL} are applied if the
first coercion is not an intermediate coercion.
The rules \rnp{CC-InjId}, \rnp{CC-Collapse}, \rnp{CC-Conflict}, and \rnp{CC-FailL}
are applied if the first one is a (nonground) intermediate coercion,
in which case another intermediate coercion is yielded.
The rules \rnp{CC-Collapse} and \rnp{CC-Conflict} deal with
cases where an injection and a projection meet and perform tag checks.
If type tags do not match, a failure coercion arises.

Failure coercions are necessary for eager coercion composition
to preserve the behavior of \lamC.
The term $\ottnt{M}  \langle   \ottnt{G} \texttt{!}   \rangle  \langle   \ottnt{H} \texttt{?}^{ \ottnt{p} }   \rangle$ (if $\ottnt{G}  \ne  \ottnt{H}$) in \lamC
evaluates to $\ottkw{blame} \, \ottnt{p}$---\emph{only after $\ottnt{M}$ evaluates to a value.}
By contrast,
the two coercions $ \ottnt{G} \texttt{!} $ and $ \ottnt{H} \texttt{?}^{ \ottnt{p} } $ in the term $\ottnt{M}  \langle   \ottkw{id} _{ \ottnt{G} }   \ottsym{;}   \ottnt{G} \texttt{!}   \rangle  \langle   \ottnt{H} \texttt{?}^{ \ottnt{p} }   \ottsym{;}   \ottkw{id} _{ \ottnt{H} }   \rangle$
are eagerly composed in \lamS.
Raising $\ottkw{blame} \, \ottnt{p}$ immediately would not match the semantics of \lamC
because $\ottnt{M}$ may evaluate to another blame or even diverge, in which case $\ottnt{p}$
is not blamed.
Thus, $ \bot^{ \ottnt{G}   \ottnt{p}   \ottnt{H} } $ must raise $\ottkw{blame} \, \ottnt{p}$
only after $\ottnt{M}$ evaluates to a value.

The rules \rnp{CC-FailR} and \rnp{CC-InjR} are applied if a ground
coercion and an intermediate coercion are composed to another
intermediate coercion.  
The rules \rnp{CC-FailL} and \rnp{CC-FailR} represent the propagation
of a failure to the context, somewhat similarly to exceptions.  The
rule \rnp{CC-InjR} represents associativity of sequential compositions
but $ \fatsemi $ is propagated to the inside.

The rules \rnp{CC-IdL}, \rnp{CC-IdR}, and \rnp{CC-Fun} are applied
if two ground coercions are composed to another ground coercion.
They are straightforward except that $ \ottkw{id} _{ \ottnt{A} }   \rightarrow   \ottkw{id} _{ \ottnt{B} } $ has to be
normalized to $ \ottkw{id} _{ \ottnt{A}  \rightarrow  \ottnt{B} } $ (\rn{CC-Fun}).

We present a few examples of coercion composition below:
\begin{align*}
  \ottsym{(}   \ottkw{id} _{ \ottkw{bool} }   \ottsym{;}   \ottkw{bool} \texttt{!}   \ottsym{)}  \fatsemi  \ottsym{(}   \ottkw{bool} \texttt{?}^{ \ottnt{p} }   \ottsym{;}   \ottkw{id} _{ \ottkw{bool} }   \ottsym{)}
  &=  \ottkw{id} _{ \ottkw{bool} }   \fatsemi   \ottkw{id} _{ \ottkw{bool} }  =  \ottkw{id} _{ \ottkw{bool} }  \\
  \ottsym{(}   \ottkw{id} _{ \mathord{\star}  \rightarrow  \mathord{\star} }   \ottsym{;}   \ottsym{(}  \mathord{\star}  \rightarrow  \mathord{\star}  \ottsym{)} \texttt{!}   \ottsym{)}  \fatsemi  \ottsym{(}   \ottkw{int} \texttt{?}^{ \ottnt{p} }   \ottsym{;}   \ottkw{id} _{ \ottkw{int} }   \ottsym{)}
  &=  \bot^{ \mathord{\star}  \rightarrow  \mathord{\star}   \ottnt{p}   \ottkw{int} }  \\
  \ottsym{(}  \ottsym{(}   \iota \texttt{?}^{ \ottnt{p} }   \ottsym{;}   \ottkw{id} _{ \iota }   \ottsym{)}  \rightarrow  \ottsym{(}   \ottkw{id} _{ \iota' }   \ottsym{;}   \iota' \texttt{!}   \ottsym{)}  \ottsym{)}  \fatsemi  \ottsym{(}  \ottsym{(}   \ottkw{id} _{ \iota }   \ottsym{;}   \iota \texttt{!}   \ottsym{)}  \rightarrow   \ottkw{id} _{ \mathord{\star} }   \ottsym{)}
  &= \ottsym{(}  \ottsym{(}   \ottkw{id} _{ \iota }   \ottsym{;}   \iota \texttt{!}   \ottsym{)}  \fatsemi  \ottsym{(}   \iota \texttt{?}^{ \ottnt{p} }   \ottsym{;}   \ottkw{id} _{ \iota }   \ottsym{)}  \ottsym{)}  \rightarrow  \ottsym{(}  \ottsym{(}   \ottkw{id} _{ \iota' }   \ottsym{;}   \iota' \texttt{!}   \ottsym{)}  \fatsemi   \ottkw{id} _{ \mathord{\star} }   \ottsym{)} \\
  &=  \ottkw{id} _{ \iota }   \rightarrow  \ottsym{(}   \ottkw{id} _{ \iota' }   \ottsym{;}   \iota' \texttt{!}   \ottsym{)}
\end{align*}
These examples involve situations where an injection meets a projection
by \rnp{CC-Collapse} or \rnp{CC-Conflict}.
The third example is by \rnp{CC-Fun}.
\begin{align*}
  \ottsym{(}   \iota \texttt{?}^{ \ottnt{p} }   \ottsym{;}   \ottkw{id} _{ \iota }   \ottsym{)}  \fatsemi  \ottsym{(}   \ottkw{id} _{ \iota }   \ottsym{;}   \iota \texttt{!}   \ottsym{)}
  &=  \iota \texttt{?}^{ \ottnt{p} }   \ottsym{;}  \ottsym{(}   \ottkw{id} _{ \iota }   \fatsemi  \ottsym{(}   \ottkw{id} _{ \iota }   \ottsym{;}   \iota \texttt{!}   \ottsym{)}  \ottsym{)}
  =  \iota \texttt{?}^{ \ottnt{p} }   \ottsym{;}  \ottsym{(}  \ottsym{(}   \ottkw{id} _{ \iota }   \fatsemi   \ottkw{id} _{ \iota }   \ottsym{)}  \ottsym{;}   \iota \texttt{!}   \ottsym{)}
  =  \iota \texttt{?}^{ \ottnt{p} }   \ottsym{;}  \ottsym{(}   \ottkw{id} _{ \iota }   \ottsym{;}   \iota \texttt{!}   \ottsym{)}
  \\
  \ottsym{(}   \ottkw{id} _{ \iota }   \ottsym{;}   \iota \texttt{!}   \ottsym{)}  \fatsemi  \ottsym{(}   \iota \texttt{?}^{ \ottnt{p} }   \ottsym{;}  \ottsym{(}   \ottkw{id} _{ \iota }   \ottsym{;}   \iota \texttt{!}   \ottsym{)}  \ottsym{)}
  &=  \ottkw{id} _{ \iota }   \fatsemi  \ottsym{(}   \ottkw{id} _{ \iota }   \ottsym{;}   \iota \texttt{!}   \ottsym{)}
  = \ottsym{(}   \ottkw{id} _{ \iota }   \fatsemi   \ottkw{id} _{ \iota }   \ottsym{)}  \ottsym{;}   \iota \texttt{!}  =  \ottkw{id} _{ \iota }   \ottsym{;}   \iota \texttt{!} 
\end{align*}
As the fourth example shows,
a projection followed by an injection does not collapse
since the projection might fail.
Such a coercion is simplified when it is preceded by another injection (the fifth example).

The following lemma states that composition is defined for two well-formed coercions with
matching target and source types.
\begin{lemma}[restate=lemCmpWelldef,name=] \label{lem:sc-cmp-welldef}
  If $\ottnt{s}  \ottsym{:}  \ottnt{A}  \rightsquigarrow  \ottnt{B}$ and $\ottnt{t}  \ottsym{:}  \ottnt{B}  \rightsquigarrow  \ottnt{C}$, then $\ottsym{(}  \ottnt{s}  \fatsemi  \ottnt{t}  \ottsym{)}  \ottsym{:}  \ottnt{A}  \rightsquigarrow  \ottnt{C}$.
\end{lemma}

\subsubsection{Evaluation} 

We give a small-step operational semantics to \lamS
consisting of two relations on closed terms:
the reduction relation $ \ottnt{M}    \reduces_{\mathsf{S} }    \ottnt{N} $ for basic computation, and
the evaluation relation $ \ottnt{M}    \mathbin{  \evalto_{\mathsf{S} }  }    \ottnt{N} $ for computing subterms and raising errors.

We show the reduction rules and the evaluation rules of \lamS in Figure~\ref{fig:semS}.
The reduction/evaluation rules are labeled either \textsf{e} or \textsf{c}.
The label \textsf{e} is for essential computation, and
the label \textsf{c} is for coercion applications.
As we see later, this distinction is important in our correctness proof.
We write $ \reduces_{\mathsf{S} } $ for $ \mathbin{\accentset{\mathsf{e} }{\reduces}_{\mathsf{S} } }  \cup  \mathbin{\accentset{\mathsf{c} }{\reduces}_{\mathsf{S} } } $, and
$ \evalto_{\mathsf{S} } $ for $ \accentset{\mathsf{e} }{\evalto}_{\mathsf{S} }  \cup  \accentset{\mathsf{c} }{\evalto}_{\mathsf{S} } $.  We sometimes call $ \accentset{\mathsf{e} }{\evalto}_{\mathsf{S} } $
and $ \accentset{\mathsf{c} }{\evalto}_{\mathsf{S} } $ e-evaluation and c-evaluation, respectively.

The rule \rnp{R-Op} applies to primitive operations.
Here, $ \delta $ is a (partial) function
that takes an operator $\ottnt{op}$ and two constants $\ottnt{a_{{\mathrm{1}}}},\ottnt{a_{{\mathrm{2}}}}$,
and returns the resulting constant of the primitive operation.
We assume that if $ \metafun{ty} ( \ottnt{op} )   \ottsym{=}  \iota_{{\mathrm{1}}}  \rightarrow  \iota_{{\mathrm{2}}}  \rightarrow  \iota$ and
$ \metafun{ty} ( \ottnt{a_{{\mathrm{1}}}} )   \ottsym{=}  \iota_{{\mathrm{1}}}$ and $ \metafun{ty} ( \ottnt{a_{{\mathrm{2}}}} )   \ottsym{=}  \iota_{{\mathrm{2}}}$, then
$\delta \, \ottsym{(}  \ottnt{op}  \ottsym{,}  \ottnt{a_{{\mathrm{1}}}}  \ottsym{,}  \ottnt{a_{{\mathrm{2}}}}  \ottsym{)}  \ottsym{=}  \ottnt{a}$ and $ \metafun{ty} ( \ottnt{a} )   \ottsym{=}  \iota$ for some constant $\ottnt{a}$.

The rule \rnp{R-Beta} performs the standard call-by-value $\beta$-reduction.
We write $\ottnt{M}  [  \ottmv{x}  \ottsym{:=}  \ottnt{V}  ]$ for capture-avoiding substitution of $\ottnt{V}$ for free occurrences of $\ottmv{x}$ in $\ottnt{M}$.
The definition of substitution is standard and thus omitted.

The rule \rnp{R-Wrap} applies to applications of wrapped function
$\ottnt{U}  \langle\!\langle  \ottnt{s}  \rightarrow  \ottnt{t}  \rangle\!\rangle$ to value $\ottnt{V}$.  In this case, we first apply
coercion $\ottnt{s}$ on the argument to $\ottnt{V}$, and get $\ottnt{V}  \langle  \ottnt{s}  \rangle$. We
next apply function $\ottnt{U}$ to $\ottnt{V}  \langle  \ottnt{s}  \rangle$, and get $\ottnt{U} \, \ottsym{(}  \ottnt{V}  \langle  \ottnt{s}  \rangle  \ottsym{)}$. We
then apply coercion $\ottnt{t}$ on the returned value, hence
$\ottsym{(}  \ottnt{U} \, \ottsym{(}  \ottnt{V}  \langle  \ottnt{s}  \rangle  \ottsym{)}  \ottsym{)}  \langle  \ottnt{t}  \rangle$.

The rule \rnp{R-Id} represents that identity coercion $ \ottkw{id} _{ \ottnt{A} } $ returns
the input value $\ottnt{U}$ as it is.
The rule \rnp{R-Fail} applies to applications of failure coercion $ \bot^{ \ottnt{G}   \ottnt{p}   \ottnt{H} } $
to uncoerced value $\ottnt{U}$, which reduces to $\ottkw{blame} \, \ottnt{p}$.
The rule \rnp{R-Crc} applies to applications $\ottnt{U}  \langle  \ottnt{d}  \rangle$
of delayed coercion $\ottnt{d}$ to uncoerced value $\ottnt{U}$,
which reduces to a coerced value $\ottnt{U}  \langle\!\langle  \ottnt{d}  \rangle\!\rangle$.

The rules \rnp{R-MergeC} and \rnp{R-MergeV} apply to
two consecutive coercion applications, and
the two coercions are merged by the composition operation.
These rules are key to space efficiency.
Thanks to \rnp{R-MergeV}, we can assume that
there is at most one coercion in a value.
Since $\ottnt{d}  \fatsemi  \ottnt{t}$ may or may not be a delayed coercion, the right-hand side
has to be $\ottnt{U}  \langle  \ottnt{d}  \fatsemi  \ottnt{t}  \rangle$, rather than $\ottnt{U}  \langle\!\langle  \ottnt{d}  \fatsemi  \ottnt{t}  \rangle\!\rangle$.
The outermost nested coercion applications are merged by \rnp{R-MergeC}.

Now, we explain \emph{evaluation contexts}, ranged over by $\mathcal{E}$,
shown in the top of Figure~\ref{fig:semS}.
Following Siek et al.~\cite{DBLP:conf/pldi/SiekTW15}, we define them in the
so-called ``inside-out''
style~\cite{DBLP:conf/lfp/FelleisenWFD88,DBLP:journals/tcs/DanvyN01}.
Evaluation contexts represent that function calls in \lamS are
call-by-value and that primitive operations and function applications
are evaluated from left to right.  The grammar is mutually recursive
with $\mathcal{F}$, which stands for evaluation contexts whose innermost
frames are not a coercion application, whereas $\mathcal{E}$ may contain a
coercion application as the innermost frame.\footnote{%
  $\mathcal{F}  [  \square \, \langle  \ottnt{s}  \rangle  ]$ (instead of $\mathcal{F}  [  \square \, \langle  \ottnt{f}  \rangle  ]$) in the definition of
  $\mathcal{E}$ fixes a problem in Siek et al.~\cite{DBLP:conf/pldi/SiekTW15} that an
  identity coercion applied to a nonvalue gets stuck (personal
  communication).  }  Careful inspection will reveal that both $\mathcal{E}$
and $\mathcal{F}$ contain no consecutive coercion applications.  As usual,
we write $\mathcal{E}  [  \ottnt{M}  ]$ for the term obtained by replacing the hole in
$\mathcal{E}$ with $\ottnt{M}$, similarly for $\mathcal{F}  [  \ottnt{M}  ]$.
(We omit their definitions.)

We present a few examples of evaluation contexts below:
\begin{align*}
  \mathcal{F}_{{\mathrm{1}}} &=  \square  &
  \mathcal{E}_{{\mathrm{1}}} &= \mathcal{F}_{{\mathrm{1}}}  [  \square \, \langle  \ottnt{s}  \rangle  ] = \square \, \langle  \ottnt{s}  \rangle \\
  \mathcal{F}_{{\mathrm{2}}} &= \mathcal{E}_{{\mathrm{1}}}  [  \ottnt{V} \, \square  ] = \ottsym{(}  \ottnt{V} \, \square \, \ottsym{)}  \langle  \ottnt{s}  \rangle &
  \mathcal{E}_{{\mathrm{2}}} &= \mathcal{F}_{{\mathrm{2}}}  [  \square \, \langle  \ottnt{t}  \rangle  ] = \ottsym{(}  \ottnt{V} \, \ottsym{(}  \square \, \langle  \ottnt{t}  \rangle  \ottsym{)}  \ottsym{)}  \langle  \ottnt{s}  \rangle \\
  \mathcal{F}_{{\mathrm{3}}} &= \mathcal{E}_{{\mathrm{2}}}  [  \square \, \ottnt{M}  ] = \ottsym{(}  \ottnt{V} \, \ottsym{(}  \ottsym{(} \, \square \, \ottnt{M}  \ottsym{)}  \langle  \ottnt{t}  \rangle  \ottsym{)}  \ottsym{)}  \langle  \ottnt{s}  \rangle
\end{align*}

We then come back to evaluation rules:
The rules \rnp{E-CtxE} and \rnp{E-CtxC} enable us to evaluate the subterm in an evaluation context.
Here, \rnp{E-CtxC} requires that
computation of coercion applications is only performed under contexts $\mathcal{F}$---otherwise,
the innermost frame may be a coercion application, in which case
\rnp{R-MergeC} has to be applied first.  For example,
$\ottnt{U}  \langle  \ottnt{d}  \rangle  \langle  \ottnt{t}  \rangle$ reduces to $\ottnt{U}  \langle  \ottnt{d}  \fatsemi  \ottnt{t}  \rangle$ rather than $\ottnt{U}  \langle\!\langle  \ottnt{d}  \rangle\!\rangle  \langle  \ottnt{t}  \rangle$.
The rule \rnp{E-Abort} halts the evaluation of a program if it raises blame.

\begin{example}\label{example:source}
  Let $\ottnt{U}$ be $ \lambda   \ottmv{x} .\,  \ottsym{(}  \ottmv{x}  \langle   \ottkw{int} \texttt{?}^{ \ottnt{p} }   \rangle  \ottsym{+}  \ottsym{2}  \ottsym{)}  \langle   \ottkw{int} \texttt{!}   \rangle $.  Term
  $\ottsym{(}  \ottsym{(}  \ottnt{U}  \langle   \ottkw{int} \texttt{!}   \rightarrow   \ottkw{int} \texttt{?}^{ \ottnt{p} }   \rangle  \ottsym{)} \, \ottsym{3}  \ottsym{)}  \langle   \ottkw{int} \texttt{!}   \rangle$ evaluates to $\ottsym{5}  \langle\!\langle   \ottkw{int} \texttt{!}   \rangle\!\rangle$ as follows:
\begin{align*}
\mathrlap{\ottsym{(}  \ottsym{(}  \ottnt{U}  \langle   \ottkw{int} \texttt{!}   \rightarrow   \ottkw{int} \texttt{?}^{ \ottnt{p} }   \rangle  \ottsym{)} \, \ottsym{3}  \ottsym{)}  \langle   \ottkw{int} \texttt{!}   \rangle} \\
\ifluxuryspace
& \evalto  \ottsym{(}  \ottsym{(}  \ottnt{U}  \langle\!\langle   \ottkw{int} \texttt{!}   \rightarrow   \ottkw{int} \texttt{?}^{ \ottnt{p} }   \rangle\!\rangle  \ottsym{)} \, \ottsym{3}  \ottsym{)}  \langle   \ottkw{int} \texttt{!}   \rangle &\text{by \rnp{R-Crc}}\\
& \evalto  \ottsym{(}  \ottnt{U} \, \ottsym{(}  \ottsym{3}  \langle   \ottkw{int} \texttt{!}   \rangle  \ottsym{)}  \ottsym{)}  \langle   \ottkw{int} \texttt{?}^{ \ottnt{p} }   \rangle  \langle   \ottkw{int} \texttt{!}   \rangle &\text{by \rnp{R-Wrap}}\\
& \evalto  \ottsym{(}  \ottnt{U} \, \ottsym{(}  \ottsym{3}  \langle   \ottkw{int} \texttt{!}   \rangle  \ottsym{)}  \ottsym{)}  \langle   \ottkw{int} \texttt{?}^{ \ottnt{p} }   \ottsym{;}  \ottkw{id}  \ottsym{;}   \ottkw{int} \texttt{!}   \rangle &\text{by \rnp{R-MergeC}}\\
& \evalto  \ottsym{(}  \ottnt{U} \, \ottsym{(}  \ottsym{3}  \langle\!\langle   \ottkw{int} \texttt{!}   \rangle\!\rangle  \ottsym{)}  \ottsym{)}  \langle   \ottkw{int} \texttt{?}^{ \ottnt{p} }   \ottsym{;}  \ottkw{id}  \ottsym{;}   \ottkw{int} \texttt{!}   \rangle &\text{by \rnp{R-Crc}}\\
& \evalto  \ottsym{(}  \ottsym{3}  \langle\!\langle   \ottkw{int} \texttt{!}   \rangle\!\rangle  \langle   \ottkw{int} \texttt{?}^{ \ottnt{p} }   \rangle  \ottsym{+}  \ottsym{2}  \ottsym{)}  \langle   \ottkw{int} \texttt{!}   \rangle  \langle   \ottkw{int} \texttt{?}^{ \ottnt{p} }   \ottsym{;}  \ottkw{id}  \ottsym{;}   \ottkw{int} \texttt{!}   \rangle &\text{by \rnp{R-Beta}}\\
& \evalto  \ottsym{(}  \ottsym{3}  \langle\!\langle   \ottkw{int} \texttt{!}   \rangle\!\rangle  \langle   \ottkw{int} \texttt{?}^{ \ottnt{p} }   \rangle  \ottsym{+}  \ottsym{2}  \ottsym{)}  \langle   \ottkw{int} \texttt{!}   \rangle &\text{by \rnp{R-MergeC}}\\
& \evalto  \ottsym{(}  \ottsym{3}  \langle  \ottkw{id}  \rangle  \ottsym{+}  \ottsym{2}  \ottsym{)}  \langle   \ottkw{int} \texttt{!}   \rangle &\text{by \rnp{R-MergeV}}\\
& \evalto  \ottsym{(}  \ottsym{3}  \ottsym{+}  \ottsym{2}  \ottsym{)}  \langle   \ottkw{int} \texttt{!}   \rangle &\text{by \rnp{R-Id}}\\
& \evalto  \ottsym{5}  \langle   \ottkw{int} \texttt{!}   \rangle &\text{by \rnp{R-Op}}\\
& \evalto  \ottsym{5}  \langle\!\langle   \ottkw{int} \texttt{!}   \rangle\!\rangle &\text{by \rnp{R-Crc}}
\else
& \mathbin{ \evalto ^*}  \ottsym{(}  \ottnt{U} \, \ottsym{(}  \ottsym{3}  \langle   \ottkw{int} \texttt{!}   \rangle  \ottsym{)}  \ottsym{)}  \langle   \ottkw{int} \texttt{?}^{ \ottnt{p} }   \rangle  \langle   \ottkw{int} \texttt{!}   \rangle &\text{by \rnp{R-Crc}, \rnp{R-Wrap}}\\
& \evalto  \ottsym{(}  \ottnt{U} \, \ottsym{(}  \ottsym{3}  \langle   \ottkw{int} \texttt{!}   \rangle  \ottsym{)}  \ottsym{)}  \langle   \ottkw{int} \texttt{?}^{ \ottnt{p} }   \ottsym{;}  \ottkw{id}  \ottsym{;}   \ottkw{int} \texttt{!}   \rangle &\text{by \rnp{R-MergeC}}\\
& \mathbin{ \evalto ^*}  \ottsym{(}  \ottsym{3}  \langle\!\langle   \ottkw{int} \texttt{!}   \rangle\!\rangle  \langle   \ottkw{int} \texttt{?}^{ \ottnt{p} }   \rangle  \ottsym{+}  \ottsym{2}  \ottsym{)}  \langle   \ottkw{int} \texttt{!}   \rangle  \langle   \ottkw{int} \texttt{?}^{ \ottnt{p} }   \ottsym{;}  \ottkw{id}  \ottsym{;}   \ottkw{int} \texttt{!}   \rangle &\text{by \rnp{R-Crc}, \rnp{R-Beta}}\\
& \mathbin{ \evalto ^*}  \ottsym{(}  \ottsym{3}  \langle  \ottkw{id}  \rangle  \ottsym{+}  \ottsym{2}  \ottsym{)}  \langle   \ottkw{int} \texttt{!}   \rangle &\text{by \rnp{R-MergeC}, \rnp{R-MergeV}}\\
& \mathbin{ \evalto ^*}  \ottsym{5}  \langle\!\langle   \ottkw{int} \texttt{!}   \rangle\!\rangle &\text{by \rnp{R-Id}, \rnp{R-Op}, \rnp{R-Crc}.}
\fi
\end{align*}
\end{example}

\subsection{Properties}

We state a few important properties of \lamS, including determinacy of
the evaluation relation and type safety via progress and
preservation~\citep{WrightFelleisenIC94}.  We write $ \mathbin{  \evalto_{\mathsf{S} }  ^*} $ for
the reflexive and transitive closure of $ \evalto_{\mathsf{S} } $, and $ \mathbin{  \evalto_{\mathsf{S} }  ^+} $
for the transitive closure of $ \evalto_{\mathsf{S} } $.  We say that \lamS-term
$\ottnt{M}$ \emph{diverges}, denoted by $\ottnt{M} \,  \mathord{\Uparrow_{\mathsf{S} } } $, if there exists
an infinite evaluation sequence from $\ottnt{M}$.

\iffull
Proofs of the stated properties are in Appendix~\ref{sec:appendix}.
\fi

\begin{lemma}[name=Determinacy,restate=lemDeterminacyS]\label{lem:determinism-S}
  If $ \ottnt{M}    \mathbin{  \evalto_{\mathsf{S} }  }    \ottnt{N} $ and $ \ottnt{M}    \mathbin{  \evalto_{\mathsf{S} }  }    \ottnt{N'} $, then $\ottnt{N}  \ottsym{=}  \ottnt{N'}$.
\end{lemma}

\ifluxuryspace
\begin{theorem}[name=Progress,restate=thmProgressS]\label{thm:progress-S}
  If $  \emptyset     \vdash_{\mathsf{S} }    \ottnt{M}  :  \ottnt{A} $, then one of the following holds.
  \begin{enumerate}
  \item $ \ottnt{M}    \mathbin{  \evalto_{\mathsf{S} }  }    \ottnt{M'} $ for some $\ottnt{M'}$.
  \item $\ottnt{M}  \ottsym{=}  \ottnt{V}$ for some $\ottnt{V}$.
  \item $\ottnt{M}  \ottsym{=}  \ottkw{blame} \, \ottnt{p}$ for some $\ottnt{p}$.
  \end{enumerate}
\end{theorem}
\else
\begin{theorem}[name=Progress,restate=thmProgressS]\label{thm:progress-S}
  If $  \emptyset     \vdash_{\mathsf{S} }    \ottnt{M}  :  \ottnt{A} $, then one of the following holds:
(1) $ \ottnt{M}    \mathbin{  \evalto_{\mathsf{S} }  }    \ottnt{M'} $ for some $\ottnt{M'}$;
(2) $\ottnt{M}  \ottsym{=}  \ottnt{V}$ for some $\ottnt{V}$; or
(3) $\ottnt{M}  \ottsym{=}  \ottkw{blame} \, \ottnt{p}$ for some $\ottnt{p}$.
\end{theorem}
\fi 

\begin{theorem}[name=Preservation,restate=thmPreservationS]\label{thm:preservation-S}
  If $  \emptyset     \vdash_{\mathsf{S} }    \ottnt{M}  :  \ottnt{A} $ and $ \ottnt{M}    \mathbin{  \evalto_{\mathsf{S} }  }    \ottnt{N} $, then
  $  \emptyset     \vdash_{\mathsf{S} }    \ottnt{N}  :  \ottnt{A} $.
\end{theorem}

\ifluxuryspace
\begin{corollary}[name=Type Safety,restate=corSafetyS]\label{cor:safety-S}
  If $  \emptyset     \vdash_{\mathsf{S} }    \ottnt{M}  :  \ottnt{A} $, then one of the following holds.
  \begin{enumerate}
  \item $ \ottnt{M}    \mathbin{  \evalto_{\mathsf{S} }  ^*}    \ottnt{V} $ and $  \emptyset     \vdash_{\mathsf{S} }    \ottnt{V}  :  \ottnt{A} $ for some $\ottnt{V}$.
  \item $ \ottnt{M}    \mathbin{  \evalto_{\mathsf{S} }  ^*}    \ottkw{blame} \, \ottnt{p} $ for some $\ottnt{p}$.
  \item $\ottnt{M} \,  \mathord{\Uparrow_{\mathsf{S} } } $.
  \end{enumerate}
\end{corollary}
\else
\begin{corollary}[name=Type Safety,restate=corSafetyS]\label{cor:safety-S}
  If $  \emptyset     \vdash_{\mathsf{S} }    \ottnt{M}  :  \ottnt{A} $, then one of the following holds:
(1) $ \ottnt{M}    \mathbin{  \evalto_{\mathsf{S} }  ^*}    \ottnt{V} $ and $  \emptyset     \vdash_{\mathsf{S} }    \ottnt{V}  :  \ottnt{A} $ for some $\ottnt{V}$;
(2) $ \ottnt{M}    \mathbin{  \evalto_{\mathsf{S} }  ^*}    \ottkw{blame} \, \ottnt{p} $ for some $\ottnt{p}$; or
(3) $\ottnt{M} \,  \mathord{\Uparrow_{\mathsf{S} } } $.
\end{corollary}
\fi 

\section{Space-Efficient First-Class Coercion Calculus}
\label{sec:target}

In this section, we introduce \lamSx, a new space-efficient coercion
calculus with first-class coercions; \lamSx serves as the target
calculus of the translation into coercion-passing style.  The design
of \lamSx is tailored to coercion-passing style and, as a result,
first-class coercions are not as general as one might expect: for
example, 
coercions for coercions are restricted to identity coercions
(e.g., $ \ottkw{id} _{ \iota  \rightsquigarrow  \iota } $).

Since coercions are first-class in \lamSx, the use of
(space-efficient) coercions $\ottnt{s}$ is not limited to coercion
applications $\ottnt{M}  \langle  \ottnt{s}  \rangle$; they can be passed to a function as an
argument, for example.  We equip \lamS with the infix (object-level)
operator $\ottnt{M}  \mathbin{;\!;}  \ottnt{N}$ to compute the composition of two coercions: if
$\ottnt{M}$ and $\ottnt{N}$ evaluate to coercions $\ottnt{s}$ and $\ottnt{t}$,
respectively, then $\ottnt{M}  \mathbin{;\!;}  \ottnt{N}$ reduces to their composition
$\ottnt{s}  \fatsemi  \ottnt{t}$, which is another space-efficient coercion.  The type of
(first-class) coercions from $\ottnt{A}$ to $\ottnt{B}$ is written
$\ottnt{A}  \rightsquigarrow  \ottnt{B}$.\footnote{%
  In \lamS, $ \rightsquigarrow $ is the symbol used in the three-place judgment
  form $\ottnt{c}  \ottsym{:}  \ottnt{A}  \rightsquigarrow  \ottnt{B}$, whereas $ \rightsquigarrow $ is also a type constructor in
  \lamSx.  }

In \lamSx, every function abstraction takes two arguments, one of
which is a parameter for a continuation coercion to be applied to the
value returned from this abstraction.  For example, $ \lambda   \ottmv{x} .\,  \ottsym{1} $ in
\lamS corresponds to $ \lambda  ( \ottmv{x} , \kappa ).\,  \ottsym{1}  \langle  \kappa  \rangle $ in \lamSx---here, $\kappa$
is a coercion parameter.  Correspondingly, a function application
takes the form $ \ottnt{M} \, ( \ottnt{N} , \ottnt{L} ) $, which calls function $\ottnt{M}$ with an
argument pair $(\ottnt{N}, \ottnt{L})$, in which $\ottnt{L}$ is a coercion
argument, which is applied to the value returned from $\ottnt{M}$.  For
example, $\ottsym{(}  \ottnt{f} \, \ottsym{3}  \ottsym{)}  \langle  \ottnt{s}  \rangle$ in \lamS corresponds to $ \ottnt{f} \, ( \ottsym{3} , \ottnt{s} ) $ in
\lamSx; $\ottsym{(}  \ottnt{f} \, \ottsym{3}  \ottsym{)}$ (without a coercion application) corresponds to
$ \ottnt{f} \, ( \ottsym{3} , \ottkw{id} ) $.

The type of a function abstraction in \lamSx is written $\ottnt{A}  \Rightarrow  \ottnt{B}$,
which means that the type of the first argument is the type $\ottnt{A}$ and
the source type of the second coercion argument is $\ottnt{B}$.  An
abstraction is polymorphic over the target type of the coercion
argument; so, if a function of type $\ottnt{A}  \Rightarrow  \ottnt{B}$ is applied to a pair of
$\ottnt{A}$ and $\ottnt{B}  \rightsquigarrow  \ottnt{C}$, then the type of the application will be
$\ottnt{C}$.  Polymorphism is useful---and in fact required---for
coercion-passing translation to work because coercions with different
target types may be passed to calls to the same function in \lamS.
Intuitively, $\ottnt{A}  \Rightarrow  \ottnt{B}$ means
$\forall \ottmv{X}.(\ottnt{A} \times (\ottnt{B}  \rightsquigarrow  \ottmv{X})) \rightarrow \ottmv{X}$ but we do not
introduce $\forall$-types explicitly because our use of $\forall$ is
limited to the target-type polymorphism.  However, we do have to introduce type
variables for typing function abstractions.

Following the change to function types, function coercions in \lamSx
take the form $\ottnt{s}  \Rightarrow  \ottnt{t}$.  Roughly speaking, its meaning is the same:
it coerces an input to a function by $\ottnt{s}$ and coerces an output by
$\ottnt{t}$.  However, due to the coercion passing semantics, there is
slight change in how $\ottnt{t}$ is used at a function call.  Consider
$\ottnt{f}  \langle\!\langle  \ottnt{s}  \Rightarrow  \ottnt{t}  \rangle\!\rangle$, i.e., coercion-passing function $\ottnt{f}$ wrapped by coercion
$\ottnt{s}  \Rightarrow  \ottnt{t}$.  If the wrapped function is applied to $(\ottnt{V},\ottnt{t'})$,
$\ottnt{V}$ is coerced by $\ottnt{s}$ before passing to $\ottnt{f}$ as in \lamS;
instead of coercing the return value by $\ottnt{t}$, however, $\ottnt{t}$ is
prepended to $\ottnt{t'}$ and passed to $\ottnt{f}$ (together with the coerced
$\ottnt{V}$) so that the return value is coerced by $\ottnt{t}$ and then
$\ottnt{t'}$.  In the reduction rule, prepending $\ottnt{t}$ to $\ottnt{t'}$ is
represented by composition $\ottnt{t}  \mathbin{;\!;}  \ottnt{t'}$.

\subsection{Syntax}

\begin{figure}[tb]\small
  \begin{align*}
    \text{Variables} && \ottmv{x}, \ottmv{y}, \kappa &\hgap
    \text{Type variables} \hgap \ottmv{X}, \ottmv{Y}
    \\
    \text{Types} &&
    \ottnt{A},\ottnt{B},\ottnt{C} &\grmeq
     \mathord{\star}  \grmor
    \graytext{\iota} \grmor
    \ottnt{A}  \rightsquigarrow  \ottnt{B} \grmor
    \ottnt{A}  \Rightarrow  \ottnt{B} \grmor
    \ottmv{X}
    \\
    \text{Ground types} &&
    \ottnt{G},\ottnt{H} &\grmeq
    \iota \grmor
    \mathord{\star}  \Rightarrow  \mathord{\star}
    \\
    \text{Space-efficient coercions} &&
    \ottnt{s}, \ottnt{t} &\grmeq
     \ottkw{id} _{ \mathord{\star} }  \grmor
     \ottnt{G} \texttt{?}^{ \ottnt{p} }   \ottsym{;}  \ottnt{i} \grmor
    \ottnt{i}
    \\
    \text{Intermediate coercions} &&
    \ottnt{i} &\grmeq
    \ottnt{g}  \ottsym{;}   \ottnt{G} \texttt{!}  \grmor
    \ottnt{g} \grmor
     \bot^{ \ottnt{G}   \ottnt{p}   \ottnt{H} } 
    \\
    \text{Ground coercions} &&
    \ottnt{g}, \ottnt{h} &\grmeq
     \ottkw{id} _{ \ottnt{A} }  \graytext{\text{ (if $\ottnt{A}  \ne  \mathord{\star}$)}} \grmor
    \ottnt{s}  \Rightarrow  \ottnt{t} \graytext{\text{ (if $\ottnt{s}  \ne  \ottkw{id}$ or $\ottnt{t}  \ne  \ottkw{id}$)}}
    \\
    \text{Delayed coercions} &&
    \ottnt{d} &\grmeq
    \ottnt{g}  \ottsym{;}   \ottnt{G} \texttt{!}  \grmor
    \ottnt{s}  \Rightarrow  \ottnt{t} \graytext{\text{ (if $\ottnt{s}  \ne  \ottkw{id}$ or $\ottnt{t}  \ne  \ottkw{id}$)}}
    \\
    \text{Terms}&&
    \ottnt{L},\ottnt{M},\ottnt{N} &\grmeq
    \graytext{
      \ottnt{V} \grmor
      \ottnt{op}  \ottsym{(}  \ottnt{M}  \ottsym{,}  \ottnt{N}  \ottsym{)}
    } \grmor
     \ottnt{L} \, ( \ottnt{M} , \ottnt{N} )  \grmor
     \ottkw{let} \,  \ottmv{x} = \ottnt{M} \, \ottkw{in}\,  \ottnt{N} 
    \\ && & \grmsp\grmor
    \ottnt{M}  \mathbin{;\!;}  \ottnt{N} \grmor
    \ottnt{M}  \langle  \ottnt{N}  \rangle \grmor
    \ottkw{blame} \, \ottnt{p}
    \\
    \text{Values}&&
    \ottnt{V}, \ottnt{W}, \ottnt{K} &\grmeq
    \graytext{
      \ottmv{x} \grmor
      \ottnt{U}
    } \grmor
    \ottnt{U}  \langle\!\langle  \ottnt{d}  \rangle\!\rangle
    \\
    \text{Uncoerced values} &&
    \ottnt{U} &\grmeq
    \graytext{
      \ottnt{a}
    } \grmor
     \lambda  ( \ottmv{x} , \kappa ).\,  \ottnt{M}  \grmor
    \ottnt{s}
    \\
    \text{Type environments} &&
    \Gamma &\grmeq
    \graytext{
       \emptyset  \grmor
      \Gamma  \ottsym{,}  \ottmv{x}  \ottsym{:}  \ottnt{A}
    }
  \end{align*}
  \caption{Syntax of \lamSx.}
  \label{fig:syntaxS1}
\end{figure}

We show the syntax of \lamSx in Figure~\ref{fig:syntaxS1}.
We reuse the same metavariables from \lamS.
We also use $\kappa$ for variables, and $\ottnt{K}$ for values.

We replace $\ottnt{A}  \rightarrow  \ottnt{B}$ with $\ottnt{A}  \Rightarrow  \ottnt{B}$ and add $\ottnt{A}  \rightsquigarrow  \ottnt{B}$ and type
variables to types.  The syntax for ground types and space-efficient,
intermediate, ground, and delayed coercions is the same except that
$ \rightarrow $ is replaced with $ \Rightarrow $, similarly to types.  As we have
mentioned, we replace abstractions and applications with two-argument
versions.  We also add let-expressions (although they could be
introduced as derived forms) and coercion composition $\ottnt{M}  \mathbin{;\!;}  \ottnt{N}$.
The syntax for coercion applications is now $\ottnt{M}  \langle  \ottnt{N}  \rangle$, where $\ottnt{N}$
is a general term (of type $\ottnt{A}  \rightsquigarrow  \ottnt{B}$).  Uncoerced values now include
space-efficient coercions.

The term $ \lambda  ( \ottmv{x} , \kappa ).\,  \ottnt{M} $ binds $\ottmv{x}$ and $\kappa$ in $\ottnt{M}$, and
the term $ \ottkw{let} \,  \ottmv{x} = \ottnt{M} \, \ottkw{in}\,  \ottnt{N} $ binds $\ottmv{x}$ in $\ottnt{N}$.  The
definitions of free variables and $\alpha$-equivalence of terms are
standard, and thus we omit them.  We identify $\alpha$-equivalent
terms.

The definition of type environments, ranged over by $\Gamma$, is the same as \lamS.

\subsection{Type System}

\begin{figure}[tb]\small
  \textbf{Well-formed coercions (replacement)} \hfill\fbox{$\ottnt{c}  \ottsym{:}  \ottnt{A}  \rightsquigarrow  \ottnt{B}$}
  \begin{center}
    \infrule[CT-Fun]{
      \ottnt{c_{{\mathrm{1}}}}  \ottsym{:}  \ottnt{A'}  \rightsquigarrow  \ottnt{A} \andalso
      \ottnt{c_{{\mathrm{2}}}}  \ottsym{:}  \ottnt{B}  \rightsquigarrow  \ottnt{B'}
    }{
      \ottnt{c_{{\mathrm{1}}}}  \Rightarrow  \ottnt{c_{{\mathrm{2}}}}  \ottsym{:}  \ottnt{A}  \Rightarrow  \ottnt{B}  \rightsquigarrow  \ottnt{A'}  \Rightarrow  \ottnt{B'}
    }
  \end{center}
  \textbf{Term typing (excerpt)} \hfill\fbox{$ \Gamma    \vdash_{\mathsf{S_1} }    \ottnt{M}  :  \ottnt{A} $}
  \begin{center}
    \infrule[T-Crcn]{
      \ottnt{s}  \ottsym{:}  \ottnt{A}  \rightsquigarrow  \ottnt{B}
    }{
       \Gamma   \vdash   \ottnt{s}  :  \ottnt{A}  \rightsquigarrow  \ottnt{B} 
    } \hgap
    \infrule[T-Cmp]{
       \Gamma   \vdash   \ottnt{M}  :  \ottnt{A}  \rightsquigarrow  \ottnt{B}  \andalso
       \Gamma   \vdash   \ottnt{N}  :  \ottnt{B}  \rightsquigarrow  \ottnt{C} 
    }{
       \Gamma   \vdash   \ottnt{M}  \mathbin{;\!;}  \ottnt{N}  :  \ottnt{A}  \rightsquigarrow  \ottnt{C} 
    } \\\vgap
    \infrule[T-Crc]{
       \Gamma   \vdash   \ottnt{M}  :  \ottnt{A}  \andalso
       \Gamma   \vdash   \ottnt{N}  :  \ottnt{A}  \rightsquigarrow  \ottnt{B} 
    }{
       \Gamma   \vdash   \ottnt{M}  \langle  \ottnt{N}  \rangle  :  \ottnt{B} 
    } \hgap
    \infrule[T-CrcV]{
        \emptyset    \vdash   \ottnt{U}  :  \ottnt{A}  \andalso
        \emptyset    \vdash   \ottnt{d}  :  \ottnt{A}  \rightsquigarrow  \ottnt{B} 
    }{
        \emptyset    \vdash   \ottnt{U}  \langle\!\langle  \ottnt{d}  \rangle\!\rangle  :  \ottnt{B} 
    } \\\vgap
    \infrule[T-Abs]{
       \Gamma  \ottsym{,}  \ottmv{x}  \ottsym{:}  \ottnt{A}  \ottsym{,}  \kappa  \ottsym{:}  \ottnt{B}  \rightsquigarrow  \ottmv{X}   \vdash   \ottnt{M}  :  \ottmv{X}  \andalso
      \text{($X$ does not appear in $\Gamma,\ottnt{A},\ottnt{B}$)}
    }{
       \Gamma   \vdash    \lambda  ( \ottmv{x} , \kappa ).\,  \ottnt{M}   :  \ottnt{A}  \Rightarrow  \ottnt{B} 
    } \\\vgap
    \infrule[T-Let]{
       \Gamma   \vdash   \ottnt{M}  :  \ottnt{A}  \andalso
       \Gamma  \ottsym{,}  \ottmv{x}  \ottsym{:}  \ottnt{A}   \vdash   \ottnt{N}  :  \ottnt{B} 
    }{
       \Gamma   \vdash    \ottkw{let} \,  \ottmv{x} = \ottnt{M} \, \ottkw{in}\,  \ottnt{N}   :  \ottnt{B} 
    } \hgap
    \infrule[T-App]{
       \Gamma   \vdash   \ottnt{L}  :  \ottnt{A}  \Rightarrow  \ottnt{B}  \andalso
       \Gamma   \vdash   \ottnt{M}  :  \ottnt{A}  \andalso
       \Gamma   \vdash   \ottnt{N}  :  \ottnt{B}  \rightsquigarrow  \ottnt{C} 
    }{
       \Gamma   \vdash    \ottnt{L} \, ( \ottnt{M} , \ottnt{N} )   :  \ottnt{C} 
    }
  \end{center}
  \caption{Typing rules of \lamSx.}
  \label{fig:typingS1}
\end{figure}

Figure~\ref{fig:typingS1} shows
the main typing rules of \lamSx, which are a straightforward adaption from
\lamS.  

The relation $\ottnt{c}  \ottsym{:}  \ottnt{A}  \rightsquigarrow  \ottnt{B}$ is mostly the same as that of \lamS.
We replace the rule \rnp{CT-Fun} as shown.
As in \lamS, function coercions are contravariant in their argument coercions
and covariant in their return-value coercions.

The judgment $ \Gamma    \vdash_{\mathsf{S_1} }    \ottnt{M}  :  \ottnt{A} $ means that
term $\ottnt{M}$ of \lamSx has type $\ottnt{A}$ under type environment $\Gamma$.
The rules \rnp{T-Const}, \rnp{T-Op}, \rnp{T-Var}, and \rnp{T-Blame}
are the same as \lamS, and so we omit them.
The rule \rnp{T-Let} is standard.

The rules \rnp{T-Abs} and \rnp{T-App} look involved but the intuition
that $\ottnt{A}  \Rightarrow  \ottnt{B}$ corresponds to
$\forall \ottmv{X}.\,(\ottnt{A} \times (\ottnt{B}  \rightsquigarrow  \ottmv{X})) \rightarrow \ottmv{X}$ should help
to understand them.  The rule \rnp{T-Abs} assigns type $\ottnt{A}  \Rightarrow  \ottnt{B}$ to an
abstraction $ \lambda  ( \ottmv{x} , \kappa ).\,  \ottnt{M} $ if the body is well typed under the
assumption that $\ottmv{x}$ is of type $\ottnt{A}$ and $\kappa$ is of type
$\ottnt{B}  \rightsquigarrow  \ottmv{X}$ for fresh $\ottmv{X}$.  The type variable $\ottmv{X}$ must not
appear in $\Gamma,\ottnt{A},\ottnt{B}$ so that the target type can be
polymorphic at call sites.
The rule \rnp{T-App} for applications is already explained.

The rule \rnp{T-Crcn} assigns type $\ottnt{A}  \rightsquigarrow  \ottnt{B}$ to
space-efficient coercion $\ottnt{s}$ if it converts a value from type $\ottnt{A}$
to type $\ottnt{B}$.
The rules \rnp{T-Crc} and \rnp{T-CrcV} are similar to the corresponding rules
of \lamS, but adjusted to first-class coercions.

\subsection{Operational Semantics}

\begin{figure}[tb]\small
  \textbf{Coercion composition (replacement)} \hfill\fbox{$\ottnt{s}  \fatsemi  \ottnt{t}  \ottsym{=}  \ottnt{s'}$}
  \begin{align*}
    \ottsym{(}  \ottnt{s}  \Rightarrow  \ottnt{t}  \ottsym{)}  \fatsemi  \ottsym{(}  \ottnt{s'}  \Rightarrow  \ottnt{t'}  \ottsym{)} &= \begin{cases}
       \ottkw{id} _{ \ottnt{A}  \Rightarrow  \ottnt{B} }  &\text{if $\ottnt{s'}  \fatsemi  \ottnt{s}  \ottsym{=}   \ottkw{id} _{ \ottnt{A} } $ and $\ottnt{t}  \fatsemi  \ottnt{t'}  \ottsym{=}   \ottkw{id} _{ \ottnt{B} } $} \\
      \ottsym{(}  \ottnt{s'}  \fatsemi  \ottnt{s}  \ottsym{)}  \Rightarrow  \ottsym{(}  \ottnt{t}  \fatsemi  \ottnt{t'}  \ottsym{)} &\text{otherwise}
    \end{cases} &\rn{CC-Fun}
  \end{align*}
  \textbf{Evaluation contexts}
  \begin{align*}
    \mathcal{E} &\grmeq
    \graytext{
       \square 
    } \grmor
    \mathcal{E}  [   \square \, (  \ottnt{M} , \ottnt{N}  )   ] \grmor
    \mathcal{E}  [   \ottnt{V} \, (  \square , \ottnt{N}  )   ] \grmor
    \mathcal{E}  [   \ottnt{V} \, (  \ottnt{W} , \square  )   ] 
         \grmor
    \graytext{
      \mathcal{E}  [  \ottnt{op}  \ottsym{(}  \square  \ottsym{,}  \ottnt{M}  \ottsym{)}  ] \grmor
      \mathcal{E}  [  \ottnt{op}  \ottsym{(}  \ottnt{V}  \ottsym{,} \, \square \, \ottsym{)}  ]
    } \\ & \grmsp \grmor
    \mathcal{E}  [   \ottkw{let} \,  \ottmv{x} = \square \, \ottkw{in}\,  \ottnt{M}   ] \grmor
    \mathcal{E}  [  \square  \mathbin{;\!;}  \ottnt{M}  ] \grmor
    \mathcal{E}  [  \ottnt{V}  \mathbin{;\!;}  \square  ] \grmor
    \mathcal{E}  [  \square \, \langle  \ottnt{M}  \rangle  ] \grmor
    \mathcal{E}  [  \ottnt{V}  \langle \, \square \, \rangle  ]
  \end{align*}
  \textbf{Reduction} \hfill\fbox{$ \ottnt{M}    \mathbin{\accentset{\mathsf{e} }{\reduces}_{\mathsf{S_1} } }    \ottnt{N} $}\quad\fbox{$ \ottnt{M}    \mathbin{\accentset{\mathsf{c} }{\reduces}_{\mathsf{S_1} } }    \ottnt{N} $}
  \begin{align*}
    \ottnt{op}  \ottsym{(}  \ottnt{a}  \ottsym{,}  \ottnt{b}  \ottsym{)} & \mathbin{\accentset{\mathsf{e} }{\reduces} }  \delta \, \ottsym{(}  \ottnt{op}  \ottsym{,}  \ottnt{a}  \ottsym{,}  \ottnt{b}  \ottsym{)} &\rn{R-Op}\\
     \ottsym{(}   \lambda  ( \ottmv{x} , \kappa ).\,  \ottnt{M}   \ottsym{)} \, ( \ottnt{V} , \ottnt{W} )  & \mathbin{\accentset{\mathsf{e} }{\reduces} }  \ottnt{M}  [  \ottmv{x}  \ottsym{:=}  \ottnt{V}  \ottsym{,}  \kappa  \ottsym{:=}  \ottnt{W}  ] &\rn{R-Beta}\\
     \ottsym{(}  \ottnt{U}  \langle\!\langle  \ottnt{s}  \Rightarrow  \ottnt{t}  \rangle\!\rangle  \ottsym{)} \, ( \ottnt{V} , \ottnt{W} )  & \mathbin{\accentset{\mathsf{e} }{\reduces} } 
      \ottkw{let} \,  \kappa = \ottnt{t}  \mathbin{;\!;}  \ottnt{W} \, \ottkw{in}\,  \ottnt{U}  \, ( \ottnt{V}  \langle  \ottnt{s}  \rangle , \kappa )  &\rn{R-Wrap} 
  \end{align*}
  \begin{align*}
     \ottkw{let} \,  \ottmv{x} = \ottnt{V} \, \ottkw{in}\,  \ottnt{M}  & \mathbin{\accentset{\mathsf{c} }{\reduces} }  \ottnt{M}  [  \ottmv{x}  \ottsym{:=}  \ottnt{V}  ] &\rn{R-Let} &&
    \ottnt{s}  \mathbin{;\!;}  \ottnt{t} & \mathbin{\accentset{\mathsf{c} }{\reduces} }  \ottnt{s}  \fatsemi  \ottnt{t} &\rn{R-Cmp} \\
    \ottnt{U}  \langle   \ottkw{id} _{ \ottnt{A} }   \rangle & \mathbin{\accentset{\mathsf{c} }{\reduces} }  \ottnt{U} &\rn{R-Id} &&
    \ottnt{U}  \langle   \bot^{ \ottnt{G}   \ottnt{p}   \ottnt{H} }   \rangle & \mathbin{\accentset{\mathsf{c} }{\reduces} }  \ottkw{blame} \, \ottnt{p} &\rn{R-Fail} \\
    \ottnt{U}  \langle  \ottnt{d}  \rangle & \mathbin{\accentset{\mathsf{c} }{\reduces} }  \ottnt{U}  \langle\!\langle  \ottnt{d}  \rangle\!\rangle &\rn{R-Crc} &&
    \ottnt{U}  \langle\!\langle  \ottnt{d}  \rangle\!\rangle  \langle  \ottnt{t}  \rangle & \mathbin{\accentset{\mathsf{c} }{\reduces} }  \ottnt{U}  \langle  \ottnt{d}  \mathbin{;\!;}  \ottnt{t}  \rangle &\rn{R-MergeV}
  \end{align*}
  \textbf{Evaluation} \hfill\fbox{$ \ottnt{M}    \mathbin{  \accentset{\mathsf{e} }{\evalto}_{\mathsf{S_1} }  }    \ottnt{N} $}\quad\fbox{$ \ottnt{M}    \mathbin{  \accentset{\mathsf{c} }{\evalto}_{\mathsf{S_1} }  }    \ottnt{N} $}
  \begin{center}
    \infrule[E-Ctx]{
       \ottnt{M}    \mathbin{\accentset{\mathcal{X} }{\reduces} }    \ottnt{N}  \andalso
       \mathcal{X}\in   \set{\mathsf{e},\mathsf{c} }  
    }{
       \mathcal{E}  [  \ottnt{M}  ]    \mathbin{  \accentset{\mathcal{X} }{\evalto}  }    \mathcal{E}  [  \ottnt{N}  ] 
    } \hgap
    \infrule[E-Abort]{
      \mathcal{E}  \ne  \square
    }{
       \mathcal{E}  [  \ottkw{blame} \, \ottnt{p}  ]    \mathbin{  \accentset{\mathsf{e} }{\evalto}  }    \ottkw{blame} \, \ottnt{p} 
    }
  \end{center}
  \caption{Reduction/evaluation rules of \lamSx.}
  \label{fig:semS1}
\end{figure}

The composition function $\ottnt{s}  \fatsemi  \ottnt{t}$ is mostly the same as that of \lamS.
We only replace \rnp{CC-Fun} as shown in Figure~\ref{fig:semS1}.

Similarly to \lamS, we give a small-step operational semantics to \lamSx
consisting of two relations on closed terms:
the reduction relation $ \ottnt{M}    \reduces_{\mathsf{S_1} }    \ottnt{N} $ and
the evaluation relation $ \ottnt{M}    \mathbin{  \evalto_{\mathsf{S_1} }  }    \ottnt{N} $.
We show the reduction/evaluation rules of \lamSx in Figure~\ref{fig:semS1}.
As in \lamS, they are labeled either \textsf{e} or \textsf{c}.
We write $ \reduces_{\mathsf{S_1} } $ for $ \mathbin{\accentset{\mathsf{e} }{\reduces}_{\mathsf{S_1} } }  \cup  \mathbin{\accentset{\mathsf{c} }{\reduces}_{\mathsf{S_1} } } $, and
$ \evalto_{\mathsf{S_1} } $ for $ \accentset{\mathsf{e} }{\evalto}_{\mathsf{S_1} }  \cup  \accentset{\mathsf{c} }{\evalto}_{\mathsf{S_1} } $.

The rules \rnp{R-Op} and \rnp{R-Beta} are standard.
Note that \rnp{R-Beta} is adjusted for pair arguments.
We write $\ottnt{M}  [  \ottmv{x}  \ottsym{:=}  \ottnt{V}  \ottsym{,}  \kappa  \ottsym{:=}  \ottnt{K}  ]$ for capture-avoiding simultaneous substitution
of $\ottnt{V}$ and $\ottnt{K}$ for $\ottmv{x}$ and $\kappa$, respectively, in $\ottnt{M}$.

The rule \rnp{R-Wrap} applies to
applications of wrapped function $\ottnt{U}  \langle\!\langle  \ottnt{s}  \Rightarrow  \ottnt{t}  \rangle\!\rangle$ to value $\ottnt{V}$.
Since coercion $\ottnt{s}$ is for function arguments,
it is applied to $\ottnt{V}$, as in \lamS.
Additionally, we compose coercion $\ottnt{t}$ on the return value with continuation coercion $\ottnt{W}$.
Thus, $\ottnt{V}  \langle  \ottnt{s}  \rangle$ and $\ottnt{t}  \mathbin{;\!;}  \ottnt{W}$ are passed to function $\ottnt{U}$.
Note that we use a let expression to evaluate the second argument $\ottnt{t}  \mathbin{;\!;}  \ottnt{W}$ before
$\ottnt{V}  \langle  \ottnt{s}  \rangle$.
It is a necessary adjustment for the semantics of \lamS and \lamSx to match.

The rule \rnp{R-Let} is standard; it is labeled as \textsf{c} because
we use let-expressions only for coercion compositions.
The rule \rnp{R-Cmp} applies to coercion compositions $\ottnt{s}  \mathbin{;\!;}  \ottnt{t}$, which
is evaluated by meta-level coercion composition function $\ottnt{s}  \fatsemi  \ottnt{t}$.
The rules \rnp{R-Id}, \rnp{R-Fail}, \rnp{R-Crc}, and \rnp{R-MergeV}
are the same as \lamS. 

Evaluation contexts, ranged over by $\mathcal{E}$, are defined also in Figure~\ref{fig:semS1}.
In contrast to \lamS, evaluation contexts are standard in \lamSx.
The definition \ifluxuryspace of evaluation contexts $\mathcal{E}$ \fi represents that
function calls in \lamSx are call-by-value, and
primitive operations, function applications,
coercion compositions, and coercion applications
are all evaluated from left to right.

We then come back to evaluation rules:
The evaluation rules \rnp{E-Ctx} and \rnp{E-Abort} are the same as
\lamS.  (However, evaluation contexts in \rnp{E-Ctx} are more
straightforward in \lamSx.)

Finally, we should emphasize that we no longer need \rnp{R-MergeC} in
\lamSx.  So, \lamSx is an ordinary call-by-value language and its
semantics should be easy to implement.

\begin{example}\label{example:target}
  Let $\ottnt{U}$ be
  $ \lambda  ( \ottmv{x} , \kappa ).\,   \ottkw{let} \,  \kappa' =  \ottkw{int} \texttt{!}   \mathbin{;\!;}  \kappa \, \ottkw{in}\,  \ottsym{(}  \ottmv{x}  \langle   \ottkw{int} \texttt{?}^{ \ottnt{p} }   \rangle  \ottsym{+}  \ottsym{2}  \ottsym{)}   \langle  \kappa'  \rangle $,
  which corresponds to the \lamS-term $ \lambda   \ottmv{x} .\,  \ottsym{(}  \ottmv{x}  \langle   \ottkw{int} \texttt{?}^{ \ottnt{p} }   \rangle  \ottsym{+}  \ottsym{2}  \ottsym{)}  \langle   \ottkw{int} \texttt{!}   \rangle $
  in Example~\ref{example:source}.  In fact, we will obtain this term
  as a result of our coercion-passing translation defined in the next
  section.  The term $ \ottsym{(}  \ottnt{U}  \langle   \ottkw{int} \texttt{!}   \Rightarrow   \ottkw{int} \texttt{?}^{ \ottnt{p} }   \rangle  \ottsym{)} \, ( \ottsym{3} ,  \ottkw{int} \texttt{!}  ) $ evaluates to
  $\ottsym{5}  \langle\!\langle   \ottkw{int} \texttt{!}   \rangle\!\rangle$ as follows:
\begin{align*}
\mathrlap{ \ottsym{(}  \ottnt{U}  \langle   \ottkw{int} \texttt{!}   \Rightarrow   \ottkw{int} \texttt{?}^{ \ottnt{p} }   \rangle  \ottsym{)} \, ( \ottsym{3} ,  \ottkw{int} \texttt{!}  ) } \\
\ifluxuryspace
&  \evalto   \ottsym{(}  \ottnt{U}  \langle\!\langle   \ottkw{int} \texttt{!}   \Rightarrow   \ottkw{int} \texttt{?}^{ \ottnt{p} }   \rangle\!\rangle  \ottsym{)} \, ( \ottsym{3} ,  \ottkw{int} \texttt{!}  )  &\text{by \rnp{R-Crc}}\\
&  \evalto    \ottkw{let} \,  \kappa'' =  \ottkw{int} \texttt{?}^{ \ottnt{p} }   \mathbin{;\!;}   \ottkw{int} \texttt{!}  \, \ottkw{in}\,  \ottnt{U}  \, ( \ottsym{3}  \langle   \ottkw{int} \texttt{!}   \rangle , \kappa'' )  &\text{by \rnp{R-Wrap}}\\
&  \evalto    \ottkw{let} \,  \kappa'' =  \ottkw{int} \texttt{?}^{ \ottnt{p} }   \ottsym{;}  \ottkw{id}  \ottsym{;}   \ottkw{int} \texttt{!}  \, \ottkw{in}\,  \ottnt{U}  \, ( \ottsym{3}  \langle   \ottkw{int} \texttt{!}   \rangle , \kappa'' )  &\text{by \rnp{R-Cmp}} \\
&  \evalto   \ottnt{U} \, ( \ottsym{3}  \langle   \ottkw{int} \texttt{!}   \rangle , \ottsym{(}   \ottkw{int} \texttt{?}^{ \ottnt{p} }   \ottsym{;}  \ottkw{id}  \ottsym{;}   \ottkw{int} \texttt{!}   \ottsym{)} )  &\text{by \rnp{R-Let}}\\
&  \evalto   \ottnt{U} \, ( \ottsym{3}  \langle\!\langle   \ottkw{int} \texttt{!}   \rangle\!\rangle , \ottsym{(}   \ottkw{int} \texttt{?}^{ \ottnt{p} }   \ottsym{;}  \ottkw{id}  \ottsym{;}   \ottkw{int} \texttt{!}   \ottsym{)} )  &\text{by \rnp{R-Crc}}\\
&  \evalto   \ottkw{let} \,  \kappa' =  \ottkw{int} \texttt{!}   \mathbin{;\!;}  \ottsym{(}   \ottkw{int} \texttt{?}^{ \ottnt{p} }   \ottsym{;}  \ottkw{id}  \ottsym{;}   \ottkw{int} \texttt{!}   \ottsym{)} \, \ottkw{in}\,  \ottsym{(}  \ottsym{3}  \langle\!\langle   \ottkw{int} \texttt{!}   \rangle\!\rangle  \langle   \ottkw{int} \texttt{?}^{ \ottnt{p} }   \rangle  \ottsym{+}  \ottsym{2}  \ottsym{)}  \langle  \kappa'  \rangle  &\text{by \rnp{R-Beta}}\\
&  \evalto   \ottkw{let} \,  \kappa' =  \ottkw{int} \texttt{!}  \, \ottkw{in}\,  \ottsym{(}  \ottsym{3}  \langle\!\langle   \ottkw{int} \texttt{!}   \rangle\!\rangle  \langle   \ottkw{int} \texttt{?}^{ \ottnt{p} }   \rangle  \ottsym{+}  \ottsym{2}  \ottsym{)}  \langle  \kappa'  \rangle  &\text{by \rnp{R-Cmp}}\\
&  \evalto  \ottsym{(}  \ottsym{3}  \langle\!\langle   \ottkw{int} \texttt{!}   \rangle\!\rangle  \langle   \ottkw{int} \texttt{?}^{ \ottnt{p} }   \rangle  \ottsym{+}  \ottsym{2}  \ottsym{)}  \langle   \ottkw{int} \texttt{!}   \rangle &\text{by \rnp{R-Let}}\\
&  \evalto  \ottsym{(}  \ottsym{3}  \langle  \ottkw{id}  \rangle  \ottsym{+}  \ottsym{2}  \ottsym{)}  \langle   \ottkw{int} \texttt{!}   \rangle &\text{by \rnp{R-MergeV}}\\
&  \evalto  \ottsym{(}  \ottsym{3}  \ottsym{+}  \ottsym{2}  \ottsym{)}  \langle   \ottkw{int} \texttt{!}   \rangle &\text{by \rnp{R-Id}}\\
&  \evalto  \ottsym{5}  \langle   \ottkw{int} \texttt{!}   \rangle &\text{by \rnp{R-Op}}\\
&  \evalto  \ottsym{5}  \langle\!\langle   \ottkw{int} \texttt{!}   \rangle\!\rangle &\text{by \rnp{R-Crc}}
\else
&  \mathbin{ \evalto ^*}    \ottkw{let} \,  \kappa'' =  \ottkw{int} \texttt{?}^{ \ottnt{p} }   \mathbin{;\!;}   \ottkw{int} \texttt{!}  \, \ottkw{in}\,  \ottnt{U}  \, ( \ottsym{3}  \langle   \ottkw{int} \texttt{!}   \rangle , \kappa'' )  &\text{by \rnp{R-Crc}, \rnp{R-Wrap}}\\
&  \evalto    \ottkw{let} \,  \kappa'' =  \ottkw{int} \texttt{?}^{ \ottnt{p} }   \ottsym{;}  \ottkw{id}  \ottsym{;}   \ottkw{int} \texttt{!}  \, \ottkw{in}\,  \ottnt{U}  \, ( \ottsym{3}  \langle   \ottkw{int} \texttt{!}   \rangle , \kappa'' )  &\text{by \rnp{R-Cmp}} \\
&  \mathbin{ \evalto ^*}   \ottnt{U} \, ( \ottsym{3}  \langle\!\langle   \ottkw{int} \texttt{!}   \rangle\!\rangle , \ottsym{(}   \ottkw{int} \texttt{?}^{ \ottnt{p} }   \ottsym{;}  \ottkw{id}  \ottsym{;}   \ottkw{int} \texttt{!}   \ottsym{)} )  &\text{by \rnp{R-Let}, \rnp{R-Crc}}\\
&  \evalto  \mathrlap{ \ottkw{let} \,  \kappa' =  \ottkw{int} \texttt{!}   \mathbin{;\!;}  \ottsym{(}   \ottkw{int} \texttt{?}^{ \ottnt{p} }   \ottsym{;}  \ottkw{id}  \ottsym{;}   \ottkw{int} \texttt{!}   \ottsym{)} \, \ottkw{in}\,  \ottsym{(}  \ottsym{3}  \langle\!\langle   \ottkw{int} \texttt{!}   \rangle\!\rangle  \langle   \ottkw{int} \texttt{?}^{ \ottnt{p} }   \rangle  \ottsym{+}  \ottsym{2}  \ottsym{)}  \langle  \kappa'  \rangle } &\text{by \rnp{R-Beta}}\\
&  \mathbin{ \evalto ^*}  \ottsym{(}  \ottsym{3}  \langle\!\langle   \ottkw{int} \texttt{!}   \rangle\!\rangle  \langle   \ottkw{int} \texttt{?}^{ \ottnt{p} }   \rangle  \ottsym{+}  \ottsym{2}  \ottsym{)}  \langle   \ottkw{int} \texttt{!}   \rangle &\text{by \rnp{R-Cmp}, \rnp{R-Let}}\\
&  \mathbin{ \evalto ^*}  \ottsym{5}  \langle\!\langle   \ottkw{int} \texttt{!}   \rangle\!\rangle &\text{by \rnp{R-MergeV}, \rnp{R-ID}, \rnp{R-Op}, \rnp{R-Crc}}
\fi                          
\end{align*}
It is easy to see that the steps by \rnp{R-MergeC} in Example~\ref{example:source}
are simulated by \rnp{R-Cmp} followed by \rnp{R-Let}.
\end{example}

\subsection{Properties}

We state a few properties of \lamSx below.
\iffull
Their proofs are in Appendix~\ref{sec:appendix}.
\fi

\begin{lemma}[name=Determinacy,restate=lemDeterminacySx]\label{lem:determinism-S1}
  If $ \ottnt{M}    \mathbin{  \evalto_{\mathsf{S_1} }  }    \ottnt{N} $ and $ \ottnt{M}    \mathbin{  \evalto_{\mathsf{S_1} }  }    \ottnt{N'} $, then $\ottnt{N}  \ottsym{=}  \ottnt{N'}$.
\end{lemma}

\ifluxuryspace
\begin{theorem}[name=Progress,restate=thmProgressSx]\label{thm:progress-S1}
  If $  \emptyset     \vdash_{\mathsf{S_1} }    \ottnt{M}  :  \ottnt{A} $, then one of the following holds.
  \begin{enumerate}
  \item $ \ottnt{M}    \mathbin{  \evalto_{\mathsf{S_1} }  }    \ottnt{M'} $ for some $\ottnt{M'}$.
  \item $\ottnt{M}  \ottsym{=}  \ottnt{V}$ for some $\ottnt{V}$.
  \item $\ottnt{M}  \ottsym{=}  \ottkw{blame} \, \ottnt{p}$ for some $\ottnt{p}$.
  \end{enumerate}
\end{theorem}
\else
\begin{theorem}[name=Progress,restate=thmProgressSx]\label{thm:progress-S1}
  If $  \emptyset     \vdash_{\mathsf{S_1} }    \ottnt{M}  :  \ottnt{A} $, then one of the following holds:
  (1) $ \ottnt{M}    \mathbin{  \evalto_{\mathsf{S_1} }  }    \ottnt{M'} $ for some $\ottnt{M'}$;
  (2) $\ottnt{M}  \ottsym{=}  \ottnt{V}$ for some $\ottnt{V}$; or
  (3) $\ottnt{M}  \ottsym{=}  \ottkw{blame} \, \ottnt{p}$ for some $\ottnt{p}$.
\end{theorem}
\fi

\begin{theorem}[name=Preservation,restate=thmPreservationSx]\label{thm:preservation-S1}
  If $  \emptyset     \vdash_{\mathsf{S_1} }    \ottnt{M}  :  \ottnt{A} $ and $ \ottnt{M}    \mathbin{  \evalto_{\mathsf{S_1} }  }    \ottnt{N} $, then
  $  \emptyset     \vdash_{\mathsf{S_1} }    \ottnt{N}  :  \ottnt{A} $.
\end{theorem}

\ifluxuryspace
\begin{corollary}[name=Type Safety,restate=corSafetySx]\label{cor:safety-S1}
  If $  \emptyset     \vdash_{\mathsf{S_1} }    \ottnt{M}  :  \ottnt{A} $, then one of the following holds.
  \begin{enumerate}
  \item $ \ottnt{M}    \mathbin{  \evalto_{\mathsf{S_1} }  ^*}    \ottnt{V} $ and $  \emptyset     \vdash_{\mathsf{S_1} }    \ottnt{V}  :  \ottnt{A} $ for some $\ottnt{V}$.
  \item $ \ottnt{M}    \mathbin{  \evalto_{\mathsf{S_1} }  ^*}    \ottkw{blame} \, \ottnt{p} $ for some $\ottnt{p}$.
  \item $\ottnt{M} \,  \mathord{\Uparrow_{\mathsf{S_1} } } $.
  \end{enumerate}
\end{corollary}
\else
\begin{corollary}[name=Type Safety,restate=corSafetySx]\label{cor:safety-S1}
  If $  \emptyset     \vdash_{\mathsf{S_1} }    \ottnt{M}  :  \ottnt{A} $, then one of the following holds:
  (1) $ \ottnt{M}    \mathbin{  \evalto_{\mathsf{S_1} }  ^*}    \ottnt{V} $ and $  \emptyset     \vdash_{\mathsf{S_1} }    \ottnt{V}  :  \ottnt{A} $ for some $\ottnt{V}$;
  (2) $ \ottnt{M}    \mathbin{  \evalto_{\mathsf{S_1} }  ^*}    \ottkw{blame} \, \ottnt{p} $ for some $\ottnt{p}$; or
  (3) $\ottnt{M} \,  \mathord{\Uparrow_{\mathsf{S_1} } } $.
\end{corollary}
\fi

\section{Translation into Coercion-Passing Style}
\label{sec:translation}

In this section, we formalize a translation into coercion-passing
style as a translation from \lamS to \lamSx and state its correctness.
As its name suggests, this translation is similar to transformations
into continuation-passing style (CPS transformations) for the
call-by-value
\(\lambda\)-calculus~\citep{DBLP:journals/tcs/Plotkin75}.

\subsection{Definition of Translation}
\label{subsec:trans-def}

\UseColorTranstrue  
\begin{figure}[tb]\small
  \textbf{Type translation} \hfill\fbox{$ \Psi (\maybebluetext{ \ottnt{A} })   \ottsym{=}  \ottnt{A'}$}
  \[
     \Psi (\maybebluetext{ \mathord{\star} })  =  \mathord{\star}  \hgap
     \Psi (\maybebluetext{ \iota })  = \iota \hgap
     \Psi (\maybebluetext{ \ottnt{A}  \rightarrow  \ottnt{B} })  =  \Psi (\maybebluetext{ \ottnt{A} })   \Rightarrow   \Psi (\maybebluetext{ \ottnt{B} }) 
  \]
  \begin{minipage}[t]{0.4\textwidth}
    \textbf{Coercion translation} \hfill\fbox{$ \Psi (\maybebluetext{ \ottnt{s} })   \ottsym{=}  \ottnt{s'}$}
    \begin{align*}
       \Psi (\maybebluetext{  \ottkw{id} _{ \ottnt{A} }  })  &=  \ottkw{id} _{  \Psi (\maybebluetext{ \ottnt{A} })  }  \\
       \Psi (\maybebluetext{ \ottnt{g}  \ottsym{;}   \ottnt{G} \texttt{!}  })  &=  \Psi (\maybebluetext{ \ottnt{g} })   \ottsym{;}    \Psi (\maybebluetext{ \ottnt{G} })  \texttt{!}  \\
       \Psi (\maybebluetext{  \ottnt{G} \texttt{?}^{ \ottnt{p} }   \ottsym{;}  \ottnt{i} })  &=   \Psi (\maybebluetext{ \ottnt{G} })  \texttt{?}^{ \ottnt{p} }   \ottsym{;}   \Psi (\maybebluetext{ \ottnt{i} })  \\
       \Psi (\maybebluetext{ \ottnt{s}  \rightarrow  \ottnt{t} })  &=  \Psi (\maybebluetext{ \ottnt{s} })   \Rightarrow   \Psi (\maybebluetext{ \ottnt{t} })  \\
       \Psi (\maybebluetext{  \bot^{ \ottnt{G}   \ottnt{p}   \ottnt{H} }  })  &=  \bot^{ \ottnt{G}   \ottnt{p}   \ottnt{H} } 
    \end{align*}
  \end{minipage}%
  \hfill
  \begin{minipage}[t]{0.5\textwidth}
    \textbf{Value translation} \hfill\fbox{$ \Psi (\maybebluetext{ \ottnt{V} })   \ottsym{=}  \ottnt{V'}$}
    \begin{align*}
       \Psi (\maybebluetext{ \ottmv{x} })  &= \ottmv{x} \\
       \Psi (\maybebluetext{ \ottnt{a} })  &= \ottnt{a} \\
       \Psi (\maybebluetext{  \lambda   \ottmv{x} .\,  \ottnt{M}  })  &=  \lambda  ( \ottmv{x} , \kappa ).\,  \ottsym{(}   \mathscr{K}\llbracket \maybebluetext{ \ottnt{M} } \rrbracket  \kappa   \ottsym{)}  \\
       \Psi (\maybebluetext{ \ottnt{U}  \langle\!\langle  \ottnt{d}  \rangle\!\rangle })  &=  \Psi (\maybebluetext{ \ottnt{U} })   \langle\!\langle   \Psi (\maybebluetext{ \ottnt{d} })   \rangle\!\rangle
    \end{align*}
  \end{minipage}
  \\
  \textbf{Term translation} \hfill\fbox{$ \mathscr{C}\llbracket \maybebluetext{ \ottnt{M} } \rrbracket   \ottsym{=}  \ottnt{M'}$}\quad\fbox{$ \mathscr{K}\llbracket \maybebluetext{ \ottnt{M} } \rrbracket  \ottnt{K}   \ottsym{=}  \ottnt{M'}$}
  \begin{align*}
     \mathscr{C}\llbracket \maybebluetext{ \ottnt{V} } \rrbracket  &=  \Psi (\maybebluetext{ \ottnt{V} })  &&&\rn{TrC-Val}\\
     \mathscr{C}\llbracket \maybebluetext{ \ottnt{M}  \langle  \ottnt{s}  \rangle } \rrbracket  &=  \mathscr{K}\llbracket \maybebluetext{ \ottnt{M} } \rrbracket   \Psi (\maybebluetext{ \ottnt{s} })   &&&\rn{TrC-Crc}\\
     \mathscr{C}\llbracket \maybebluetext{  \ottnt{M} ^{ \ottnt{A} }  } \rrbracket  &=  \mathscr{K}\llbracket \maybebluetext{ \ottnt{M} } \rrbracket   \ottkw{id} _{  \Psi (\maybebluetext{ \ottnt{A} })  }  
    &\text{otherwise} &&\rn{TrC-Else}
    \\[0.5em]
     \mathscr{K}\llbracket \maybebluetext{ \ottnt{V} } \rrbracket  \ottnt{K}  &=  \Psi (\maybebluetext{ \ottnt{V} })   \langle  \ottnt{K}  \rangle &&&\rn{Tr-Val}\\
     \mathscr{K}\llbracket \maybebluetext{ \ottnt{op}  \ottsym{(}  \ottnt{M}  \ottsym{,}  \ottnt{N}  \ottsym{)} } \rrbracket  \ottnt{K} 
    &= \ottnt{op}  \ottsym{(}   \mathscr{C}\llbracket \maybebluetext{ \ottnt{M} } \rrbracket   \ottsym{,}   \mathscr{C}\llbracket \maybebluetext{ \ottnt{N} } \rrbracket   \ottsym{)}  \langle  \ottnt{K}  \rangle &&&\rn{Tr-Op}\\
     \mathscr{K}\llbracket \maybebluetext{ \ottnt{M} \, \ottnt{N} } \rrbracket  \ottnt{K} 
    &=  \ottsym{(}   \mathscr{C}\llbracket \maybebluetext{ \ottnt{M} } \rrbracket   \ottsym{)} \, (  \mathscr{C}\llbracket \maybebluetext{ \ottnt{N} } \rrbracket  , \ottnt{K} )  &&&\rn{Tr-App}\\
     \mathscr{K}\llbracket \maybebluetext{ \ottnt{M}  \langle  \ottnt{s}  \rangle } \rrbracket  \ottnt{K}  &=  \ottkw{let} \,  \kappa =  \Psi (\maybebluetext{ \ottnt{s} })   \mathbin{;\!;}  \ottnt{K} \, \ottkw{in}\,  \ottsym{(}   \mathscr{K}\llbracket \maybebluetext{ \ottnt{M} } \rrbracket  \kappa   \ottsym{)} 
    &&&\rn{Tr-Crc}\\
     \mathscr{K}\llbracket \maybebluetext{ \ottkw{blame} \, \ottnt{p} } \rrbracket  \ottnt{K}  &= \ottkw{blame} \, \ottnt{p} &&&\rn{Tr-Blame}
  \end{align*}
  \caption{Translation into coercion-passing style (from \lamS to \lamSx).}
  \label{fig:transSS1}
\end{figure}
\UseColorTransfalse

We give the translation into coercion-passing style
by the translation rules presented in Figure~\ref{fig:transSS1}.
In order to distinguish metavariables of \lamS and \lamSx,
we often use \bluetext{blue} for the source calculus \lamS.
When we need static type information in translation rules, we write
$ \ottnt{M} ^{ \ottnt{A} } $ to indicate that term $\ottnt{M}$ has type $\ottnt{A}$.  Thus,
strictly speaking, the translation is defined for type derivations in
\lamS.

Translations for types $ \Psi (\maybebluetext{ \ottnt{A} }) $ and coercions $ \Psi (\maybebluetext{ \ottnt{s} }) $ are
very straightforward, thanks to the special type/coercion constructor
$ \Rightarrow $: they just recursively replace \ifluxuryspace type/coercion constructor \fi
$ \rightarrow $ with $ \Rightarrow $.

Value translation $ \Psi (\maybebluetext{ \ottnt{V} }) $ and term translation $ \mathscr{K}\llbracket \maybebluetext{ \ottnt{M} } \rrbracket  \ottnt{K} $ are
defined in a mutually recursive manner.  In $ \mathscr{K}\llbracket \maybebluetext{ \ottnt{M} } \rrbracket  \ottnt{K} $, $\ottnt{M}$ is a
\lamS-term whereas $\ottnt{K}$ is a \lamSx-term, which is either a
variable or a \lamSx-coercion.  $ \mathscr{K}\llbracket \maybebluetext{ \ottnt{M} } \rrbracket  \ottnt{K} $ returns a \lamSx-term---in
coercion-passing style---that applies $\ottnt{K}$ to the value of $\ottnt{M}$.

Value translation $ \Psi (\maybebluetext{ \ottnt{V} }) $ is straightforward:
every function $ \lambda   \ottmv{x} .\,  \ottnt{M} $ is translated to
a \lamSx-ab\-strac\-tion that takes as the second argument $\kappa$ a coercion which
is to be applied to the return value.
So, the body is translated by term translation $ \mathscr{K}\llbracket \maybebluetext{ \ottnt{M} } \rrbracket  \kappa $.

We now describe the translation for terms.
We write $ \mathscr{K}\llbracket \maybebluetext{ \ottnt{M} } \rrbracket  \ottnt{K} $ for the translation of \lamS-term $\ottnt{M}$
with continuation coercion $\ottnt{K}$.
We first explain the basic transformation scheme given by
the recursive function $\mathscr{K}'$ defined by the following simpler rules:
\UseColorTranstrue
\begin{align*}
   \mathscr{K}'\llbracket \maybebluetext{ \ottnt{V} } \rrbracket  \ottnt{K}  &=  \Psi (\maybebluetext{ \ottnt{V} })   \langle  \ottnt{K}  \rangle &\rnFakeTr{Val}\\
   \mathscr{K}'\llbracket \maybebluetext{ \ottnt{op}  \ottsym{(}   \ottnt{M} ^{ \iota_{{\mathrm{1}}} }   \ottsym{,}   \ottnt{N} ^{ \iota_{{\mathrm{2}}} }   \ottsym{)} } \rrbracket  \ottnt{K}  &= \ottnt{op}  \ottsym{(}   \mathscr{K}'\llbracket \maybebluetext{ \ottnt{M} } \rrbracket   \ottkw{id} _{ \iota_{{\mathrm{1}}} }    \ottsym{,}   \mathscr{K}'\llbracket \maybebluetext{ \ottnt{N} } \rrbracket   \ottkw{id} _{ \iota_{{\mathrm{2}}} }    \ottsym{)}  \langle  \ottnt{K}  \rangle &\rnFakeTr{Op}\\
   \mathscr{K}'\llbracket \maybebluetext{   \ottnt{M} ^{ \ottnt{A}  \rightarrow  \ottnt{B} }  \, \ottnt{N} ^{ \ottnt{A} }  } \rrbracket  \ottnt{K}  &=  \ottsym{(}   \mathscr{K}'\llbracket \maybebluetext{ \ottnt{M} } \rrbracket   \ottkw{id} _{  \Psi (\maybebluetext{ \ottnt{A}  \rightarrow  \ottnt{B} })  }    \ottsym{)} \, (  \mathscr{K}'\llbracket \maybebluetext{ \ottnt{N} } \rrbracket   \ottkw{id} _{  \Psi (\maybebluetext{ \ottnt{A} })  }   , \ottnt{K} )  &\rnFakeTr{App}\\
   \mathscr{K}'\llbracket \maybebluetext{ \ottnt{M}  \langle  \ottnt{s}  \rangle } \rrbracket  \ottnt{K}  &=  \ottkw{let} \,  \kappa =  \Psi (\maybebluetext{ \ottnt{s} })   \mathbin{;\!;}  \ottnt{K} \, \ottkw{in}\,  \ottsym{(}   \mathscr{K}'\llbracket \maybebluetext{ \ottnt{M} } \rrbracket  \kappa   \ottsym{)}  &\rnFakeTr{Crc}\\
   \mathscr{K}'\llbracket \maybebluetext{ \ottkw{blame} \, \ottnt{p} } \rrbracket  \ottnt{K}  &= \ottkw{blame} \, \ottnt{p} &\rnFakeTr{Blame}
\end{align*}
\UseColorTransfalse
(We put a prime on $\mathscr{K}$ to distinguish with the final version.)

The rule \rnpFakeTr{Val} applies to values $\ottnt{V}$, where we apply
coercion $\ottnt{K}$ to the result of value translation $ \Psi (\maybebluetext{ \ottnt{V} }) $.

The rule \rnpFakeTr{Op} applies to primitive operations $\ottnt{op}  \ottsym{(}  \ottnt{M}  \ottsym{,}  \ottnt{N}  \ottsym{)}$.
We translate the arguments $\ottnt{M}$ and $\ottnt{N}$ with identity
continuation coercions by $ \mathscr{K}'\llbracket \maybebluetext{ \ottnt{M} } \rrbracket  \ottkw{id} $ and $ \mathscr{K}'\llbracket \maybebluetext{ \ottnt{N} } \rrbracket  \ottkw{id} $ and pass
them to the primitive operation.  The given continuation coercion
$\ottnt{K}$ is applied to the result.  Translating subexpressions with
$\ottkw{id}$ is one of the main differences from CPS transformation.
While continuations in continuation-passing style capture the whole
rest of computation, continuation coercions in coercion-passing style
capture only the coercion applied right after the current computation.
Since neither $\ottnt{M}$ nor $\ottnt{N}$ is surrounded by a coercion, they
are translated with identity coercions of appropriate types.  (Cases
where a subexpression itself is a coercion application will be
discussed shortly.)  Careful readers may notice at this point that
left-to-right evaluation of arguments is enforced by the semantics (or
the definition of evaluation contexts) of \lamS, not by the translation.
In other words, the correctness of the translation relies on the fact that
\lamS evaluation is left-to-right and call-by-value.  This is another
point that is different from CPS transformation, which dismisses the
distinction of call-by-name and call-by-value.

The rule \rnpFakeTr{App} applies to function applications $\ottnt{M} \, \ottnt{N}$.  We
translate function $\ottnt{M}$ and argument $\ottnt{N}$ with identity
continuation coercions just like the case for primitive operations.
We then pass the continuation coercion $\ottnt{K}$ as the second argument to
function $ \mathscr{K}'\llbracket \maybebluetext{ \ottnt{M} } \rrbracket  \ottkw{id} $.

The rule \rnpFakeTr{Crc} applies to coercion applications $\ottnt{M}  \langle  \ottnt{s}  \rangle$.
We can think of the sequential composition of $ \Psi (\maybebluetext{ \ottnt{s} }) $ and $\ottnt{K}$
as the continuation coercion for $\ottnt{M}$.  Thus, we first compute the
composition $ \Psi (\maybebluetext{ \ottnt{s} })   \mathbin{;\!;}  \ottnt{K}$, bind its result to $\kappa$, and
translate $\ottnt{M}$ with continuation $\kappa$.  The let-expression is
necessary to compose $ \Psi (\maybebluetext{ \ottnt{s} }) $ and $\ottnt{K}$ before evaluating
$ \mathscr{K}'\llbracket \maybebluetext{ \ottnt{M} } \rrbracket  \kappa $.  In general, it is not necessarily the case that
$ \mathscr{K}'\llbracket \maybebluetext{ \ottnt{M} } \rrbracket  \ottnt{K} $ evaluates $\ottnt{K}$ first, so if we set
$ \mathscr{K}'\llbracket \maybebluetext{ \ottnt{M}  \langle  \ottnt{s}  \rangle } \rrbracket  \ottnt{K}  = \ottsym{(}   \mathscr{K}'\llbracket \maybebluetext{ \ottnt{M} } \rrbracket  \ottsym{(}   \Psi (\maybebluetext{ \ottnt{s} })   \mathbin{;\!;}  \ottnt{K}  \ottsym{)}   \ottsym{)}$, then the order of
computation would change by the translation and correctness of translation
would be harder to show.

Lastly, the rule \rnpFakeTr{Blame} means that continuation $\ottnt{K}$ is discarded for $\ottkw{blame} \, \ottnt{p}$.

The translation $\mathscr{K}'$ seems acceptable but, just as na\"ive CPS
transformation leaves administrative redexes, it leaves many
applications of $\ottkw{id}$, which we call \emph{administrative
  coercions}.  We expect $\ottnt{M}$ and $ \mathscr{K}'\llbracket \maybebluetext{ \ottnt{M} } \rrbracket  \ottnt{K} $ to ``behave similarly''
but administrative redexes make it hard to show such semantic
correspondence.
Therefore, we will optimize the translation so that
administrative coercions are eliminated, similarly to CPS
transformations that eliminate administrative
redexes~\cite{DBLP:journals/tcs/Plotkin75, DBLP:books/cu/Appel1992,
  DBLP:conf/mfps/Wand91, DBLP:journals/lisp/SabryF93,
  DBLP:journals/tcs/DanvyN03, DBLP:journals/mscs/DanvyF92,
  DBLP:journals/toplas/SabryW97}.

The bottom of Figure~\ref{fig:transSS1} shows the optimized
translation rules.  The idea to eliminate administrative coercions is
close to the colon translation by
Plotkin~\cite{DBLP:journals/tcs/Plotkin75}: we avoid translating
values with administrative coercions.  So, we introduce an auxiliary
translation function $ \mathscr{C}\llbracket \maybebluetext{ \ottnt{M} } \rrbracket $, which, if $\ottnt{M}$ is a value
$\ottnt{V}$, returns $ \Psi (\maybebluetext{ \ottnt{V} }) $---without a coercion application---and,
if $\ottnt{M}$ is a coercion application $\ottnt{N}  \langle  \ottnt{s}  \rangle$, returns
$ \mathscr{K}\llbracket \maybebluetext{ \ottnt{N} } \rrbracket   \Psi (\maybebluetext{ \ottnt{s} })  $---with the trivial composition $ \Psi (\maybebluetext{ \ottnt{s} })   \fatsemi  \ottkw{id}$
optimized away---and returns $ \mathscr{K}\llbracket \maybebluetext{ \ottnt{M} } \rrbracket  \ottkw{id} $ otherwise.  Translation
rules for primitive operations and function applications are adapted
so that they use $ \mathscr{C}\llbracket \maybebluetext{ \ottnt{M} } \rrbracket $ to translate subexpressions.

In other words, $ \mathscr{C}\llbracket \maybebluetext{ \ottnt{M} } \rrbracket $ helps us precisely distinguish
between $\ottkw{id}$ introduced by the translation
and $\ottkw{id}$ that was present in the original term.
Whenever we introduce $\ottkw{id}$ as an initial coercion for the translation,
we first apply $ \mathscr{C}\llbracket \maybebluetext{ \ottnt{M} } \rrbracket $ and then apply $ \mathscr{K}\llbracket \maybebluetext{ \ottnt{M} } \rrbracket  \ottkw{id} $ only if necessary.
We note that $  \mathscr{K}\llbracket \maybebluetext{ \ottnt{M} } \rrbracket  \ottkw{id}     \mathbin{  \evalto_{\mathsf{S_1} }  }     \mathscr{C}\llbracket \maybebluetext{ \ottnt{M} } \rrbracket  $ holds. \iffull(Lemma~\ref{lem:trans-id-CV})\fi

We present a few examples of the translation below:
\begin{align*}
  \ifluxuryspace
   \Psi (\maybebluetext{ \ottsym{5} })  &= \ottsym{5} \\
  \fi
   \Psi (\maybebluetext{  \lambda   \ottmv{x} .\,  \ottmv{x}  \ottsym{+}  \ottsym{1}  })  &=  \lambda  ( \ottmv{x} , \kappa ).\,  \ottsym{(}  \ottmv{x}  \ottsym{+}  \ottsym{1}  \ottsym{)}  \langle  \kappa  \rangle  \\
   \mathscr{K}\llbracket \maybebluetext{ \ottsym{(}   \lambda   \ottmv{x} .\,  \ottmv{x}   \ottsym{)} \, \ottsym{5} } \rrbracket   \ottkw{int} \texttt{!}   &=  \ottsym{(}   \lambda  ( \ottmv{x} , \kappa ).\,  \ottmv{x}   \langle  \kappa  \rangle  \ottsym{)} \, ( \ottsym{5} ,  \ottkw{int} \texttt{!}  )  \\
   \mathscr{K}\llbracket \maybebluetext{ \ottsym{(}  \ottsym{(}   \lambda   \ottmv{x} .\,  \ottmv{x}   \ottsym{)} \, \ottsym{5}  \ottsym{)}  \langle   \ottkw{int} \texttt{!}   \rangle } \rrbracket   \ottkw{int} \texttt{?}^{ \ottnt{p} }   &=   \ottkw{let} \,  \kappa =  \ottkw{int} \texttt{!}   \mathbin{;\!;}   \ottkw{int} \texttt{?}^{ \ottnt{p} }  \, \ottkw{in}\,  \ottsym{(}   \lambda  ( \ottmv{x} , \kappa ).\,  \ottmv{x}   \langle  \kappa  \rangle  \ottsym{)}  \, ( \ottsym{5} , \kappa ) 
\end{align*}

The following example shows the translation of the \lamS-term
in Example~\ref{example:source} will be the \lamSx-term in Example~\ref{example:target}.

\begin{example}\label{example:trans}
  Let $\ottnt{U}$ be a \lamS-term $ \lambda   \ottmv{x} .\,  \ottsym{(}  \ottmv{x}  \langle   \ottkw{int} \texttt{?}^{ \ottnt{p} }   \rangle  \ottsym{+}  \ottsym{2}  \ottsym{)}  \langle   \ottkw{int} \texttt{!}   \rangle $.
\[
\begin{array}{rll}
   \Psi (\maybebluetext{ \ottnt{U} })  & = &  \lambda  ( \ottmv{x} , \kappa ).\,  \ottsym{(}   \mathscr{K}\llbracket \maybebluetext{ \ottsym{(}  \ottmv{x}  \langle   \ottkw{int} \texttt{?}^{ \ottnt{p} }   \rangle  \ottsym{+}  \ottsym{2}  \ottsym{)}  \langle   \ottkw{int} \texttt{!}   \rangle } \rrbracket  \kappa   \ottsym{)}  \\
             &= &  \lambda  ( \ottmv{x} , \kappa ).\,   \ottkw{let} \,  \kappa' =  \ottkw{int} \texttt{!}   \mathbin{;\!;}  \kappa \, \ottkw{in}\,  \ottsym{(}   \mathscr{K}\llbracket \maybebluetext{ \ottsym{(}  \ottmv{x}  \langle   \ottkw{int} \texttt{?}^{ \ottnt{p} }   \rangle  \ottsym{+}  \ottsym{2}  \ottsym{)} } \rrbracket  \kappa'   \ottsym{)}   \\
&= &  \lambda  ( \ottmv{x} , \kappa ).\,   \ottkw{let} \,  \kappa' =  \ottkw{int} \texttt{!}   \mathbin{;\!;}  \kappa \, \ottkw{in}\,  \ottsym{(}  \ottmv{x}  \langle   \ottkw{int} \texttt{?}^{ \ottnt{p} }   \rangle  \ottsym{+}  \ottsym{2}  \ottsym{)}   \langle  \kappa'  \rangle  \\ \\
 \mathscr{K}\llbracket \maybebluetext{ \ottsym{(}  \ottsym{(}  \ottnt{U}  \langle   \ottkw{int} \texttt{!}   \rightarrow   \ottkw{int} \texttt{?}^{ \ottnt{p} }   \rangle  \ottsym{)} \, \ottsym{3}  \ottsym{)} } \rrbracket  \ottkw{id} 
&= &  \ottsym{(}   \mathscr{K}\llbracket \maybebluetext{ \ottsym{(}  \ottnt{U}  \langle   \ottkw{int} \texttt{!}   \rightarrow   \ottkw{int} \texttt{?}^{ \ottnt{p} }   \rangle  \ottsym{)} } \rrbracket  \ottkw{id}   \ottsym{)} \, (  \mathscr{K}\llbracket \maybebluetext{ \ottsym{3} } \rrbracket  \ottkw{id}  , \ottkw{id} )  \\
&= &  \ottsym{(}   \mathscr{K}\llbracket \maybebluetext{ \ottnt{U} } \rrbracket  \ottsym{(}   \ottkw{int} \texttt{!}   \rightarrow   \ottkw{int} \texttt{?}^{ \ottnt{p} }   \ottsym{)}   \ottsym{)} \, ( \ottsym{3} , \ottkw{id} )  \\
&= &  \ottsym{(}   \Psi (\maybebluetext{ \ottnt{U} })   \langle   \ottkw{int} \texttt{!}   \Rightarrow   \ottkw{int} \texttt{?}^{ \ottnt{p} }   \rangle  \ottsym{)} \, ( \ottsym{3} , \ottkw{id} ) 
\end{array}
\]
\end{example}

\subsection{Correctness of Translation}
\label{subsec:correctness}

Having defined the translation, we now state its correctness properties with auxiliary lemmas.
\iffull
(Their proofs are in Appendix~\ref{sec:appendix}.)
\fi

To begin with, the translation preserves typing.
Here, we write $ \Psi ( \Gamma ) $ for the type environment
\ifluxuryspace
satisfying the following:
\begin{center}
  $ ( \ottmv{x}  :  \ottnt{A} ) \in  \Gamma $ if and only if $ ( \ottmv{x}  :   \Psi (\maybebluetext{ \ottnt{A} })  ) \in   \Psi ( \Gamma )  $.
\end{center}
\else
satisfying:
  $ ( \ottmv{x}  :  \ottnt{A} ) \in  \Gamma $ if and only if $ ( \ottmv{x}  :   \Psi (\maybebluetext{ \ottnt{A} })  ) \in   \Psi ( \Gamma )  $.
\fi

\begin{theorem}[name=Translation Preserves Typing,restate=thmTransTyping] \label{thm:trans-typing}
  \leavevmode
  \begin{enumerate}
  \item   If $ \Gamma    \vdash_{\mathsf{S} }    \ottnt{M}  :  \ottnt{A} $ and $\ottnt{s}  \ottsym{:}  \ottnt{A}  \rightsquigarrow  \ottnt{B}$ ,
  then $  \Psi ( \Gamma )     \vdash_{\mathsf{S_1} }    \ottsym{(}   \mathscr{K}\llbracket \maybebluetext{ \ottnt{M} } \rrbracket   \Psi (\maybebluetext{ \ottnt{s} })    \ottsym{)}  :   \Psi (\maybebluetext{ \ottnt{B} })  $.
  \item If $ \Gamma    \vdash_{\mathsf{S} }    \ottnt{V}  :  \ottnt{A} $,
  then $  \Psi ( \Gamma )     \vdash_{\mathsf{S_1} }     \Psi (\maybebluetext{ \ottnt{V} })   :   \Psi (\maybebluetext{ \ottnt{A} })  $.
  \end{enumerate}
\end{theorem}

As for the preservation of semantics, we will prove the following theorem that states
the semantics is preserved by the translation:

\ifluxuryspace
\begin{theorem}[name=Translation Preserves Semantics,restate=thmTransSem] \label{thm:trans-correctness}
  If $  \emptyset     \vdash_{\mathsf{S} }    \ottnt{M}  :  \iota $, then
  \begin{enumerate}
  \item $ \ottnt{M}    \mathbin{  \evalto_{\mathsf{S} }  ^*}    \ottnt{a} $ iff $  \mathscr{C}\llbracket \maybebluetext{ \ottnt{M} } \rrbracket     \mathbin{  \evalto_{\mathsf{S_1} }  ^*}    \ottnt{a} $;
  \item $ \ottnt{M}    \mathbin{  \evalto_{\mathsf{S} }  ^*}    \ottkw{blame} \, \ottnt{p} $ iff $  \mathscr{C}\llbracket \maybebluetext{ \ottnt{M} } \rrbracket     \mathbin{  \evalto_{\mathsf{S_1} }  ^*}    \ottkw{blame} \, \ottnt{p} $; and
  \item $\ottnt{M} \,  \mathord{\Uparrow_{\mathsf{S} } } $ iff $ \mathscr{C}\llbracket \maybebluetext{ \ottnt{M} } \rrbracket  \,  \mathord{\Uparrow_{\mathsf{S_1} } } $.
  \end{enumerate}
\end{theorem}
\else
\begin{theorem}[name=Translation Preserves Semantics,restate=thmTransSem] \label{thm:trans-correctness}
  If $  \emptyset     \vdash_{\mathsf{S} }    \ottnt{M}  :  \iota $, then
  (1) $ \ottnt{M}    \mathbin{  \evalto_{\mathsf{S} }  ^*}    \ottnt{a} $ iff $  \mathscr{C}\llbracket \maybebluetext{ \ottnt{M} } \rrbracket     \mathbin{  \evalto_{\mathsf{S_1} }  ^*}    \ottnt{a} $;
  (2) $ \ottnt{M}    \mathbin{  \evalto_{\mathsf{S} }  ^*}    \ottkw{blame} \, \ottnt{p} $ iff $  \mathscr{C}\llbracket \maybebluetext{ \ottnt{M} } \rrbracket     \mathbin{  \evalto_{\mathsf{S_1} }  ^*}    \ottkw{blame} \, \ottnt{p} $; and
  (3) $\ottnt{M} \,  \mathord{\Uparrow_{\mathsf{S} } } $ iff $ \mathscr{C}\llbracket \maybebluetext{ \ottnt{M} } \rrbracket  \,  \mathord{\Uparrow_{\mathsf{S_1} } } $.
\end{theorem}
\fi

To prove this theorem, it suffices to show the left-to-right direction
(Theorem~\ref{thm:trans-soundness} below) for each item because the
other direction follows from Theorem~\ref{thm:trans-soundness} together with other properties: for
example, if $  \emptyset     \vdash_{\mathsf{S} }    \ottnt{M}  :  \iota $ and
$ \mathscr{C}\llbracket \maybebluetext{ \ottnt{M} } \rrbracket  \,  \mathord{\Uparrow_{\mathsf{S_1} } } $, then $\ottnt{M}$ can neither get stuck (by
type soundness of \lamS) nor terminate (as it contradicts the
left-to-right direction and the fact that $ \evalto_{\mathsf{S_1} } $ is
deterministic).

\ifluxuryspace
\begin{theorem}[name=Translation Soundness,restate=thmTransSoundness] \label{thm:trans-soundness} Suppose $ \Gamma    \vdash_{\mathsf{S} }    \ottnt{M}  :  \ottnt{A} $.
  \begin{enumerate}
  \item If $ \ottnt{M}    \mathbin{  \evalto_{\mathsf{S} }  ^*}    \ottnt{V} $, then $  \mathscr{C}\llbracket \maybebluetext{ \ottnt{M} } \rrbracket     \mathbin{  \evalto_{\mathsf{S_1} }  ^*}     \Psi (\maybebluetext{ \ottnt{V} })  $.
  \item If $ \ottnt{M}    \mathbin{  \evalto_{\mathsf{S} }  ^*}    \ottkw{blame} \, \ottnt{p} $, then $  \mathscr{C}\llbracket \maybebluetext{ \ottnt{M} } \rrbracket     \mathbin{  \evalto_{\mathsf{S_1} }  ^*}    \ottkw{blame} \, \ottnt{p} $.
  \item If $\ottnt{M} \,  \mathord{\Uparrow_{\mathsf{S} } } $, then $ \mathscr{C}\llbracket \maybebluetext{ \ottnt{M} } \rrbracket  \,  \mathord{\Uparrow_{\mathsf{S_1} } } $.
  \end{enumerate}
\end{theorem}
\else
\begin{theorem}[name=Translation Soundness,restate=thmTransSoundness] \label{thm:trans-soundness}
  Suppose $ \Gamma    \vdash_{\mathsf{S} }    \ottnt{M}  :  \ottnt{A} $.  (1) If $ \ottnt{M}    \mathbin{  \evalto_{\mathsf{S} }  ^*}    \ottnt{V} $, then
  $  \mathscr{C}\llbracket \maybebluetext{ \ottnt{M} } \rrbracket     \mathbin{  \evalto_{\mathsf{S_1} }  ^*}     \Psi (\maybebluetext{ \ottnt{V} })  $; (2) if $ \ottnt{M}    \mathbin{  \evalto_{\mathsf{S} }  ^*}    \ottkw{blame} \, \ottnt{p} $, then
  $  \mathscr{C}\llbracket \maybebluetext{ \ottnt{M} } \rrbracket     \mathbin{  \evalto_{\mathsf{S_1} }  ^*}    \ottkw{blame} \, \ottnt{p} $; and (3) if $\ottnt{M} \,  \mathord{\Uparrow_{\mathsf{S} } } $, then
  $ \mathscr{C}\llbracket \maybebluetext{ \ottnt{M} } \rrbracket  \,  \mathord{\Uparrow_{\mathsf{S_1} } } $.
\end{theorem}
\fi

A standard proof strategy would be to show that single-step evaluation in the source language
is simulated by multi-step evaluation in the target language.  In fact, we prove
the following lemma:
\begin{lemma}[name=Simulation,restate=lemTransEval]\label{lem:trans-eval}\leavevmode
  \begin{enumerate}
  \item If $ \ottnt{M}    \mathbin{  \accentset{\mathsf{e} }{\evalto}_{\mathsf{S} }  }    \ottnt{N} $, then
    $ \mathscr{C}\llbracket \maybebluetext{ \ottnt{M} } \rrbracket   \mathbin{  \accentset{\mathsf{e} }{\evalto}_{\mathsf{S_1} }     \accentset{\mathsf{c} }{\evalto}_{\mathsf{S_1} }  ^*}   \mathscr{C}\llbracket \maybebluetext{ \ottnt{N} } \rrbracket $.
  \item If $ \ottnt{M}    \mathbin{  \accentset{\mathsf{c} }{\evalto}_{\mathsf{S} }  }    \ottnt{N} $, then $  \mathscr{C}\llbracket \maybebluetext{ \ottnt{M} } \rrbracket     \mathbin{  \accentset{\mathsf{c} }{\evalto}_{\mathsf{S_1} }  ^+}     \mathscr{C}\llbracket \maybebluetext{ \ottnt{N} } \rrbracket  $.
  \end{enumerate}
\[
\xymatrix@=50pt{
    \ottnt{M} \ar@{|->}[rr]^{\mathsf{e}}_>>{\mathsf{S}} \ar[d]^{\mathscr{C} \llbracket \_  \rrbracket}& & \ottnt{N} \ar[d]^{\mathscr{C} \llbracket \_  \rrbracket} \\
     \mathscr{C}\llbracket \maybebluetext{ \ottnt{M} } \rrbracket  \ar@{|.>}[r]^-{\mathsf{e}} _>>{\mathsf{S}_1} & \ar@{|.>}[r]^-{\mathsf{c}} ^>>{*}_>>{\mathsf{S}_1} &  \mathscr{C}\llbracket \maybebluetext{ \ottnt{N} } \rrbracket 
}
\qquad
\xymatrix@=50pt{
    \ottnt{M} \ar@{|->}[r]^{\mathsf{c}}_>>{\mathsf{S}} \ar[d]^{\mathscr{C} \llbracket \_  \rrbracket} & \ottnt{N} \ar[d]^{\mathscr{C} \llbracket \_  \rrbracket} \\
     \mathscr{C}\llbracket \maybebluetext{ \ottnt{M} } \rrbracket  \ar@{|.>}[r]^{\mathsf{c}} ^>>{+}_>>{\mathsf{S}_1} &   \mathscr{C}\llbracket \maybebluetext{ \ottnt{N} } \rrbracket 
}
\]
\end{lemma}

The straightforward simulation property below follows from Lemma~\ref{lem:trans-eval}.
\begin{lemma}[restate=lemSimulation,name=]\label{lem:trans-reduction-preserve}
  If $ \ottnt{M}    \mathbin{  \evalto_{\mathsf{S} }  }    \ottnt{N} $, then $  \mathscr{C}\llbracket \maybebluetext{ \ottnt{M} } \rrbracket     \mathbin{  \evalto_{\mathsf{S_1} }  ^+}     \mathscr{C}\llbracket \maybebluetext{ \ottnt{N} } \rrbracket  $.
\end{lemma}

As is the case for simulation proofs for CPS
translation~\cite{DBLP:journals/tcs/Plotkin75,
  DBLP:books/cu/Appel1992, DBLP:conf/mfps/Wand91,
  DBLP:journals/lisp/SabryF93, DBLP:journals/tcs/DanvyN03,
  DBLP:journals/mscs/DanvyF92, DBLP:journals/toplas/SabryW97}, the
simulation property\footnote{%
  If we had been interested only in the property that translation preserves term
  equivalence, we could have simplified the technical development by,
  say, removing the distinction between $\ottnt{U}  \langle  \ottnt{s}  \rangle$ and $\ottnt{U}  \langle\!\langle  \ottnt{s}  \rangle\!\rangle$.
  However, simulation is crucial for showing that divergence is preserved
  by the translation.  } is quite subtle. We discuss this subtlety below.

\begin{sloppypar}
First, it is important that the translation removes administrative
identity coercions by distinguishing values and nonvalues in
$ \mathscr{C}\llbracket \maybebluetext{ \ottnt{M} } \rrbracket $.  For example, $ \ottsym{(}   \lambda   \ottmv{x} .\,  \ottmv{x}   \ottsym{)} \, \ottsym{5}    \mathbin{  \accentset{\mathsf{e} }{\evalto}  }    \ottsym{5} $ holds in \lamS, but
the translation $ \mathscr{K}'\llbracket \maybebluetext{ \ottsym{(}   \lambda   \ottmv{x} .\,  \ottmv{x}   \ottsym{)} \, \ottsym{5} } \rrbracket  \ottnt{K} $ without removing administrative
redexes would yield $ \ottsym{(}  \ottsym{(}   \lambda  ( \ottmv{x} , \kappa ).\,  \ottmv{x}   \langle  \kappa  \rangle  \ottsym{)}  \langle  \ottkw{id}  \rangle  \ottsym{)} \, ( \ottsym{5}  \langle  \ottkw{id}  \rangle , \ottnt{K} ) $, which
performs c-evaluation before calling the function.
We avoid such a situation. More formally, we
prove the following lemma, which means the redex in the source is also
the redex in the target.
\end{sloppypar}
\begin{lemma}[restate=lemTransCtx,name=]\label{lem:trans-ctx}\leavevmode
  \begin{enumerate}
  \item For any $\mathcal{F}$, there exists $\mathcal{E}'$
    such that for any $\ottnt{M}$, $ \mathscr{C}\llbracket \maybebluetext{ \mathcal{F}  [  \ottnt{M}  ] } \rrbracket   \ottsym{=}  \mathcal{E}'  [   \mathscr{C}\llbracket \maybebluetext{ \ottnt{M} } \rrbracket   ]$.
  \item For any $\mathcal{F}$ and $\ottnt{s}$, there exists $\mathcal{E}'$ such that
    for any $\ottnt{M}$, $ \mathscr{C}\llbracket \maybebluetext{ \mathcal{F}  [  \ottnt{M}  \langle  \ottnt{s}  \rangle  ] } \rrbracket   \ottsym{=}  \mathcal{E}'  [   \mathscr{K}\llbracket \maybebluetext{ \ottnt{M} } \rrbracket   \Psi (\maybebluetext{ \ottnt{s} })    ]$.
  \end{enumerate}
\end{lemma}

To prove this lemma, the rule \rnp{TrC-Crc} also plays an important
role: for example, if we removed \rnp{TrC-Crc},
$ \mathscr{K}\llbracket \maybebluetext{ \ottsym{(}  \ottsym{1}  \ottsym{+}  \ottsym{1}  \ottsym{)}  \langle   \ottkw{int} \texttt{!}   \rangle } \rrbracket  \ottkw{id} $ would translate
to $ \ottkw{let} \,  \kappa =  \ottkw{int} \texttt{!}   \mathbin{;\!;}  \ottkw{id} \, \ottkw{in}\,  \ottsym{(}  \ottsym{1}  \ottsym{+}  \ottsym{1}  \ottsym{)}  \langle  \kappa  \rangle $, which performs
c-evaluation before adding 1 and 1, which is the first thing the
original term $\ottsym{(}  \ottsym{1}  \ottsym{+}  \ottsym{1}  \ottsym{)}  \langle   \ottkw{int} \texttt{!}   \rangle$ will do.

\begin{sloppypar}
Second,
optimizing too many (identity) coercions can break simulation.
We should only remove administrative identity coercions, and
keep identity coercions that were present in the original term.
Consider
$\ottnt{M} \defeq \ottsym{(}  \ottsym{(}  \ottsym{(}   \lambda   \ottmv{x} .\,  \ottnt{M_{{\mathrm{1}}}}   \ottsym{)}  \langle\!\langle   \ottkw{id} _{ \iota }   \rightarrow   \iota \texttt{!}   \rangle\!\rangle  \ottsym{)} \, \ottnt{a}  \ottsym{)}  \langle   \iota \texttt{?}^{ \ottnt{p} }   \rangle$ and
$\ottnt{N} \defeq \ottsym{(}  \ottsym{(}   \lambda   \ottmv{x} .\,  \ottnt{M_{{\mathrm{1}}}}   \ottsym{)} \, \ottsym{(}  \ottnt{a}  \langle   \ottkw{id} _{ \iota }   \rangle  \ottsym{)}  \ottsym{)}  \langle   \iota \texttt{!}   \rangle  \langle   \iota \texttt{?}^{ \ottnt{p} }   \rangle$, for which
$ \ottnt{M}    \mathbin{  \evalto_{\mathsf{S} }  }    \ottnt{N} $ holds by \rnp{R-Wrap}.  Then,
\[
\begin{array}{lll}
 \mathscr{C}\llbracket \maybebluetext{ \ottnt{M} } \rrbracket  &=&  \mathscr{K}\llbracket \maybebluetext{ \ottnt{M} } \rrbracket  \ottkw{id}  =  \ottsym{(}  \ottsym{(}   \mathscr{K}\llbracket \maybebluetext{  \lambda  ( \ottmv{x} , \kappa ).\,  \ottnt{M_{{\mathrm{1}}}}  } \rrbracket  \kappa   \ottsym{)}  \langle\!\langle   \ottkw{id} _{ \iota }   \Rightarrow   \iota \texttt{!}   \rangle\!\rangle  \ottsym{)} \, ( \ottnt{a} ,  \iota \texttt{?}^{ \ottnt{p} }  )  \\
           & \evalto_{\mathsf{S_1} } &   \ottkw{let} \,  \kappa' =  \iota \texttt{!}   \mathbin{;\!;}   \iota \texttt{?}^{ \ottnt{p} }  \, \ottkw{in}\,  \ottsym{(}   \mathscr{K}\llbracket \maybebluetext{  \lambda  ( \ottmv{x} , \kappa ).\,  \ottnt{M_{{\mathrm{1}}}}  } \rrbracket  \kappa   \ottsym{)}  \, ( \ottnt{a}  \langle   \ottkw{id} _{ \iota }   \rangle , \kappa' ) 
           =  \mathscr{C}\llbracket \maybebluetext{ \ottnt{N} } \rrbracket .
\end{array}
\]
At one point, we defined the translation (let's call it $\mathscr{K}''$)
so that applications of identity coercions would be removed as much as
possible, namely,
\[
   \mathscr{K}''\llbracket \maybebluetext{ \ottnt{N} } \rrbracket  \ottkw{id}  =   \ottkw{let} \,  \kappa' =  \iota \texttt{!}   \mathbin{;\!;}   \iota \texttt{?}^{ \ottnt{p} }  \, \ottkw{in}\,  \ottsym{(}   \mathscr{K}''\llbracket \maybebluetext{  \lambda  ( \ottmv{x} , \kappa ).\,  \ottnt{M_{{\mathrm{1}}}}  } \rrbracket  \kappa   \ottsym{)}  \, ( \ottnt{a} , \kappa' ) 
\]
(notice that $\langle   \ottkw{id} _{ \iota }   \rangle$ on $\ottnt{a}$ is removed).  Although $ \mathscr{K}''\llbracket \maybebluetext{ \ottnt{M} } \rrbracket  \ottkw{id} $ and
$ \mathscr{K}''\llbracket \maybebluetext{ \ottnt{N} } \rrbracket  \ottkw{id} $ reduced to the same term, we did not quite have $  \mathscr{K}''\llbracket \maybebluetext{ \ottnt{M} } \rrbracket  \ottkw{id}     \mathbin{ \evalto ^+}     \mathscr{K}''\llbracket \maybebluetext{ \ottnt{N} } \rrbracket  \ottkw{id}  $ as
we had desired.

Third, the distinction between $\ottnt{U}  \langle  \ottnt{s}  \rangle$ and $\ottnt{U}  \langle\!\langle  \ottnt{s}  \rangle\!\rangle$ is crucial
for ensuring that substitution commutes with the translation:
\begin{lemma}[name=Substitution,restate=lemSubstTrans]\label{lem:subst-trans}
  If $ \kappa \notin   \metafun{FV} ( \ottnt{M} )   \cup  \metafun{FV} ( \ottnt{V} ) $, then
  $\ottsym{(}   \mathscr{K}\llbracket \maybebluetext{ \ottnt{M} } \rrbracket  \kappa   \ottsym{)}  [  \ottmv{x}  \ottsym{:=}   \Psi (\maybebluetext{ \ottnt{V} })   \ottsym{,}  \kappa  \ottsym{:=}  \ottnt{K}  ]  \ottsym{=}   \mathscr{K}\llbracket \maybebluetext{ \ottnt{M}  [  \ottmv{x}  \ottsym{:=}  \ottnt{V}  ] } \rrbracket  \ottnt{K} $.
\end{lemma}
Roughly speaking, if we identified a value $\ottnt{U}  \langle\!\langle  \ottnt{s}  \rangle\!\rangle$ and an
application $\ottnt{U}  \langle  \ottnt{s}  \rangle$ of $\ottnt{s}$ to an uncoerced value $\ottnt{U}$, then the term
$\ottnt{U}  \langle  \ottnt{s}  \rangle  \langle  \ottnt{t}  \rangle$ would allow two interpretations: an application of $\ottnt{t}$ to
a value $\ottnt{U}  \langle  \ottnt{s}  \rangle$ and applications of $\ottnt{s}$ and $\ottnt{t}$ to $\ottnt{U}$
and committing to either interpretation would break Lemma~\ref{lem:subst-trans}.
\end{sloppypar}

\section{Implementation and Evaluation}
\label{sec:implementation}

\subsection{Implementation}
We have implemented the coercion-passing translation described in
Section~\ref{sec:translation} and the semantics of \lamSx for
Grift~\citep{DBLP:conf/pldi/KuhlenschmidtAS19}\footnote{%
  The semantics of coercions in Grift is so-called
  D~\cite{DBLP:conf/esop/SiekGT09}, which is slightly different from
  that of \lamSx, which is UD.  Since the main difference is in the
  coercion composition, our technique can be applied to Grift.}, %
an experimental compiler for gradually typed languages.  GTLC+, the
language that the Grift compiler implements, supports integers,
floating-point numbers, Booleans, higher-order functions, local
binding by \texttt{let}, (mutually) recursive definitions by
\texttt{letrec}, conditional expressions, iterations, sequencing,
mutable references, and vectors (mutable arrays).

The Grift compiler compiles a GTLC+ program into the C language
where coercions are represented as values of a \texttt{struct} type, and
operations such as coercion application and coercion composition are C functions.
The compiler supports different run-time check schemes, those based on
type-based casts~\cite{conf/scheme/SiekT06} and space-efficient
coercions~\cite{DBLP:conf/pldi/SiekTW15}.
Note that, although meta-level composition $\ottnt{s_{{\mathrm{1}}}}  \fatsemi  \ottnt{s_{{\mathrm{2}}}}$ is implemented, only
nested coercions on \emph{values} are composed; in other words,
\rnp{R-MergeC} was not implemented.  Thus, implicit run-time checks may
break tail calls and seemingly tail-recursive functions may cause
stack overflow.

We modify the compiler phases for run-time checking based on
the space-efficient coercions.  After typechecking a user program, the
compiler inserts type-based casts to the program and converts
type-based casts to space-efficient coercions, following the
translation from blame calculus \lamB to \lamS~\cite{DBLP:conf/pldi/SiekTW15}.  Our
implementation performs the coercion-passing translation after the translation into \lamS.  It is
straightforward to extend the translation scheme to language features
that are not present in \lamS.  For example, here is translation for
conditional expressions:
\[
 \mathscr{K}\llbracket \maybebluetext{  \textsf{if }  \ottnt{M}  \textsf{ then }  \ottnt{N_{{\mathrm{1}}}}  \textsf{ else }  \ottnt{N_{{\mathrm{2}}}}  } \rrbracket  \ottnt{K}  =  \textsf{if }   \mathscr{C}\llbracket \maybebluetext{ \ottnt{M} } \rrbracket   \textsf{ then }  \ottsym{(}   \mathscr{K}\llbracket \maybebluetext{ \ottnt{N_{{\mathrm{1}}}} } \rrbracket  \ottnt{K}   \ottsym{)}  \textsf{ else }  \ottsym{(}   \mathscr{K}\llbracket \maybebluetext{ \ottnt{N_{{\mathrm{2}}}} } \rrbracket  \ottnt{K}   \ottsym{)} .
\]
Since coercions are represented as \texttt{struct}s, we did not have
to do anything special to make coercions first-class.  We
modify another compiler phase that generates operations on coercions
such as $\ottnt{M}  \mathbin{;\!;}  \ottnt{N}$ and \rnp{R-Wrap}.  The current implementation, which
generates C code and uses \texttt{clang}\footnote{\url{https://clang.llvm.org/}} for compilation to machine code,
relies on the C compiler to perform tail-call optimizations.  We have
found the original compiler's handling of recursive types hampers
tail-call optimizations,\footnote{%
  The C function to compose coercions takes a pointer to a \emph{stack-allocated} object
  as an argument and writes into the object when recursive coercions are composed.
  Although those stack-allocated objects never escape and tail-call optimization is safe,
  the C compiler is not powerful enough to see it.
} so our implementation does not deal
with recursive types.  We leave their implementation for future work.

\subsection{Even and Odd Functions}

We first inspected the tail-recursive even--odd functions in GTLC+:
\begin{lstlisting}[language=Grift]
(letrec ([even (lambda ([n : $\ottnt{A_{{\mathrm{1}}}}$]) : $\ottnt{A_{{\mathrm{3}}}}$
           (if (= 0 n) #t (odd (- n 1))))]
          [odd (lambda ([n : $\ottnt{A_{{\mathrm{2}}}}$]) : $\ottnt{A_{{\mathrm{4}}}}$
           (if (= 0 n) #f (even (- n 1))))])
  (odd $\ottmv{n}$))
\end{lstlisting}
where $\ottnt{A_{{\mathrm{1}}}}$ and $\ottnt{A_{{\mathrm{2}}}}$ are either \texttt{Int} or \texttt{Dyn},
and $\ottnt{A_{{\mathrm{3}}}}$ and $\ottnt{A_{{\mathrm{4}}}}$ are either \texttt{Bool} or \texttt{Dyn}.
We run this program with the original and modified compilers
for all combinations of $\ottnt{A_{{\mathrm{1}}}},\ottnt{A_{{\mathrm{2}}}},\ottnt{A_{{\mathrm{3}}}}$, and $\ottnt{A_{{\mathrm{4}}}}$.
We call the program compiled by the original compiler Base,
the program compiled by the modified compiler CrcPS.

We have confirmed that, as $\ottmv{n}$ increases, 12 of 16
configurations of Base cause stack overflow.\footnote{%
  The size of the run-time stack is 8 MB.}  In the four configurations
that survived, both $\ottnt{A_{{\mathrm{3}}}}$ and $\ottnt{A_{{\mathrm{4}}}}$ are set to \texttt{Bool}.
CrcPS \emph{never} causes stack overflow for any configuration.

Although we expected that Base would crash if $\ottnt{A_{{\mathrm{3}}}}$ and $\ottnt{A_{{\mathrm{4}}}}$
are different, it is our surprise that Base causes stack overflow even
when $\ottnt{A_{{\mathrm{3}}}} = \ottnt{A_{{\mathrm{4}}}} = \texttt{Dyn}$.  We have found that it is due
to the typing rule of Grift for conditional expressions.  In Grift, if
one of the branches is given a static type, say \texttt{Bool}, and the
other is \texttt{Dyn}, the whole \texttt{if}-expression is given the
static type and the compiler put a cast from \texttt{Dyn} on the
branch of type \texttt{Dyn}.  If both $\ottnt{A_{{\mathrm{3}}}}$ and $\ottnt{A_{{\mathrm{4}}}}$ are
\texttt{Dyn}, the recursive calls in the two else-branches will
involve casts $ \ottkw{bool} \texttt{?}^{ \ottnt{p} } $ from \texttt{Dyn} to \texttt{Bool} because
the two then-branches are Boolean constants and the if-expressions are given type \texttt{Bool}.
However, since the
return types are declared to be \texttt{Dyn}, the whole
\texttt{if}-expressions are cast back to \texttt{Dyn}, inserting
injections $ \ottkw{bool} \texttt{!} $.  Thus, every recursive call involves a
projection immediately followed by an injection, as shown below,
eventually causing stack overflow.
\begin{lstlisting}[language=Grift]
(letrec ([even (lambda ([n : Dyn]) : Dyn
           (if (= 0 n$\langle   \ottkw{int} \texttt{?}^{ \ottnt{p_{{\mathrm{1}}}} }   \rangle$) #t
               (odd (- n$\langle   \ottkw{int} \texttt{?}^{ \ottnt{p_{{\mathrm{2}}}} }   \rangle$ 1))$\langle   \ottkw{bool} \texttt{?}^{ \ottnt{p_{{\mathrm{3}}}} }   \rangle$)$\langle   \ottkw{bool} \texttt{!}   \rangle$)]
          [odd (lambda ([n : Dyn) : Dyn
           (if (= 0 n$\langle   \ottkw{int} \texttt{?}^{ \ottnt{p_{{\mathrm{4}}}} }   \rangle$) #f
               (even (- n$\langle   \ottkw{int} \texttt{?}^{ \ottnt{p_{{\mathrm{5}}}} }   \rangle$ 1))$\langle   \ottkw{bool} \texttt{?}^{ \ottnt{p_{{\mathrm{6}}}} }   \rangle$)$\langle   \ottkw{bool} \texttt{!}   \rangle$))])
  (odd $\ottmv{n}$))
\end{lstlisting}

\subsection{Evaluation}

We have conducted some experiments to measure the overhead of
the coercion-passing style translation.  The benchmark programs we have used
are taken from Kuhlenschmidt et al.~\cite{DBLP:conf/pldi/KuhlenschmidtAS19}\footnote{\url{https://github.com/Gradual-Typing/benchmarks}};
we excluded the sieve program because of the use of recursive types.
We also include the even/odd program only for reference,
which is relatively small compared to other programs.

We compare the running time of a benchmark program between Base and CrcPS.
To take many partially typed configurations for each benchmark program into account,
we focus on the so-called \emph{fine-grained} approach,
where everywhere a type is required is given either the dynamic type \texttt{Dyn}
or an appropriate static type.\footnote{%
  The other approach is called \emph{coarse-grained},
  where functions in each module are all statically or all dynamically typed.
}
In the fine-grained approach, the number of configurations is $2^n$
where $n$ is the number of type annotations.
When this number is very large, we consider uniformly sampled configurations.
We use the sampling algorithm\footnote{\url{https://github.com/Gradual-Typing/Dynamizer}}
from \citep{DBLP:conf/pldi/KuhlenschmidtAS19}.

We describe the (sampled) number of partially typed configurations and
main language features used for each benchmark program below.
(Each benchmark program has one additional type annotation for the return type
of the 0-ary main function.)  For more detailed description of
benchmark programs, we refer readers to Kuhlenschmidt et
al.~\cite{DBLP:conf/pldi/KuhlenschmidtAS19}.

\begin{center}
\begin{tabular}{lll}
name & \# of configurations & description \\ \hline
even--odd & all $32 = 2^5$ &  mutually tail-recursive functions \\
n-body & 300 out of $2^{136}$ & vectors \\  
tak & all $256 = 2^8$  & recursive function \\
ray & 300 out of $2^{280}$  & tuples and iterations \\
blackscholes & 300 out of $2^{128}$  & vectors and iterations \\
matmult & 300 out of $2^{33}$  & vectors and iterations \\
quicksort & 300 out of $2^{44}$  & vectors \\
fft & 300 out of $2^{67}$ & vectors
\end{tabular}
\end{center}

Our benchmark method is as follows:
For each partially typed configuration of a benchmark program,
we measure its running time by taking the average of 5 runs for Base and CrcPS,
and compute the ratio of CrcPS to Base.
We use a machine with a 8-core 3.6 GHz Intel Core i7-7700 and 16 GB memory,
and run the benchmark programs within a Docker 
container (Docker version 19.03.5) which runs Arch Linux.
The generated C code is compiled by clang version
9.0.0 with \texttt{-O3} so that tail-call optimization is applied.
The size of the run-time stack is set as unlimited.

\begin{figure}[t]
  \centering
  \includegraphics[draft=false,width=\columnwidth]{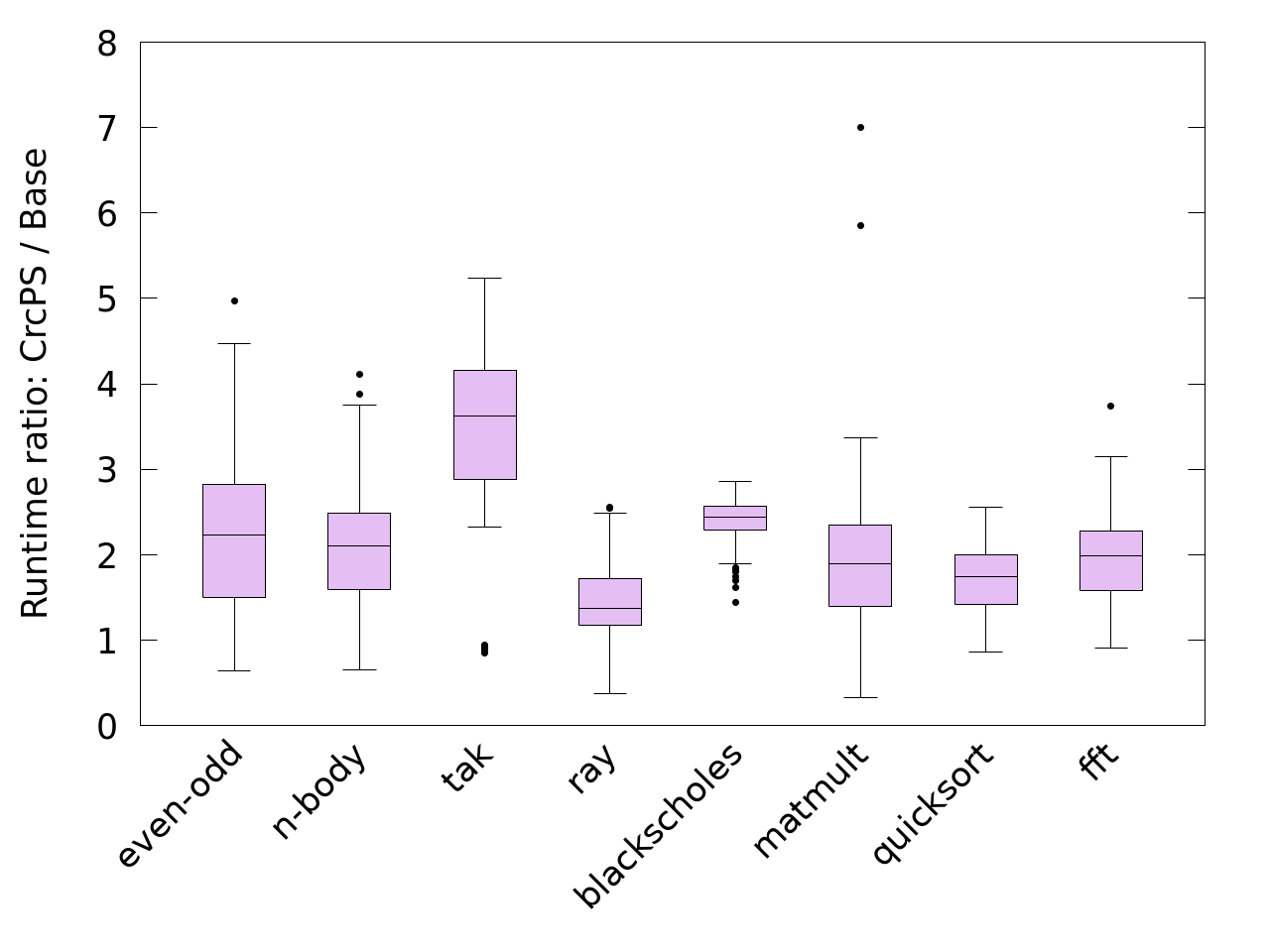}
  \caption{A box plot for the running time ratios of CrcPS to Base
    across (sampled) partially typed configurations of the benchmark programs.
    (As is standard,
    the lower/upper end of a box indicates the first/third quartile, respectively,
    and the middle line in a box indicates the median.
    The length of each whisker is below 1.5 times of interquartile range, and
    outliers are plotted individually.)
  }
  \label{fig:boxplot}
\end{figure}

Figure~\ref{fig:boxplot} shows the result in box plots.
\iffull
(Detailed plots for each benchmark are shown in Appendix~\ref{sec:detailedbenchmark}.)
\fi
It shows that, except for tak (and even--odd),
practical programs in CrcPS run up to three times as slow as Base, for most configurations.
It is natural because coercion-passing style translation adds an extra coercion argument
to each function.
In fact, tak and even--odd, which have a lot of function calls, have large overhead
compared with other programs.
In even--odd, CrcPS performs many coercion composition operations (and one coercion application)
while Base performs many coercion applications (without any coercion composition).\footnote{%
  An application of a projection coercion to an injected value
  is always computed by coercion composition in CrcPS,
  while the implementation of Base is slightly optimized for first-order types.
}
Thus, the difference between Base and CrcPS for even--odd is partially due to the difference
of the cost of coercion application and coercion composition.

The benchmark programs other than tak and even--odd
mainly concern vectors and iterations over them.
Vector operations are treated in the translation as primitive operations, which
we consider do not have much overhead by the translation.
In fact, our translation implementation optimizes the rule \rnp{Tr-Op}
when its continuation is $\ottkw{id}$: $ \mathscr{K}\llbracket \maybebluetext{ \ottnt{op}  \ottsym{(}  \ottnt{M}  \ottsym{,}  \ottnt{N}  \ottsym{)} } \rrbracket  \ottkw{id}   \ottsym{=}  \ottnt{op}  \ottsym{(}   \mathscr{C}\llbracket \maybebluetext{ \ottnt{M} } \rrbracket   \ottsym{,}   \mathscr{C}\llbracket \maybebluetext{ \ottnt{N} } \rrbracket   \ottsym{)}$
without an application of an identity coercion.

There are several configurations in which CrcPS is faster than Base but
we have not figured out why this is the case.

\section{Related Work}
\label{sec:related}

\subsection{Space-Efficient Coercion/Cast Calculi}
\begin{sloppypar}
As we have already mentioned, it is fairly well known that
coercions~\cite{DBLP:journals/scp/Henglein94} and
casts~\cite{DBLP:conf/esop/WadlerF09} hamper tail-call
optimization and make the space complexity of the execution of a
program worse than the execution under an unchecked semantics.
We discuss below a few pieces of work~\cite{%
  DBLP:conf/sfp/HermanTF07,DBLP:journals/lisp/HermanTF10,
  DBLP:conf/esop/SiekGT09,
  DBLP:conf/popl/SiekW10,DBLP:conf/icfp/Garcia13,DBLP:conf/pldi/SiekTW15}
addressing the problem.
\end{sloppypar}

To the best of our knowledge, Herman et
al.~\cite{DBLP:conf/sfp/HermanTF07,DBLP:journals/lisp/HermanTF10} were
the first to observe the space-efficiency problem of inserted dynamic
checks.  They developed a variant of Henglein's coercion calculus with
semantics such that a sequence of coercion applications is eagerly
composed to reduce the size of coercions.  However, they identified
two coercions $\ottsym{(}  \ottnt{c_{{\mathrm{1}}}}  \ottsym{;}  \ottnt{c_{{\mathrm{2}}}}  \ottsym{)}  \ottsym{;}  \ottnt{c_{{\mathrm{3}}}}$ and $\ottnt{c_{{\mathrm{1}}}}  \ottsym{;}  \ottsym{(}  \ottnt{c_{{\mathrm{2}}}}  \ottsym{;}  \ottnt{c_{{\mathrm{3}}}}  \ottsym{)}$ (note that
$\ottnt{c_{{\mathrm{1}}}}  \ottsym{;}  \ottnt{c_{{\mathrm{2}}}}$ is not a meta-level operator but only a \emph{formal}
composition constructor); thus, an algorithm for computing coercion
composition was not very clear.  They did not take blame
tracking~\cite{DBLP:conf/icfp/FindlerF02} into account, either.

Later, Siek et al.~\cite{DBLP:conf/esop/SiekGT09} extended
Herman et al.~\cite{DBLP:conf/sfp/HermanTF07,DBLP:journals/lisp/HermanTF10} with a
few different blame tracking strategies.  The issue of identifying
$\ottsym{(}  \ottnt{c_{{\mathrm{1}}}}  \ottsym{;}  \ottnt{c_{{\mathrm{2}}}}  \ottsym{)}  \ottsym{;}  \ottnt{c_{{\mathrm{3}}}}$ and $\ottnt{c_{{\mathrm{1}}}}  \ottsym{;}  \ottsym{(}  \ottnt{c_{{\mathrm{2}}}}  \ottsym{;}  \ottnt{c_{{\mathrm{3}}}}  \ottsym{)}$ remained.  According to
their terminology, our work, which follows previous
work~\cite{DBLP:conf/pldi/SiekTW15}, adopts the UD semantics, which
allows only $\mathord{\star}  \rightarrow  \mathord{\star}$ as a tag to functional values, as opposed to
the D semantics, which allows any function types to be used as a tag.

Siek and Wadler~\cite{DBLP:conf/popl/SiekW10} introduced threesomes to
a blame calculus as another solution to the space-efficiency problem.
Threesome casts have a third type (called a mediating type) in
addition to the source and target types; a threesome cast is
considered a downcast from the source type to the mediating, followed
by an upcast from the mediating type to the target.  Threesome casts
allow a simple recursive algorithm to compose two threesome casts but
blame tracking is rather complicated.

Garcia~\cite{DBLP:conf/icfp/Garcia13} gave a translation from coercion
calculi to threesome calculi and show that the two solutions to the
space-efficiency problem are equivalent in some sense.  He introduced
supercoercions and a recursive algorithm to compute composition of
supercoercions but they were complex, too.

Siek et al.~\cite{DBLP:conf/pldi/SiekTW15} proposed yet another space-efficient
coercion calculus \lamS, in which they succeeded in developing a
simple recursive algorithm for coercion composition by restricting
coercions to be in certain canonical forms---what they call
space-efficient coercions.  They also gave a translation from blame
calculus \lamB to \lamS (via Henglein's coercion calculus \lamC) and showed
that the translation is fully abstract.  As we have discussed already,
our \lamS has introduced syntax that distinguishes an application
$\ottnt{U}  \langle  \ottnt{s}  \rangle$ of a coercion to (uncoerced) values from $\ottnt{U}  \langle\!\langle  \ottnt{d}  \rangle\!\rangle$ for a
value wrapped by a delayed coercion.  Such distinction, which can be seen
in some blame calculi~\cite{DBLP:conf/esop/WadlerF09}, is not just
an aesthetic choice but crucial for proving correctness of the translation.

All the above-mentioned calculi adopt a nonstandard reduction rule to
compose coercions or casts even before the subject evaluates to a
value, together with a nonstandard form of evaluation contexts, and as a result it
has not been clear how to implement them efficiently.
Herman et al.~\cite{DBLP:conf/sfp/HermanTF07,DBLP:journals/lisp/HermanTF10}
sketched a few possible implementation strategies, including coercion
passing, but details were not discussed.
Siek and Garcia~\cite{DBLP:conf/icfp/SiekG12} showed an interpreter which performs
coercion composition at tail calls.  Although not showing correctness
of the interpreter, their interpreter would give a hint to direct
low-level implementation of space-efficient coercions.  Our work
addresses the problem of the nonstandard semantics in a different
way---by translating a program into coercion-passing style.  The
difference, however, may not be so large as it may appear at first: in
Siek and Garcia~\cite{DBLP:conf/icfp/SiekG12}, a state of the abstract machine
includes an evaluation context, which contains the information on a
coercion to be applied to a return value and such a coercion roughly
corresponds to our continuation coercions.  More detailed analysis of
the relationship between the two implementation schemes is left for
future work.

Kuhlenschmidt et al.~\cite{DBLP:conf/pldi/KuhlenschmidtAS19} built an experimental compiler
Grift for gradual typing with structural types.  It supports run-time
checking with the space-efficient coercions of \lamS but does not support
composition of coercions at tail positions.  We have implemented
our coercion-passing translation for the Grift compiler.

Greenberg~\cite{DBLP:conf/popl/Greenberg15} has studied the same
space-efficiency problem in the context of manifest contract
calculi~\cite{DBLP:journals/toplas/KnowlesF10,
  DBLP:conf/popl/GreenbergPW10, DBLP:journals/jfp/GreenbergPW12} and
proposed a few semantics for composing casts that involve contract
checking.  Feltey et al.~\cite{DBLP:journals/pacmpl/FelteyGSFS18} recently
implemented Greenberg's eidetic contracts on top of Typed
Racket~\cite{DBLP:conf/popl/Tobin-HochstadtF08} but, similarly to
Kuhlenschmidt et al.~\cite{DBLP:conf/pldi/KuhlenschmidtAS19}, composition is limited to a
sequence of contracts applied to values.

There is other recent work for making gradual typing
efficient~\cite{DBLP:journals/pacmpl/BaumanBST17,
  DBLP:journals/pacmpl/MuehlboeckT17,
  DBLP:journals/pacmpl/RichardsAT17,DBLP:conf/popl/RastogiSFBV15} but
as far as we know, none of them addresses the problem caused by
run-time checking applied to tail positions.
Additionally, Castagna et al.~\cite{conf/ifl/CastagnaDLS19} implemented a virtual machine
for space-efficient gradual typing in presence of set-theoretic types,
but without blame tracking.
They address the problem caused by casts applied to tail positions
by an approach similar to the one in the interpreter by Siek and Garcia~\cite{DBLP:conf/icfp/SiekG12}.
They implemented their virtual machine and evaluated their implementation
by benchmarks such as the even--odd functions.

\subsection{Continuation-Passing Style}

Our coercion-passing style translation is inspired by
continuation-passing style translation, first formalized
by Plotkin~\cite{DBLP:journals/tcs/Plotkin75}.  However, coercions represent
only a part of the rest of computation and are, in this sense, closer
to delimited continuations~\cite{DBLP:conf/lfp/DanvyF90}.  Roughly
speaking, translating a subexpression with $\ottkw{id}$ corresponds to the
reset operation~\cite{DBLP:conf/lfp/DanvyF90} to delimit
continuations.  Unlike (delimited) continuations, which are usually
expressed by first-class functions, coercions have compact
representations and compactness can be preserved by composition.

Wallach and Felten~\cite{DBLP:conf/sp/WallachF98} proposed security-passing style to
implement Java stack inspection~\cite{JVMspec}.  The idea is indeed
similar to ours: each function is augmented by an additional argument
to pass information on run-time security checking.

In CPS, it is crucial to eliminate administrative redexes to achieve a
simulation property~\cite{DBLP:journals/tcs/Plotkin75,
  DBLP:books/cu/Appel1992, DBLP:conf/mfps/Wand91,
  DBLP:journals/lisp/SabryF93, DBLP:journals/tcs/DanvyN03,
  DBLP:journals/mscs/DanvyF92, DBLP:journals/toplas/SabryW97}, which
says that a reduction in the source is simulated by a sequence of
(one-directional) reductions in the translation.  Simulation is
usually achieved by applying different translations to an application
$\ottnt{M} \, \ottnt{N}$, depending on whether $\ottnt{M}$ and $\ottnt{N}$ are values or not.
In addition to such value/nonvalue distinction, our coercion-passing style
translation also relies on whether subterms are coercion applications or not.

Continuation-passing style eliminates the difference between
call-by-name and call-by-value but our coercion-passing style
translation works only under the call-by-value semantics of the target
language because coercions have to be eagerly composed.  It would be
interesting to investigate call-by-name for either the source or
the target language, or both.

\subsection{First-Class Coercions}

The idea of first-class coercions is also found in Cretin and R\'emy~\cite{DBLP:conf/popl/CretinR12}.
Their language $\mathrm{F}_\iota$ is equipped with abstraction over coercions.
However, their coercions are not for gradual typing but for parametric polymorphism and
subtyping polymorphism.

\section{Conclusion}
\label{sec:conclusion}

We have developed a new coercion calculus \lamSx with first-class
coercions as a target language of coercion-passing style translation
from \lamS, an existing space-efficient coercion calculus.  We have
proved the translation preserves both typing and semantics.  To
achieve a simulation property, it is important to reduce
administrative coercions, just as in CPS transformations.  Our
coercion-passing style translation solves the difficulty in implementing
the semantics of \lamS in a faithful manner and, with the help of
first-class coercions, makes it possible to implement in a compiler
for a call-by-value language.  We have modified an existing compiler
for a gradually typed language and conducted some experiments.
We have confirmed that our implementation successfully overcomes
stack overflow caused by coercions at tail positions,
which Kuhlenschmidt at al.~\cite{DBLP:conf/pldi/KuhlenschmidtAS19} did not support.
Our experiment has shown that for practical programs (without heavy use of function calls),
the coercion-passing style translation
causes slowdown up to 3 times for most partially typed configurations.

Aside from completing the implementation by adding recursive types,
which the original Grift compiler supports, more efficient
implementation is an obvious direction of future work.  Our coercion-passing style
translation introduces several identity coercions and optimizing
operations on coercions will be necessary.

From a theoretical point of view, it would be interesting to extend
the technique to gradual typing in the presence of parametric
polymorphism~\cite{DBLP:conf/popl/AhmedFSW11,
  DBLP:journals/pacmpl/AhmedJSW17, DBLP:journals/pacmpl/IgarashiSI17,
  DBLP:conf/esop/XieBO18,
  DBLP:journals/pacmpl/ToroLT19}, for which a polymorphic coercion
calculus has to be studied first---Luo~\cite{DBLP:journals/mscs/Luo08}
and Kie{\ss}ling and Luo~\cite{DBLP:conf/types/KiesslingL03}, who study coercive subtyping in
polymorphic settings, may be relevant.  The present design of \lamSx
is geared towards coercion-passing style.  For example, in \lamSx,
trivial (namely identity) coercions for coercion types $\ottnt{A}  \rightsquigarrow  \ottnt{B}$ are
allowed; passing coercions to dynamically typed code is prohibited;
variables cannot appear as an argument to coercion constructors, like
$\ottmv{x}  \Rightarrow  \ottnt{s}$. It may be interesting to study more general first-class
coercions without such restrictions.

\bibliography{biblio}

\iffull
\clearpage
\appendix
\section{Proofs}
\label{sec:appendix}

\begin{figure}[h]
  Type consistency \hfill\fbox{$\ottnt{A}  \sim  \ottnt{B}$}
  \begin{center}
    \infrule[C-Base]{}{
      \iota  \sim  \iota
    } \hfill
    \infrule[C-DynR]{}{
      \ottnt{A}  \sim  \mathord{\star}
    } \hfill
    \infrule[C-DynL]{}{
      \mathord{\star}  \sim  \ottnt{A}
    } \hfill
    \infrule[C-Fun]{            
      \ottnt{A}  \sim  \ottnt{A'} \andalso
      \ottnt{B}  \sim  \ottnt{B'}
    }{
      \ottnt{A}  \rightarrow  \ottnt{B}  \sim  \ottnt{A'}  \rightarrow  \ottnt{B'}
    }
  \end{center}
  \caption{Type consistency.}
  \label{fig:consistency}
\end{figure}

\subsection{Properties of \lamS}

\iffull\propSrcTgt*
\else
\begin{proposition}[Source and Target Types]\leavevmode
  \begin{enumerate}
  \item If $\ottnt{i}  \ottsym{:}  \ottnt{A}  \rightsquigarrow  \ottnt{B}$ then $\ottnt{A}  \ne  \mathord{\star}$.
  \item If $\ottnt{g}  \ottsym{:}  \ottnt{A}  \rightsquigarrow  \ottnt{B}$, then $\ottnt{A}  \ne  \mathord{\star}$ and $\ottnt{B}  \ne  \mathord{\star}$
    and there exists a unique $\ottnt{G}$ such that $\ottnt{A}  \sim  \ottnt{G}$ and $\ottnt{G}  \sim  \ottnt{B}$.
  \end{enumerate}
\end{proposition}
\fi
\begin{proof}
  (1) By case analysis on $\ottnt{i}$ with (2).
  (2) By case analysis on $\ottnt{g}$.
\end{proof}

\begin{proposition}
  Coercion composition $\ottnt{s}  \fatsemi  \ottnt{t}$ is terminating.
\end{proposition}
\begin{proof}
  The sum of sizes of two arguments gets smaller
  at each recursive call of $ \fatsemi $.
\end{proof}

\iffull\lemCmpWelldef*
\else
\begin{lemma} \label{lem:sc-cmp-welldef}
  If $\ottnt{s}  \ottsym{:}  \ottnt{A}  \rightsquigarrow  \ottnt{B}$ and $\ottnt{t}  \ottsym{:}  \ottnt{B}  \rightsquigarrow  \ottnt{C}$, then $\ottsym{(}  \ottnt{s}  \fatsemi  \ottnt{t}  \ottsym{)}  \ottsym{:}  \ottnt{A}  \rightsquigarrow  \ottnt{C}$.
\end{lemma}
\fi
\begin{proof}
  We prove the following four items simultaneously by straightforward induction:
  \begin{itemize}
  \item If $\ottnt{s}  \ottsym{:}  \ottnt{A}  \rightsquigarrow  \ottnt{B}$ and $\ottnt{t}  \ottsym{:}  \ottnt{B}  \rightsquigarrow  \ottnt{C}$, then $\ottsym{(}  \ottnt{s}  \fatsemi  \ottnt{t}  \ottsym{)}  \ottsym{:}  \ottnt{A}  \rightsquigarrow  \ottnt{C}$.
  \item If $\ottnt{i}  \ottsym{:}  \ottnt{A}  \rightsquigarrow  \ottnt{B}$ and $\ottnt{t}  \ottsym{:}  \ottnt{B}  \rightsquigarrow  \ottnt{C}$, then there exists $\ottnt{i'}$ such that $\ottnt{i'} = \ottnt{i}  \fatsemi  \ottnt{t}$ and $\ottnt{i'}  \ottsym{:}  \ottnt{A}  \rightsquigarrow  \ottnt{C}$.
  \item If $\ottnt{g}  \ottsym{:}  \ottnt{A}  \rightsquigarrow  \ottnt{B}$ and $\ottnt{i}  \ottsym{:}  \ottnt{B}  \rightsquigarrow  \ottnt{C}$, then there exists $\ottnt{i'}$ such that $\ottnt{i'} = \ottnt{g}  \fatsemi  \ottnt{i}$ and $\ottnt{i'}  \ottsym{:}  \ottnt{A}  \rightsquigarrow  \ottnt{C}$.
  \item If $\ottnt{g_{{\mathrm{1}}}}  \ottsym{:}  \ottnt{A}  \rightsquigarrow  \ottnt{B}$ and $\ottnt{g_{{\mathrm{2}}}}  \ottsym{:}  \ottnt{B}  \rightsquigarrow  \ottnt{C}$, then there exists $\ottnt{g_{{\mathrm{3}}}}$ such that $\ottnt{g_{{\mathrm{3}}}} = \ottnt{g_{{\mathrm{1}}}}  \fatsemi  \ottnt{g_{{\mathrm{2}}}}$ and $\ottnt{g_{{\mathrm{3}}}}  \ottsym{:}  \ottnt{A}  \rightsquigarrow  \ottnt{C}$. \qedhere
  \end{itemize}
\end{proof}

\begin{lemma} \label{lem:Sc-terminate}
  $ \accentset{\mathsf{c} }{\evalto}_{\mathsf{S} } $ is terminating.
\end{lemma}
\begin{proof}
  Consider a metric $f(M) = 4(k + l) + 2m + n$ of a term $\ottnt{M}$ where:
  \begin{itemize}
  \item $k$ is the sum of the sizes of coercions in $ \langle \cdot \rangle $ in $\ottnt{M}$
  \item $l$ is the sum of the sizes of coercions in $ \langle\!\langle \cdot \rangle\!\rangle $ in $\ottnt{M}$
  \item $m$ is the number of $ \langle \cdot \rangle $
  \item $n$ is the number of $ \langle\!\langle \cdot \rangle\!\rangle $ in $\ottnt{M}$
  \end{itemize}
  It is easy to show that if $ \ottnt{M}    \mathbin{  \accentset{\mathsf{c} }{\evalto}_{\mathsf{S} }  }    \ottnt{N} $ then $f(M) > f(N)$.
\end{proof}

\begin{lemma}[Unique Decomposition]\label{lem:decom-S}
  If $  \emptyset     \vdash_{\mathsf{S} }    \ottnt{M}  :  \ottnt{A} $, then one of the following holds.
  \begin{enumerate}
  \item There uniquely exist a redex $\ottnt{M_{{\mathrm{1}}}}$ and an evaluation context $\mathcal{E}$
    such that $\ottnt{M}  \ottsym{=}  \mathcal{E}  [  \ottnt{M_{{\mathrm{1}}}}  ]$.
  \item $\ottnt{M}  \ottsym{=}  \ottnt{V}$ for some $\ottnt{V}$.
  \item $\ottnt{M}  \ottsym{=}  \ottkw{blame} \, \ottnt{p}$ for some $\ottnt{p}$.
  \end{enumerate}
\end{lemma}
\begin{proof}
  By induction on the derivation of $  \emptyset     \vdash_{\mathsf{S} }    \ottnt{M}  :  \ottnt{A} $
  with case analysis on the rule applied last.
\end{proof}

\iffull\lemDeterminacyS*
\else
\begin{lemma}[Determinacy]\label{lem:determinism-S}
  If $ \ottnt{M}    \mathbin{  \evalto_{\mathsf{S} }  }    \ottnt{N} $ and $ \ottnt{M}    \mathbin{  \evalto_{\mathsf{S} }  }    \ottnt{N'} $, then $\ottnt{N}  \ottsym{=}  \ottnt{N'}$.
\end{lemma}
\fi
\begin{proof}
  By Lemma~\ref{lem:decom-S}.
\end{proof}

We state type safety for \lamS with auxiliary lemmas.

\begin{lemma}[Canonical Forms]\label{lem:canonical-S}
  If $  \emptyset     \vdash_{\mathsf{S} }    \ottnt{V}  :  \ottnt{A} $, then one of the following holds.
  \begin{enumerate}
  \item $\ottnt{V}  \ottsym{=}  \ottnt{a}$ and $\ottnt{A}  \ottsym{=}  \iota$ for some $\ottnt{a},\iota$.
  \item $\ottnt{V}  \ottsym{=}   \lambda   \ottmv{x} .\,  \ottnt{M} $ and $\ottnt{A}  \ottsym{=}  \ottnt{A_{{\mathrm{1}}}}  \rightarrow  \ottnt{A_{{\mathrm{2}}}}$ for some $\ottmv{x},\ottnt{M},\ottnt{A_{{\mathrm{1}}}},\ottnt{A_{{\mathrm{2}}}}$.
  \item $\ottnt{V}  \ottsym{=}  \ottnt{U}  \langle\!\langle  \ottnt{s}  \rightarrow  \ottnt{t}  \rangle\!\rangle$ and $\ottnt{A}  \ottsym{=}  \ottnt{A_{{\mathrm{1}}}}  \rightarrow  \ottnt{A_{{\mathrm{2}}}}$
    for some $\ottnt{U},\ottnt{s},\ottnt{t},\ottnt{A_{{\mathrm{1}}}},\ottnt{A_{{\mathrm{2}}}}$. 
  \item $\ottnt{V}  \ottsym{=}  \ottnt{U}  \langle\!\langle  \ottnt{g}  \ottsym{;}   \ottnt{G} \texttt{!}   \rangle\!\rangle$ and $\ottnt{A}  \ottsym{=}  \mathord{\star}$ for some $\ottnt{U},\ottnt{g},\ottnt{G}$.
  \end{enumerate}
\end{lemma}

\iffull\thmProgressS*
\else
\begin{theorem}[Progress]\label{thm:progress-S}
  If $  \emptyset     \vdash_{\mathsf{S} }    \ottnt{M}  :  \ottnt{A} $, then one of the following holds.
  \begin{enumerate}
  \item $ \ottnt{M}    \mathbin{  \evalto_{\mathsf{S} }  }    \ottnt{M'} $ for some $\ottnt{M'}$.
  \item $\ottnt{M}  \ottsym{=}  \ottnt{V}$ for some $\ottnt{V}$.
  \item $\ottnt{M}  \ottsym{=}  \ottkw{blame} \, \ottnt{p}$ for some $\ottnt{p}$.
  \end{enumerate}
\end{theorem}
\fi
\begin{proof}
  By Lemma~\ref{lem:decom-S}.
\end{proof}

\begin{lemma}[Preservation of Types under Substitution]\label{lem:preservation-subst-S}\leavevmode
  If $ \Gamma  \ottsym{,}  \ottmv{x}  \ottsym{:}  \ottnt{A}    \vdash_{\mathsf{S} }    \ottnt{M}  :  \ottnt{B} $ and $ \Gamma    \vdash_{\mathsf{S} }    \ottnt{V}  :  \ottnt{A} $, then
  $ \Gamma    \vdash_{\mathsf{S} }    \ottnt{M}  [  \ottmv{x}  \ottsym{:=}  \ottnt{V}  ]  :  \ottnt{B} $.
\end{lemma}

\begin{lemma}[Preservation for Reduction]
  If $  \emptyset     \vdash_{\mathsf{S} }    \ottnt{M}  :  \ottnt{A} $ and $ \ottnt{M}    \reduces_{\mathsf{S} }    \ottnt{N} $, then $  \emptyset     \vdash_{\mathsf{S} }    \ottnt{N}  :  \ottnt{A} $.
\end{lemma}
\begin{proof}
  By case analysis on the reduction rule applied to $ \ottnt{M}    \reduces_{\mathsf{S} }    \ottnt{N} $.
  (Similar to Lemma~\ref{lem:preservation-red-S1}.)
\end{proof}

\iffull\thmPreservationS*
\else
\begin{theorem}[Preservation]\label{thm:preservation-S}
  If $  \emptyset     \vdash_{\mathsf{S} }    \ottnt{M}  :  \ottnt{A} $ and $ \ottnt{M}    \mathbin{  \evalto_{\mathsf{S} }  }    \ottnt{N} $, then $  \emptyset     \vdash_{\mathsf{S} }    \ottnt{N}  :  \ottnt{A} $.
\end{theorem}
\fi
\begin{proof}
  By case analysis on the evaluation rule applied to $ \ottnt{M}    \mathbin{  \evalto_{\mathsf{S} }  }    \ottnt{N} $.
  (Similar to Theorem~\ref{thm:preservation-S1}.)
\end{proof}

\iffull\corSafetyS*
\else
\begin{corollary}[Type Safety]\label{cor:safety-S}
  If $  \emptyset     \vdash_{\mathsf{S} }    \ottnt{M}  :  \ottnt{A} $, then one of the following holds.
  \begin{enumerate}
  \item $ \ottnt{M}    \mathbin{  \evalto_{\mathsf{S} }  ^*}    \ottnt{V} $ and $  \emptyset     \vdash_{\mathsf{S} }    \ottnt{V}  :  \ottnt{A} $ for some $\ottnt{V}$.
  \item $ \ottnt{M}    \mathbin{  \evalto_{\mathsf{S} }  ^*}    \ottkw{blame} \, \ottnt{p} $ for some $\ottnt{p}$.
  \item $\ottnt{M} \,  \mathord{\Uparrow_{\mathsf{S} } } $.
  \end{enumerate}
\end{corollary}
\fi
\begin{proof}
  By Theorem~\ref{thm:progress-S} and Theorem~\ref{thm:preservation-S}.
\end{proof}

\subsection{Properties of \lamSx}

The properties for \lamS-coercions still hold in \lamSx.
We do not repeat all of them.

\iffull
\else
\begin{proposition}[Source and Target Types]\leavevmode
  \begin{enumerate}
  \item If $\ottnt{i}  \ottsym{:}  \ottnt{A}  \rightsquigarrow  \ottnt{B}$ then $\ottnt{A}  \ne  \mathord{\star}$.
  \item If $\ottnt{g}  \ottsym{:}  \ottnt{A}  \rightsquigarrow  \ottnt{B}$, then $\ottnt{A}  \ne  \mathord{\star}$ and $\ottnt{B}  \ne  \mathord{\star}$
    and there exists a unique $\ottnt{G}$ such that $\ottnt{A}  \sim  \ottnt{G}$ and $\ottnt{G}  \sim  \ottnt{B}$.
  \end{enumerate}
\end{proposition}
\fi

\begin{lemma}\label{lem:cmp-ty}
  If $\ottnt{s}  \ottsym{:}  \ottnt{A}  \rightsquigarrow  \ottnt{B}$ and $\ottnt{t}  \ottsym{:}  \ottnt{B}  \rightsquigarrow  \ottnt{C}$, $\ottsym{(}  \ottnt{s}  \fatsemi  \ottnt{t}  \ottsym{)}  \ottsym{:}  \ottnt{A}  \rightsquigarrow  \ottnt{C}$.
\end{lemma}
\begin{proof}
  Similar to Lemma~\ref{lem:sc-cmp-welldef}.
\end{proof}

We explicitly state a few lemmas on evaluation contexts.

The following lemma ensures that the composition of evaluation contexts
is also an evaluation context.
Here, we note that
\[
  \ottsym{(}  \mathcal{E}_{{\mathrm{1}}}  [  \mathcal{E}_{{\mathrm{2}}}  ]  \ottsym{)}  [  \ottnt{L}  ] = \mathcal{E}_{{\mathrm{1}}}  [  \mathcal{E}_{{\mathrm{2}}}  [  \ottnt{L}  ]  ]
\]
where $\ottnt{L}$ is a term that may contain at most one hole $ \square $.
For example,
\[
  \ottsym{(}  \ottnt{op}  \ottsym{(}  \square  \ottsym{,}  \ottnt{M}  \ottsym{)}  [  \ottnt{op}  \ottsym{(}  \ottnt{V}  \ottsym{,} \, \square \, \ottsym{)}  ]  \ottsym{)}  [  \ottnt{L}  ] = \ottsym{(}  \ottnt{op}  \ottsym{(}  \ottnt{op}  \ottsym{(}  \ottnt{V}  \ottsym{,} \, \square \, \ottsym{)}  \ottsym{,}  \ottnt{M}  \ottsym{)}  \ottsym{)}  [  \ottnt{L}  ]
  = \ottnt{op}  \ottsym{(}  \ottnt{op}  \ottsym{(}  \ottnt{V}  \ottsym{,}  \ottnt{L}  \ottsym{)}  \ottsym{,}  \ottnt{M}  \ottsym{)}.
\]

\begin{lemma}[Composition of Contexts]\label{lem:ctx-cmp}
  For any evaluation contexts $\mathcal{E}_{{\mathrm{1}}},\mathcal{E}_{{\mathrm{2}}}$ of \lamSx,
  there exists an evaluation context $\mathcal{E}$ such that $\mathcal{E}_{{\mathrm{1}}}  [  \mathcal{E}_{{\mathrm{2}}}  ]  \ottsym{=}  \mathcal{E}$.
\end{lemma}
\begin{proof}
  By induction on $\mathcal{E}_{{\mathrm{2}}}$. We only show one case.
  \begin{description}
  \item[\Case{$\mathcal{E}_{{\mathrm{2}}}  \ottsym{=}  \mathcal{E}'_{{\mathrm{2}}}  [   \square \, (  \ottnt{M} , \ottnt{N}  )   ]$}]
    \[
      \mathcal{E}_{{\mathrm{1}}}  [  \mathcal{E}_{{\mathrm{2}}}  ] = \mathcal{E}_{{\mathrm{1}}}  [  \mathcal{E}'_{{\mathrm{2}}}  [   \square \, (  \ottnt{M} , \ottnt{N}  )   ]  ]
      = \ottsym{(}  \mathcal{E}_{{\mathrm{1}}}  [  \mathcal{E}'_{{\mathrm{2}}}  ]  \ottsym{)}  [   \square \, (  \ottnt{M} , \ottnt{N}  )   ]
    \]
    By the IH, we have $\mathcal{E}_{{\mathrm{1}}}  [  \mathcal{E}'_{{\mathrm{2}}}  ]  \ottsym{=}  \mathcal{E}'$ for some $\mathcal{E}'$.
    Then, $\mathcal{E}_{{\mathrm{1}}}  [  \mathcal{E}_{{\mathrm{2}}}  ]  \ottsym{=}  \mathcal{E}'  [   \square \, (  \ottnt{M} , \ottnt{N}  )   ]$ is an evaluation context. \qedhere
  \end{description}
\end{proof}

The following lemma is useful in Lemma~\ref{lem:trans-eval}.

\begin{lemma}\label{lem:ctx-multi}  
  Assume $\ottnt{N}  \ne  \ottkw{blame} \, \ottnt{p}$ and $ \mathcal{X}\in   \set{\mathsf{e},\mathsf{c} }  $.
  If $ \ottnt{M}    \mathbin{  \accentset{\mathcal{X} }{\evalto}_{\mathsf{S_1} }  }    \ottnt{N} $, then $ \mathcal{E}  [  \ottnt{M}  ]    \mathbin{  \accentset{\mathcal{X} }{\evalto}_{\mathsf{S_1} }  }    \mathcal{E}  [  \ottnt{N}  ] $.
\end{lemma}
\begin{proof}
  By case analysis on the evaluation rule applied to $ \ottnt{M}    \mathbin{  \accentset{\mathcal{X} }{\evalto}_{\mathsf{S_1} }  }    \ottnt{N} $.
  \begin{description}
  \item[\Case{\rnp{E-Ctx}}]
    We are given
    \[
       \ottnt{M_{{\mathrm{1}}}}    \mathbin{\accentset{\mathcal{X} }{\reduces} }    \ottnt{N_{{\mathrm{1}}}}  \hgap
      \ottnt{M}  \ottsym{=}  \mathcal{E}_{{\mathrm{1}}}  [  \ottnt{M_{{\mathrm{1}}}}  ] \hgap
      \ottnt{N}  \ottsym{=}  \mathcal{E}_{{\mathrm{1}}}  [  \ottnt{N_{{\mathrm{1}}}}  ] \hgap
    \]
    for some $\mathcal{E},\ottnt{M_{{\mathrm{1}}}},\ottnt{N_{{\mathrm{1}}}}$.
    By Lemma~\ref{lem:ctx-cmp}, we have $\mathcal{E}  [  \mathcal{E}_{{\mathrm{1}}}  ]  \ottsym{=}  \mathcal{E}'$ for some $\mathcal{E}'$.
    \[
      \mathcal{E}'  [  \ottnt{M_{{\mathrm{1}}}}  ]  \ottsym{=}  \mathcal{E}  [  \mathcal{E}_{{\mathrm{1}}}  [  \ottnt{M_{{\mathrm{1}}}}  ]  ] = \mathcal{E}  [  \ottnt{M}  ] \hgap
      \mathcal{E}'  [  \ottnt{N_{{\mathrm{1}}}}  ]  \ottsym{=}  \mathcal{E}  [  \ottnt{N}  ]
    \]
    By \rnp{E-Ctx} with evaluation context $\mathcal{E}  [  \mathcal{E}_{{\mathrm{1}}}  ]  \ottsym{=}  \mathcal{E}'$,
    we have $ \mathcal{E}  [  \ottnt{M}  ]    \mathbin{  \accentset{\mathcal{X} }{\evalto}  }    \mathcal{E}  [  \ottnt{N}  ] $.

  \item[\Case{\rnp{E-Abort}}]
    Cannot happen (since $\ottnt{N}  \ne  \ottkw{blame} \, \ottnt{p}$). \qedhere
  \end{description}
\end{proof}

We note that the following property (for any natural number $\ottmv{n}$)
follows from Lemma~\ref{lem:ctx-multi}.
(By $\ottnt{N}  \ne  \ottkw{blame} \, \ottnt{p}$, we can assume no use of \rnp{E-Abort} in the derivation of $\ottnt{M}  \mathbin{  \accentset{\mathcal{X} }{\evalto}_{\mathsf{S_1} }  ^ \ottmv{n} }  \ottnt{N}$.)
\begin{center}
  Assume $\ottnt{N}  \ne  \ottkw{blame} \, \ottnt{p}$ and $ \mathcal{X}\in   \set{\mathsf{e},\mathsf{c} }  $.
  If $\ottnt{M}  \mathbin{  \accentset{\mathcal{X} }{\evalto}_{\mathsf{S_1} }  ^ \ottmv{n} }  \ottnt{N}$, then $\mathcal{E}  [  \ottnt{M}  ]  \mathbin{  \accentset{\mathcal{X} }{\evalto}_{\mathsf{S_1} }  ^ \ottmv{n} }  \mathcal{E}  [  \ottnt{N}  ]$.
\end{center}

\begin{lemma}[Unique Decomposition]\label{lem:decom-S1}
  If $  \emptyset     \vdash_{\mathsf{S_1} }    \ottnt{M}  :  \ottnt{A} $, then one of the following holds.
  \begin{enumerate}
  \item There uniquely exist a redex $\ottnt{M_{{\mathrm{1}}}}$ and an evaluation context $\mathcal{E}$
    such that $\ottnt{M}  \ottsym{=}  \mathcal{E}  [  \ottnt{M_{{\mathrm{1}}}}  ]$.
  \item $\ottnt{M}  \ottsym{=}  \ottnt{V}$ for some $\ottnt{V}$.
  \item $\ottnt{M}  \ottsym{=}  \ottkw{blame} \, \ottnt{p}$ for some $\ottnt{p}$.
  \end{enumerate}
\end{lemma}
\begin{proof}
  Similarly for Theorem~\ref{thm:progress-S1}.\footnote{%
    It would be a bit more involved so as to show the uniqueness.
    We avoid involvedness to write them down.
    (Theorem~\ref{thm:progress-S1} would follow Lemma~\ref{lem:decom-S1}.)
  }
\end{proof}

\iffull\lemDeterminacySx*
\else
\begin{lemma}[Determinacy]\label{lem:determinism-S1}
  If $ \ottnt{M}    \mathbin{  \evalto_{\mathsf{S_1} }  }    \ottnt{N} $ and $ \ottnt{M}    \mathbin{  \evalto_{\mathsf{S_1} }  }    \ottnt{N'} $, then $\ottnt{N}  \ottsym{=}  \ottnt{N'}$.
\end{lemma}
\fi
\begin{proof}
  By Lemma~\ref{lem:decom-S1}.
\end{proof}

We state type safety for \lamSx with auxiliary lemmas.
We omit inversion lemmas for the typing judgments.

\begin{lemma}[Canonical Forms]\label{lem:canonical-S1}
  If $  \emptyset     \vdash_{\mathsf{S_1} }    \ottnt{V}  :  \ottnt{A} $, then one of the following holds.
  \begin{enumerate}
  \item $\ottnt{V}  \ottsym{=}  \ottnt{a}$ and $\ottnt{A}  \ottsym{=}  \iota$ for some $\ottnt{a},\iota$.
  \item $\ottnt{V}  \ottsym{=}   \lambda  ( \ottmv{x} , \kappa ).\,  \ottnt{M} $ and $\ottnt{A}  \ottsym{=}  \ottnt{A_{{\mathrm{1}}}}  \Rightarrow  \ottnt{A_{{\mathrm{2}}}}$ for some $\ottmv{x},\kappa,\ottnt{M},\ottnt{A_{{\mathrm{1}}}},\ottnt{A_{{\mathrm{2}}}}$.
  \item $\ottnt{V}  \ottsym{=}  \ottnt{U}  \langle\!\langle  \ottnt{s}  \Rightarrow  \ottnt{t}  \rangle\!\rangle$ and $\ottnt{A}  \ottsym{=}  \ottnt{A_{{\mathrm{1}}}}  \Rightarrow  \ottnt{A_{{\mathrm{2}}}}$ for some $\ottnt{U},\ottnt{s},\ottnt{t},\ottnt{A_{{\mathrm{1}}}},\ottnt{A_{{\mathrm{2}}}}$.
  \item $\ottnt{V}  \ottsym{=}  \ottnt{U}  \langle\!\langle  \ottnt{g}  \ottsym{;}   \ottnt{G} \texttt{!}   \rangle\!\rangle$ and $\ottnt{A}  \ottsym{=}  \mathord{\star}$ for some $\ottnt{U},\ottnt{g},\ottnt{G}$.
  \item $\ottnt{V}  \ottsym{=}  \ottnt{s}$ and $\ottnt{A}  \ottsym{=}  \ottnt{A_{{\mathrm{1}}}}  \rightsquigarrow  \ottnt{A_{{\mathrm{2}}}}$ for some $\ottnt{s},\ottnt{A_{{\mathrm{1}}}},\ottnt{A_{{\mathrm{2}}}}$.
  \end{enumerate}
\end{lemma}
\begin{proof}
  By case analysis on the typing rule applied to $  \emptyset     \vdash_{\mathsf{S_1} }    \ottnt{V}  :  \ottnt{A} $.
\end{proof}

In the proof of the following theorem,
we only write down case \rnp{T-Op} in detail
and write ``Similar.'' for the other cases.

\iffull\thmProgressSx*
\else
\begin{theorem}[Progress]\label{thm:progress-S1}
  If $  \emptyset     \vdash_{\mathsf{S_1} }    \ottnt{M}  :  \ottnt{A} $, then one of the following holds.
  \begin{enumerate}
  \item $ \ottnt{M}    \mathbin{  \evalto_{\mathsf{S_1} }  }    \ottnt{M'} $ for some $\ottnt{M'}$.
  \item $\ottnt{M}  \ottsym{=}  \ottnt{V}$ for some $\ottnt{V}$.
  \item $\ottnt{M}  \ottsym{=}  \ottkw{blame} \, \ottnt{p}$ for some $\ottnt{p}$.
  \end{enumerate}
\end{theorem}
\fi
\begin{proof}
  By induction on the derivation of $  \emptyset     \vdash_{\mathsf{S_1} }    \ottnt{M}  :  \ottnt{A} $
  with case analysis on the rule applied last.
  \begin{description}
  \item[\Case{\rnp{T-Var}}] Cannot happen.
  \item[\Case{\rnp{T-Const}}] Immediate. ($\ottnt{M}  \ottsym{=}  \ottnt{a}$ is a value.)
  \item[\Case{\rnp{T-Abs}}] Immediate. ($\ottnt{M}  \ottsym{=}   \lambda  ( \ottmv{x} , \kappa ).\,  \ottnt{M_{{\mathrm{1}}}} $ is a value.)
  \item[\Case{\rnp{T-Op}}]
    We are given
    \[
      \ottnt{M}  \ottsym{=}  \ottnt{op}  \ottsym{(}  \ottnt{N_{{\mathrm{1}}}}  \ottsym{,}  \ottnt{N_{{\mathrm{2}}}}  \ottsym{)} \hgap
        \emptyset    \vdash   \ottnt{N_{{\mathrm{1}}}}  :  \iota_{{\mathrm{1}}}  \hgap
        \emptyset    \vdash   \ottnt{N_{{\mathrm{2}}}}  :  \iota_{{\mathrm{2}}} 
    \]
    for some $\ottnt{N_{{\mathrm{1}}}},\ottnt{N_{{\mathrm{2}}}},\iota_{{\mathrm{1}}},\iota_{{\mathrm{2}}}$.
    We have the IHs for $  \emptyset    \vdash   \ottnt{N_{{\mathrm{1}}}}  :  \iota_{{\mathrm{1}}} $ and $  \emptyset    \vdash   \ottnt{N_{{\mathrm{2}}}}  :  \iota_{{\mathrm{2}}} $.
    We proceed by case analysis on $\ottnt{N_{{\mathrm{1}}}},\ottnt{N_{{\mathrm{2}}}}$.
    \begin{description}
    \item[\SubCase{$ \ottnt{N_{{\mathrm{1}}}}    \mathbin{ \evalto }    \ottnt{N'_{{\mathrm{1}}}} $}]
      By case analysis on the evaluation rule applied to $\ottnt{N_{{\mathrm{1}}}}$.
      \begin{description}
      \item[\SubSubCase{\rnp{E-Ctx}}]
        We are given
        \[
           \ottnt{N_{{\mathrm{11}}}}   \reduces   \ottnt{N'_{{\mathrm{11}}}}  \hgap
          \ottnt{N_{{\mathrm{1}}}}  \ottsym{=}  \mathcal{E}_{{\mathrm{1}}}  [  \ottnt{N_{{\mathrm{11}}}}  ] \hgap
          \ottnt{N'_{{\mathrm{1}}}}  \ottsym{=}  \mathcal{E}_{{\mathrm{1}}}  [  \ottnt{N'_{{\mathrm{11}}}}  ]
        \]
        for some $\mathcal{E}_{{\mathrm{1}}},\ottnt{N_{{\mathrm{11}}}},\ottnt{N'_{{\mathrm{11}}}}$.
        Take $\mathcal{E}  \ottsym{=}  \ottsym{(}  \ottnt{op}  \ottsym{(}  \square  \ottsym{,}  \ottnt{N_{{\mathrm{2}}}}  \ottsym{)}  \ottsym{)}  [  \mathcal{E}_{{\mathrm{1}}}  ]$ by Lemma~\ref{lem:ctx-cmp}.
        \begin{align*}
          \mathcal{E}  [  \ottnt{N_{{\mathrm{11}}}}  ] &=
          \ottsym{(}  \ottnt{op}  \ottsym{(}  \square  \ottsym{,}  \ottnt{N_{{\mathrm{2}}}}  \ottsym{)}  \ottsym{)}  [  \mathcal{E}_{{\mathrm{1}}}  [  \ottnt{N_{{\mathrm{11}}}}  ]  ] = \ottsym{(}  \ottnt{op}  \ottsym{(}  \square  \ottsym{,}  \ottnt{N_{{\mathrm{2}}}}  \ottsym{)}  \ottsym{)}  [  \ottnt{N_{{\mathrm{1}}}}  ] = \ottnt{op}  \ottsym{(}  \ottnt{N_{{\mathrm{1}}}}  \ottsym{,}  \ottnt{N_{{\mathrm{2}}}}  \ottsym{)} \\
          \mathcal{E}  [  \ottnt{N'_{{\mathrm{11}}}}  ] &= \ottnt{op}  \ottsym{(}  \ottnt{N'_{{\mathrm{1}}}}  \ottsym{,}  \ottnt{N_{{\mathrm{2}}}}  \ottsym{)}
        \end{align*}
        By \rnp{E-Ctx} with $\mathcal{E}  \ottsym{=}  \ottsym{(}  \ottnt{op}  \ottsym{(}  \square  \ottsym{,}  \ottnt{N_{{\mathrm{2}}}}  \ottsym{)}  \ottsym{)}  [  \mathcal{E}_{{\mathrm{1}}}  ]$,
        we have $ \ottnt{op}  \ottsym{(}  \ottnt{N_{{\mathrm{1}}}}  \ottsym{,}  \ottnt{N_{{\mathrm{2}}}}  \ottsym{)}    \mathbin{ \evalto }    \ottnt{op}  \ottsym{(}  \ottnt{N'_{{\mathrm{1}}}}  \ottsym{,}  \ottnt{N_{{\mathrm{2}}}}  \ottsym{)} $.
        Take $\ottnt{M'}  \ottsym{=}  \ottnt{op}  \ottsym{(}  \ottnt{N'_{{\mathrm{1}}}}  \ottsym{,}  \ottnt{N_{{\mathrm{2}}}}  \ottsym{)}$.
      \item[\SubSubCase{\rnp{E-Abort}}]
        We are given
        \[
          \ottnt{N_{{\mathrm{1}}}}  \ottsym{=}  \mathcal{E}_{{\mathrm{1}}}  [  \ottkw{blame} \, \ottnt{p}  ] \hgap
          \ottnt{N'_{{\mathrm{1}}}}  \ottsym{=}  \ottkw{blame} \, \ottnt{p} \hgap
          \mathcal{E}  \ne  \square
        \]
        for some $\mathcal{E}_{{\mathrm{1}}},\ottnt{p}$.
        By Lemma~\ref{lem:ctx-cmp},
        we have
        \[
          \ottsym{(}  \ottnt{op}  \ottsym{(}  \square  \ottsym{,}  \ottnt{N_{{\mathrm{2}}}}  \ottsym{)}  [  \mathcal{E}_{{\mathrm{1}}}  ]  \ottsym{)}  [  \ottkw{blame} \, \ottnt{p}  ] =
          \ottnt{op}  \ottsym{(}  \square  \ottsym{,}  \ottnt{N_{{\mathrm{2}}}}  \ottsym{)}  [  \mathcal{E}_{{\mathrm{1}}}  [  \ottkw{blame} \, \ottnt{p}  ]  ]  \ottsym{=}  \ottnt{op}  \ottsym{(}  \ottnt{N_{{\mathrm{1}}}}  \ottsym{,}  \ottnt{N_{{\mathrm{2}}}}  \ottsym{)}
          .
        \]
        By \rnp{E-Abort} with $\mathcal{E}  \ottsym{=}  \ottsym{(}  \ottnt{op}  \ottsym{(}  \square  \ottsym{,}  \ottnt{N_{{\mathrm{2}}}}  \ottsym{)}  \ottsym{)}  [  \mathcal{E}_{{\mathrm{1}}}  ]$,
        we have $ \ottnt{op}  \ottsym{(}  \ottnt{N_{{\mathrm{1}}}}  \ottsym{,}  \ottnt{N_{{\mathrm{2}}}}  \ottsym{)}    \mathbin{ \evalto }    \ottkw{blame} \, \ottnt{p} $.
        Take $\ottnt{M'}  \ottsym{=}  \ottkw{blame} \, \ottnt{p}$.
      \end{description}
    \item[\SubCase{$\ottnt{N_{{\mathrm{1}}}}  \ottsym{=}  \ottkw{blame} \, \ottnt{p}$}]
      Take $\ottnt{M'}  \ottsym{=}  \ottkw{blame} \, \ottnt{p}$.
      By \rnp{E-Abort} with $\mathcal{E}  \ottsym{=}  \ottnt{op}  \ottsym{(}  \square  \ottsym{,}  \ottnt{N_{{\mathrm{2}}}}  \ottsym{)}$, we have
      $ \ottnt{op}  \ottsym{(}  \ottkw{blame} \, \ottnt{p}  \ottsym{,}  \ottnt{N_{{\mathrm{2}}}}  \ottsym{)}    \mathbin{ \evalto }    \ottkw{blame} \, \ottnt{p} $;
      i.e., $ \ottnt{M}    \mathbin{ \evalto }    \ottnt{M'} $.
    \item[\SubCase{$\ottnt{N_{{\mathrm{1}}}}  \ottsym{=}  \ottnt{V_{{\mathrm{1}}}}$ and $ \ottnt{N_{{\mathrm{2}}}}    \mathbin{ \evalto }    \ottnt{N'_{{\mathrm{2}}}} $}]
      By case analysis on the evaluation rule applied to $\ottnt{N_{{\mathrm{2}}}}$.
      \begin{description}
      \item[\SubSubCase{\rnp{E-Ctx}}]
        Similarly, take $\ottnt{M'}  \ottsym{=}  \ottnt{op}  \ottsym{(}  \ottnt{V_{{\mathrm{1}}}}  \ottsym{,}  \ottnt{N'_{{\mathrm{2}}}}  \ottsym{)}$.
      \item[\SubSubCase{\rnp{E-Abort}}]
        Similarly, take $\ottnt{M'}  \ottsym{=}  \ottkw{blame} \, \ottnt{p}$.
      \end{description}
    \item[\SubCase{$\ottnt{N_{{\mathrm{1}}}}  \ottsym{=}  \ottnt{V_{{\mathrm{1}}}}$ and $\ottnt{N_{{\mathrm{2}}}}  \ottsym{=}  \ottkw{blame} \, \ottnt{p}$}]
      Take $\ottnt{M'}  \ottsym{=}  \ottkw{blame} \, \ottnt{p}$.
      By \rnp{E-Abort} with $\mathcal{E}  \ottsym{=}  \ottnt{op}  \ottsym{(}  \ottnt{V_{{\mathrm{1}}}}  \ottsym{,} \, \square \, \ottsym{)}$,
      $ \ottnt{op}  \ottsym{(}  \ottnt{V_{{\mathrm{1}}}}  \ottsym{,}  \ottkw{blame} \, \ottnt{p}  \ottsym{)}    \mathbin{ \evalto }    \ottkw{blame} \, \ottnt{p} $;
      i.e., $ \ottnt{M}    \mathbin{ \evalto }    \ottnt{M'} $.
    \item[\SubCase{$\ottnt{N_{{\mathrm{1}}}}  \ottsym{=}  \ottnt{V_{{\mathrm{1}}}}$ and $\ottnt{N_{{\mathrm{2}}}}  \ottsym{=}  \ottnt{V_{{\mathrm{2}}}}$}] 
      By $  \emptyset    \vdash   \ottnt{V_{{\mathrm{1}}}}  :  \iota_{{\mathrm{1}}} $ and $  \emptyset    \vdash   \ottnt{V_{{\mathrm{2}}}}  :  \iota_{{\mathrm{2}}} $
      and Lemma~\ref{lem:canonical-S1}, we have
      $\ottnt{V_{{\mathrm{1}}}}  \ottsym{=}  \ottnt{a_{{\mathrm{1}}}}$ and $\ottnt{V_{{\mathrm{2}}}}  \ottsym{=}  \ottnt{a_{{\mathrm{2}}}}$ for some $\ottnt{a_{{\mathrm{1}}}},\ottnt{a_{{\mathrm{2}}}}$.
      Then, \rnp{R-Op} finishes.
    \end{description}

  \item[\Case{\rnp{T-App}}]
    We are given
    \[
      \ottnt{M}  \ottsym{=}   \ottnt{N_{{\mathrm{1}}}} \, ( \ottnt{N_{{\mathrm{2}}}} , \ottnt{N_{{\mathrm{3}}}} )  \hgap
        \emptyset    \vdash   \ottnt{N_{{\mathrm{1}}}}  :  \ottnt{A_{{\mathrm{2}}}}  \Rightarrow  \ottnt{B}  \hgap
        \emptyset    \vdash   \ottnt{N_{{\mathrm{2}}}}  :  \ottnt{A_{{\mathrm{2}}}}  \hgap
        \emptyset    \vdash   \ottnt{N_{{\mathrm{3}}}}  :  \ottnt{B}  \rightsquigarrow  \ottnt{A} 
    \]
    for some $\ottnt{N_{{\mathrm{1}}}},\ottnt{N_{{\mathrm{2}}}},\ottnt{N_{{\mathrm{3}}}},\ottnt{A_{{\mathrm{2}}}},\ottnt{B}$.
    We have the IHs for three typing derivations.
    We proceed by case analysis on $\ottnt{N_{{\mathrm{1}}}},\ottnt{N_{{\mathrm{2}}}},\ottnt{N_{{\mathrm{3}}}}$.
    \begin{description}
    \item[\SubCase{$\ottnt{N_{{\mathrm{1}}}}  \ottsym{=}  \ottnt{V_{{\mathrm{1}}}}$ and $\ottnt{N_{{\mathrm{2}}}}  \ottsym{=}  \ottnt{V_{{\mathrm{2}}}}$ and $\ottnt{N_{{\mathrm{3}}}}  \ottsym{=}  \ottnt{V_{{\mathrm{3}}}}$}]
      By $  \emptyset    \vdash   \ottnt{V_{{\mathrm{1}}}}  :  \ottnt{A_{{\mathrm{2}}}}  \Rightarrow  \ottnt{A} $ and Lemma~\ref{lem:canonical-S1}, we have either
      \[
        \ottnt{V_{{\mathrm{1}}}}  \ottsym{=}   \lambda  ( \ottmv{x} , \kappa ).\,  \ottnt{L}  \hgap
        \ottnt{V_{{\mathrm{1}}}}  \ottsym{=}  \ottnt{U}  \langle  \ottnt{V_{{\mathrm{11}}}}  \Rightarrow  \ottnt{V_{{\mathrm{12}}}}  \rangle.
      \]
      Then, \rnp{R-Beta} or \rnp{R-Wrap} finishes the case.
    \item[\Otherwise] Similar.
    \end{description}

  \item[\Case{\rnp{T-Let}}]
    We are given
    \[
      \ottnt{M}  \ottsym{=}   \ottkw{let} \,  \ottmv{x} = \ottnt{N_{{\mathrm{1}}}} \, \ottkw{in}\,  \ottnt{N_{{\mathrm{2}}}}  \hgap
        \emptyset    \vdash   \ottnt{N_{{\mathrm{1}}}}  :  \ottnt{A_{{\mathrm{1}}}}  \hgap
       \ottmv{x}  \ottsym{:}  \ottnt{A_{{\mathrm{1}}}}   \vdash   \ottnt{N_{{\mathrm{2}}}}  :  \ottnt{A} 
    \]
    for some $\ottnt{N_{{\mathrm{1}}}},\ottnt{N_{{\mathrm{2}}}},\ottmv{x},\ottnt{A_{{\mathrm{1}}}}$.
    We use the IH with $  \emptyset    \vdash   \ottnt{N_{{\mathrm{1}}}}  :  \ottnt{A_{{\mathrm{1}}}} $.
    We proceed by case analysis on $\ottnt{N_{{\mathrm{1}}}}$.
    \begin{description}
    \item[\SubCase{$\ottnt{N_{{\mathrm{1}}}}  \ottsym{=}  \ottnt{V_{{\mathrm{1}}}}$}] Take $\ottnt{M'}  \ottsym{=}  \ottnt{N_{{\mathrm{2}}}}  [  \ottmv{x}  \ottsym{:=}  \ottnt{V_{{\mathrm{1}}}}  ]$ by \rnp{R-Let}.
    \item[\SubCase{$ \ottnt{N_{{\mathrm{1}}}}    \mathbin{ \evalto }    \ottnt{N'_{{\mathrm{1}}}} $}] Similar.
    \item[\SubCase{$\ottnt{N_{{\mathrm{1}}}}  \ottsym{=}  \ottkw{blame} \, \ottnt{p}$}] Similar.
    \end{description}

  \item[\Case{\rnp{T-Cmp}}]
    We are given
    \[
      \ottnt{M}  \ottsym{=}  \ottnt{N_{{\mathrm{1}}}}  \fatsemi  \ottnt{N_{{\mathrm{2}}}} \hgap
        \emptyset    \vdash   \ottnt{N_{{\mathrm{1}}}}  :  \ottnt{A_{{\mathrm{1}}}}  \rightsquigarrow  \ottnt{B}  \hgap
        \emptyset    \vdash   \ottnt{N_{{\mathrm{2}}}}  :  \ottnt{B}  \rightsquigarrow  \ottnt{A_{{\mathrm{2}}}}  \hgap
      \ottnt{A}  \ottsym{=}  \ottnt{A_{{\mathrm{1}}}}  \rightsquigarrow  \ottnt{A_{{\mathrm{2}}}}
    \]
    for some $\ottnt{N_{{\mathrm{1}}}},\ottnt{N_{{\mathrm{2}}}},\ottnt{A_{{\mathrm{1}}}},\ottnt{A_{{\mathrm{2}}}},\ottnt{B}$.
    We have the IHs for $  \emptyset    \vdash   \ottnt{N_{{\mathrm{1}}}}  :  \ottnt{A_{{\mathrm{1}}}}  \rightsquigarrow  \ottnt{B} $ and $  \emptyset    \vdash   \ottnt{N_{{\mathrm{2}}}}  :  \ottnt{B}  \rightsquigarrow  \ottnt{A_{{\mathrm{2}}}} $.
    We proceed by case analysis on $\ottnt{N_{{\mathrm{1}}}},\ottnt{N_{{\mathrm{2}}}}$.
    \begin{description}
    \item[\SubCase{$\ottnt{N_{{\mathrm{1}}}}  \ottsym{=}  \ottnt{V_{{\mathrm{1}}}}$ and $\ottnt{N_{{\mathrm{2}}}}  \ottsym{=}  \ottnt{V_{{\mathrm{2}}}}$}]
      By $  \emptyset    \vdash   \ottnt{V_{{\mathrm{1}}}}  :  \ottnt{A_{{\mathrm{1}}}}  \rightsquigarrow  \ottnt{B} $ and $  \emptyset    \vdash   \ottnt{V_{{\mathrm{2}}}}  :  \ottnt{B}  \rightsquigarrow  \ottnt{A_{{\mathrm{2}}}} $
      and Lemma~\ref{lem:canonical-S1}, we have
      $\ottnt{V_{{\mathrm{1}}}}  \ottsym{=}  \ottnt{s_{{\mathrm{1}}}}$ and $\ottnt{V_{{\mathrm{2}}}}  \ottsym{=}  \ottnt{s_{{\mathrm{2}}}}$ for some $\ottnt{s_{{\mathrm{1}}}},\ottnt{s_{{\mathrm{2}}}}$.
      Take $\ottnt{M'}  \ottsym{=}  \ottnt{s}  \fatsemi  \ottnt{t}$ by \rnp{R-Op}.
      (Here, $\ottnt{s}  \fatsemi  \ottnt{t}$ is defined by Lemma~\ref{lem:sc-cmp-welldef}.)
    \item[\Otherwise] Similar.
    \end{description}

  \item[\Case{\rnp{T-Crc}}]
    We are given
    \[
      \ottnt{M}  \ottsym{=}  \ottnt{N_{{\mathrm{1}}}}  \langle  \ottnt{N_{{\mathrm{2}}}}  \rangle \hgap
        \emptyset    \vdash   \ottnt{N_{{\mathrm{1}}}}  :  \ottnt{A_{{\mathrm{1}}}}  \hgap
        \emptyset    \vdash   \ottnt{N_{{\mathrm{2}}}}  :  \ottnt{A_{{\mathrm{1}}}}  \rightsquigarrow  \ottnt{A} 
    \]
    for some $\ottnt{N_{{\mathrm{1}}}},\ottnt{N_{{\mathrm{2}}}},\ottnt{A_{{\mathrm{1}}}}$.
    We have the IHs for $  \emptyset    \vdash   \ottnt{N_{{\mathrm{1}}}}  :  \ottnt{A_{{\mathrm{1}}}} $ and $  \emptyset    \vdash   \ottnt{N_{{\mathrm{2}}}}  :  \ottnt{A_{{\mathrm{1}}}}  \rightsquigarrow  \ottnt{A} $.
    We proceed by case analysis on $\ottnt{N_{{\mathrm{1}}}},\ottnt{N_{{\mathrm{2}}}}$.
    \begin{description}
    \item[\SubCase{$\ottnt{N_{{\mathrm{1}}}}  \ottsym{=}  \ottnt{V_{{\mathrm{1}}}}$ and $\ottnt{N_{{\mathrm{2}}}}  \ottsym{=}  \ottnt{V_{{\mathrm{2}}}}$}]
      By $  \emptyset    \vdash   \ottnt{N_{{\mathrm{2}}}}  :  \ottnt{A_{{\mathrm{1}}}}  \rightsquigarrow  \ottnt{A} $ and Lemma~\ref{lem:canonical-S1},
      we have $\ottnt{N_{{\mathrm{2}}}}  \ottsym{=}  \ottnt{t}$ for some $\ottnt{t}$.
      We proceed by case analysis on closed value $\ottnt{V_{{\mathrm{1}}}}$.
      \begin{description}
      \item[\SubCase{$\ottnt{N_{{\mathrm{1}}}}  \ottsym{=}  \ottnt{U}$}]
        By $  \emptyset    \vdash   \ottnt{U}  :  \ottnt{A_{{\mathrm{1}}}} $, we have $\ottnt{A_{{\mathrm{1}}}}  \ne  \mathord{\star}$.
        As the source type of $\ottnt{t}$ is nondynamic,
        we have either $\ottnt{t}  \ottsym{=}  \ottkw{id}, \bot^{ \ottnt{G}   \ottnt{p}   \ottnt{H} } ,\ottnt{d}$.
        Then, \rnp{R-Id} or \rnp{R-Fail} or \rnp{R-Crc} finishes the case.
        (Note: it might be the case that $\ottnt{U}  \ottsym{=}  \ottnt{s}$ for some $\ottnt{s}$;
        e.g., $  \ottkw{id} _{ \iota }   \langle   \ottkw{id} _{ \iota  \rightsquigarrow  \iota }   \rangle    \mathbin{ \evalto }     \ottkw{id} _{ \iota }  $.)

      \item[\SubCase{$\ottnt{N_{{\mathrm{1}}}}  \ottsym{=}  \ottnt{U}  \langle\!\langle  \ottnt{d}  \rangle\!\rangle$}]
        Take $\ottnt{M'}  \ottsym{=}  \ottnt{U}  \langle  \ottnt{d}  \mathbin{;\!;}  \ottnt{t}  \rangle$ by \rnp{R-MergeV}.
      \end{description}
    \item[\Otherwise] Similar.
    \end{description}

  \item[\Case{\rnp{T-CrcV}}] Immediate. ($\ottnt{M}  \ottsym{=}  \ottnt{U}  \langle\!\langle  \ottnt{V}  \rangle\!\rangle$ is a value.)
  \item[\Case{\rnp{T-Crcn}}] Immediate. ($\ottnt{M}  \ottsym{=}  \ottnt{s}$ is a value.)
  \item[\Case{\rnp{T-Blame}}] Immediate. ($\ottnt{M}  \ottsym{=}  \ottkw{blame} \, \ottnt{p}$) \qedhere
  \end{description}
\end{proof}

\begin{lemma}[Preservation of Types under Substitution]\label{lem:preservation-subst-S1}\leavevmode
  If $ \Gamma  \ottsym{,}  \ottmv{x}  \ottsym{:}  \ottnt{A}    \vdash_{\mathsf{S_1} }    \ottnt{M}  :  \ottnt{B} $ and $ \Gamma    \vdash_{\mathsf{S_1} }    \ottnt{V}  :  \ottnt{A} $, then
    $ \Gamma    \vdash_{\mathsf{S_1} }    \ottnt{M}  [  \ottmv{x}  \ottsym{:=}  \ottnt{V}  ]  :  \ottnt{B} $.
\end{lemma}
\begin{proof}
  By straightforward induction on the derivation of $ \Gamma  \ottsym{,}  \ottmv{x}  \ottsym{:}  \ottnt{A}    \vdash_{\mathsf{S_1} }    \ottnt{M}  :  \ottnt{B} $
  with case analysis on the rule applied last.
\end{proof}

\begin{lemma}[Preservation for Reduction]\label{lem:preservation-red-S1}
  If $  \emptyset     \vdash_{\mathsf{S_1} }    \ottnt{M}  :  \ottnt{A} $ and $ \ottnt{M}    \reduces_{\mathsf{S_1} }    \ottnt{N} $, then $  \emptyset     \vdash_{\mathsf{S_1} }    \ottnt{N}  :  \ottnt{A} $.
\end{lemma}
\begin{proof}
  By case analysis on the reduction rule applied to $ \ottnt{M}    \reduces_{\mathsf{S_1} }    \ottnt{N} $.
  \begin{description} 
    \item[\Case{\rnp{R-Op}}]
      We are given
      \[
        \ottnt{M}  \ottsym{=}  \ottnt{op}  \ottsym{(}  \ottnt{a_{{\mathrm{1}}}}  \ottsym{,}  \ottnt{a_{{\mathrm{2}}}}  \ottsym{)} \hgap
        \ottnt{N}  \ottsym{=}  \delta \, \ottsym{(}  \ottnt{op}  \ottsym{,}  \ottnt{a_{{\mathrm{1}}}}  \ottsym{,}  \ottnt{a_{{\mathrm{2}}}}  \ottsym{)}
      \]
      for some $\ottnt{a_{{\mathrm{1}}}},\ottnt{a_{{\mathrm{2}}}},\ottnt{a}$.
      By inversion on $  \emptyset    \vdash   \ottnt{op}  \ottsym{(}  \ottnt{a_{{\mathrm{1}}}}  \ottsym{,}  \ottnt{a_{{\mathrm{2}}}}  \ottsym{)}  :  \ottnt{A} $,
      \[
        \ottnt{A}  \ottsym{=}  \iota \hgap
         \metafun{ty} ( \ottnt{op} )   \ottsym{=}  \iota_{{\mathrm{1}}}  \rightarrow  \iota_{{\mathrm{2}}}  \rightarrow  \iota \hgap
          \emptyset    \vdash   \ottnt{a_{{\mathrm{1}}}}  :  \iota_{{\mathrm{1}}}  \hgap
          \emptyset    \vdash   \ottnt{a_{{\mathrm{2}}}}  :  \iota_{{\mathrm{2}}} 
      \]
      for some $\iota_{{\mathrm{1}}},\iota_{{\mathrm{2}}},\iota$.
      Assumptions on $ \delta $ (called $\delta$-typability) ensure that
      \[
        \delta \, \ottsym{(}  \ottnt{op}  \ottsym{,}  \ottnt{a_{{\mathrm{1}}}}  \ottsym{,}  \ottnt{a_{{\mathrm{2}}}}  \ottsym{)}  \ottsym{=}  \ottnt{a} \hgap
         \metafun{ty} ( \ottnt{a} )   \ottsym{=}  \iota
      \]
      for some constant $\ottnt{a}$.
      By \rnp{T-Const}, we have $  \emptyset    \vdash   \ottnt{a}  :  \iota $.

    \item[\Case{\rnp{R-Beta}}]
      We are given
      \[
        \ottnt{M}  \ottsym{=}   \ottsym{(}   \lambda  ( \ottmv{x} , \kappa ).\,  \ottnt{M_{{\mathrm{1}}}}   \ottsym{)} \, ( \ottnt{V} , \ottnt{W} )  \hgap
        \ottnt{N}  \ottsym{=}  \ottnt{M_{{\mathrm{1}}}}  [  \ottmv{x}  \ottsym{:=}  \ottnt{V}  \ottsym{,}  \kappa  \ottsym{:=}  \ottnt{W}  ]
      \]
      for some $\ottmv{x},\kappa,\ottnt{M_{{\mathrm{1}}}},\ottnt{V},\ottnt{W}$.
      By inversion on $  \emptyset    \vdash    \ottsym{(}   \lambda  ( \ottmv{x} , \kappa ).\,  \ottnt{M_{{\mathrm{1}}}}   \ottsym{)} \, ( \ottnt{V} , \ottnt{W} )   :  \ottnt{A} $,
      \[
          \emptyset    \vdash    \lambda  ( \ottmv{x} , \kappa ).\,  \ottnt{M_{{\mathrm{1}}}}   :  \ottnt{A_{{\mathrm{1}}}}  \Rightarrow  \ottnt{A_{{\mathrm{2}}}}  \hgap
          \emptyset    \vdash   \ottnt{V}  :  \ottnt{A_{{\mathrm{1}}}}  \hgap
          \emptyset    \vdash   \ottnt{W}  :  \ottnt{A_{{\mathrm{2}}}}  \rightsquigarrow  \ottnt{A} 
      \]
      for some $\ottnt{A_{{\mathrm{1}}}},\ottnt{A_{{\mathrm{2}}}}$.
      By inversion on the left judgment,
      \[
         \ottmv{x}  \ottsym{:}  \ottnt{A_{{\mathrm{1}}}}  \ottsym{,}  \kappa  \ottsym{:}  \ottnt{A_{{\mathrm{2}}}}  \rightsquigarrow  \ottmv{X}   \vdash   \ottnt{M_{{\mathrm{1}}}}  :  \ottmv{X} 
      \]
      for some $\ottmv{X}$. Thus, we have  
      \[
         \ottmv{x}  \ottsym{:}  \ottnt{A_{{\mathrm{1}}}}  \ottsym{,}  \kappa  \ottsym{:}  \ottnt{A_{{\mathrm{2}}}}  \rightsquigarrow  \ottnt{A}   \vdash   \ottnt{M_{{\mathrm{1}}}}  :  \ottnt{A} 
      \]
      (by type substitution of $\ottnt{A}$ for $\ottmv{X}$).
      By Lemma~\ref{lem:preservation-subst-S1} (twice),
      $  \emptyset    \vdash   \ottnt{M_{{\mathrm{1}}}}  [  \ottmv{x}  \ottsym{:=}  \ottnt{V}  \ottsym{,}  \kappa  \ottsym{:=}  \ottnt{W}  ]  :  \ottnt{A} $ follows.
    \item[\Case{\rnp{R-Wrap}}]
      We are given
      \[
        \ottnt{M}  \ottsym{=}   \ottsym{(}  \ottnt{U}  \langle\!\langle  \ottnt{s}  \Rightarrow  \ottnt{t}  \rangle\!\rangle  \ottsym{)} \, ( \ottnt{V} , \ottnt{W} )  \hgap
        \ottnt{N}  \ottsym{=}    \ottkw{let} \,  \kappa = \ottnt{t}  \mathbin{;\!;}  \ottnt{W} \, \ottkw{in}\,  \ottnt{U}  \, ( \ottnt{V}  \langle  \ottnt{s}  \rangle , \kappa ) 
      \]
      for some $\ottnt{U},\ottnt{s},\ottnt{t},\ottnt{V},\ottnt{W},\kappa$.
      By inversion on $  \emptyset    \vdash    \ottsym{(}  \ottnt{U}  \langle\!\langle  \ottnt{s}  \rightarrow  \ottnt{t}  \rangle\!\rangle  \ottsym{)} \, ( \ottnt{V} , \ottnt{W} )   :  \ottnt{A} $,
      \[
          \emptyset    \vdash   \ottnt{U}  \langle\!\langle  \ottnt{s}  \Rightarrow  \ottnt{t}  \rangle\!\rangle  :  \ottnt{A_{{\mathrm{1}}}}  \Rightarrow  \ottnt{A_{{\mathrm{2}}}}  \hgap
          \emptyset    \vdash   \ottnt{V}  :  \ottnt{A_{{\mathrm{1}}}}  \hgap
          \emptyset    \vdash   \ottnt{W}  :  \ottnt{A_{{\mathrm{2}}}}  \rightsquigarrow  \ottnt{A} .
      \]
      for some $\ottnt{A_{{\mathrm{1}}}},\ottnt{A_{{\mathrm{2}}}}$.
      By inversion on the left judgment,  
      \[
          \emptyset    \vdash   \ottnt{U}  :  \ottnt{A'}  \hgap
          \emptyset    \vdash   \ottnt{s}  \Rightarrow  \ottnt{t}  :  \ottnt{A'}  \rightsquigarrow  \ottsym{(}  \ottnt{A_{{\mathrm{1}}}}  \Rightarrow  \ottnt{A_{{\mathrm{2}}}}  \ottsym{)} .
      \]
      for some $\ottnt{A'}$.
      By inversion on the right judgment,  
      \[
        \ottnt{A'}  \ottsym{=}  \ottnt{A'_{{\mathrm{1}}}}  \Rightarrow  \ottnt{A'_{{\mathrm{2}}}} \hgap
          \emptyset    \vdash   \ottnt{s}  :  \ottnt{A_{{\mathrm{1}}}}  \rightsquigarrow  \ottnt{A'_{{\mathrm{1}}}}  \hgap
          \emptyset    \vdash   \ottnt{t}  :  \ottnt{A'_{{\mathrm{2}}}}  \rightsquigarrow  \ottnt{A_{{\mathrm{2}}}} 
      \]
      for some $\ottnt{A'_{{\mathrm{1}}}},\ottnt{A'_{{\mathrm{2}}}}$.
      By \rnp{T-Crc} and \rnp{T-Cmp}, we have
      \[
          \emptyset    \vdash   \ottnt{V}  \langle  \ottnt{s}  \rangle  :  \ottnt{A'_{{\mathrm{1}}}}  \hgap
          \emptyset    \vdash   \ottnt{t}  \mathbin{;\!;}  \ottnt{W}  :  \ottnt{A'_{{\mathrm{2}}}}  \rightsquigarrow  \ottnt{A} .
      \]
      Then, \rnp{T-Let} and \rnp{T-App} finish this case.

    \item[\Case{\rnp{R-Let}}]
      We are given
      \[
        \ottnt{M}  \ottsym{=}   \ottkw{let} \,  \ottmv{x} = \ottnt{V} \, \ottkw{in}\,  \ottnt{M_{{\mathrm{1}}}}  \hgap
        \ottnt{N}  \ottsym{=}  \ottnt{M_{{\mathrm{1}}}}  [  \ottmv{x}  \ottsym{:=}  \ottnt{V}  ]
      \]
      for some $\ottmv{x},\ottnt{V},\ottnt{M_{{\mathrm{1}}}}$.
      By inversion on $  \emptyset    \vdash    \ottkw{let} \,  \ottmv{x} = \ottnt{V} \, \ottkw{in}\,  \ottnt{M_{{\mathrm{1}}}}   :  \ottnt{A} $,
      \[
          \emptyset    \vdash   \ottnt{V}  :  \ottnt{A_{{\mathrm{1}}}}  \hgap
         \ottmv{x}  \ottsym{:}  \ottnt{A_{{\mathrm{1}}}}   \vdash   \ottnt{M_{{\mathrm{1}}}}  :  \ottnt{A} 
      \]
      for some $\ottnt{A_{{\mathrm{1}}}}$.
      By Lemma~\ref{lem:preservation-subst-S1},
      $  \emptyset    \vdash   \ottnt{M_{{\mathrm{1}}}}  [  \ottmv{x}  \ottsym{:=}  \ottnt{V}  ]  :  \ottnt{A} $ follows.

    \item[\Case{\rnp{R-Cmp}}]
      We are given
      \[
        \ottnt{M}  \ottsym{=}  \ottnt{s}  \mathbin{;\!;}  \ottnt{t} \hgap
        \ottnt{N}  \ottsym{=}  \ottnt{s}  \fatsemi  \ottnt{t}
      \]
      for some $\ottnt{s},\ottnt{t}$.
      By inversion on $  \emptyset    \vdash   \ottnt{s}  \mathbin{;\!;}  \ottnt{t}  :  \ottnt{A} $,
      \[
          \emptyset    \vdash   \ottnt{s}  :  \ottnt{A_{{\mathrm{1}}}}  \rightsquigarrow  \ottnt{B}  \hgap
          \emptyset    \vdash   \ottnt{t}  :  \ottnt{B}  \rightsquigarrow  \ottnt{A_{{\mathrm{2}}}}  \hgap
        \ottnt{A}  \ottsym{=}  \ottnt{A_{{\mathrm{1}}}}  \rightsquigarrow  \ottnt{A_{{\mathrm{2}}}}
      \]
      for some $\ottnt{A_{{\mathrm{1}}}},\ottnt{A_{{\mathrm{2}}}},\ottnt{B}$.
      By Lemma~\ref{lem:cmp-ty}, we have $  \emptyset    \vdash   \ottsym{(}  \ottnt{s}  \fatsemi  \ottnt{t}  \ottsym{)}  :  \ottnt{A_{{\mathrm{1}}}}  \rightsquigarrow  \ottnt{A_{{\mathrm{2}}}} $.

    \item[\Case{\rnp{R-Id}}]
      We are given
      \[
        \ottnt{M}  \ottsym{=}  \ottnt{U}  \langle   \ottkw{id} _{ \ottnt{A} }   \rangle \hgap
        \ottnt{N}  \ottsym{=}  \ottnt{U}
      \]
      for some $\ottnt{U},\ottnt{A}$.
      By inversion on $  \emptyset    \vdash   \ottnt{U}  \langle   \ottkw{id} _{ \ottnt{A} }   \rangle  :  \ottnt{A} $,
      \[
          \emptyset    \vdash   \ottnt{U}  :  \ottnt{A'}  \hgap
          \emptyset    \vdash    \ottkw{id} _{ \ottnt{A} }   :  \ottnt{A'}  \rightsquigarrow  \ottnt{A} 
      \]
      for some $\ottnt{A'}$.
      By $\ottnt{A'}  \ottsym{=}  \ottnt{A}$, we have $  \emptyset    \vdash   \ottnt{U}  :  \ottnt{A} $.

    \item[\Case{\rnp{R-Fail}}] By \rnp{T-Blame}.
    \item[\Case{\rnp{R-Crc}}] By \rnp{T-CrcV}.  

    \item[\Case{\rnp{R-MergeV}}] 
      We are given
      \[
        \ottnt{M}  \ottsym{=}  \ottnt{U}  \langle\!\langle  \ottnt{d}  \rangle\!\rangle  \langle  \ottnt{t}  \rangle \hgap
        \ottnt{N}  \ottsym{=}  \ottnt{U}  \langle  \ottnt{d}  \mathbin{;\!;}  \ottnt{t}  \rangle
      \]
      for some $\ottnt{U},\ottnt{d},\ottnt{t}$.
      By inversion on $  \emptyset    \vdash   \ottnt{U}  \langle\!\langle  \ottnt{d}  \rangle\!\rangle  \langle  \ottnt{t}  \rangle  :  \ottnt{A} $,
      \[
          \emptyset    \vdash   \ottnt{U}  \langle\!\langle  \ottnt{d}  \rangle\!\rangle  :  \ottnt{A'}  \hgap
          \emptyset    \vdash   \ottnt{t}  :  \ottnt{A'}  \rightsquigarrow  \ottnt{A} 
      \]
      for some $\ottnt{A'}$.
      By inversion on the left judgment,
      \[
          \emptyset    \vdash   \ottnt{U}  :  \ottnt{A''}  \hgap
          \emptyset    \vdash   \ottnt{d}  :  \ottnt{A''}  \rightsquigarrow  \ottnt{A'} 
      \]
      for some $\ottnt{A''}$.
      By \rnp{T-Cmp}, $  \emptyset    \vdash   \ottsym{(}  \ottnt{d}  \mathbin{;\!;}  \ottnt{t}  \ottsym{)}  :  \ottnt{A''}  \rightsquigarrow  \ottnt{A} $.
      By \rnp{T-Crc}, $  \emptyset    \vdash   \ottnt{U}  \langle  \ottnt{d}  \mathbin{;\!;}  \ottnt{t}  \rangle  :  \ottnt{A} $. \qedhere
  \end{description}
\end{proof}

\iffull\thmPreservationSx*
\else
\begin{theorem}[Preservation]\label{thm:preservation-S1}
  If $  \emptyset     \vdash_{\mathsf{S_1} }    \ottnt{M}  :  \ottnt{A} $ and $ \ottnt{M}    \mathbin{  \evalto_{\mathsf{S_1} }  }    \ottnt{N} $, then $  \emptyset     \vdash_{\mathsf{S_1} }    \ottnt{N}  :  \ottnt{A} $.
\end{theorem}
\fi
\begin{proof}
  By case analysis on the evaluation rule applied to $ \ottnt{M}    \mathbin{  \evalto_{\mathsf{S_1} }  }    \ottnt{N} $.
  \begin{description}
  \item[\Case{\rnp{E-Ctx}}]  
    We are given
    \[
       \ottnt{M_{{\mathrm{1}}}}   \reduces   \ottnt{N_{{\mathrm{1}}}}  \hgap
      \ottnt{M}  \ottsym{=}  \mathcal{E}  [  \ottnt{M_{{\mathrm{1}}}}  ] \hgap
      \ottnt{N}  \ottsym{=}  \mathcal{E}  [  \ottnt{N_{{\mathrm{1}}}}  ]
    \]
    for some $\mathcal{E},\ottnt{M_{{\mathrm{1}}}},\ottnt{N_{{\mathrm{1}}}}$.
    We have derivation $\mathcal{D}$ of $  \emptyset    \vdash   \mathcal{E}  [  \ottnt{M_{{\mathrm{1}}}}  ]  :  \ottnt{A} $.
    In derivation $\mathcal{D}$, there exists
    subderivation $\mathcal{D}_1$ of $  \emptyset    \vdash   \ottnt{M_{{\mathrm{1}}}}  :  \ottnt{A_{{\mathrm{1}}}} $ for some $\ottnt{A_{{\mathrm{1}}}}$.
    By $ \ottnt{M_{{\mathrm{1}}}}   \reduces   \ottnt{N_{{\mathrm{1}}}} $ and Lemma~\ref{lem:preservation-red-S1},
    we have derivation $\mathcal{D}_2$ of $  \emptyset    \vdash   \ottnt{N_{{\mathrm{1}}}}  :  \ottnt{A_{{\mathrm{1}}}} $.
    Thus, we can form derivation of $  \emptyset    \vdash   \mathcal{E}  [  \ottnt{N_{{\mathrm{1}}}}  ]  :  \ottnt{A} $
    by substituting $\mathcal{D}_2$ for $\mathcal{D}_1$ in $\mathcal{D}$.
    We have $  \emptyset    \vdash   \ottnt{N}  :  \ottnt{A} $.
    (More precisely, by induction on $\mathcal{E}$.)
  \item[\Case{\rnp{E-Abort}}]
    We are given
    \[
      \ottnt{M}  \ottsym{=}  \mathcal{E}  [  \ottkw{blame} \, \ottnt{p}  ] \hgap
      \ottnt{N}  \ottsym{=}  \ottkw{blame} \, \ottnt{p}
    \]
    for some $\mathcal{E},\ottnt{p}$.
    By \rnp{T-Blame}, $  \emptyset    \vdash   \ottkw{blame} \, \ottnt{p}  :  \ottnt{A} $. \qedhere
  \end{description}
\end{proof}

\iffull\corSafetySx*
\else
\begin{corollary}[Type Safety]\label{cor:safety-S1}
  If $  \emptyset     \vdash_{\mathsf{S_1} }    \ottnt{M}  :  \ottnt{A} $, then one of the following holds.
  \begin{enumerate}
  \item $ \ottnt{M}    \mathbin{  \evalto_{\mathsf{S_1} }  ^*}    \ottnt{V} $ and $  \emptyset     \vdash_{\mathsf{S_1} }    \ottnt{V}  :  \ottnt{A} $ for some $\ottnt{V}$.
  \item $ \ottnt{M}    \mathbin{  \evalto_{\mathsf{S_1} }  ^*}    \ottkw{blame} \, \ottnt{p} $ for some $\ottnt{p}$.
  \item $\ottnt{M} \,  \mathord{\Uparrow_{\mathsf{S_1} } } $.
  \end{enumerate}
\end{corollary}
\fi
\begin{proof}
  By Theorem~\ref{thm:progress-S1} and Theorem~\ref{thm:preservation-S1}.
\end{proof}

\subsection{Translation Preserves Typing}

\iffull
\else
\begin{definition}
  Let $ \Psi ( \Gamma ) $ be the type environment satisfying the following:
  \begin{center}
    $ ( \ottmv{x}  :  \ottnt{A} ) \in  \Gamma $ if and only if $ ( \ottmv{x}  :   \Psi (\maybebluetext{ \ottnt{A} })  ) \in   \Psi ( \Gamma )  $.
  \end{center}
\end{definition}
\fi

\begin{lemma}
  If $\ottnt{s}  \ottsym{:}  \ottnt{A}  \rightsquigarrow  \ottnt{B}$ in \lamS,
  then $  \emptyset     \vdash_{\mathsf{S_1} }     \Psi (\maybebluetext{ \ottnt{s} })   :   \Psi (\maybebluetext{ \ottnt{A} })   \rightsquigarrow   \Psi (\maybebluetext{ \ottnt{B} })  $.
\end{lemma}
\begin{proof}
  By case analysis on $\ottnt{s}$.
\end{proof}

\iffull\thmTransTyping*
\else
\begin{theorem}[Translation Preserves Typing] 
  \leavevmode
  \begin{enumerate}
  \item   If $ \Gamma    \vdash_{\mathsf{S} }    \ottnt{M}  :  \ottnt{A} $ and $\ottnt{s}  \ottsym{:}  \ottnt{A}  \rightsquigarrow  \ottnt{B}$ ,
  then $  \Psi ( \Gamma )     \vdash_{\mathsf{S_1} }    \ottsym{(}   \mathscr{K}\llbracket \maybebluetext{ \ottnt{M} } \rrbracket   \Psi (\maybebluetext{ \ottnt{s} })    \ottsym{)}  :   \Psi (\maybebluetext{ \ottnt{B} })  $.
  \item If $ \Gamma    \vdash_{\mathsf{S} }    \ottnt{V}  :  \ottnt{A} $,
  then $  \Psi ( \Gamma )     \vdash_{\mathsf{S_1} }     \Psi (\maybebluetext{ \ottnt{V} })   :   \Psi (\maybebluetext{ \ottnt{A} })  $.
  \end{enumerate}
\end{theorem}
\fi
\begin{proof}
  Simultaneously proved by induction on the derivation of $ \Gamma    \vdash_{\mathsf{S} }    \ottnt{M}  :  \ottnt{A} $ and $ \Gamma    \vdash_{\mathsf{S} }    \ottnt{V}  :  \ottnt{A} $.
\end{proof}

\subsection{Translation Preserves Semantics}

\begin{lemma}[Composition]\label{lem:cmp-trans}
  If $\ottnt{s}  \fatsemi  \ottnt{t}  \ottsym{=}  \ottnt{s'}$ in \lamS, then
  $ \Psi (\maybebluetext{ \ottnt{s} })   \fatsemi   \Psi (\maybebluetext{ \ottnt{t} })   \ottsym{=}   \Psi (\maybebluetext{ \ottnt{s'} }) $.
\end{lemma}
\begin{proof}
  By induction on the derivation of $\ottnt{s}  \fatsemi  \ottnt{t}  \ottsym{=}  \ottnt{s'}$.
\end{proof}

\begin{lemma} \label{lem:merge-id}
  If $\ottnt{s}  \ottsym{:}  \ottnt{A}  \rightsquigarrow  \ottnt{B}$ and $\ottkw{id}  \ottsym{:}   \Psi (\maybebluetext{ \ottnt{B} })   \rightsquigarrow   \Psi (\maybebluetext{ \ottnt{B} }) $, then
  $ \Psi (\maybebluetext{ \ottnt{s} })   \fatsemi  \ottkw{id}  \ottsym{=}   \Psi (\maybebluetext{ \ottnt{s} }) $.
\end{lemma}
\begin{proof}
  Easy case analysis on $\ottnt{s}$.
\end{proof}

\begin{lemma}\label{lem:val-coe-val}
  $ \Psi (\maybebluetext{ \ottnt{U} }) $ is an uncoerced value and $ \Psi (\maybebluetext{ \ottnt{V} }) $ is a value.
\end{lemma}
\begin{proof}
  Easy.
\end{proof}

\begin{lemma}\label{lem:val-coe-id}
  If $  \emptyset     \vdash_{\mathsf{S} }     \Psi (\maybebluetext{ \ottnt{V} })   :  \ottnt{A} $ and $\ottkw{id}  \ottsym{:}  \ottnt{A}  \rightsquigarrow  \ottnt{A}$, then
  $  \Psi (\maybebluetext{ \ottnt{V} })   \langle  \ottkw{id}  \rangle    \mathbin{  \accentset{\mathsf{c} }{\evalto}_{\mathsf{S_1} }  ^*}     \Psi (\maybebluetext{ \ottnt{V} })  $.
\end{lemma}
\begin{proof}
  By case analysis on $\ottnt{V}$.
  (Note that $\ottnt{V}$ is closed.)
  \begin{description}
  \item[\Case{$\ottnt{V}  \ottsym{=}  \ottmv{x}$}] Cannot happen.
  \item[\Case{$\ottnt{V}  \ottsym{=}  \ottnt{U}$}] By \rnp{R-Id}.  
  \item[\Case{$\ottnt{V}  \ottsym{=}  \ottnt{U}  \langle\!\langle  \ottnt{d}  \rangle\!\rangle$}]
    \[
    \begin{array}[b]{llll}
       \Psi (\maybebluetext{ \ottnt{U}  \langle\!\langle  \ottnt{d}  \rangle\!\rangle })   \langle  \ottkw{id}  \rangle
      &=&  \Psi (\maybebluetext{ \ottnt{U} })   \langle\!\langle   \Psi (\maybebluetext{ \ottnt{d} })   \rangle\!\rangle  \langle  \ottkw{id}  \rangle \\
      & \accentset{\mathsf{c} }{\evalto} &  \Psi (\maybebluetext{ \ottnt{U} })   \langle   \Psi (\maybebluetext{ \ottnt{d} })   \mathbin{;\!;}  \ottkw{id}  \rangle &\text{by \rnp{R-Merge}} \\
      & \accentset{\mathsf{c} }{\evalto} &  \Psi (\maybebluetext{ \ottnt{U} })   \langle   \Psi (\maybebluetext{ \ottnt{d} })   \fatsemi  \ottkw{id}  \rangle &\text{by \rnp{R-Cmp}} \\
      &=&  \Psi (\maybebluetext{ \ottnt{U} })   \langle   \Psi (\maybebluetext{ \ottnt{d} })   \rangle &\text{by Lemma~\ref{lem:merge-id}} \\
      & \accentset{\mathsf{c} }{\evalto} &  \Psi (\maybebluetext{ \ottnt{U} })   \langle\!\langle   \Psi (\maybebluetext{ \ottnt{d} })   \rangle\!\rangle &\text{by \rnp{R-Crc}} \\
      &=&  \Psi (\maybebluetext{ \ottnt{U}  \langle\!\langle  \ottnt{d}  \rangle\!\rangle }) .
    \end{array}\qedhere
    \]
  \end{description}
\end{proof}

\subsubsection{Substitution}

The definition of substitution is standard.
Here are selected cases from its definition:
\begin{align*}
  \ottsym{(}  \ottnt{M}  \mathbin{;\!;}  \ottnt{N}  \ottsym{)}  [  \ottmv{x}  \ottsym{:=}  \ottnt{V}  ] &= \ottsym{(}  \ottnt{M}  [  \ottmv{x}  \ottsym{:=}  \ottnt{V}  ]  \ottsym{)}  \mathbin{;\!;}  \ottsym{(}  \ottnt{N}  [  \ottmv{x}  \ottsym{:=}  \ottnt{V}  ]  \ottsym{)} \\
  \ottsym{(}   \ottkw{let} \,  \kappa = \ottnt{M} \, \ottkw{in}\,  \ottnt{N}   \ottsym{)}  [  \ottmv{x}  \ottsym{:=}  \ottnt{V}  ] &=  \ottkw{let} \,  \kappa = \ottnt{M}  [  \ottmv{x}  \ottsym{:=}  \ottnt{V}  ] \, \ottkw{in}\,  \ottnt{N}   [  \ottmv{x}  \ottsym{:=}  \ottnt{V}  ] \\
  \ottnt{M}  \langle  \ottnt{N}  \rangle  [  \ottmv{x}  \ottsym{:=}  \ottnt{V}  ] &= \ottnt{M}  [  \ottmv{x}  \ottsym{:=}  \ottnt{V}  ]  \langle  \ottnt{N}  [  \ottmv{x}  \ottsym{:=}  \ottnt{V}  ]  \rangle \\
  \ottsym{(}   \ottnt{L} \, ( \ottnt{M} , \ottnt{N} )   \ottsym{)}  [  \ottmv{x}  \ottsym{:=}  \ottnt{V}  ] &=  \ottsym{(}  \ottnt{L}  [  \ottmv{x}  \ottsym{:=}  \ottnt{V}  ]  \ottsym{)} \, ( \ottnt{M}  [  \ottmv{x}  \ottsym{:=}  \ottnt{V}  ] , \ottnt{N}  [  \ottmv{x}  \ottsym{:=}  \ottnt{V}  ] ) .
\end{align*}

\begin{lemma} \label{lem:fv-colon}
  $ \metafun{FV} (  \mathscr{K}\llbracket \maybebluetext{ \ottnt{M} } \rrbracket  \ottnt{K}  )   \ottsym{=}   \metafun{FV} ( \ottnt{M} )  \, \cup \,  \metafun{FV} ( \ottnt{K} ) $
\end{lemma}
\begin{proof}
  By induction on the derivation of $ \mathscr{K}\llbracket \maybebluetext{ \ottnt{M} } \rrbracket  \ottnt{K} $.
\end{proof}

\begin{lemma}[Substitution for a non-continuation variable]\label{lem:subst-x} \leavevmode
  \begin{enumerate}
  \item $ \mathscr{C}\llbracket \maybebluetext{ \ottnt{M} } \rrbracket   [  \ottmv{x}  \ottsym{:=}   \Psi (\maybebluetext{ \ottnt{V} })   ]  \ottsym{=}   \mathscr{C}\llbracket \maybebluetext{ \ottnt{M}  [  \ottmv{x}  \ottsym{:=}  \ottnt{V}  ] } \rrbracket $
  \item If $ \ottmv{x} \notin   \metafun{FV} ( \ottnt{K} )  $, then
    $\ottsym{(}   \mathscr{K}\llbracket \maybebluetext{ \ottnt{M} } \rrbracket  \ottnt{K}   \ottsym{)}  [  \ottmv{x}  \ottsym{:=}   \Psi (\maybebluetext{ \ottnt{V} })   ] =  \mathscr{K}\llbracket \maybebluetext{ \ottnt{M}  [  \ottmv{x}  \ottsym{:=}  \ottnt{V}  ] } \rrbracket  \ottnt{K} $.
  \end{enumerate}
\end{lemma}
\begin{proof}
  The two items are simultaneously proved 
  by induction on the derivations of $ \mathscr{C}\llbracket \maybebluetext{ \ottnt{M} } \rrbracket $ and $ \mathscr{K}\llbracket \maybebluetext{ \ottnt{M} } \rrbracket  \ottnt{K} $.

  (1) By case analysis on the form of $\ottnt{M}$.
  We first consider the cases where $\ottnt{M}$ is a value: $\ottnt{M}  \ottsym{=}  \ottnt{W}$.
  We must show $ \Psi (\maybebluetext{ \ottnt{W} })   [  \ottmv{x}  \ottsym{:=}   \Psi (\maybebluetext{ \ottnt{V} })   ]  \ottsym{=}   \Psi (\maybebluetext{ \ottnt{W}  [  \ottmv{x}  \ottsym{:=}  \ottnt{V}  ] }) $.
  (Note that $\ottnt{W}  [  \ottmv{x}  \ottsym{:=}  \ottnt{V}  ]$ is a value.)  
  \begin{description}
  \item[\Case{$\ottnt{W}  \ottsym{=}  \ottmv{x}$}]
    \[
    \begin{array}{llll}
       \Psi (\maybebluetext{ \ottmv{x} })   [  \ottmv{x}  \ottsym{:=}   \Psi (\maybebluetext{ \ottnt{V} })   ] &=& \ottmv{x}  [  \ottmv{x}  \ottsym{:=}   \Psi (\maybebluetext{ \ottnt{V} })   ] =  \Psi (\maybebluetext{ \ottnt{V} })  & \text{and} \\
       \Psi (\maybebluetext{ \ottmv{x}  [  \ottmv{x}  \ottsym{:=}  \ottnt{V}  ] })  &=&  \Psi (\maybebluetext{ \ottnt{V} }) .
    \end{array}
    \]
  \item[\Case{$\ottnt{W}  \ottsym{=}  \ottmv{y} \ne x$}]
    \[
    \begin{array}{llll}
       \Psi (\maybebluetext{ \ottmv{y} })   [  \ottmv{x}  \ottsym{:=}   \Psi (\maybebluetext{ \ottnt{V} })   ] &=& \ottmv{y}  [  \ottmv{x}  \ottsym{:=}   \Psi (\maybebluetext{ \ottnt{V} })   ] = \ottmv{y}  & \text{and} \\
       \Psi (\maybebluetext{ \ottmv{y}  [  \ottmv{x}  \ottsym{:=}  \ottnt{V}  ] })  &=&  \Psi (\maybebluetext{ \ottmv{y} })  = \ottmv{y}.
    \end{array}
    \]
  \item[\Case{$\ottnt{W}  \ottsym{=}  \ottnt{a}$}]
    \[
    \begin{array}{llll}
       \Psi (\maybebluetext{ \ottnt{a} })   [  \ottmv{x}  \ottsym{:=}   \Psi (\maybebluetext{ \ottnt{V} })   ] &=& \ottnt{a}  [  \ottmv{x}  \ottsym{:=}   \Psi (\maybebluetext{ \ottnt{V} })   ] = \ottnt{a}  & \text{and} \\
       \Psi (\maybebluetext{ \ottnt{a}  [  \ottmv{x}  \ottsym{:=}  \ottnt{V}  ] })  &=&  \Psi (\maybebluetext{ \ottnt{a} })  = \ottnt{a}.
    \end{array}
    \]
  \item[\Case{$\ottnt{W}  \ottsym{=}   \lambda   \ottmv{y} .\,  \ottnt{N} $}]
    We can assume $\ottmv{y}  \ne  \ottmv{x}$.
    \[
    \begin{array}{llll}
       \Psi (\maybebluetext{  \lambda   \ottmv{y} .\,  \ottnt{N}  })   [  \ottmv{x}  \ottsym{:=}   \Psi (\maybebluetext{ \ottnt{V} })   ]
      &=& \ottsym{(}   \lambda  ( \ottmv{y} , \kappa ).\,  \ottsym{(}   \mathscr{K}\llbracket \maybebluetext{ \ottnt{N} } \rrbracket  \kappa   \ottsym{)}   \ottsym{)}  [  \ottmv{x}  \ottsym{:=}   \Psi (\maybebluetext{ \ottnt{V} })   ] \\
      &=&  \lambda  ( \ottmv{y} , \kappa ).\,  \ottsym{(}  \ottsym{(}   \mathscr{K}\llbracket \maybebluetext{ \ottnt{N} } \rrbracket  \kappa   \ottsym{)}  [  \ottmv{x}  \ottsym{:=}   \Psi (\maybebluetext{ \ottnt{V} })   ]  \ottsym{)} .
    \end{array}
    \]
    Then,
    \[
    \begin{array}{llll}
       \Psi (\maybebluetext{ \ottsym{(}   \lambda   \ottmv{y} .\,  \ottnt{N}   \ottsym{)}  [  \ottmv{x}  \ottsym{:=}  \ottnt{V}  ] }) 
      &=&  \Psi (\maybebluetext{  \lambda   \ottmv{y} .\,  \ottnt{N}  [  \ottmv{x}  \ottsym{:=}  \ottnt{V}  ]  })  \\
      &=&  \lambda  ( \ottmv{y} , \kappa ).\,  \ottsym{(}   \mathscr{K}\llbracket \maybebluetext{ \ottnt{N}  [  \ottmv{x}  \ottsym{:=}  \ottnt{V}  ] } \rrbracket  \kappa   \ottsym{)} .
    \end{array}
    \]
    By the IH,
    $\ottsym{(}   \mathscr{K}\llbracket \maybebluetext{ \ottnt{N} } \rrbracket  \kappa   \ottsym{)}  [  \ottmv{x}  \ottsym{:=}   \Psi (\maybebluetext{ \ottnt{V} })   ]  \ottsym{=}   \mathscr{K}\llbracket \maybebluetext{ \ottnt{N}  [  \ottmv{x}  \ottsym{:=}  \ottnt{V}  ] } \rrbracket  \kappa $,
    which finishes this case.

    \item[\Case{$\ottnt{W}  \ottsym{=}  \ottnt{U}  \langle\!\langle  \ottnt{d}  \rangle\!\rangle$}]
    \[
    \begin{array}{llll}
       \Psi (\maybebluetext{ \ottnt{U}  \langle\!\langle  \ottnt{d}  \rangle\!\rangle })   [  \ottmv{x}  \ottsym{:=}   \Psi (\maybebluetext{ \ottnt{V} })   ]
      &=&  \Psi (\maybebluetext{ \ottnt{U} })   \langle\!\langle   \Psi (\maybebluetext{ \ottnt{d} })   \rangle\!\rangle  [  \ottmv{x}  \ottsym{:=}   \Psi (\maybebluetext{ \ottnt{V} })   ] \\
      &=&  \Psi (\maybebluetext{ \ottnt{U} })   [  \ottmv{x}  \ottsym{:=}   \Psi (\maybebluetext{ \ottnt{V} })   ]  \langle\!\langle   \Psi (\maybebluetext{ \ottnt{d} })   \rangle\!\rangle.
    \end{array}
    \]
    Since $\ottnt{U}  [  \ottmv{x}  \ottsym{:=}  \ottnt{V}  ]$ is an uncoerced value,
    \[
    \begin{array}{llll}
       \Psi (\maybebluetext{ \ottnt{U}  \langle\!\langle  \ottnt{d}  \rangle\!\rangle  [  \ottmv{x}  \ottsym{:=}  \ottnt{V}  ] }) 
      &=&  \Psi (\maybebluetext{ \ottnt{U}  [  \ottmv{x}  \ottsym{:=}  \ottnt{V}  ]  \langle\!\langle  \ottnt{d}  \rangle\!\rangle })  \\
      &=&  \Psi (\maybebluetext{ \ottnt{U}  [  \ottmv{x}  \ottsym{:=}  \ottnt{V}  ] })   \langle\!\langle   \Psi (\maybebluetext{ \ottnt{d} })   \rangle\!\rangle.
    \end{array}
    \]
    By the IH,
    $ \Psi (\maybebluetext{ \ottnt{U} })   [  \ottmv{x}  \ottsym{:=}   \Psi (\maybebluetext{ \ottnt{V} })   ]  \ottsym{=}   \Psi (\maybebluetext{ \ottnt{U}  [  \ottmv{x}  \ottsym{:=}  \ottnt{V}  ] }) $,
    which finishes this subcase.
  \end{description}
  We then consider the cases where $\ottnt{M}$ is not a value.
  \begin{description}
  \item[\Case{$\ottnt{M}  \ottsym{=}  \ottnt{N}  \langle  \ottnt{s}  \rangle$}]
    \[
    \begin{array}{lll}
       \mathscr{C}\llbracket \maybebluetext{ \ottnt{N}  \langle  \ottnt{s}  \rangle } \rrbracket   [  \ottmv{x}  \ottsym{:=}   \Psi (\maybebluetext{ \ottnt{V} })   ]
      &=& \ottsym{(}   \mathscr{K}\llbracket \maybebluetext{ \ottnt{N} } \rrbracket   \Psi (\maybebluetext{ \ottnt{s} })    \ottsym{)}  [  \ottmv{x}  \ottsym{:=}   \Psi (\maybebluetext{ \ottnt{V} })   ]
    \end{array}
    \]
    Then,
    \[
    \begin{array}{lll}
       \mathscr{C}\llbracket \maybebluetext{ \ottsym{(}  \ottnt{N}  \langle  \ottnt{s}  \rangle  \ottsym{)}  [  \ottmv{x}  \ottsym{:=}  \ottnt{V}  ] } \rrbracket 
      &=&  \mathscr{C}\llbracket \maybebluetext{ \ottsym{(}  \ottnt{N}  [  \ottmv{x}  \ottsym{:=}  \ottnt{V}  ]  \ottsym{)}  \langle  \ottnt{s}  \rangle } \rrbracket  \\
      &=&  \mathscr{K}\llbracket \maybebluetext{ \ottnt{N}  [  \ottmv{x}  \ottsym{:=}  \ottnt{V}  ] } \rrbracket   \Psi (\maybebluetext{ \ottnt{s} })  
    \end{array}
    \]
    We have $ \ottmv{x} \notin   \metafun{FV} (  \Psi (\maybebluetext{ \ottnt{s} })  )   = \emptyset$.
    By the IH, $\ottsym{(}   \mathscr{K}\llbracket \maybebluetext{ \ottnt{N} } \rrbracket   \Psi (\maybebluetext{ \ottnt{s} })    \ottsym{)}  [  \ottmv{x}  \ottsym{:=}   \Psi (\maybebluetext{ \ottnt{V} })   ]  \ottsym{=}   \mathscr{K}\llbracket \maybebluetext{ \ottnt{N}  [  \ottmv{x}  \ottsym{:=}  \ottnt{V}  ] } \rrbracket   \Psi (\maybebluetext{ \ottnt{s} })  $,
    which finishes this case.

  \item[\Otherwise]
    Since $\ottnt{M}$ is neither a value nor a coercion application,
    \[
       \mathscr{C}\llbracket \maybebluetext{ \ottnt{M} } \rrbracket   [  \ottmv{x}  \ottsym{:=}   \Psi (\maybebluetext{ \ottnt{V} })   ]
      = \ottsym{(}   \mathscr{K}\llbracket \maybebluetext{ \ottnt{M} } \rrbracket  \ottkw{id}   \ottsym{)}  [  \ottmv{x}  \ottsym{:=}   \Psi (\maybebluetext{ \ottnt{V} })   ]
    \]
    Since $\ottnt{M}  [  \ottmv{x}  \ottsym{:=}  \ottnt{V}  ]$ is neither a value nor a coercion application,
    \[
       \mathscr{C}\llbracket \maybebluetext{ \ottnt{M}  [  \ottmv{x}  \ottsym{:=}  \ottnt{V}  ] } \rrbracket 
      =  \mathscr{K}\llbracket \maybebluetext{ \ottnt{M}  [  \ottmv{x}  \ottsym{:=}  \ottnt{V}  ] } \rrbracket  \ottkw{id} 
    \]
    We have $ \ottmv{x} \notin   \metafun{FV} ( \ottkw{id} )   = \emptyset$.
    By the IH, $\ottsym{(}   \mathscr{K}\llbracket \maybebluetext{ \ottnt{M} } \rrbracket  \ottkw{id}   \ottsym{)}  [  \ottmv{x}  \ottsym{:=}   \Psi (\maybebluetext{ \ottnt{V} })   ]  \ottsym{=}   \mathscr{K}\llbracket \maybebluetext{ \ottnt{M}  [  \ottmv{x}  \ottsym{:=}  \ottnt{V}  ] } \rrbracket  \ottkw{id} $,
    which finishes this case.
  \end{description}

  (2) By case analysis on the form of $\ottnt{M}$.
  \begin{description}
  \item[\Case{$\ottnt{M}  \ottsym{=}  \ottnt{W}$}] We have
    \[
    \begin{array}{llll}
      \ottsym{(}   \mathscr{K}\llbracket \maybebluetext{ \ottnt{W} } \rrbracket  \ottnt{K}   \ottsym{)}  [  \ottmv{x}  \ottsym{:=}   \Psi (\maybebluetext{ \ottnt{V} })   ]
      &=& \ottsym{(}   \Psi (\maybebluetext{ \ottnt{W} })   \langle  \ottnt{K}  \rangle  \ottsym{)}  [  \ottmv{x}  \ottsym{:=}   \Psi (\maybebluetext{ \ottnt{V} })   ] \\
      &=& \ottsym{(}   \Psi (\maybebluetext{ \ottnt{W} })   [  \ottmv{x}  \ottsym{:=}   \Psi (\maybebluetext{ \ottnt{V} })   ]  \ottsym{)}  \langle  \ottnt{K}  \rangle.
    \end{array}
    \]
    Since $\ottnt{W}  [  \ottmv{x}  \ottsym{:=}  \ottnt{V}  ]$ is a value, 
    \[
    \begin{array}{llll}
       \mathscr{K}\llbracket \maybebluetext{ \ottnt{W}  [  \ottmv{x}  \ottsym{:=}  \ottnt{V}  ] } \rrbracket  \ottnt{K} 
      &=&  \Psi (\maybebluetext{ \ottnt{W}  [  \ottmv{x}  \ottsym{:=}  \ottnt{V}  ] })   \langle  \ottnt{K}  \rangle.
    \end{array}
    \]
    By the IH,
    $ \Psi (\maybebluetext{ \ottnt{W} })   [  \ottmv{x}  \ottsym{:=}   \Psi (\maybebluetext{ \ottnt{V} })   ]  \ottsym{=}   \Psi (\maybebluetext{ \ottnt{W}  [  \ottmv{x}  \ottsym{:=}  \ottnt{V}  ] }) $,
    which finishes this case.

  \item[\Case{$\ottnt{M}  \ottsym{=}  \ottnt{op}  \ottsym{(}  \ottnt{N_{{\mathrm{1}}}}  \ottsym{,}  \ottnt{N_{{\mathrm{2}}}}  \ottsym{)}$}] 
    By $ \ottmv{x} \notin   \metafun{FV} ( \ottnt{K} )  $,
    \[
    \begin{array}{llll}
      \ottsym{(}   \mathscr{K}\llbracket \maybebluetext{ \ottnt{op}  \ottsym{(}  \ottnt{N_{{\mathrm{1}}}}  \ottsym{,}  \ottnt{N_{{\mathrm{2}}}}  \ottsym{)} } \rrbracket  \ottnt{K}   \ottsym{)}  [  \ottmv{x}  \ottsym{:=}   \Psi (\maybebluetext{ \ottnt{V} })   ]
      &=& \ottnt{op}  \ottsym{(}   \mathscr{C}\llbracket \maybebluetext{ \ottnt{N_{{\mathrm{1}}}} } \rrbracket   \ottsym{,}   \mathscr{C}\llbracket \maybebluetext{ \ottnt{N_{{\mathrm{2}}}} } \rrbracket   \ottsym{)}  \langle  \ottnt{K}  \rangle  [  \ottmv{x}  \ottsym{:=}   \Psi (\maybebluetext{ \ottnt{V} })   ] \\
      &=& \ottnt{op}  \ottsym{(}   \mathscr{C}\llbracket \maybebluetext{ \ottnt{N_{{\mathrm{1}}}} } \rrbracket   [  \ottmv{x}  \ottsym{:=}   \Psi (\maybebluetext{ \ottnt{V} })   ]  \ottsym{,}   \mathscr{C}\llbracket \maybebluetext{ \ottnt{N_{{\mathrm{2}}}} } \rrbracket   [  \ottmv{x}  \ottsym{:=}   \Psi (\maybebluetext{ \ottnt{V} })   ]  \ottsym{)}  \langle  \ottnt{K}  \rangle.
    \end{array}
    \]
    Then,
    \[
    \begin{array}{llll}
       \mathscr{K}\llbracket \maybebluetext{ \ottsym{(}  \ottnt{op}  \ottsym{(}  \ottnt{N_{{\mathrm{1}}}}  \ottsym{,}  \ottnt{N_{{\mathrm{2}}}}  \ottsym{)}  \ottsym{)}  [  \ottmv{x}  \ottsym{:=}  \ottnt{V}  ] } \rrbracket  \ottnt{K} 
      &=&  \mathscr{K}\llbracket \maybebluetext{ \ottnt{op}  \ottsym{(}  \ottnt{N_{{\mathrm{1}}}}  [  \ottmv{x}  \ottsym{:=}  \ottnt{V}  ]  \ottsym{,}  \ottnt{N_{{\mathrm{2}}}}  [  \ottmv{x}  \ottsym{:=}  \ottnt{V}  ]  \ottsym{)} } \rrbracket  \ottnt{K}  \\
      &=& \ottnt{op}  \ottsym{(}   \mathscr{C}\llbracket \maybebluetext{ \ottnt{N_{{\mathrm{1}}}}  [  \ottmv{x}  \ottsym{:=}  \ottnt{V}  ] } \rrbracket   \ottsym{,}   \mathscr{C}\llbracket \maybebluetext{ \ottnt{N_{{\mathrm{2}}}}  [  \ottmv{x}  \ottsym{:=}  \ottnt{V}  ] } \rrbracket   \ottsym{)}  \langle  \ottnt{K}  \rangle.
    \end{array}
    \]
    By the IHs,
    \[
       \mathscr{C}\llbracket \maybebluetext{ \ottnt{N_{{\mathrm{1}}}} } \rrbracket   [  \ottmv{x}  \ottsym{:=}   \Psi (\maybebluetext{ \ottnt{V} })   ]  \ottsym{=}   \mathscr{C}\llbracket \maybebluetext{ \ottnt{N_{{\mathrm{1}}}}  [  \ottmv{x}  \ottsym{:=}  \ottnt{V}  ] } \rrbracket  \hgap
       \mathscr{C}\llbracket \maybebluetext{ \ottnt{N_{{\mathrm{2}}}} } \rrbracket   [  \ottmv{x}  \ottsym{:=}   \Psi (\maybebluetext{ \ottnt{V} })   ]  \ottsym{=}   \mathscr{C}\llbracket \maybebluetext{ \ottnt{N_{{\mathrm{2}}}}  [  \ottmv{x}  \ottsym{:=}  \ottnt{V}  ] } \rrbracket 
    \]
    which finish this case.

  \item[\Case{$\ottnt{M}  \ottsym{=}  \ottnt{N_{{\mathrm{1}}}} \, \ottnt{N_{{\mathrm{2}}}}$}] 
    By $ \ottmv{x} \notin   \metafun{FV} ( \ottnt{K} )  $,
    \[
    \begin{array}{llll}
      \ottsym{(}   \mathscr{K}\llbracket \maybebluetext{ \ottnt{N_{{\mathrm{1}}}} \, \ottnt{N_{{\mathrm{2}}}} } \rrbracket  \ottnt{K}   \ottsym{)}  [  \ottmv{x}  \ottsym{:=}   \Psi (\maybebluetext{ \ottnt{V} })   ]
      &=& \ottsym{(}    \mathscr{C}\llbracket \maybebluetext{ \ottnt{N_{{\mathrm{1}}}} } \rrbracket  \, (  \mathscr{C}\llbracket \maybebluetext{ \ottnt{N_{{\mathrm{2}}}} } \rrbracket  , \ottnt{K} )   \ottsym{)}  [  \ottmv{x}  \ottsym{:=}   \Psi (\maybebluetext{ \ottnt{V} })   ] \\
      &=&  \ottsym{(}   \mathscr{C}\llbracket \maybebluetext{ \ottnt{N_{{\mathrm{1}}}} } \rrbracket   [  \ottmv{x}  \ottsym{:=}   \Psi (\maybebluetext{ \ottnt{V} })   ]  \ottsym{)} \, (  \mathscr{C}\llbracket \maybebluetext{ \ottnt{N_{{\mathrm{2}}}} } \rrbracket   [  \ottmv{x}  \ottsym{:=}   \Psi (\maybebluetext{ \ottnt{V} })   ] , \ottnt{K} ) 
    \end{array}
    \]
    Then,
    \[
    \begin{array}{llll}
       \mathscr{K}\llbracket \maybebluetext{ \ottsym{(}  \ottnt{N_{{\mathrm{1}}}} \, \ottnt{N_{{\mathrm{2}}}}  \ottsym{)}  [  \ottmv{x}  \ottsym{:=}  \ottnt{V}  ] } \rrbracket  \ottnt{K} 
      &=&  \mathscr{K}\llbracket \maybebluetext{ \ottsym{(}  \ottnt{N_{{\mathrm{1}}}}  [  \ottmv{x}  \ottsym{:=}  \ottnt{V}  ]  \ottsym{)} \, \ottsym{(}  \ottnt{N_{{\mathrm{2}}}}  [  \ottmv{x}  \ottsym{:=}  \ottnt{V}  ]  \ottsym{)} } \rrbracket  \ottnt{K}  \\
      &=&   \mathscr{C}\llbracket \maybebluetext{ \ottnt{N_{{\mathrm{1}}}}  [  \ottmv{x}  \ottsym{:=}  \ottnt{V}  ] } \rrbracket  \, (  \mathscr{C}\llbracket \maybebluetext{ \ottnt{N_{{\mathrm{2}}}}  [  \ottmv{x}  \ottsym{:=}  \ottnt{V}  ] } \rrbracket  , \ottnt{K} ) 
    \end{array}
    \]
    By the IHs,
    \[
       \mathscr{C}\llbracket \maybebluetext{ \ottnt{N_{{\mathrm{1}}}} } \rrbracket   [  \ottmv{x}  \ottsym{:=}   \Psi (\maybebluetext{ \ottnt{V} })   ]  \ottsym{=}   \mathscr{C}\llbracket \maybebluetext{ \ottnt{N_{{\mathrm{1}}}}  [  \ottmv{x}  \ottsym{:=}  \ottnt{V}  ] } \rrbracket  \hgap
       \mathscr{C}\llbracket \maybebluetext{ \ottnt{N_{{\mathrm{2}}}} } \rrbracket   [  \ottmv{x}  \ottsym{:=}   \Psi (\maybebluetext{ \ottnt{V} })   ]  \ottsym{=}   \mathscr{C}\llbracket \maybebluetext{ \ottnt{N_{{\mathrm{2}}}}  [  \ottmv{x}  \ottsym{:=}  \ottnt{V}  ] } \rrbracket ,
    \]
    which finish this case.

  \item[\Case{$\ottnt{M}  \ottsym{=}  \ottnt{N}  \langle  \ottnt{s}  \rangle$}]
    By $ \ottmv{x} \notin   \metafun{FV} ( \ottnt{K} )  $ and $ \ottmv{x} \notin   \metafun{FV} (  \Psi (\maybebluetext{ \ottnt{s} })  )   = \emptyset$,
    \[
    \begin{array}{llll}
      \ottsym{(}   \mathscr{K}\llbracket \maybebluetext{ \ottnt{N}  \langle  \ottnt{s}  \rangle } \rrbracket  \ottnt{K}   \ottsym{)}  [  \ottmv{x}  \ottsym{:=}   \Psi (\maybebluetext{ \ottnt{V} })   ]
      &=& \ottsym{(}   \ottkw{let} \,  \kappa =  \Psi (\maybebluetext{ \ottnt{s} })   \mathbin{;\!;}  \ottnt{K} \, \ottkw{in}\,  \ottsym{(}   \mathscr{K}\llbracket \maybebluetext{ \ottnt{N} } \rrbracket  \kappa   \ottsym{)}   \ottsym{)}  [  \ottmv{x}  \ottsym{:=}   \Psi (\maybebluetext{ \ottnt{V} })   ] \\
      &=&  \ottkw{let} \,  \kappa = \ottsym{(}   \Psi (\maybebluetext{ \ottnt{s} })   \mathbin{;\!;}  \ottnt{K}  \ottsym{)}  [  \ottmv{x}  \ottsym{:=}   \Psi (\maybebluetext{ \ottnt{V} })   ] \, \ottkw{in}\,  \ottsym{(}   \mathscr{K}\llbracket \maybebluetext{ \ottnt{N} } \rrbracket  \kappa   \ottsym{)}   [  \ottmv{x}  \ottsym{:=}   \Psi (\maybebluetext{ \ottnt{V} })   ] \\
      &=&  \ottkw{let} \,  \kappa =  \Psi (\maybebluetext{ \ottnt{s} })   \mathbin{;\!;}  \ottnt{K} \, \ottkw{in}\,  \ottsym{(}  \ottsym{(}   \mathscr{K}\llbracket \maybebluetext{ \ottnt{N} } \rrbracket  \kappa   \ottsym{)}  [  \ottmv{x}  \ottsym{:=}   \Psi (\maybebluetext{ \ottnt{V} })   ]  \ottsym{)} 
    \end{array}
    \]
    Then,
    \[
    \begin{array}{llll}
       \mathscr{K}\llbracket \maybebluetext{ \ottnt{N}  \langle  \ottnt{s}  \rangle  [  \ottmv{x}  \ottsym{:=}  \ottnt{V}  ] } \rrbracket  \ottnt{K} 
      &=&  \mathscr{K}\llbracket \maybebluetext{ \ottnt{N}  [  \ottmv{x}  \ottsym{:=}  \ottnt{V}  ]  \langle  \ottnt{s}  \rangle } \rrbracket  \ottnt{K}  \\
      &=&  \ottkw{let} \,  \kappa =  \Psi (\maybebluetext{ \ottnt{s} })   \mathbin{;\!;}  \ottnt{K} \, \ottkw{in}\,  \ottsym{(}   \mathscr{K}\llbracket \maybebluetext{ \ottnt{N}  [  \ottmv{x}  \ottsym{:=}  \ottnt{V}  ] } \rrbracket  \kappa   \ottsym{)} .
    \end{array}
    \]
    Here, we can assume $\kappa  \ne  \ottmv{x}$.
    So, $ \ottmv{x} \notin   \metafun{FV} ( \kappa )  $.
    By the IH, $\ottsym{(}   \mathscr{K}\llbracket \maybebluetext{ \ottnt{N} } \rrbracket  \kappa   \ottsym{)}  [  \ottmv{x}  \ottsym{:=}   \Psi (\maybebluetext{ \ottnt{V} })   ]  \ottsym{=}   \mathscr{K}\llbracket \maybebluetext{ \ottnt{N}  [  \ottmv{x}  \ottsym{:=}  \ottnt{V}  ] } \rrbracket  \kappa $,
    which finishes this case.

  \item[\Case{$\ottnt{M}  \ottsym{=}  \ottkw{blame} \, \ottnt{p}$}]
    \[
    \begin{array}[b]{llll}
      \ottsym{(}   \mathscr{K}\llbracket \maybebluetext{ \ottkw{blame} \, \ottnt{p} } \rrbracket  \ottnt{K}   \ottsym{)}  [  \ottmv{x}  \ottsym{:=}   \Psi (\maybebluetext{ \ottnt{V} })   ]
      &=& \ottsym{(}  \ottkw{blame} \, \ottnt{p}  \ottsym{)}  [  \ottmv{x}  \ottsym{:=}   \Psi (\maybebluetext{ \ottnt{V} })   ] = \ottkw{blame} \, \ottnt{p} \\
       \mathscr{K}\llbracket \maybebluetext{ \ottsym{(}  \ottkw{blame} \, \ottnt{p}  \ottsym{)}  [  \ottmv{x}  \ottsym{:=}  \ottnt{V}  ] } \rrbracket  \ottnt{K} 
      &=&  \mathscr{K}\llbracket \maybebluetext{ \ottsym{(}  \ottkw{blame} \, \ottnt{p}  \ottsym{)} } \rrbracket  \ottnt{K}  = \ottkw{blame} \, \ottnt{p}.
    \end{array} \qedhere
    \]
  \end{description}
\end{proof}

\begin{lemma}[Substitution for a continuation variable]\label{lem:subst-kap}
  If $ \kappa \notin   \metafun{FV} ( \ottnt{M} )  $,
  then $\ottsym{(}   \mathscr{K}\llbracket \maybebluetext{ \ottnt{M} } \rrbracket  \kappa   \ottsym{)}  [  \kappa  \ottsym{:=}  \ottnt{K}  ]  \ottsym{=}   \mathscr{K}\llbracket \maybebluetext{ \ottnt{M} } \rrbracket  \ottnt{K} $.
\end{lemma}
\begin{proof}
  By case analysis on the structure of $\ottnt{M}$.
  \begin{description}
  \item[\Case{$\ottnt{M}  \ottsym{=}  \ottnt{V}$}] \leavevmode
      \[
        \ottsym{(}   \mathscr{K}\llbracket \maybebluetext{ \ottnt{V} } \rrbracket  \kappa   \ottsym{)}  [  \kappa  \ottsym{:=}  \ottnt{K}  ]
        = \ottsym{(}   \Psi (\maybebluetext{ \ottnt{V} })   \langle  \kappa  \rangle  \ottsym{)}  [  \kappa  \ottsym{:=}  \ottnt{K}  ] =  \Psi (\maybebluetext{ \ottnt{V} })   \langle  \ottnt{K}  \rangle =  \mathscr{K}\llbracket \maybebluetext{ \ottnt{V} } \rrbracket  \ottnt{K} .
      \]

  \item[\Case{$\ottnt{M}  \ottsym{=}  \ottnt{op}  \ottsym{(}  \ottnt{N_{{\mathrm{1}}}}  \ottsym{,}  \ottnt{N_{{\mathrm{2}}}}  \ottsym{)}$}]
    Since $ \kappa \notin   \metafun{FV} ( \ottnt{M} )  $,
    we have $ \kappa \notin   \metafun{FV} ( \ottnt{N_{{\mathrm{1}}}} )  $ and $ \kappa \notin   \metafun{FV} ( \ottnt{N_{{\mathrm{2}}}} )  $.
    \[
    \begin{array}{lll}
      \ottsym{(}   \mathscr{K}\llbracket \maybebluetext{ \ottnt{op}  \ottsym{(}  \ottnt{N_{{\mathrm{1}}}}  \ottsym{,}  \ottnt{N_{{\mathrm{2}}}}  \ottsym{)} } \rrbracket  \kappa   \ottsym{)}  [  \kappa  \ottsym{:=}  \ottnt{K}  ]
      &=& \ottsym{(}  \ottnt{op}  \ottsym{(}   \mathscr{C}\llbracket \maybebluetext{ \ottnt{N_{{\mathrm{1}}}} } \rrbracket   \ottsym{,}   \mathscr{C}\llbracket \maybebluetext{ \ottnt{N_{{\mathrm{2}}}} } \rrbracket   \ottsym{)}  \langle  \kappa  \rangle  \ottsym{)}  [  \kappa  \ottsym{:=}  \ottnt{K}  ] \\
      &=& \ottnt{op}  \ottsym{(}   \mathscr{C}\llbracket \maybebluetext{ \ottnt{N_{{\mathrm{1}}}} } \rrbracket   \ottsym{,}   \mathscr{C}\llbracket \maybebluetext{ \ottnt{N_{{\mathrm{2}}}} } \rrbracket   \ottsym{)}  \langle  \ottnt{K}  \rangle \\
      &=&  \mathscr{K}\llbracket \maybebluetext{ \ottnt{op}  \ottsym{(}  \ottnt{N_{{\mathrm{1}}}}  \ottsym{,}  \ottnt{N_{{\mathrm{2}}}}  \ottsym{)} } \rrbracket  \ottnt{K} .
    \end{array}
    \]

  \item[\Case{$\ottnt{M}  \ottsym{=}  \ottnt{N_{{\mathrm{1}}}} \, \ottnt{N_{{\mathrm{2}}}}$}]
    Since $ \kappa \notin   \metafun{FV} ( \ottnt{M} )  $,
    we have $ \kappa \notin   \metafun{FV} ( \ottnt{N_{{\mathrm{1}}}} )  $ and $ \kappa \notin   \metafun{FV} ( \ottnt{N_{{\mathrm{2}}}} )  $.
    \[
    \begin{array}{lll}
      \ottsym{(}   \mathscr{K}\llbracket \maybebluetext{ \ottnt{N_{{\mathrm{1}}}} \, \ottnt{N_{{\mathrm{2}}}} } \rrbracket  \kappa   \ottsym{)}  [  \kappa  \ottsym{:=}  \ottnt{K}  ]
      &=& \ottsym{(}    \mathscr{C}\llbracket \maybebluetext{ \ottnt{N_{{\mathrm{1}}}} } \rrbracket  \, (  \mathscr{C}\llbracket \maybebluetext{ \ottnt{N_{{\mathrm{2}}}} } \rrbracket  , \kappa )   \ottsym{)}  [  \kappa  \ottsym{:=}  \ottnt{K}  ]\\
      &=&   \mathscr{C}\llbracket \maybebluetext{ \ottnt{N_{{\mathrm{1}}}} } \rrbracket  \, (  \mathscr{C}\llbracket \maybebluetext{ \ottnt{N_{{\mathrm{2}}}} } \rrbracket  , \ottnt{K} ) \\
      &=& \mathscr{K}\llbracket \maybebluetext{ \ottnt{N_{{\mathrm{1}}}} \, \ottnt{N_{{\mathrm{2}}}} } \rrbracket  \ottnt{K} .
    \end{array}
    \]

  \item[\Case{$\ottnt{M}  \ottsym{=}  \ottnt{N}  \langle  \ottnt{s}  \rangle$}]
    \[
    \begin{array}{lll}
      \ottsym{(}   \mathscr{K}\llbracket \maybebluetext{ \ottnt{N}  \langle  \ottnt{s}  \rangle } \rrbracket  \kappa   \ottsym{)}  [  \kappa  \ottsym{:=}  \ottnt{K}  ]
      &=& \ottsym{(}   \ottkw{let} \,  \kappa' =  \Psi (\maybebluetext{ \ottnt{s} })   \mathbin{;\!;}  \kappa \, \ottkw{in}\,  \ottsym{(}   \mathscr{K}\llbracket \maybebluetext{ \ottnt{N} } \rrbracket  \kappa'   \ottsym{)}   \ottsym{)}  [  \kappa  \ottsym{:=}  \ottnt{K}  ] \\
      &=&  \ottkw{let} \,  \kappa' = \ottsym{(}   \Psi (\maybebluetext{ \ottnt{s} })   \mathbin{;\!;}  \kappa  \ottsym{)}  [  \kappa  \ottsym{:=}  \ottnt{K}  ] \, \ottkw{in}\,  \ottsym{(}  \ottsym{(}   \mathscr{K}\llbracket \maybebluetext{ \ottnt{N} } \rrbracket  \kappa'   \ottsym{)}  [  \kappa  \ottsym{:=}  \ottnt{K}  ]  \ottsym{)}  \\
      &=&  \ottkw{let} \,  \kappa' =  \Psi (\maybebluetext{ \ottnt{s} })   \mathbin{;\!;}  \ottnt{K} \, \ottkw{in}\,  \ottsym{(}   \mathscr{K}\llbracket \maybebluetext{ \ottnt{N} } \rrbracket  \kappa'   \ottsym{)}  \\
      &=&  \mathscr{K}\llbracket \maybebluetext{ \ottnt{N}  \langle  \ottnt{s}  \rangle } \rrbracket  \ottnt{K} .
    \end{array}
    \]
    The third equality is by $\ottsym{(}   \mathscr{K}\llbracket \maybebluetext{ \ottnt{N} } \rrbracket  \kappa'   \ottsym{)}  [  \kappa  \ottsym{:=}  \ottnt{K}  ]  \ottsym{=}   \mathscr{K}\llbracket \maybebluetext{ \ottnt{N} } \rrbracket  \kappa' $,
    which is shown as below:
    We can assume $\kappa'  \ne  \kappa$.
    Since $ \kappa \notin   \metafun{FV} ( \ottnt{M} )  $, we have $ \kappa \notin   \metafun{FV} ( \ottnt{N} )  $.
    By Lemma~\ref{lem:fv-colon},
    \[
       \kappa \notin   \metafun{FV} ( \ottnt{N} )  \, \cup \,  \metafun{FV} ( \kappa' )   =  \metafun{FV} (  \mathscr{K}\llbracket \maybebluetext{ \ottnt{N} } \rrbracket  \kappa'  ) .
    \]
    Thus, $\ottsym{(}   \mathscr{K}\llbracket \maybebluetext{ \ottnt{N} } \rrbracket  \kappa'   \ottsym{)}  [  \kappa  \ottsym{:=}  \ottnt{K}  ]  \ottsym{=}   \mathscr{K}\llbracket \maybebluetext{ \ottnt{N} } \rrbracket  \kappa' $.

  \item[\Case{$\ottnt{M}  \ottsym{=}  \ottkw{blame} \, \ottnt{p}$}]
    \[
    \begin{array}[b]{lll}
      \ottsym{(}   \mathscr{K}\llbracket \maybebluetext{ \ottkw{blame} \, \ottnt{p} } \rrbracket  \kappa   \ottsym{)}  [  \kappa  \ottsym{:=}  \ottnt{K}  ]
      &=& \ottsym{(}  \ottkw{blame} \, \ottnt{p}  \ottsym{)}  [  \kappa  \ottsym{:=}  \ottnt{K}  ] \\
      &=& \ottkw{blame} \, \ottnt{p} \\
      &=&  \mathscr{K}\llbracket \maybebluetext{ \ottkw{blame} \, \ottnt{p} } \rrbracket  \ottnt{K} .
    \end{array} \qedhere
    \]
  \end{description}
\end{proof}

\iffull\lemSubstTrans*
\else
\begin{lemma}[Substitution]\label{lem:subst-trans}
  If $ \kappa \notin   \metafun{FV} ( \ottnt{M} )   \cup  \metafun{FV} ( \ottnt{V} ) $, then
  $\ottsym{(}   \mathscr{K}\llbracket \maybebluetext{ \ottnt{M} } \rrbracket  \kappa   \ottsym{)}  [  \ottmv{x}  \ottsym{:=}   \Psi (\maybebluetext{ \ottnt{V} })   \ottsym{,}  \kappa  \ottsym{:=}  \ottnt{K}  ]  \ottsym{=}   \mathscr{K}\llbracket \maybebluetext{ \ottnt{M}  [  \ottmv{x}  \ottsym{:=}  \ottnt{V}  ] } \rrbracket  \ottnt{K} $.
\end{lemma}
\fi

\begin{proof}
  We have $ \kappa \notin   \metafun{FV} ( \ottnt{M}  [  \ottmv{x}  \ottsym{:=}  \ottnt{V}  ] )  $.
  \[
  \begin{array}[b]{llll}
    && \ottsym{(}   \mathscr{K}\llbracket \maybebluetext{ \ottnt{M} } \rrbracket  \kappa   \ottsym{)}  [  \ottmv{x}  \ottsym{:=}   \Psi (\maybebluetext{ \ottnt{V} })   \ottsym{,}  \kappa  \ottsym{:=}  \ottnt{K}  ] \\
    &=& \ottsym{(}   \mathscr{K}\llbracket \maybebluetext{ \ottnt{M} } \rrbracket  \kappa   \ottsym{)}  [  \ottmv{x}  \ottsym{:=}   \Psi (\maybebluetext{ \ottnt{V} })   ]  [  \kappa  \ottsym{:=}  \ottnt{K}  ] \\
    &=& \ottsym{(}   \mathscr{K}\llbracket \maybebluetext{ \ottnt{M}  [  \ottmv{x}  \ottsym{:=}  \ottnt{V}  ] } \rrbracket  \kappa   \ottsym{)}  [  \kappa  \ottsym{:=}  \ottnt{K}  ]
    &\text{by Lemma~\ref{lem:subst-x} with $\ottnt{K}  \ottsym{=}  \kappa$}\\
    &=&  \mathscr{K}\llbracket \maybebluetext{ \ottnt{M}  [  \ottmv{x}  \ottsym{:=}  \ottnt{V}  ] } \rrbracket  \ottnt{K} 
    &\text{by Lemma~\ref{lem:subst-kap}.}
  \end{array} \qedhere
  \]
\end{proof}

\subsubsection{Evaluation Contexts}

We can rewrite the syntax of evaluation contexts in \lamS as below:
\begin{align*}
  \mathcal{E} &\grmeq
  \mathcal{F} \grmor
  \mathcal{F}  [  \square \, \langle  \ottnt{s}  \rangle  ]
  \\
  \mathcal{F} &\grmeq
   \square  \grmor
  \mathcal{F}  [  \ottnt{op}  \ottsym{(}  \square  \ottsym{,}  \ottnt{M}  \ottsym{)}  ] \grmor
  \mathcal{F}  [  \ottnt{op}  \ottsym{(}  \ottnt{V}  \ottsym{,} \, \square \, \ottsym{)}  ] \grmor
  \mathcal{F}  [  \square \, \ottnt{M}  ] \grmor
  \mathcal{F}  [  \ottnt{V} \, \square  ] \\
  &\grmsp\grmor
  \mathcal{F}  [  \ottsym{(}  \ottnt{op}  \ottsym{(}  \square  \ottsym{,}  \ottnt{M}  \ottsym{)}  \ottsym{)}  \langle  \ottnt{s}  \rangle  ] \grmor
  \mathcal{F}  [  \ottsym{(}  \ottnt{op}  \ottsym{(}  \ottnt{V}  \ottsym{,} \, \square \, \ottsym{)}  \ottsym{)}  \langle  \ottnt{s}  \rangle  ] \grmor
  \mathcal{F}  [  \ottsym{(} \, \square \, \ottnt{M}  \ottsym{)}  \langle  \ottnt{s}  \rangle  ] \grmor
  \mathcal{F}  [  \ottsym{(}  \ottnt{V} \, \square \, \ottsym{)}  \langle  \ottnt{s}  \rangle  ]
\end{align*}

In the following lemma, the first item concerns the case $\mathcal{E}  \ottsym{=}  \mathcal{F}$,
and the second concerns the case $\mathcal{E}  \ottsym{=}  \mathcal{F}  [  \square \, \langle  \ottnt{s}  \rangle  ]$.
Note that $\ottsym{(}  \mathcal{F}  [  \square \, \langle  \ottnt{s}  \rangle  ]  \ottsym{)}  [  \ottnt{M}  ]  \ottsym{=}  \mathcal{F}  [  \ottnt{M}  \langle  \ottnt{s}  \rangle  ]$.

\iffull\lemTransCtx*
\else
\begin{lemma}[Contexts]\label{lem:trans-ctx}\leavevmode
  \begin{enumerate}
  \item For any $\mathcal{F}$, there exists $\mathcal{E}'$
    such that for any $\ottnt{M}$, $ \mathscr{C}\llbracket \maybebluetext{ \mathcal{F}  [  \ottnt{M}  ] } \rrbracket   \ottsym{=}  \mathcal{E}'  [   \mathscr{C}\llbracket \maybebluetext{ \ottnt{M} } \rrbracket   ]$.
  \item For any $\mathcal{F}$ and $\ottnt{s}$, there exists $\mathcal{E}'$ such that
    for any $\ottnt{M}$, $ \mathscr{C}\llbracket \maybebluetext{ \mathcal{F}  [  \ottnt{M}  \langle  \ottnt{s}  \rangle  ] } \rrbracket   \ottsym{=}  \mathcal{E}'  [   \mathscr{K}\llbracket \maybebluetext{ \ottnt{M} } \rrbracket   \Psi (\maybebluetext{ \ottnt{s} })    ]$.
  \end{enumerate}
\end{lemma}
\fi
\begin{proof}
  Two items are simultaneously proved
  by induction on the structure of $\mathcal{F}$.

  (1) By case analysis on the structure of $\mathcal{F}$.
  \begin{description}
  \item[\Case{$\mathcal{F}  \ottsym{=}  \square$}]
    By $\square  [  \ottnt{M}  ]  \ottsym{=}  \ottnt{M}$, we must show $ \mathscr{C}\llbracket \maybebluetext{ \ottnt{M} } \rrbracket   \ottsym{=}  \mathcal{E}'  [   \mathscr{C}\llbracket \maybebluetext{ \ottnt{M} } \rrbracket   ]$.
    Take $\mathcal{E}' =  \square $.

  \item[\Case{$\mathcal{F}  \ottsym{=}  \mathcal{F}_{{\mathrm{1}}}  [  \ottnt{op}  \ottsym{(}  \square  \ottsym{,}  \ottnt{N}  \ottsym{)}  ]$}]    
    By the IH (item 1), there exists $\mathcal{E}'_{{\mathrm{1}}}$ such that
    $ \mathscr{C}\llbracket \maybebluetext{ \mathcal{F}_{{\mathrm{1}}}  [  \ottnt{L}  ] } \rrbracket  = \mathcal{E}'_{{\mathrm{1}}}  [   \mathscr{C}\llbracket \maybebluetext{ \ottnt{L} } \rrbracket   ]$ for any $\ottnt{L}$.
    We have
    \[
      \mathcal{F}  [  \ottnt{M}  ]  \ottsym{=}  \ottsym{(}  \mathcal{F}_{{\mathrm{1}}}  [  \ottnt{op}  \ottsym{(}  \square  \ottsym{,}  \ottnt{N}  \ottsym{)}  ]  \ottsym{)}  [  \ottnt{M}  ] = \mathcal{F}_{{\mathrm{1}}}  [  \ottnt{op}  \ottsym{(}  \ottnt{M}  \ottsym{,}  \ottnt{N}  \ottsym{)}  ]
    \]
    and so
    \[
    \begin{array}{llll}
       \mathscr{C}\llbracket \maybebluetext{ \mathcal{F}  [  \ottnt{M}  ] } \rrbracket 
      &=&  \mathscr{C}\llbracket \maybebluetext{ \mathcal{F}_{{\mathrm{1}}}  [  \ottnt{op}  \ottsym{(}  \ottnt{M}  \ottsym{,}  \ottnt{N}  \ottsym{)}  ] } \rrbracket  \\
      &=& \mathcal{E}'_{{\mathrm{1}}}  [   \mathscr{C}\llbracket \maybebluetext{ \ottnt{op}  \ottsym{(}  \ottnt{M}  \ottsym{,}  \ottnt{N}  \ottsym{)} } \rrbracket   ] & \text{by IH with $\ottnt{L}  \ottsym{=}  \ottnt{op}  \ottsym{(}  \ottnt{M}  \ottsym{,}  \ottnt{N}  \ottsym{)}$} \\
      &=& \mathcal{E}'_{{\mathrm{1}}}  [   \mathscr{K}\llbracket \maybebluetext{ \ottnt{op}  \ottsym{(}  \ottnt{M}  \ottsym{,}  \ottnt{N}  \ottsym{)} } \rrbracket  \ottkw{id}   ]  \\ 
      &=& \mathcal{E}'_{{\mathrm{1}}}  [  \ottnt{op}  \ottsym{(}   \mathscr{C}\llbracket \maybebluetext{ \ottnt{M} } \rrbracket   \ottsym{,}   \mathscr{C}\llbracket \maybebluetext{ \ottnt{N} } \rrbracket   \ottsym{)}  \langle  \ottkw{id}  \rangle  ] &\text{by \rnp{Tr-Op}.}
    \end{array}
    \]
    (Note that $\ottnt{op}  \ottsym{(}  \ottnt{M}  \ottsym{,}  \ottnt{N}  \ottsym{)}$ is neither a value nor a coercion application.)
    Take $\mathcal{E}'  \ottsym{=}  \mathcal{E}'_{{\mathrm{1}}}  [  \ottnt{op}  \ottsym{(} \, \square \, \ottsym{,}   \mathscr{C}\llbracket \maybebluetext{ \ottnt{N} } \rrbracket   \ottsym{)}  \langle  \ottkw{id}  \rangle  ]$; then we have $\mathcal{E}'  [   \mathscr{C}\llbracket \maybebluetext{ \ottnt{M} } \rrbracket   ]  \ottsym{=}   \mathscr{C}\llbracket \maybebluetext{ \mathcal{F}  [  \ottnt{M}  ] } \rrbracket $.

  \item[\Case{$\mathcal{F}  \ottsym{=}  \mathcal{F}_{{\mathrm{1}}}  [  \ottnt{op}  \ottsym{(} \, \square \, \ottsym{,}  \ottnt{N}  \ottsym{)}  \langle  \ottnt{t}  \rangle  ]$}]
    By the IH (item 2), there exists $\mathcal{E}'_{{\mathrm{1}}}$ such that
    $ \mathscr{C}\llbracket \maybebluetext{ \mathcal{F}_{{\mathrm{1}}}  [  \ottnt{L}  \langle  \ottnt{t}  \rangle  ] } \rrbracket  = \mathcal{E}'_{{\mathrm{1}}}  [   \mathscr{K}\llbracket \maybebluetext{ \ottnt{L} } \rrbracket   \Psi (\maybebluetext{ \ottnt{t} })    ]$ for any $\ottnt{L}$.
    We have
    \[
      \mathcal{F}  [  \ottnt{M}  ]  \ottsym{=}  \ottsym{(}  \mathcal{F}_{{\mathrm{1}}}  [  \ottnt{op}  \ottsym{(} \, \square \, \ottsym{,}  \ottnt{N}  \ottsym{)}  \langle  \ottnt{t}  \rangle  ]  \ottsym{)}  [  \ottnt{M}  ] = \mathcal{F}_{{\mathrm{1}}}  [  \ottnt{op}  \ottsym{(}  \ottnt{M}  \ottsym{,}  \ottnt{N}  \ottsym{)}  \langle  \ottnt{t}  \rangle  ]
    \]
    and so
    \[
    \begin{array}{llll}
       \mathscr{C}\llbracket \maybebluetext{ \mathcal{F}  [  \ottnt{M}  ] } \rrbracket 
      &=&  \mathscr{C}\llbracket \maybebluetext{ \mathcal{F}_{{\mathrm{1}}}  [  \ottnt{op}  \ottsym{(}  \ottnt{M}  \ottsym{,}  \ottnt{N}  \ottsym{)}  \langle  \ottnt{t}  \rangle  ] } \rrbracket  \\
      &=& \mathcal{E}'_{{\mathrm{1}}}  [   \mathscr{K}\llbracket \maybebluetext{ \ottnt{op}  \ottsym{(}  \ottnt{M}  \ottsym{,}  \ottnt{N}  \ottsym{)} } \rrbracket   \Psi (\maybebluetext{ \ottnt{t} })    ] & \text{by IH with $\ottnt{L}  \ottsym{=}  \ottnt{op}  \ottsym{(}  \ottnt{M}  \ottsym{,}  \ottnt{N}  \ottsym{)}$} \\
      &=& \mathcal{E}'_{{\mathrm{1}}}  [  \ottnt{op}  \ottsym{(}   \mathscr{C}\llbracket \maybebluetext{ \ottnt{M} } \rrbracket   \ottsym{,}   \mathscr{C}\llbracket \maybebluetext{ \ottnt{N} } \rrbracket   \ottsym{)}  \langle   \Psi (\maybebluetext{ \ottnt{t} })   \rangle  ] &\text{by \rnp{Tr-Op}.}
    \end{array}
    \]
    Take $\mathcal{E}'  \ottsym{=}  \mathcal{E}'_{{\mathrm{1}}}  [  \ottnt{op}  \ottsym{(} \, \square \, \ottsym{,}   \mathscr{C}\llbracket \maybebluetext{ \ottnt{N} } \rrbracket   \ottsym{)}  \langle   \Psi (\maybebluetext{ \ottnt{t} })   \rangle  ]$; then we have $\mathcal{E}'  [   \mathscr{C}\llbracket \maybebluetext{ \ottnt{M} } \rrbracket   ]  \ottsym{=}   \mathscr{C}\llbracket \maybebluetext{ \mathcal{F}  [  \ottnt{M}  ] } \rrbracket $.

  \item[\Case{$\mathcal{F}  \ottsym{=}  \mathcal{F}_{{\mathrm{1}}}  [  \ottnt{op}  \ottsym{(}  \ottnt{V}  \ottsym{,} \, \square \, \ottsym{)}  ]$}]    
    \ifshownotes
    By the IH (item 1), there exists $\mathcal{E}'_{{\mathrm{1}}}$ such that
    $ \mathscr{C}\llbracket \maybebluetext{ \mathcal{F}_{{\mathrm{1}}}  [  \ottnt{L}  ] } \rrbracket  = \mathcal{E}'_{{\mathrm{1}}}  [   \mathscr{C}\llbracket \maybebluetext{ \ottnt{L} } \rrbracket   ]$ for any $\ottnt{L}$.
    We have
    \[
      \mathcal{F}  [  \ottnt{M}  ]  \ottsym{=}  \ottsym{(}  \mathcal{F}_{{\mathrm{1}}}  [  \ottnt{op}  \ottsym{(}  \ottnt{V}  \ottsym{,} \, \square \, \ottsym{)}  ]  \ottsym{)}  [  \ottnt{M}  ] = \mathcal{F}_{{\mathrm{1}}}  [  \ottnt{op}  \ottsym{(}  \ottnt{V}  \ottsym{,}  \ottnt{M}  \ottsym{)}  ]
    \]
    and so
    \[
    \begin{array}{llll}
       \mathscr{C}\llbracket \maybebluetext{ \mathcal{F}  [  \ottnt{M}  ] } \rrbracket 
      &=&  \mathscr{C}\llbracket \maybebluetext{ \mathcal{F}_{{\mathrm{1}}}  [  \ottnt{op}  \ottsym{(}  \ottnt{V}  \ottsym{,}  \ottnt{M}  \ottsym{)}  ] } \rrbracket  \\
      &=& \mathcal{E}'_{{\mathrm{1}}}  [   \mathscr{C}\llbracket \maybebluetext{ \ottnt{op}  \ottsym{(}  \ottnt{V}  \ottsym{,}  \ottnt{M}  \ottsym{)} } \rrbracket   ] & \text{by IH with $\ottnt{L}  \ottsym{=}  \ottnt{op}  \ottsym{(}  \ottnt{V}  \ottsym{,}  \ottnt{M}  \ottsym{)}$} \\
      &=& \mathcal{E}'_{{\mathrm{1}}}  [   \mathscr{K}\llbracket \maybebluetext{ \ottnt{op}  \ottsym{(}  \ottnt{V}  \ottsym{,}  \ottnt{M}  \ottsym{)} } \rrbracket  \ottkw{id}   ]  \\
      &=& \mathcal{E}'_{{\mathrm{1}}}  [  \ottnt{op}  \ottsym{(}   \mathscr{C}\llbracket \maybebluetext{ \ottnt{V} } \rrbracket   \ottsym{,}   \mathscr{C}\llbracket \maybebluetext{ \ottnt{M} } \rrbracket   \ottsym{)}  \langle  \ottkw{id}  \rangle  ] & \text{by \rnp{Tr-Op}.}
    \end{array}
    \]
    $ \mathscr{C}\llbracket \maybebluetext{ \ottnt{V} } \rrbracket  =  \Psi (\maybebluetext{ \ottnt{V} }) $ is a value.
    Take $\mathcal{E}'  \ottsym{=}  \mathcal{E}'_{{\mathrm{1}}}  [  \ottnt{op}  \ottsym{(}   \mathscr{C}\llbracket \maybebluetext{ \ottnt{V} } \rrbracket   \ottsym{,}  \square  \ottsym{)}  \langle  \ottkw{id}  \rangle  ]$; then we have $\mathcal{E}'  [   \mathscr{C}\llbracket \maybebluetext{ \ottnt{M} } \rrbracket   ]  \ottsym{=}   \mathscr{C}\llbracket \maybebluetext{ \mathcal{F}  [  \ottnt{M}  ] } \rrbracket $.
    \else
    Similar.
    \fi

  \item[\Case{$\mathcal{F}  \ottsym{=}  \mathcal{F}_{{\mathrm{1}}}  [  \ottnt{op}  \ottsym{(}  \ottnt{V}  \ottsym{,} \, \square \, \ottsym{)}  \langle  \ottnt{t}  \rangle  ]$}]
    \ifshownotes
    By the IH (item 2), there exists $\mathcal{E}'_{{\mathrm{1}}}$ such that
    $ \mathscr{C}\llbracket \maybebluetext{ \mathcal{F}_{{\mathrm{1}}}  [  \ottnt{L}  \langle  \ottnt{t}  \rangle  ] } \rrbracket  = \mathcal{E}'_{{\mathrm{1}}}  [   \mathscr{K}\llbracket \maybebluetext{ \ottnt{L} } \rrbracket   \Psi (\maybebluetext{ \ottnt{t} })    ]$ for any $\ottnt{L}$.
    We have
    \[
      \mathcal{F}  [  \ottnt{M}  ]  \ottsym{=}  \ottsym{(}  \mathcal{F}_{{\mathrm{1}}}  [  \ottnt{op}  \ottsym{(}  \ottnt{V}  \ottsym{,} \, \square \, \ottsym{)}  \langle  \ottnt{t}  \rangle  ]  \ottsym{)}  [  \ottnt{M}  ] = \mathcal{F}_{{\mathrm{1}}}  [  \ottnt{op}  \ottsym{(}  \ottnt{V}  \ottsym{,}  \ottnt{M}  \ottsym{)}  \langle  \ottnt{t}  \rangle  ]
    \]
    and so
    \[
    \begin{array}{llll}
       \mathscr{C}\llbracket \maybebluetext{ \mathcal{F}  [  \ottnt{M}  ] } \rrbracket 
      &=&  \mathscr{C}\llbracket \maybebluetext{ \mathcal{F}_{{\mathrm{1}}}  [  \ottnt{op}  \ottsym{(}  \ottnt{V}  \ottsym{,}  \ottnt{M}  \ottsym{)}  \langle  \ottnt{t}  \rangle  ] } \rrbracket  \\
      &=& \mathcal{E}'_{{\mathrm{1}}}  [   \mathscr{K}\llbracket \maybebluetext{ \ottnt{op}  \ottsym{(}  \ottnt{V}  \ottsym{,}  \ottnt{M}  \ottsym{)} } \rrbracket   \Psi (\maybebluetext{ \ottnt{t} })    ] & \text{by IH with $\ottnt{L}  \ottsym{=}  \ottnt{op}  \ottsym{(}  \ottnt{V}  \ottsym{,}  \ottnt{M}  \ottsym{)}$} \\
      &=& \mathcal{E}'_{{\mathrm{1}}}  [  \ottnt{op}  \ottsym{(}   \mathscr{C}\llbracket \maybebluetext{ \ottnt{V} } \rrbracket   \ottsym{,}   \mathscr{C}\llbracket \maybebluetext{ \ottnt{M} } \rrbracket   \ottsym{)}  \langle   \Psi (\maybebluetext{ \ottnt{t} })   \rangle  ] & \text{by \rnp{Tr-Op}.}
    \end{array}
    \]
    Take $\mathcal{E}'  \ottsym{=}  \mathcal{E}'_{{\mathrm{1}}}  [  \ottnt{op}  \ottsym{(}   \mathscr{C}\llbracket \maybebluetext{ \ottnt{V} } \rrbracket   \ottsym{,}  \square  \ottsym{)}  \langle   \Psi (\maybebluetext{ \ottnt{t} })   \rangle  ]$; then we have $\mathcal{E}'  [   \mathscr{C}\llbracket \maybebluetext{ \ottnt{M} } \rrbracket   ]  \ottsym{=}   \mathscr{C}\llbracket \maybebluetext{ \mathcal{F}  [  \ottnt{M}  ] } \rrbracket $.
    \else
    Similar.
    \fi

  \item[\Case{$\mathcal{F}  \ottsym{=}  \mathcal{F}_{{\mathrm{1}}}  [  \square \, \ottnt{N}  ]$}]  
    By the IH (item 1), there exists $\mathcal{E}'_{{\mathrm{1}}}$ such that
    $ \mathscr{C}\llbracket \maybebluetext{ \mathcal{F}_{{\mathrm{1}}}  [  \ottnt{L}  ] } \rrbracket  = \mathcal{E}'_{{\mathrm{1}}}  [   \mathscr{C}\llbracket \maybebluetext{ \ottnt{L} } \rrbracket   ]$ for any $\ottnt{L}$.
    We have
    \[
      \mathcal{F}  [  \ottnt{M}  ]  \ottsym{=}  \ottsym{(}  \mathcal{F}_{{\mathrm{1}}}  [  \square \, \ottnt{N}  ]  \ottsym{)}  [  \ottnt{M}  ] = \mathcal{F}_{{\mathrm{1}}}  [  \ottnt{M} \, \ottnt{N}  ]
    \]
    and so
    \[
    \begin{array}{llll}
       \mathscr{C}\llbracket \maybebluetext{ \mathcal{F}  [  \ottnt{M}  ] } \rrbracket 
      &=&  \mathscr{C}\llbracket \maybebluetext{ \mathcal{F}_{{\mathrm{1}}}  [  \ottnt{M} \, \ottnt{N}  ] } \rrbracket  \\
      &=& \mathcal{E}'_{{\mathrm{1}}}  [   \mathscr{C}\llbracket \maybebluetext{ \ottnt{M} \, \ottnt{N} } \rrbracket   ] & \text{by IH with $\ottnt{L}  \ottsym{=}  \ottnt{M} \, \ottnt{N}$} \\
      &=& \mathcal{E}'_{{\mathrm{1}}}  [   \mathscr{K}\llbracket \maybebluetext{ \ottnt{M} \, \ottnt{N} } \rrbracket  \ottkw{id}   ]  \\
      &=& \mathcal{E}'_{{\mathrm{1}}}  [    \mathscr{C}\llbracket \maybebluetext{ \ottnt{M} } \rrbracket  \, (  \mathscr{C}\llbracket \maybebluetext{ \ottnt{N} } \rrbracket  , \ottkw{id} )   ] &\text{by \rnp{Tr-App}.}
    \end{array}
  \]
  Take $\mathcal{E}'  \ottsym{=}  \mathcal{E}'_{{\mathrm{1}}}  [   \square \, (   \mathscr{C}\llbracket \maybebluetext{ \ottnt{N} } \rrbracket  , \ottkw{id}  )   ]$; then we have $\mathcal{E}'  [   \mathscr{C}\llbracket \maybebluetext{ \ottnt{M} } \rrbracket   ]  \ottsym{=}   \mathscr{C}\llbracket \maybebluetext{ \mathcal{F}  [  \ottnt{M}  ] } \rrbracket $.

  \item[\Case{$\mathcal{F}  \ottsym{=}  \mathcal{F}_{{\mathrm{1}}}  [  \ottsym{(} \, \square \, \ottnt{N}  \ottsym{)}  \langle  \ottnt{t}  \rangle  ]$}]
    By the IH (item 2), there exists $\mathcal{E}'_{{\mathrm{1}}}$ such that
    $ \mathscr{C}\llbracket \maybebluetext{ \mathcal{F}_{{\mathrm{1}}}  [  \ottnt{L}  \langle  \ottnt{t}  \rangle  ] } \rrbracket  = \mathcal{E}'_{{\mathrm{1}}}  [   \mathscr{K}\llbracket \maybebluetext{ \ottnt{L} } \rrbracket   \Psi (\maybebluetext{ \ottnt{t} })    ]$ for any $\ottnt{L}$.
    We have
    \[
      \mathcal{F}  [  \ottnt{M}  ]  \ottsym{=}  \ottsym{(}  \mathcal{F}_{{\mathrm{1}}}  [  \ottsym{(} \, \square \, \ottnt{N}  \ottsym{)}  \langle  \ottnt{t}  \rangle  ]  \ottsym{)}  [  \ottnt{M}  ] = \mathcal{F}_{{\mathrm{1}}}  [  \ottsym{(}  \ottnt{M} \, \ottnt{N}  \ottsym{)}  \langle  \ottnt{t}  \rangle  ]
    \]
    and so
    \[
    \begin{array}{llll}
       \mathscr{C}\llbracket \maybebluetext{ \mathcal{F}  [  \ottnt{M}  ] } \rrbracket 
      &=&  \mathscr{C}\llbracket \maybebluetext{ \mathcal{F}_{{\mathrm{1}}}  [  \ottsym{(}  \ottnt{M} \, \ottnt{N}  \ottsym{)}  \langle  \ottnt{t}  \rangle  ] } \rrbracket  \\
      &=& \mathcal{E}'_{{\mathrm{1}}}  [   \mathscr{K}\llbracket \maybebluetext{ \ottsym{(}  \ottnt{M} \, \ottnt{N}  \ottsym{)} } \rrbracket   \Psi (\maybebluetext{ \ottnt{t} })    ] & \text{by IH with $\ottnt{L}  \ottsym{=}  \ottnt{M} \, \ottnt{N}$} \\
      &=& \mathcal{E}'_{{\mathrm{1}}}  [    \mathscr{C}\llbracket \maybebluetext{ \ottnt{M} } \rrbracket  \, (  \mathscr{C}\llbracket \maybebluetext{ \ottnt{N} } \rrbracket  ,  \Psi (\maybebluetext{ \ottnt{t} })  )   ] &\text{by \rnp{Tr-App}.}
    \end{array}
    \]
    Take $\mathcal{E}'  \ottsym{=}  \mathcal{E}'_{{\mathrm{1}}}  [   \square \, (   \mathscr{C}\llbracket \maybebluetext{ \ottnt{N} } \rrbracket  ,  \Psi (\maybebluetext{ \ottnt{t} })   )   ]$; then we have $\mathcal{E}'  [   \mathscr{C}\llbracket \maybebluetext{ \ottnt{M} } \rrbracket   ]  \ottsym{=}   \mathscr{C}\llbracket \maybebluetext{ \mathcal{F}  [  \ottnt{M}  ] } \rrbracket $.

  \item[\Case{$\mathcal{F}  \ottsym{=}  \mathcal{F}_{{\mathrm{1}}}  [  \ottnt{V} \, \square  ]$}]  
    \ifshownotes
    By the IH (item 2), there exists $\mathcal{E}'_{{\mathrm{1}}}$ such that
    $ \mathscr{C}\llbracket \maybebluetext{ \mathcal{F}_{{\mathrm{1}}}  [  \ottnt{L}  ] } \rrbracket  = \mathcal{E}'_{{\mathrm{1}}}  [   \mathscr{C}\llbracket \maybebluetext{ \ottnt{L} } \rrbracket   ]$ for any $\ottnt{L}$.
    We have
    \[
      \mathcal{F}  [  \ottnt{M}  ]  \ottsym{=}  \ottsym{(}  \mathcal{F}_{{\mathrm{1}}}  [  \ottnt{V} \, \square  ]  \ottsym{)}  [  \ottnt{M}  ] = \mathcal{F}_{{\mathrm{1}}}  [  \ottnt{V} \, \ottnt{M}  ]
    \]
    and so
    \[
    \begin{array}{llll}
       \mathscr{C}\llbracket \maybebluetext{ \mathcal{F}  [  \ottnt{M}  ] } \rrbracket 
      &=&  \mathscr{C}\llbracket \maybebluetext{ \mathcal{F}_{{\mathrm{1}}}  [  \ottnt{V} \, \ottnt{M}  ] } \rrbracket  \\
      &=& \mathcal{E}'_{{\mathrm{1}}}  [   \mathscr{C}\llbracket \maybebluetext{ \ottnt{V} \, \ottnt{M} } \rrbracket   ] & \text{by IH with $\ottnt{L}  \ottsym{=}  \ottnt{V} \, \ottnt{M}$} \\
      &=& \mathcal{E}'_{{\mathrm{1}}}  [   \mathscr{K}\llbracket \maybebluetext{ \ottnt{V} \, \ottnt{M} } \rrbracket  \ottkw{id}   ]  \\
      &=& \mathcal{E}'_{{\mathrm{1}}}  [    \mathscr{C}\llbracket \maybebluetext{ \ottnt{V} } \rrbracket  \, (  \mathscr{C}\llbracket \maybebluetext{ \ottnt{M} } \rrbracket  , \ottkw{id} )   ] &\text{by \rnp{Tr-App}.}
    \end{array}
    \]
    Take $\mathcal{E}'  \ottsym{=}  \mathcal{E}'_{{\mathrm{1}}}  [    \mathscr{C}\llbracket \maybebluetext{ \ottnt{V} } \rrbracket  \, ( \square , \ottkw{id} )   ]$; then we have $\mathcal{E}'  [   \mathscr{C}\llbracket \maybebluetext{ \ottnt{M} } \rrbracket   ]  \ottsym{=}   \mathscr{C}\llbracket \maybebluetext{ \mathcal{F}  [  \ottnt{M}  ] } \rrbracket $.
    \else
    Similar.
    \fi

  \item[\Case{$\mathcal{F}  \ottsym{=}  \mathcal{F}_{{\mathrm{1}}}  [  \ottsym{(}  \ottnt{V} \, \square \, \ottsym{)}  \langle  \ottnt{t}  \rangle  ]$}]
    \ifshownotes
    By the IH (item 2), there exists $\mathcal{E}'_{{\mathrm{1}}}$ such that
    $ \mathscr{C}\llbracket \maybebluetext{ \mathcal{F}_{{\mathrm{1}}}  [  \ottnt{L}  \langle  \ottnt{t}  \rangle  ] } \rrbracket  = \mathcal{E}'_{{\mathrm{1}}}  [   \mathscr{K}\llbracket \maybebluetext{ \ottnt{L} } \rrbracket   \Psi (\maybebluetext{ \ottnt{t} })    ]$ for any $\ottnt{L}$.
    We have
    \[
      \mathcal{F}  [  \ottnt{M}  ]  \ottsym{=}  \ottsym{(}  \mathcal{F}_{{\mathrm{1}}}  [  \ottsym{(}  \ottnt{V} \, \square \, \ottsym{)}  \langle  \ottnt{t}  \rangle  ]  \ottsym{)}  [  \ottnt{M}  ] = \mathcal{F}_{{\mathrm{1}}}  [  \ottsym{(}  \ottnt{V} \, \ottnt{M}  \ottsym{)}  \langle  \ottnt{t}  \rangle  ]
    \]
    and so
    \[
    \begin{array}{llll}
       \mathscr{C}\llbracket \maybebluetext{ \mathcal{F}  [  \ottnt{M}  ] } \rrbracket 
      &=&  \mathscr{C}\llbracket \maybebluetext{ \mathcal{F}_{{\mathrm{1}}}  [  \ottsym{(}  \ottnt{V} \, \ottnt{M}  \ottsym{)}  \langle  \ottnt{t}  \rangle  ] } \rrbracket  \\
      &=& \mathcal{E}'_{{\mathrm{1}}}  [   \mathscr{K}\llbracket \maybebluetext{ \ottsym{(}  \ottnt{V} \, \ottnt{M}  \ottsym{)} } \rrbracket   \Psi (\maybebluetext{ \ottnt{t} })    ] & \text{by IH with $\ottnt{L}  \ottsym{=}  \ottnt{V} \, \ottnt{M}$} \\
      &=& \mathcal{E}'_{{\mathrm{1}}}  [    \mathscr{C}\llbracket \maybebluetext{ \ottnt{V} } \rrbracket  \, (  \mathscr{C}\llbracket \maybebluetext{ \ottnt{M} } \rrbracket  ,  \Psi (\maybebluetext{ \ottnt{t} })  )   ] &\text{by \rnp{Tr-App}.}
    \end{array}
    \]
    Take $\mathcal{E}'  \ottsym{=}  \mathcal{E}'_{{\mathrm{1}}}  [    \mathscr{C}\llbracket \maybebluetext{ \ottnt{V} } \rrbracket  \, ( \square ,  \Psi (\maybebluetext{ \ottnt{t} })  )   ]$; then we have $\mathcal{E}'  [   \mathscr{C}\llbracket \maybebluetext{ \ottnt{M} } \rrbracket   ]  \ottsym{=}   \mathscr{C}\llbracket \maybebluetext{ \mathcal{F}  [  \ottnt{M}  ] } \rrbracket $.
    \else
    Similar.
    \fi
  \end{description}

  (2) By case analysis on the structure of $\mathcal{F}$.
  \begin{description}
  \item[\Case{$\mathcal{F}  \ottsym{=}  \square$}]
    By $\square  [  \ottnt{M}  \langle  \ottnt{s}  \rangle  ]  \ottsym{=}  \ottnt{M}  \langle  \ottnt{s}  \rangle$, we must show $ \mathscr{C}\llbracket \maybebluetext{ \ottnt{M}  \langle  \ottnt{s}  \rangle } \rrbracket   \ottsym{=}  \mathcal{E}'  [   \mathscr{K}\llbracket \maybebluetext{ \ottnt{M} } \rrbracket   \Psi (\maybebluetext{ \ottnt{s} })    ]$.
    We now have $ \mathscr{C}\llbracket \maybebluetext{ \ottnt{M}  \langle  \ottnt{s}  \rangle } \rrbracket  =  \mathscr{K}\llbracket \maybebluetext{ \ottnt{M} } \rrbracket   \Psi (\maybebluetext{ \ottnt{s} })  $.
    Take $\mathcal{E}' =  \square $.

  \item[\Case{$\mathcal{F}  \ottsym{=}  \mathcal{F}_{{\mathrm{1}}}  [  \ottnt{op}  \ottsym{(}  \square  \ottsym{,}  \ottnt{N}  \ottsym{)}  ]$}]    
    By the IH (item 1), there exists $\mathcal{E}'_{{\mathrm{1}}}$ such that
    $ \mathscr{C}\llbracket \maybebluetext{ \mathcal{F}_{{\mathrm{1}}}  [  \ottnt{L}  ] } \rrbracket  = \mathcal{E}'_{{\mathrm{1}}}  [   \mathscr{C}\llbracket \maybebluetext{ \ottnt{L} } \rrbracket   ]$ for any $\ottnt{L}$.
    We have
    \[
      \mathcal{F}  [  \ottnt{M}  \langle  \ottnt{s}  \rangle  ]  \ottsym{=}  \ottsym{(}  \mathcal{F}_{{\mathrm{1}}}  [  \ottnt{op}  \ottsym{(}  \square  \ottsym{,}  \ottnt{N}  \ottsym{)}  ]  \ottsym{)}  [  \ottnt{M}  \langle  \ottnt{s}  \rangle  ] = \mathcal{F}_{{\mathrm{1}}}  [  \ottnt{op}  \ottsym{(}  \ottnt{M}  \langle  \ottnt{s}  \rangle  \ottsym{,}  \ottnt{N}  \ottsym{)}  ].
    \]
    and so
    \[
    \begin{array}{llll}
       \mathscr{C}\llbracket \maybebluetext{ \mathcal{F}  [  \ottnt{M}  \langle  \ottnt{s}  \rangle  ] } \rrbracket 
      &=&  \mathscr{C}\llbracket \maybebluetext{ \mathcal{F}_{{\mathrm{1}}}  [  \ottnt{op}  \ottsym{(}  \ottnt{M}  \langle  \ottnt{s}  \rangle  \ottsym{,}  \ottnt{N}  \ottsym{)}  ] } \rrbracket  \\
      &=& \mathcal{E}'_{{\mathrm{1}}}  [   \mathscr{C}\llbracket \maybebluetext{ \ottnt{op}  \ottsym{(}  \ottnt{M}  \langle  \ottnt{s}  \rangle  \ottsym{,}  \ottnt{N}  \ottsym{)} } \rrbracket   ] & \text{by IH with $\ottnt{L}  \ottsym{=}  \ottnt{op}  \ottsym{(}  \ottnt{M}  \langle  \ottnt{s}  \rangle  \ottsym{,}  \ottnt{N}  \ottsym{)}$} \\
      &=& \mathcal{E}'_{{\mathrm{1}}}  [   \mathscr{K}\llbracket \maybebluetext{ \ottnt{op}  \ottsym{(}  \ottnt{M}  \langle  \ottnt{s}  \rangle  \ottsym{,}  \ottnt{N}  \ottsym{)} } \rrbracket  \ottkw{id}   ]  \\ 
      &=& \mathcal{E}'_{{\mathrm{1}}}  [  \ottnt{op}  \ottsym{(}   \mathscr{C}\llbracket \maybebluetext{ \ottnt{M}  \langle  \ottnt{s}  \rangle } \rrbracket   \ottsym{,}   \mathscr{C}\llbracket \maybebluetext{ \ottnt{N} } \rrbracket   \ottsym{)}  \langle  \ottkw{id}  \rangle  ] & \text{by \rnp{Tr-Op}}\\
      &=& \mathcal{E}'_{{\mathrm{1}}}  [  \ottnt{op}  \ottsym{(}   \mathscr{K}\llbracket \maybebluetext{ \ottnt{M} } \rrbracket   \Psi (\maybebluetext{ \ottnt{s} })    \ottsym{,}   \mathscr{C}\llbracket \maybebluetext{ \ottnt{N} } \rrbracket   \ottsym{)}  \langle  \ottkw{id}  \rangle  ].
    \end{array}
    \]
    Take $\mathcal{E}'  \ottsym{=}  \mathcal{E}'_{{\mathrm{1}}}  [  \ottnt{op}  \ottsym{(} \, \square \, \ottsym{,}   \mathscr{C}\llbracket \maybebluetext{ \ottnt{N} } \rrbracket   \ottsym{)}  \langle  \ottkw{id}  \rangle  ]$; then we have $\mathcal{E}'  [   \mathscr{K}\llbracket \maybebluetext{ \ottnt{M} } \rrbracket   \Psi (\maybebluetext{ \ottnt{s} })    ]  \ottsym{=}   \mathscr{C}\llbracket \maybebluetext{ \mathcal{F}  [  \ottnt{M}  \langle  \ottnt{s}  \rangle  ] } \rrbracket $.

  \item[\Case{$\mathcal{F}  \ottsym{=}  \mathcal{F}_{{\mathrm{1}}}  [  \ottnt{op}  \ottsym{(} \, \square \, \ottsym{,}  \ottnt{N}  \ottsym{)}  \langle  \ottnt{t}  \rangle  ]$}]
    By the IH (item 2), there exists $\mathcal{E}'_{{\mathrm{1}}}$ such that
    $ \mathscr{C}\llbracket \maybebluetext{ \mathcal{F}_{{\mathrm{1}}}  [  \ottnt{L}  \langle  \ottnt{t}  \rangle  ] } \rrbracket  = \mathcal{E}'_{{\mathrm{1}}}  [   \mathscr{K}\llbracket \maybebluetext{ \ottnt{L} } \rrbracket   \Psi (\maybebluetext{ \ottnt{t} })    ]$ for any $\ottnt{L}$.
    We have
    \[
      \mathcal{F}  [  \ottnt{M}  \langle  \ottnt{s}  \rangle  ]  \ottsym{=}  \ottsym{(}  \mathcal{F}_{{\mathrm{1}}}  [  \ottnt{op}  \ottsym{(} \, \square \, \ottsym{,}  \ottnt{N}  \ottsym{)}  \langle  \ottnt{t}  \rangle  ]  \ottsym{)}  [  \ottnt{M}  \langle  \ottnt{s}  \rangle  ] = \mathcal{F}_{{\mathrm{1}}}  [  \ottnt{op}  \ottsym{(}  \ottnt{M}  \langle  \ottnt{s}  \rangle  \ottsym{,}  \ottnt{N}  \ottsym{)}  \langle  \ottnt{t}  \rangle  ].
    \]
    and so
    \[
    \begin{array}{llll}
       \mathscr{C}\llbracket \maybebluetext{ \mathcal{F}  [  \ottnt{M}  \langle  \ottnt{s}  \rangle  ] } \rrbracket 
      &=&  \mathscr{C}\llbracket \maybebluetext{ \mathcal{F}_{{\mathrm{1}}}  [  \ottnt{op}  \ottsym{(}  \ottnt{M}  \langle  \ottnt{s}  \rangle  \ottsym{,}  \ottnt{N}  \ottsym{)}  \langle  \ottnt{t}  \rangle  ] } \rrbracket  \\
      &=& \mathcal{E}'_{{\mathrm{1}}}  [   \mathscr{K}\llbracket \maybebluetext{ \ottnt{op}  \ottsym{(}  \ottnt{M}  \langle  \ottnt{s}  \rangle  \ottsym{,}  \ottnt{N}  \ottsym{)} } \rrbracket   \Psi (\maybebluetext{ \ottnt{t} })    ] & \text{by IH with $\ottnt{L}  \ottsym{=}  \ottnt{op}  \ottsym{(}  \ottnt{M}  \langle  \ottnt{s}  \rangle  \ottsym{,}  \ottnt{N}  \ottsym{)}$} \\
      &=& \mathcal{E}'_{{\mathrm{1}}}  [  \ottnt{op}  \ottsym{(}   \mathscr{C}\llbracket \maybebluetext{ \ottnt{M}  \langle  \ottnt{s}  \rangle } \rrbracket   \ottsym{,}   \mathscr{C}\llbracket \maybebluetext{ \ottnt{N} } \rrbracket   \ottsym{)}  \langle   \Psi (\maybebluetext{ \ottnt{t} })   \rangle  ] & \text{by \rnp{Tr-Op}}\\
      &=& \mathcal{E}'_{{\mathrm{1}}}  [  \ottnt{op}  \ottsym{(}   \mathscr{K}\llbracket \maybebluetext{ \ottnt{M} } \rrbracket   \Psi (\maybebluetext{ \ottnt{s} })    \ottsym{,}   \mathscr{C}\llbracket \maybebluetext{ \ottnt{N} } \rrbracket   \ottsym{)}  \langle   \Psi (\maybebluetext{ \ottnt{t} })   \rangle  ].
    \end{array}
    \]
    Take $\mathcal{E}'  \ottsym{=}  \mathcal{E}'_{{\mathrm{1}}}  [  \ottnt{op}  \ottsym{(} \, \square \, \ottsym{,}   \mathscr{C}\llbracket \maybebluetext{ \ottnt{N} } \rrbracket   \ottsym{)}  \langle   \Psi (\maybebluetext{ \ottnt{t} })   \rangle  ]$; then we have $\mathcal{E}'  [   \mathscr{K}\llbracket \maybebluetext{ \ottnt{M} } \rrbracket   \Psi (\maybebluetext{ \ottnt{s} })    ]  \ottsym{=}   \mathscr{C}\llbracket \maybebluetext{ \mathcal{F}  [  \ottnt{M}  \langle  \ottnt{s}  \rangle  ] } \rrbracket $.

  \item[\Case{$\mathcal{F}  \ottsym{=}  \mathcal{F}_{{\mathrm{1}}}  [  \ottnt{op}  \ottsym{(}  \ottnt{V}  \ottsym{,} \, \square \, \ottsym{)}  ]$}]    
    \ifshownotes
    By the IH (item 1), there exists $\mathcal{E}'_{{\mathrm{1}}}$ such that
    $ \mathscr{C}\llbracket \maybebluetext{ \mathcal{F}_{{\mathrm{1}}}  [  \ottnt{L}  ] } \rrbracket  = \mathcal{E}'_{{\mathrm{1}}}  [   \mathscr{C}\llbracket \maybebluetext{ \ottnt{L} } \rrbracket   ]$ for any $\ottnt{L}$.
    We have
    \[
      \mathcal{F}  [  \ottnt{M}  \langle  \ottnt{s}  \rangle  ]  \ottsym{=}  \ottsym{(}  \mathcal{F}_{{\mathrm{1}}}  [  \ottnt{op}  \ottsym{(}  \ottnt{V}  \ottsym{,} \, \square \, \ottsym{)}  ]  \ottsym{)}  [  \ottnt{M}  \langle  \ottnt{s}  \rangle  ] = \mathcal{F}_{{\mathrm{1}}}  [  \ottnt{op}  \ottsym{(}  \ottnt{V}  \ottsym{,}  \ottnt{M}  \langle  \ottnt{s}  \rangle  \ottsym{)}  ].
    \]
    and so
    \[
    \begin{array}{llll}
       \mathscr{C}\llbracket \maybebluetext{ \mathcal{F}  [  \ottnt{M}  \langle  \ottnt{s}  \rangle  ] } \rrbracket 
      &=&  \mathscr{C}\llbracket \maybebluetext{ \mathcal{F}_{{\mathrm{1}}}  [  \ottnt{op}  \ottsym{(}  \ottnt{V}  \ottsym{,}  \ottnt{M}  \langle  \ottnt{s}  \rangle  \ottsym{)}  ] } \rrbracket  \\
      &=& \mathcal{E}'_{{\mathrm{1}}}  [   \mathscr{C}\llbracket \maybebluetext{ \ottnt{op}  \ottsym{(}  \ottnt{V}  \ottsym{,}  \ottnt{M}  \langle  \ottnt{s}  \rangle  \ottsym{)} } \rrbracket   ] & \text{by IH with $\ottnt{L}  \ottsym{=}  \ottnt{op}  \ottsym{(}  \ottnt{V}  \ottsym{,}  \ottnt{M}  \langle  \ottnt{s}  \rangle  \ottsym{)}$} \\
      &=& \mathcal{E}'_{{\mathrm{1}}}  [   \mathscr{K}\llbracket \maybebluetext{ \ottnt{op}  \ottsym{(}  \ottnt{V}  \ottsym{,}  \ottnt{M}  \langle  \ottnt{s}  \rangle  \ottsym{)} } \rrbracket  \ottkw{id}   ]  \\
      &=& \mathcal{E}'_{{\mathrm{1}}}  [  \ottnt{op}  \ottsym{(}   \mathscr{C}\llbracket \maybebluetext{ \ottnt{V} } \rrbracket   \ottsym{,}   \mathscr{C}\llbracket \maybebluetext{ \ottnt{M}  \langle  \ottnt{s}  \rangle } \rrbracket   \ottsym{)}  \langle  \ottkw{id}  \rangle  ] & \text{by \rnp{Tr-Op}}\\
      &=& \mathcal{E}'_{{\mathrm{1}}}  [  \ottnt{op}  \ottsym{(}   \Psi (\maybebluetext{ \ottnt{V} })   \ottsym{,}   \mathscr{K}\llbracket \maybebluetext{ \ottnt{M} } \rrbracket   \Psi (\maybebluetext{ \ottnt{s} })    \ottsym{)}  \langle  \ottkw{id}  \rangle  ].
    \end{array}
    \]
    Take $\mathcal{E}'  \ottsym{=}  \mathcal{E}'_{{\mathrm{1}}}  [  \ottnt{op}  \ottsym{(}   \Psi (\maybebluetext{ \ottnt{V} })   \ottsym{,} \, \square \, \ottsym{)}  \langle  \ottkw{id}  \rangle  ]$; then we have $\mathcal{E}'  [   \mathscr{K}\llbracket \maybebluetext{ \ottnt{M} } \rrbracket   \Psi (\maybebluetext{ \ottnt{s} })    ]  \ottsym{=}   \mathscr{C}\llbracket \maybebluetext{ \mathcal{F}  [  \ottnt{M}  \langle  \ottnt{s}  \rangle  ] } \rrbracket $.
    \else
    Similar.
    \fi

  \item[\Case{$\mathcal{F}  \ottsym{=}  \mathcal{F}_{{\mathrm{1}}}  [  \ottnt{op}  \ottsym{(}  \ottnt{V}  \ottsym{,} \, \square \, \ottsym{)}  \langle  \ottnt{t}  \rangle  ]$}]
    \ifshownotes
    By the IH (item 2), there exists $\mathcal{E}'_{{\mathrm{1}}}$ such that
    $ \mathscr{C}\llbracket \maybebluetext{ \mathcal{F}_{{\mathrm{1}}}  [  \ottnt{L}  \langle  \ottnt{t}  \rangle  ] } \rrbracket  = \mathcal{E}'_{{\mathrm{1}}}  [   \mathscr{K}\llbracket \maybebluetext{ \ottnt{L} } \rrbracket   \Psi (\maybebluetext{ \ottnt{t} })    ]$ for any $\ottnt{L}$.
    We have
    \[
      \mathcal{F}  [  \ottnt{M}  \langle  \ottnt{s}  \rangle  ]  \ottsym{=}  \ottsym{(}  \mathcal{F}_{{\mathrm{1}}}  [  \ottnt{op}  \ottsym{(}  \ottnt{V}  \ottsym{,} \, \square \, \ottsym{)}  \langle  \ottnt{t}  \rangle  ]  \ottsym{)}  [  \ottnt{M}  \langle  \ottnt{s}  \rangle  ] = \mathcal{F}_{{\mathrm{1}}}  [  \ottnt{op}  \ottsym{(}  \ottnt{V}  \ottsym{,}  \ottnt{M}  \langle  \ottnt{s}  \rangle  \ottsym{)}  \langle  \ottnt{t}  \rangle  ].
    \]
    and so
    \[
    \begin{array}{llll}
       \mathscr{C}\llbracket \maybebluetext{ \mathcal{F}  [  \ottnt{M}  \langle  \ottnt{s}  \rangle  ] } \rrbracket 
      &=&  \mathscr{C}\llbracket \maybebluetext{ \mathcal{F}_{{\mathrm{1}}}  [  \ottnt{op}  \ottsym{(}  \ottnt{V}  \ottsym{,}  \ottnt{M}  \langle  \ottnt{s}  \rangle  \ottsym{)}  \langle  \ottnt{t}  \rangle  ] } \rrbracket  \\
      &=& \mathcal{E}'_{{\mathrm{1}}}  [   \mathscr{K}\llbracket \maybebluetext{ \ottnt{op}  \ottsym{(}  \ottnt{V}  \ottsym{,}  \ottnt{M}  \langle  \ottnt{s}  \rangle  \ottsym{)} } \rrbracket   \Psi (\maybebluetext{ \ottnt{t} })    ] & \text{by IH with $\ottnt{L}  \ottsym{=}  \ottnt{op}  \ottsym{(}  \ottnt{M}  \ottsym{,}  \ottnt{N}  \ottsym{)}$} \\
      &=& \mathcal{E}'_{{\mathrm{1}}}  [  \ottnt{op}  \ottsym{(}   \mathscr{C}\llbracket \maybebluetext{ \ottnt{V} } \rrbracket   \ottsym{,}   \mathscr{C}\llbracket \maybebluetext{ \ottnt{M}  \langle  \ottnt{s}  \rangle } \rrbracket   \ottsym{)}  \langle   \Psi (\maybebluetext{ \ottnt{t} })   \rangle  ] & \text{by \rnp{Tr-Op}}\\
      &=& \mathcal{E}'_{{\mathrm{1}}}  [  \ottnt{op}  \ottsym{(}   \Psi (\maybebluetext{ \ottnt{V} })   \ottsym{,}   \mathscr{K}\llbracket \maybebluetext{ \ottnt{M} } \rrbracket   \Psi (\maybebluetext{ \ottnt{s} })    \ottsym{)}  \langle   \Psi (\maybebluetext{ \ottnt{t} })   \rangle  ].
    \end{array}
    \]
    Take $\mathcal{E}'  \ottsym{=}  \mathcal{E}'_{{\mathrm{1}}}  [  \ottnt{op}  \ottsym{(}   \Psi (\maybebluetext{ \ottnt{V} })   \ottsym{,} \, \square \, \ottsym{)}  \langle   \Psi (\maybebluetext{ \ottnt{t} })   \rangle  ]$; then we have $\mathcal{E}'  [   \mathscr{K}\llbracket \maybebluetext{ \ottnt{M} } \rrbracket   \Psi (\maybebluetext{ \ottnt{s} })    ]  \ottsym{=}   \mathscr{C}\llbracket \maybebluetext{ \mathcal{F}  [  \ottnt{M}  \langle  \ottnt{s}  \rangle  ] } \rrbracket $.
    \else
    Similar.
    \fi

  \item[\Case{$\mathcal{F}  \ottsym{=}  \mathcal{F}_{{\mathrm{1}}}  [  \square \, \ottnt{N}  ]$}]  
    By the IH (item 1), there exists $\mathcal{E}'_{{\mathrm{1}}}$ such that
    $ \mathscr{C}\llbracket \maybebluetext{ \mathcal{F}_{{\mathrm{1}}}  [  \ottnt{L}  ] } \rrbracket  = \mathcal{E}'_{{\mathrm{1}}}  [   \mathscr{C}\llbracket \maybebluetext{ \ottnt{L} } \rrbracket   ]$ for any $\ottnt{L}$.
    We have
    \[
      \mathcal{F}  [  \ottnt{M}  \langle  \ottnt{s}  \rangle  ]  \ottsym{=}  \ottsym{(}  \mathcal{F}_{{\mathrm{1}}}  [  \square \, \ottnt{N}  ]  \ottsym{)}  [  \ottnt{M}  \langle  \ottnt{s}  \rangle  ] = \mathcal{F}_{{\mathrm{1}}}  [  \ottsym{(}  \ottnt{M}  \langle  \ottnt{s}  \rangle  \ottsym{)} \, \ottnt{N}  ].
    \]
    and so
    \[
    \begin{array}{llll}
       \mathscr{C}\llbracket \maybebluetext{ \mathcal{F}  [  \ottnt{M}  \langle  \ottnt{s}  \rangle  ] } \rrbracket 
      &=&  \mathscr{C}\llbracket \maybebluetext{ \mathcal{F}_{{\mathrm{1}}}  [  \ottsym{(}  \ottnt{M}  \langle  \ottnt{s}  \rangle  \ottsym{)} \, \ottnt{N}  ] } \rrbracket  \\
      &=& \mathcal{E}'_{{\mathrm{1}}}  [   \mathscr{C}\llbracket \maybebluetext{ \ottsym{(}  \ottnt{M}  \langle  \ottnt{s}  \rangle  \ottsym{)} \, \ottnt{N} } \rrbracket   ] & \text{by IH with $\ottnt{L}  \ottsym{=}  \ottnt{op}  \ottsym{(}  \ottnt{M}  \langle  \ottnt{s}  \rangle  \ottsym{,}  \ottnt{N}  \ottsym{)}$} \\
      &=& \mathcal{E}'_{{\mathrm{1}}}  [   \mathscr{K}\llbracket \maybebluetext{ \ottsym{(}  \ottnt{M}  \langle  \ottnt{s}  \rangle  \ottsym{)} \, \ottnt{N} } \rrbracket  \ottkw{id}   ]  & \text{by \rnp{Tr-App}}\\
      &=& \mathcal{E}'_{{\mathrm{1}}}  [    \mathscr{C}\llbracket \maybebluetext{ \ottnt{M}  \langle  \ottnt{s}  \rangle } \rrbracket  \, (  \mathscr{C}\llbracket \maybebluetext{ \ottnt{N} } \rrbracket  , \ottkw{id} )   ] \\
      &=& \mathcal{E}'_{{\mathrm{1}}}  [   \ottsym{(}   \mathscr{K}\llbracket \maybebluetext{ \ottnt{M} } \rrbracket   \Psi (\maybebluetext{ \ottnt{s} })    \ottsym{)} \, (  \mathscr{C}\llbracket \maybebluetext{ \ottnt{N} } \rrbracket  , \ottkw{id} )   ].
    \end{array}
    \]
    Take $\mathcal{E}'  \ottsym{=}  \mathcal{E}'_{{\mathrm{1}}}  [   \square \, (   \mathscr{C}\llbracket \maybebluetext{ \ottnt{N} } \rrbracket  , \ottkw{id}  )   ]$; then we have $\mathcal{E}'  [   \mathscr{K}\llbracket \maybebluetext{ \ottnt{M} } \rrbracket   \Psi (\maybebluetext{ \ottnt{s} })    ]  \ottsym{=}   \mathscr{C}\llbracket \maybebluetext{ \mathcal{F}  [  \ottnt{M}  \langle  \ottnt{s}  \rangle  ] } \rrbracket $.

  \item[\Case{$\mathcal{F}  \ottsym{=}  \mathcal{F}_{{\mathrm{1}}}  [  \ottsym{(} \, \square \, \ottnt{N}  \ottsym{)}  \langle  \ottnt{t}  \rangle  ]$}]
    By the IH (item 2), there exists $\mathcal{E}'_{{\mathrm{1}}}$ such that
    $ \mathscr{C}\llbracket \maybebluetext{ \mathcal{F}_{{\mathrm{1}}}  [  \ottnt{L}  \langle  \ottnt{t}  \rangle  ] } \rrbracket  = \mathcal{E}'_{{\mathrm{1}}}  [   \mathscr{K}\llbracket \maybebluetext{ \ottnt{L} } \rrbracket   \Psi (\maybebluetext{ \ottnt{t} })    ]$ for any $\ottnt{L}$.
    We have
    \[
      \mathcal{F}  [  \ottnt{M}  \langle  \ottnt{s}  \rangle  ]  \ottsym{=}  \ottsym{(}  \mathcal{F}_{{\mathrm{1}}}  [  \ottsym{(} \, \square \, \ottnt{N}  \ottsym{)}  \langle  \ottnt{t}  \rangle  ]  \ottsym{)}  [  \ottnt{M}  \langle  \ottnt{s}  \rangle  ] = \mathcal{F}_{{\mathrm{1}}}  [  \ottsym{(}  \ottsym{(}  \ottnt{M}  \langle  \ottnt{s}  \rangle  \ottsym{)} \, \ottnt{N}  \ottsym{)}  \langle  \ottnt{t}  \rangle  ].
    \]
    and so
    \[
    \begin{array}{llll}
       \mathscr{C}\llbracket \maybebluetext{ \mathcal{F}  [  \ottnt{M}  \langle  \ottnt{s}  \rangle  ] } \rrbracket 
      &=&  \mathscr{C}\llbracket \maybebluetext{ \mathcal{F}_{{\mathrm{1}}}  [  \ottsym{(}  \ottsym{(}  \ottnt{M}  \langle  \ottnt{s}  \rangle  \ottsym{)} \, \ottnt{N}  \ottsym{)}  \langle  \ottnt{t}  \rangle  ] } \rrbracket  \\
      &=& \mathcal{E}'_{{\mathrm{1}}}  [   \mathscr{K}\llbracket \maybebluetext{ \ottsym{(}  \ottnt{M}  \langle  \ottnt{s}  \rangle  \ottsym{)} \, \ottnt{N} } \rrbracket   \Psi (\maybebluetext{ \ottnt{t} })    ] & \text{by IH with $\ottnt{L}  \ottsym{=}  \ottsym{(}  \ottnt{M}  \langle  \ottnt{s}  \rangle  \ottsym{)} \, \ottnt{N}$} \\
      &=& \mathcal{E}'_{{\mathrm{1}}}  [    \mathscr{C}\llbracket \maybebluetext{ \ottnt{M}  \langle  \ottnt{s}  \rangle } \rrbracket  \, (  \mathscr{C}\llbracket \maybebluetext{ \ottnt{N} } \rrbracket  ,  \Psi (\maybebluetext{ \ottnt{t} })  )   ] & \text{by \rnp{Tr-App}}\\
      &=& \mathcal{E}'_{{\mathrm{1}}}  [   \ottsym{(}   \mathscr{K}\llbracket \maybebluetext{ \ottnt{M} } \rrbracket   \Psi (\maybebluetext{ \ottnt{s} })    \ottsym{)} \, (  \mathscr{C}\llbracket \maybebluetext{ \ottnt{N} } \rrbracket  ,  \Psi (\maybebluetext{ \ottnt{t} })  )   ].
    \end{array}
    \]
    Take $\mathcal{E}'  \ottsym{=}  \mathcal{E}'_{{\mathrm{1}}}  [   \square \, (   \mathscr{C}\llbracket \maybebluetext{ \ottnt{N} } \rrbracket  ,  \Psi (\maybebluetext{ \ottnt{t} })   )   ]$; then we have $\mathcal{E}'  [   \mathscr{K}\llbracket \maybebluetext{ \ottnt{M} } \rrbracket   \Psi (\maybebluetext{ \ottnt{s} })    ]  \ottsym{=}   \mathscr{C}\llbracket \maybebluetext{ \mathcal{F}  [  \ottnt{M}  \langle  \ottnt{s}  \rangle  ] } \rrbracket $.

  \item[\Case{$\mathcal{F}  \ottsym{=}  \mathcal{F}_{{\mathrm{1}}}  [  \ottnt{V} \, \square  ]$}]  
    \ifshownotes
    By the IH (item 2), there exists $\mathcal{E}'_{{\mathrm{1}}}$ such that
    $ \mathscr{C}\llbracket \maybebluetext{ \mathcal{F}_{{\mathrm{1}}}  [  \ottnt{L}  ] } \rrbracket  = \mathcal{E}'_{{\mathrm{1}}}  [   \mathscr{C}\llbracket \maybebluetext{ \ottnt{L} } \rrbracket   ]$ for any $\ottnt{L}$.
    We have
    \[
      \mathcal{F}  [  \ottnt{M}  \langle  \ottnt{s}  \rangle  ]  \ottsym{=}  \ottsym{(}  \mathcal{F}_{{\mathrm{1}}}  [  \ottnt{V} \, \square  ]  \ottsym{)}  [  \ottnt{M}  \langle  \ottnt{s}  \rangle  ] = \mathcal{F}_{{\mathrm{1}}}  [  \ottnt{V} \, \ottsym{(}  \ottnt{M}  \langle  \ottnt{s}  \rangle  \ottsym{)}  ].
    \]
    and so
    \[
    \begin{array}{llll}
       \mathscr{C}\llbracket \maybebluetext{ \mathcal{F}  [  \ottnt{M}  \langle  \ottnt{s}  \rangle  ] } \rrbracket 
      &=&  \mathscr{C}\llbracket \maybebluetext{ \mathcal{F}_{{\mathrm{1}}}  [  \ottnt{V} \, \ottsym{(}  \ottnt{M}  \langle  \ottnt{s}  \rangle  \ottsym{)}  ] } \rrbracket  \\
      &=& \mathcal{E}'_{{\mathrm{1}}}  [   \mathscr{C}\llbracket \maybebluetext{ \ottnt{V} \, \ottsym{(}  \ottnt{M}  \langle  \ottnt{s}  \rangle  \ottsym{)} } \rrbracket   ] & \text{by IH with $\ottnt{L}  \ottsym{=}  \ottnt{V} \, \ottsym{(}  \ottnt{M}  \langle  \ottnt{s}  \rangle  \ottsym{)}$} \\
      &=& \mathcal{E}'_{{\mathrm{1}}}  [   \mathscr{K}\llbracket \maybebluetext{ \ottnt{V} \, \ottsym{(}  \ottnt{M}  \langle  \ottnt{s}  \rangle  \ottsym{)} } \rrbracket  \ottkw{id}   ]  \\
      &=& \mathcal{E}'_{{\mathrm{1}}}  [    \mathscr{C}\llbracket \maybebluetext{ \ottnt{V} } \rrbracket  \, (  \mathscr{C}\llbracket \maybebluetext{ \ottnt{M}  \langle  \ottnt{s}  \rangle } \rrbracket  , \ottkw{id} )   ] &\text{by \rnp{Tr-App}}\\
      &=& \mathcal{E}'_{{\mathrm{1}}}  [    \Psi (\maybebluetext{ \ottnt{V} })  \, (  \mathscr{K}\llbracket \maybebluetext{ \ottnt{M} } \rrbracket   \Psi (\maybebluetext{ \ottnt{s} })   , \ottkw{id} )   ].
    \end{array}
    \]
    Take $\mathcal{E}'  \ottsym{=}  \mathcal{E}'_{{\mathrm{1}}}  [    \Psi (\maybebluetext{ \ottnt{V} })  \, (  \square , \ottkw{id}  )   ]$; then we have $\mathcal{E}'  [   \mathscr{K}\llbracket \maybebluetext{ \ottnt{M} } \rrbracket   \Psi (\maybebluetext{ \ottnt{s} })    ]  \ottsym{=}   \mathscr{C}\llbracket \maybebluetext{ \mathcal{F}  [  \ottnt{M}  \langle  \ottnt{s}  \rangle  ] } \rrbracket $.
    \else
    Similar.
    \fi

  \item[\Case{$\mathcal{F}  \ottsym{=}  \mathcal{F}_{{\mathrm{1}}}  [  \ottsym{(}  \ottnt{V} \, \square \, \ottsym{)}  \langle  \ottnt{t}  \rangle  ]$}]
    \ifshownotes
    By the IH (item 2), there exists $\mathcal{E}'_{{\mathrm{1}}}$ such that
    $ \mathscr{C}\llbracket \maybebluetext{ \mathcal{F}_{{\mathrm{1}}}  [  \ottnt{L}  \langle  \ottnt{t}  \rangle  ] } \rrbracket  = \mathcal{E}'_{{\mathrm{1}}}  [   \mathscr{K}\llbracket \maybebluetext{ \ottnt{L} } \rrbracket   \Psi (\maybebluetext{ \ottnt{t} })    ]$ for any $\ottnt{L}$.
    We have
    \[
      \mathcal{F}  [  \ottnt{M}  \langle  \ottnt{s}  \rangle  ]  \ottsym{=}  \ottsym{(}  \mathcal{F}_{{\mathrm{1}}}  [  \ottsym{(}  \ottnt{V} \, \square \, \ottsym{)}  \langle  \ottnt{t}  \rangle  ]  \ottsym{)}  [  \ottnt{M}  \langle  \ottnt{s}  \rangle  ] = \mathcal{F}_{{\mathrm{1}}}  [  \ottsym{(}  \ottnt{V} \, \ottsym{(}  \ottnt{M}  \langle  \ottnt{s}  \rangle  \ottsym{)}  \ottsym{)}  \langle  \ottnt{t}  \rangle  ].
    \]
    and so
    \[
    \begin{array}{llll}
       \mathscr{C}\llbracket \maybebluetext{ \mathcal{F}  [  \ottnt{M}  \langle  \ottnt{s}  \rangle  ] } \rrbracket 
      &=&  \mathscr{C}\llbracket \maybebluetext{ \mathcal{F}_{{\mathrm{1}}}  [  \ottsym{(}  \ottnt{V} \, \ottsym{(}  \ottnt{M}  \langle  \ottnt{s}  \rangle  \ottsym{)}  \ottsym{)}  \langle  \ottnt{t}  \rangle  ] } \rrbracket  \\
      &=& \mathcal{E}'_{{\mathrm{1}}}  [   \mathscr{K}\llbracket \maybebluetext{ \ottnt{V} \, \ottsym{(}  \ottnt{M}  \langle  \ottnt{s}  \rangle  \ottsym{)} } \rrbracket   \Psi (\maybebluetext{ \ottnt{t} })    ] & \text{by IH with $\ottnt{L}  \ottsym{=}  \ottnt{V} \, \ottsym{(}  \ottnt{M}  \langle  \ottnt{s}  \rangle  \ottsym{)}$} \\
      &=& \mathcal{E}'_{{\mathrm{1}}}  [    \mathscr{C}\llbracket \maybebluetext{ \ottnt{V} } \rrbracket  \, (  \mathscr{C}\llbracket \maybebluetext{ \ottnt{M}  \langle  \ottnt{s}  \rangle } \rrbracket  ,  \Psi (\maybebluetext{ \ottnt{t} })  )   ] &\text{by \rnp{Tr-App}}\\
      &=& \mathcal{E}'_{{\mathrm{1}}}  [    \Psi (\maybebluetext{ \ottnt{V} })  \, (  \mathscr{K}\llbracket \maybebluetext{ \ottnt{M} } \rrbracket   \Psi (\maybebluetext{ \ottnt{s} })   ,  \Psi (\maybebluetext{ \ottnt{t} })  )   ]
    \end{array}
    \]
    Take $\mathcal{E}'  \ottsym{=}  \mathcal{E}'_{{\mathrm{1}}}  [    \Psi (\maybebluetext{ \ottnt{V} })  \, (  \square ,  \Psi (\maybebluetext{ \ottnt{t} })   )   ]$; then we have $\mathcal{E}'  [   \mathscr{K}\llbracket \maybebluetext{ \ottnt{M} } \rrbracket   \Psi (\maybebluetext{ \ottnt{s} })    ]  \ottsym{=}   \mathscr{C}\llbracket \maybebluetext{ \mathcal{F}  [  \ottnt{M}  \langle  \ottnt{s}  \rangle  ] } \rrbracket $.
    \else
    Similar.
    \fi
    \qedhere
  \end{description}
\end{proof}

\subsubsection{Main Theorem}

As usual, $ \mathbin{  \accentset{\mathsf{e} }{\evalto}_{\mathsf{S_1} }     \accentset{\mathsf{c} }{\evalto}_{\mathsf{S_1} }  ^*} $ denotes the relational composition.

\begin{lemma}[Simulation for Reduction]\label{lem:trans-reduce}\leavevmode
  \begin{enumerate}
  \item If $ \ottnt{M}    \mathbin{\accentset{\mathsf{e} }{\reduces}_{\mathsf{S} } }    \ottnt{N} $ 
    , then
    $ \mathscr{K}\llbracket \maybebluetext{ \ottnt{M} } \rrbracket  \ottnt{K}   \mathbin{  \accentset{\mathsf{e} }{\evalto}_{\mathsf{S_1} }     \accentset{\mathsf{c} }{\evalto}_{\mathsf{S_1} }  ^*}   \mathscr{K}\llbracket \maybebluetext{ \ottnt{N} } \rrbracket  \ottnt{K} $.
  \item If $ \ottnt{M}    \mathbin{\accentset{\mathsf{c} }{\reduces}_{\mathsf{S} } }    \ottnt{N} $, then $  \mathscr{C}\llbracket \maybebluetext{ \ottnt{M} } \rrbracket     \mathbin{  \accentset{\mathsf{c} }{\evalto}_{\mathsf{S_1} }  ^+}     \mathscr{C}\llbracket \maybebluetext{ \ottnt{N} } \rrbracket  $.
  \end{enumerate}
\end{lemma}
\begin{proof}
  (1) By case analysis on the reduction rule applied to $ \ottnt{M}    \mathbin{\accentset{\mathsf{e} }{\reduces}_{\mathsf{S} } }    \ottnt{N} $.
  \begin{description}
  \item[\Case{\rnp{R-Op}}]
    We are given
    \[
      \ottnt{M}  \ottsym{=}  \ottnt{op}  \ottsym{(}  \ottnt{a_{{\mathrm{1}}}}  \ottsym{,}  \ottnt{a_{{\mathrm{2}}}}  \ottsym{)} \hgap
      \ottnt{N}  \ottsym{=}  \ottnt{a} \hgap
      \delta \, \ottsym{(}  \ottnt{op}  \ottsym{,}  \ottnt{a_{{\mathrm{1}}}}  \ottsym{,}  \ottnt{a_{{\mathrm{2}}}}  \ottsym{)}  \ottsym{=}  \ottnt{a}
    \]
    for some $\ottnt{op},\ottnt{a_{{\mathrm{1}}}},\ottnt{a_{{\mathrm{2}}}},\ottnt{a}$.
    We assume $ \metafun{ty} ( \ottnt{a_{{\mathrm{1}}}} )   \ottsym{=}  \iota_{{\mathrm{1}}}$ and $ \metafun{ty} ( \ottnt{a_{{\mathrm{2}}}} )   \ottsym{=}  \iota_{{\mathrm{2}}}$.
    \[
    \begin{array}{llll}
       \mathscr{K}\llbracket \maybebluetext{ \ottnt{op}  \ottsym{(}  \ottnt{a_{{\mathrm{1}}}}  \ottsym{,}  \ottnt{a_{{\mathrm{2}}}}  \ottsym{)} } \rrbracket  \ottnt{K} 
      &=& \ottnt{op}  \ottsym{(}   \mathscr{C}\llbracket \maybebluetext{ \ottnt{a_{{\mathrm{1}}}} } \rrbracket   \ottsym{,}   \mathscr{C}\llbracket \maybebluetext{ \ottnt{a_{{\mathrm{2}}}} } \rrbracket   \ottsym{)}  \langle  \ottnt{K}  \rangle
      &\text{by \rnp{Tr-Op}}\\
      &=& \ottnt{op}  \ottsym{(}  \ottnt{a_{{\mathrm{1}}}}  \ottsym{,}  \ottnt{a_{{\mathrm{2}}}}  \ottsym{)}  \langle  \ottnt{K}  \rangle\\
      & \accentset{\mathsf{e} }{\evalto} & \delta \, \ottsym{(}  \ottnt{op}  \ottsym{,}  \ottnt{a_{{\mathrm{1}}}}  \ottsym{,}  \ottnt{a_{{\mathrm{2}}}}  \ottsym{)}  \langle  \ottnt{K}  \rangle &\text{by \rnp{R-Op}} \\
      &=& \ottnt{a}  \langle  \ottnt{K}  \rangle \\
      &=&  \mathscr{K}\llbracket \maybebluetext{ \ottnt{a} } \rrbracket  \ottnt{K} .
    \end{array}
    \]

  \item[\Case{\rnp{R-Beta}}]
    We are given
    \[
      \ottnt{M}  \ottsym{=}  \ottsym{(}   \lambda   \ottmv{x} .\,  \ottnt{M_{{\mathrm{1}}}}   \ottsym{)} \, \ottnt{V} \hgap
      \ottnt{N}  \ottsym{=}  \ottnt{M_{{\mathrm{1}}}}  [  \ottmv{x}  \ottsym{:=}  \ottnt{V}  ]
    \]
    for some $\ottmv{x},\ottnt{M_{{\mathrm{1}}}},\ottnt{V}$.
    Here,
    \[
    \begin{array}{llll}
       \mathscr{K}\llbracket \maybebluetext{ \ottsym{(}   \lambda   \ottmv{x} .\,  \ottnt{M_{{\mathrm{1}}}}   \ottsym{)} \, \ottnt{V} } \rrbracket  \ottnt{K} 
      &=&   \mathscr{C}\llbracket \maybebluetext{  \lambda   \ottmv{x} .\,  \ottnt{M_{{\mathrm{1}}}}  } \rrbracket  \, (  \mathscr{C}\llbracket \maybebluetext{ \ottnt{V} } \rrbracket  , \ottnt{K} ) 
      &\text{by \rnp{Tr-App}}\\
      &=&   \Psi (\maybebluetext{  \lambda   \ottmv{x} .\,  \ottnt{M_{{\mathrm{1}}}}  })  \, (  \Psi (\maybebluetext{ \ottnt{V} })  , \ottnt{K} ) 
      &\text{as $ \lambda   \ottmv{x} .\,  \ottnt{M_{{\mathrm{1}}}} $ is a value}\\
      &=&  \ottsym{(}   \lambda  ( \ottmv{x} , \kappa ).\,  \ottsym{(}   \mathscr{K}\llbracket \maybebluetext{ \ottnt{M_{{\mathrm{1}}}} } \rrbracket  \kappa   \ottsym{)}   \ottsym{)} \, (  \Psi (\maybebluetext{ \ottnt{V} })  , \ottnt{K} )  \\
      & \accentset{\mathsf{e} }{\evalto} & \ottsym{(}   \mathscr{K}\llbracket \maybebluetext{ \ottnt{M_{{\mathrm{1}}}} } \rrbracket  \kappa   \ottsym{)}  [  \ottmv{x}  \ottsym{:=}   \Psi (\maybebluetext{ \ottnt{V} })   \ottsym{,}  \kappa  \ottsym{:=}  \ottnt{K}  ] &\text{by \rnp{R-Beta}}\\
      &=&  \mathscr{K}\llbracket \maybebluetext{ \ottnt{M_{{\mathrm{1}}}}  [  \ottmv{x}  \ottsym{:=}  \ottnt{V}  ] } \rrbracket  \ottnt{K} 
      &\text{by Lemma~\ref{lem:subst-trans}.}
    \end{array}
    \]

  \item[\Case{\rnp{R-Wrap}}]
    We are given
    \[
      \ottnt{M}  \ottsym{=}  \ottsym{(}  \ottnt{U}  \langle\!\langle  \ottnt{s}  \rightarrow  \ottnt{t}  \rangle\!\rangle  \ottsym{)} \, \ottnt{V} \hgap
      \ottnt{N}  \ottsym{=}  \ottsym{(}  \ottnt{U} \, \ottsym{(}  \ottnt{V}  \langle  \ottnt{s}  \rangle  \ottsym{)}  \ottsym{)}  \langle  \ottnt{t}  \rangle
    \]
    for some $\ottnt{U},\ottnt{s},\ottnt{t},\ottnt{V}$.
    \[
    \begin{array}{llll}
      &&  \mathscr{K}\llbracket \maybebluetext{ \ottsym{(}  \ottnt{U}  \langle\!\langle  \ottnt{s}  \rightarrow  \ottnt{t}  \rangle\!\rangle  \ottsym{)} \, \ottnt{V} } \rrbracket  \ottnt{K}  \\
      &=&   \mathscr{C}\llbracket \maybebluetext{ \ottnt{U}  \langle\!\langle  \ottnt{s}  \rightarrow  \ottnt{t}  \rangle\!\rangle } \rrbracket  \, (  \mathscr{C}\llbracket \maybebluetext{ \ottnt{V} } \rrbracket  , \ottnt{K} ) 
      &\text{by \rnp{Tr-App}}\\
      &=&   \Psi (\maybebluetext{ \ottnt{U}  \langle\!\langle  \ottnt{s}  \rightarrow  \ottnt{t}  \rangle\!\rangle })  \, (  \Psi (\maybebluetext{ \ottnt{V} })  , \ottnt{K} ) 
      &\text{as $\ottnt{U}  \langle\!\langle  \ottnt{s}  \rightarrow  \ottnt{t}  \rangle\!\rangle$ is a value}\\
      &=&   \Psi (\maybebluetext{ \ottnt{U} })   \langle\!\langle   \Psi (\maybebluetext{ \ottnt{s} })   \Rightarrow   \Psi (\maybebluetext{ \ottnt{t} })   \rangle\!\rangle \, (  \Psi (\maybebluetext{ \ottnt{V} })  , \ottnt{K} )  \\
      & \accentset{\mathsf{e} }{\evalto} &   \ottkw{let} \,  \kappa =  \Psi (\maybebluetext{ \ottnt{t} })   \mathbin{;\!;}  \ottnt{K} \, \ottkw{in}\,   \Psi (\maybebluetext{ \ottnt{U} })   \, (  \Psi (\maybebluetext{ \ottnt{V} })   \langle   \Psi (\maybebluetext{ \ottnt{s} })   \rangle , \kappa ) 
      &\text{by \rnp{R-Wrap}}
    \end{array}
    \]
    Then,
    \[
    \begin{array}{llll}
      &&  \mathscr{K}\llbracket \maybebluetext{ \ottsym{(}  \ottnt{U} \, \ottsym{(}  \ottnt{V}  \langle  \ottnt{s}  \rangle  \ottsym{)}  \ottsym{)}  \langle  \ottnt{t}  \rangle } \rrbracket  \ottnt{K}  \\
      &=&  \ottkw{let} \,  \kappa =  \Psi (\maybebluetext{ \ottnt{t} })   \mathbin{;\!;}  \ottnt{K} \, \ottkw{in}\,  \ottsym{(}   \mathscr{K}\llbracket \maybebluetext{ \ottnt{U} \, \ottsym{(}  \ottnt{V}  \langle  \ottnt{s}  \rangle  \ottsym{)} } \rrbracket  \kappa   \ottsym{)} 
      &\text{by \rnp{Tr-Crc}}\\
      &=&   \ottkw{let} \,  \kappa =  \Psi (\maybebluetext{ \ottnt{t} })   \mathbin{;\!;}  \ottnt{K} \, \ottkw{in}\,   \mathscr{C}\llbracket \maybebluetext{ \ottnt{U} } \rrbracket   \, (  \mathscr{C}\llbracket \maybebluetext{ \ottnt{V}  \langle  \ottnt{s}  \rangle } \rrbracket  , \kappa ) 
      &\text{by \rnp{Tr-App}}\\
      &=&   \ottkw{let} \,  \kappa =  \Psi (\maybebluetext{ \ottnt{t} })   \mathbin{;\!;}  \ottnt{K} \, \ottkw{in}\,   \Psi (\maybebluetext{ \ottnt{U} })   \, (  \Psi (\maybebluetext{ \ottnt{V} })   \langle   \Psi (\maybebluetext{ \ottnt{s} })   \rangle , \kappa ) 
      &\text{as $ \mathscr{C}\llbracket \maybebluetext{ \ottnt{V}  \langle  \ottnt{s}  \rangle } \rrbracket  =  \mathscr{K}\llbracket \maybebluetext{ \ottnt{V} } \rrbracket   \Psi (\maybebluetext{ \ottnt{s} })   =  \Psi (\maybebluetext{ \ottnt{V} })   \langle   \Psi (\maybebluetext{ \ottnt{s} })   \rangle$.}
    \end{array}
    \]
    Thus,
    $ \mathscr{K}\llbracket \maybebluetext{ \ottsym{(}  \ottnt{U}  \langle\!\langle  \ottnt{s}  \rightarrow  \ottnt{t}  \rangle\!\rangle  \ottsym{)} \, \ottnt{V} } \rrbracket  \ottnt{K}   \accentset{\mathsf{e} }{\evalto}   \mathscr{K}\llbracket \maybebluetext{ \ottsym{(}  \ottnt{U} \, \ottsym{(}  \ottnt{V}  \langle  \ottnt{s}  \rangle  \ottsym{)}  \ottsym{)}  \langle  \ottnt{t}  \rangle } \rrbracket  \ottnt{K} $.
  \end{description}

  (2) By case analysis on the reduction rule applied to $ \ottnt{M}    \mathbin{\accentset{\mathsf{c} }{\reduces}_{\mathsf{S} } }    \ottnt{N} $.
  \begin{description}
  \item[\Case{\rnp{R-Id}}]
    We are given
    \[
      \ottnt{M}  \ottsym{=}  \ottnt{U}  \langle  \ottkw{id}  \rangle \hgap
      \ottnt{N}  \ottsym{=}  \ottnt{U}
    \]
    for some $\ottnt{U}$.
    Here,
    \[
    \begin{array}{llll}
       \mathscr{C}\llbracket \maybebluetext{ \ottnt{U}  \langle  \ottkw{id}  \rangle } \rrbracket 
      &=&  \mathscr{K}\llbracket \maybebluetext{ \ottnt{U} } \rrbracket   \Psi (\maybebluetext{ \ottkw{id} })   \\
      &=&  \mathscr{K}\llbracket \maybebluetext{ \ottnt{U} } \rrbracket  \ottkw{id}  \\
      &=&  \Psi (\maybebluetext{ \ottnt{U} })   \langle  \ottkw{id}  \rangle \\
      & \accentset{\mathsf{c} }{\evalto} &  \Psi (\maybebluetext{ \ottnt{U} })  &\text{by \rnp{R-Id}} \\
      &=&  \mathscr{C}\llbracket \maybebluetext{ \ottnt{U} } \rrbracket .
    \end{array}
    \]

  \item[\Case{\rnp{R-Fail}}]
    We are given
    \[
      \ottnt{M}  \ottsym{=}  \ottnt{U}  \langle   \bot^{ \ottnt{G}   \ottnt{p}   \ottnt{H} }   \rangle \hgap
      \ottnt{N}  \ottsym{=}  \ottkw{blame} \, \ottnt{p}
    \]
    for some $\ottnt{U},\ottnt{p},\ottnt{G},\ottnt{H}$.
    Here,
    \[
    \begin{array}{llll}
       \mathscr{C}\llbracket \maybebluetext{ \ottnt{U}  \langle   \bot^{ \ottnt{G}   \ottnt{p}   \ottnt{H} }   \rangle } \rrbracket 
      &=&  \mathscr{K}\llbracket \maybebluetext{ \ottnt{U} } \rrbracket   \Psi (\maybebluetext{  \bot^{ \ottnt{G}   \ottnt{p}   \ottnt{H} }  })   \\
      &=&  \mathscr{K}\llbracket \maybebluetext{ \ottnt{U} } \rrbracket   \bot^{ \ottnt{G}   \ottnt{p}   \ottnt{H} }   \\
      &=&  \Psi (\maybebluetext{ \ottnt{U} })   \langle   \bot^{ \ottnt{G}   \ottnt{p}   \ottnt{H} }   \rangle\\
      & \accentset{\mathsf{c} }{\evalto} & \ottkw{blame} \, \ottnt{p} &\text{by \rnp{R-Fail}}\\
       \mathscr{C}\llbracket \maybebluetext{ \ottkw{blame} \, \ottnt{p} } \rrbracket 
      &=&  \mathscr{K}\llbracket \maybebluetext{ \ottkw{blame} \, \ottnt{p} } \rrbracket  \ottkw{id}  \\
      &=& \ottkw{blame} \, \ottnt{p}.
    \end{array}
    \]
    Thus, $  \mathscr{C}\llbracket \maybebluetext{ \ottnt{U}  \langle   \bot^{ \ottnt{G}   \ottnt{p}   \ottnt{H} }   \rangle } \rrbracket     \mathbin{  \accentset{\mathsf{c} }{\evalto}  }     \mathscr{C}\llbracket \maybebluetext{ \ottkw{blame} \, \ottnt{p} } \rrbracket  $.

  \item[\Case{\rnp{R-Crc}}]
    We are given
    \[
      \ottnt{M}  \ottsym{=}  \ottnt{U}  \langle  \ottnt{d}  \rangle \hgap
      \ottnt{N}  \ottsym{=}  \ottnt{U}  \langle\!\langle  \ottnt{d}  \rangle\!\rangle
    \]
    for some $\ottnt{d}$.
    Note that $ \Psi (\maybebluetext{ \ottnt{d} }) $ is also a delayed coercion.
    \[
    \begin{array}{llll}
       \mathscr{C}\llbracket \maybebluetext{ \ottnt{U}  \langle  \ottnt{d}  \rangle } \rrbracket 
      &=&  \mathscr{K}\llbracket \maybebluetext{ \ottnt{U} } \rrbracket   \Psi (\maybebluetext{ \ottnt{d} })   \\
      &=&  \Psi (\maybebluetext{ \ottnt{U} })   \langle   \Psi (\maybebluetext{ \ottnt{d} })   \rangle \\
      & \accentset{\mathsf{c} }{\evalto} &  \Psi (\maybebluetext{ \ottnt{U} })   \langle\!\langle   \Psi (\maybebluetext{ \ottnt{d} })   \rangle\!\rangle &\text{by \rnp{R-Crc}}\\
       \mathscr{C}\llbracket \maybebluetext{ \ottnt{U}  \langle\!\langle  \ottnt{d}  \rangle\!\rangle } \rrbracket 
      &=&  \Psi (\maybebluetext{ \ottnt{U}  \langle\!\langle  \ottnt{d}  \rangle\!\rangle })  \\
      &=&  \Psi (\maybebluetext{ \ottnt{U} })   \langle\!\langle   \Psi (\maybebluetext{ \ottnt{d} })   \rangle\!\rangle.
    \end{array}
    \]
    Thus, $  \mathscr{C}\llbracket \maybebluetext{ \ottnt{U}  \langle  \ottnt{d}  \rangle } \rrbracket     \mathbin{  \accentset{\mathsf{c} }{\evalto}  }     \mathscr{C}\llbracket \maybebluetext{ \ottnt{U}  \langle\!\langle  \ottnt{d}  \rangle\!\rangle } \rrbracket  $.

  \item[\Case{\rnp{R-MergeC}}]
    We are given
    \[
      \ottnt{M}  \ottsym{=}  \ottnt{M_{{\mathrm{1}}}}  \langle  \ottnt{s}  \rangle  \langle  \ottnt{t}  \rangle \hgap
      \ottnt{N}  \ottsym{=}  \ottnt{M_{{\mathrm{1}}}}  \langle  \ottnt{s'}  \rangle \hgap
      \ottnt{s}  \fatsemi  \ottnt{t}  \ottsym{=}  \ottnt{s'}
    \]
    for some $\ottnt{M_{{\mathrm{1}}}},\ottnt{s},\ottnt{t},\ottnt{s'}$.
    Here, $ \Psi (\maybebluetext{ \ottnt{s} })   \fatsemi   \Psi (\maybebluetext{ \ottnt{t} })  =  \Psi (\maybebluetext{ \ottnt{s'} }) $ by Lemma~\ref{lem:cmp-trans}.
    \[
    \begin{array}{llll}
       \mathscr{C}\llbracket \maybebluetext{ \ottnt{M_{{\mathrm{1}}}}  \langle  \ottnt{s}  \rangle  \langle  \ottnt{t}  \rangle } \rrbracket 
      &=&  \mathscr{K}\llbracket \maybebluetext{ \ottnt{M_{{\mathrm{1}}}}  \langle  \ottnt{s}  \rangle } \rrbracket   \Psi (\maybebluetext{ \ottnt{t} })   \\
      &=&  \ottkw{let} \,  \kappa =  \Psi (\maybebluetext{ \ottnt{s} })   \mathbin{;\!;}   \Psi (\maybebluetext{ \ottnt{t} })  \, \ottkw{in}\,  \ottsym{(}   \mathscr{K}\llbracket \maybebluetext{ \ottnt{M_{{\mathrm{1}}}} } \rrbracket  \kappa   \ottsym{)} 
      &\text{by \rnp{Tr-Crc}}\\
      & \accentset{\mathsf{c} }{\evalto} &  \ottkw{let} \,  \kappa =  \Psi (\maybebluetext{ \ottnt{s} })   \fatsemi   \Psi (\maybebluetext{ \ottnt{t} })  \, \ottkw{in}\,  \ottsym{(}   \mathscr{K}\llbracket \maybebluetext{ \ottnt{M_{{\mathrm{1}}}} } \rrbracket  \kappa   \ottsym{)} 
      &\text{by \rnp{R-Cmp}}\\
      &=&  \ottkw{let} \,  \kappa =  \Psi (\maybebluetext{ \ottnt{s'} })  \, \ottkw{in}\,  \ottsym{(}   \mathscr{K}\llbracket \maybebluetext{ \ottnt{M_{{\mathrm{1}}}} } \rrbracket  \kappa   \ottsym{)} 
      &\text{by Lemma~\ref{lem:cmp-trans}}\\
      & \accentset{\mathsf{c} }{\evalto} & \ottsym{(}   \mathscr{K}\llbracket \maybebluetext{ \ottnt{M_{{\mathrm{1}}}} } \rrbracket  \kappa   \ottsym{)}  [  \kappa  \ottsym{:=}   \Psi (\maybebluetext{ \ottnt{s'} })   ] &\text{by \rnp{R-Let}}\\
      &=&  \mathscr{K}\llbracket \maybebluetext{ \ottnt{M_{{\mathrm{1}}}} } \rrbracket   \Psi (\maybebluetext{ \ottnt{s'} })  
      &\text{by Lemma~\ref{lem:subst-kap}} \\
       \mathscr{C}\llbracket \maybebluetext{ \ottnt{M_{{\mathrm{1}}}}  \langle  \ottnt{s'}  \rangle } \rrbracket 
      &=&  \mathscr{K}\llbracket \maybebluetext{ \ottnt{M_{{\mathrm{1}}}} } \rrbracket   \Psi (\maybebluetext{ \ottnt{s'} })  
    \end{array}
    \]
    Thus, $  \mathscr{C}\llbracket \maybebluetext{ \ottnt{M_{{\mathrm{1}}}}  \langle  \ottnt{s}  \rangle  \langle  \ottnt{t}  \rangle } \rrbracket     \mathbin{  \accentset{\mathsf{c} }{\evalto}  ^+}     \mathscr{C}\llbracket \maybebluetext{ \ottnt{M_{{\mathrm{1}}}}  \langle  \ottnt{s'}  \rangle } \rrbracket  $.

  \item[\Case{\rnp{R-MergeV}}]
    We are given
    \[
      \ottnt{M}  \ottsym{=}  \ottnt{U}  \langle\!\langle  \ottnt{d}  \rangle\!\rangle  \langle  \ottnt{t}  \rangle \hgap
      \ottnt{N}  \ottsym{=}  \ottnt{U}  \langle  \ottnt{s'}  \rangle \hgap
      \ottnt{d}  \fatsemi  \ottnt{t}  \ottsym{=}  \ottnt{s'}
    \]
    for some $\ottnt{U},\ottnt{d},\ottnt{t},\ottnt{s'}$.
    Here, $ \Psi (\maybebluetext{ \ottnt{d} })   \fatsemi   \Psi (\maybebluetext{ \ottnt{t} })   \ottsym{=}   \Psi (\maybebluetext{ \ottnt{s'} }) $ by Lemma~\ref{lem:cmp-trans}.
    \[
    \begin{array}{llll}
       \mathscr{C}\llbracket \maybebluetext{ \ottnt{U}  \langle\!\langle  \ottnt{d}  \rangle\!\rangle  \langle  \ottnt{t}  \rangle } \rrbracket 
      &=&  \mathscr{K}\llbracket \maybebluetext{ \ottnt{U}  \langle\!\langle  \ottnt{d}  \rangle\!\rangle } \rrbracket   \Psi (\maybebluetext{ \ottnt{t} })   \\
      &=&  \Psi (\maybebluetext{ \ottnt{U}  \langle\!\langle  \ottnt{d}  \rangle\!\rangle })   \langle   \Psi (\maybebluetext{ \ottnt{t} })   \rangle
      &\text{as $\ottnt{U}  \langle\!\langle  \ottnt{d}  \rangle\!\rangle$ is a value}\\
      &=&  \Psi (\maybebluetext{ \ottnt{U} })   \langle\!\langle   \Psi (\maybebluetext{ \ottnt{d} })   \rangle\!\rangle  \langle   \Psi (\maybebluetext{ \ottnt{t} })   \rangle \\
      & \accentset{\mathsf{c} }{\evalto} &  \Psi (\maybebluetext{ \ottnt{U} })   \langle   \Psi (\maybebluetext{ \ottnt{d} })   \mathbin{;\!;}   \Psi (\maybebluetext{ \ottnt{t} })   \rangle &\text{by \rnp{R-MergeV}}\\
      & \accentset{\mathsf{c} }{\evalto} &  \Psi (\maybebluetext{ \ottnt{U} })   \langle   \Psi (\maybebluetext{ \ottnt{d} })   \fatsemi   \Psi (\maybebluetext{ \ottnt{t} })   \rangle &\text{by \rnp{R-Cmp}}\\
      &=&  \Psi (\maybebluetext{ \ottnt{U} })   \langle   \Psi (\maybebluetext{ \ottnt{s'} })   \rangle
      &\text{by Lemma~\ref{lem:cmp-trans}}. \\
       \mathscr{C}\llbracket \maybebluetext{ \ottnt{U}  \langle  \ottnt{s'}  \rangle } \rrbracket 
      &=&  \mathscr{K}\llbracket \maybebluetext{ \ottnt{U} } \rrbracket   \Psi (\maybebluetext{ \ottnt{s'} })   \\
      &=&  \Psi (\maybebluetext{ \ottnt{U} })   \langle   \Psi (\maybebluetext{ \ottnt{s'} })   \rangle.
    \end{array}
    \]
    Thus, $  \mathscr{C}\llbracket \maybebluetext{ \ottnt{U}  \langle\!\langle  \ottnt{d}  \rangle\!\rangle  \langle  \ottnt{t}  \rangle } \rrbracket     \mathbin{  \accentset{\mathsf{c} }{\evalto}  ^+}     \mathscr{C}\llbracket \maybebluetext{ \ottnt{U}  \langle  \ottnt{s'}  \rangle } \rrbracket  $.\qedhere
  \end{description}
\end{proof}

\begin{lemma}\label{lem:trans-id-CV}
  $  \mathscr{K}\llbracket \maybebluetext{ \ottnt{M} } \rrbracket  \ottkw{id}     \mathbin{  \accentset{\mathsf{c} }{\evalto}_{\mathsf{S_1} }  ^*}     \mathscr{C}\llbracket \maybebluetext{ \ottnt{M} } \rrbracket  $
\end{lemma}
\begin{proof}
  By case analysis on the form of $\ottnt{M}$.
  \begin{description}
  \item[\Case{$\ottnt{M}  \ottsym{=}  \ottnt{V}$}]
    \[
    \begin{array}{llll}
       \mathscr{K}\llbracket \maybebluetext{ \ottnt{V} } \rrbracket  \ottkw{id} 
      &=&  \Psi (\maybebluetext{ \ottnt{V} })   \langle  \ottkw{id}  \rangle \\
      & \mathbin{  \accentset{\mathsf{c} }{\evalto}  ^*} &  \Psi (\maybebluetext{ \ottnt{V} })  &\text{by Lemma~\ref{lem:val-coe-id}} \\
      &=&  \mathscr{C}\llbracket \maybebluetext{ \ottnt{V} } \rrbracket .
    \end{array}
    \]
  \item[\Case{$\ottnt{M}  \ottsym{=}  \ottnt{N}  \langle  \ottnt{s}  \rangle$}]
    \[
    \begin{array}{llll}
       \mathscr{K}\llbracket \maybebluetext{ \ottnt{N}  \langle  \ottnt{s}  \rangle } \rrbracket  \ottkw{id} 
      &=&  \ottkw{let} \,  \kappa =  \Psi (\maybebluetext{ \ottnt{s} })   \mathbin{;\!;}  \ottkw{id} \, \ottkw{in}\,  \ottsym{(}   \mathscr{K}\llbracket \maybebluetext{ \ottnt{N} } \rrbracket  \kappa   \ottsym{)} 
      &\text{by \rnp{Tr-Crc}}\\
      & \accentset{\mathsf{c} }{\evalto} &  \ottkw{let} \,  \kappa =  \Psi (\maybebluetext{ \ottnt{s} })   \fatsemi  \ottkw{id} \, \ottkw{in}\,  \ottsym{(}   \mathscr{K}\llbracket \maybebluetext{ \ottnt{N} } \rrbracket  \kappa   \ottsym{)} 
      &\text{by \rnp{R-Cmp}}\\
      &=&  \ottkw{let} \,  \kappa =  \Psi (\maybebluetext{ \ottnt{s} })  \, \ottkw{in}\,  \ottsym{(}   \mathscr{K}\llbracket \maybebluetext{ \ottnt{N} } \rrbracket  \kappa   \ottsym{)} 
      &\text{by Lemma~\ref{lem:merge-id}}\\
      & \accentset{\mathsf{c} }{\evalto} & \ottsym{(}   \mathscr{K}\llbracket \maybebluetext{ \ottnt{N} } \rrbracket  \kappa   \ottsym{)}  [  \kappa  \ottsym{:=}   \Psi (\maybebluetext{ \ottnt{s} })   ]
      &\text{by \rnp{R-Let}}\\
      &=&  \mathscr{K}\llbracket \maybebluetext{ \ottnt{N} } \rrbracket   \Psi (\maybebluetext{ \ottnt{s} })   &\text{by Lemma~\ref{lem:subst-kap}}\\
      &=&  \mathscr{C}\llbracket \maybebluetext{ \ottnt{N}  \langle  \ottnt{s}  \rangle } \rrbracket .
    \end{array}
    \]
  \item[\Otherwise] Since $\ottnt{M}$ is neither a value nor a coercion application,
    $ \mathscr{K}\llbracket \maybebluetext{ \ottnt{M} } \rrbracket  \ottkw{id}   \ottsym{=}   \mathscr{C}\llbracket \maybebluetext{ \ottnt{M} } \rrbracket $. \qedhere
  \end{description}
\end{proof}

\iffull\lemTransEval*
\else
\begin{lemma}[Simulation]\label{lem:trans-eval}\leavevmode
  \begin{enumerate}
  \item If $ \ottnt{M}    \mathbin{  \accentset{\mathsf{e} }{\evalto}_{\mathsf{S} }  }    \ottnt{N} $, then
    $ \mathscr{C}\llbracket \maybebluetext{ \ottnt{M} } \rrbracket   \mathbin{  \accentset{\mathsf{e} }{\evalto}_{\mathsf{S_1} }     \accentset{\mathsf{c} }{\evalto}_{\mathsf{S_1} }  ^*}   \mathscr{C}\llbracket \maybebluetext{ \ottnt{N} } \rrbracket $.
  \item If $ \ottnt{M}    \mathbin{  \accentset{\mathsf{c} }{\evalto}_{\mathsf{S} }  }    \ottnt{N} $, then $  \mathscr{C}\llbracket \maybebluetext{ \ottnt{M} } \rrbracket     \mathbin{  \accentset{\mathsf{c} }{\evalto}_{\mathsf{S_1} }  ^+}     \mathscr{C}\llbracket \maybebluetext{ \ottnt{N} } \rrbracket  $.
  \end{enumerate}
\[
\xymatrix@=50pt{
    \ottnt{M} \ar@{|->}[rr]^{\mathsf{e}}_>>{\mathsf{S}} \ar[d]^{\mathscr{C} \llbracket \_  \rrbracket}& & \ottnt{N} \ar[d]^{\mathscr{C} \llbracket \_  \rrbracket} \\
     \mathscr{C}\llbracket \maybebluetext{ \ottnt{M} } \rrbracket  \ar@{|.>}[r]^-{\mathsf{e}} _>>{\mathsf{S}_1} & \ar@{|.>}[r]^-{\mathsf{c}} ^>>{*}_>>{\mathsf{S}_1} &  \mathscr{C}\llbracket \maybebluetext{ \ottnt{N} } \rrbracket 
}
\qquad
\xymatrix@=50pt{
    \ottnt{M} \ar@{|->}[r]^{\mathsf{c}}_>>{\mathsf{S}} \ar[d]^{\mathscr{C} \llbracket \_  \rrbracket} & \ottnt{N} \ar[d]^{\mathscr{C} \llbracket \_  \rrbracket} \\
     \mathscr{C}\llbracket \maybebluetext{ \ottnt{M} } \rrbracket  \ar@{|.>}[r]^{\mathsf{c}} ^>>{+}_>>{\mathsf{S}_1} &   \mathscr{C}\llbracket \maybebluetext{ \ottnt{N} } \rrbracket 
}
\]
\end{lemma}
\fi

\begin{proof}
  (1) By case analysis on the evaluation rule applied to $ \ottnt{M}    \mathbin{  \accentset{\mathsf{e} }{\evalto}_{\mathsf{S} }  }    \ottnt{N} $.
  \begin{description}
  \item[\Case{\rnp{E-CtxE} with $\mathcal{E}  \ottsym{=}  \mathcal{F}$}] 
    We are given
    \[
       \ottnt{M_{{\mathrm{1}}}}    \mathbin{\accentset{\mathsf{e} }{\reduces}_{\mathsf{S} } }    \ottnt{N_{{\mathrm{1}}}}  \hgap
      \ottnt{M}  \ottsym{=}  \mathcal{F}  [  \ottnt{M_{{\mathrm{1}}}}  ] \hgap
      \ottnt{N}  \ottsym{=}  \mathcal{F}  [  \ottnt{N_{{\mathrm{1}}}}  ]
    \]
    for some $\ottnt{M_{{\mathrm{1}}}},\ottnt{N_{{\mathrm{1}}}}$.
    By Lemma~\ref{lem:trans-ctx}~(1), there exists $\mathcal{E}'$ such that
    $ \mathscr{C}\llbracket \maybebluetext{ \mathcal{F}  [  \ottnt{L}  ] } \rrbracket   \ottsym{=}  \mathcal{E}'  [   \mathscr{C}\llbracket \maybebluetext{ \ottnt{L} } \rrbracket   ]$ for any $\ottnt{L}$. So,
    \begin{equation}
       \mathscr{C}\llbracket \maybebluetext{ \mathcal{F}  [  \ottnt{M_{{\mathrm{1}}}}  ] } \rrbracket   \ottsym{=}  \mathcal{E}'  [   \mathscr{C}\llbracket \maybebluetext{ \ottnt{M_{{\mathrm{1}}}} } \rrbracket   ] \hgap
       \mathscr{C}\llbracket \maybebluetext{ \mathcal{F}  [  \ottnt{N_{{\mathrm{1}}}}  ] } \rrbracket   \ottsym{=}  \mathcal{E}'  [   \mathscr{C}\llbracket \maybebluetext{ \ottnt{N_{{\mathrm{1}}}} } \rrbracket   ]. \label{eq:sim-trans-ctx-e1}
    \end{equation}
    Since $\ottnt{M_{{\mathrm{1}}}}$ is e-reducible, $\ottnt{M_{{\mathrm{1}}}}$ is neither a value nor a coercion application.
    (Note that a coercion application may be a c-redex, but not an e-redex.)
    So, we have $ \mathscr{C}\llbracket \maybebluetext{ \ottnt{M_{{\mathrm{1}}}} } \rrbracket  =  \mathscr{K}\llbracket \maybebluetext{ \ottnt{M_{{\mathrm{1}}}} } \rrbracket  \ottkw{id} $.  
    Next, by $ \ottnt{M_{{\mathrm{1}}}}    \mathbin{\accentset{\mathsf{e} }{\reduces}_{\mathsf{S} } }    \ottnt{N_{{\mathrm{1}}}} $ and Lemma~\ref{lem:trans-reduce}~(1) with $\ottnt{K}  \ottsym{=}  \ottkw{id}$,
    \[
       \mathscr{K}\llbracket \maybebluetext{ \ottnt{M_{{\mathrm{1}}}} } \rrbracket  \ottkw{id}   \mathbin{  \accentset{\mathsf{e} }{\evalto}_{\mathsf{S_1} }     \accentset{\mathsf{c} }{\evalto}_{\mathsf{S_1} }  ^*}   \mathscr{K}\llbracket \maybebluetext{ \ottnt{N_{{\mathrm{1}}}} } \rrbracket  \ottkw{id} . 
    \]
    Then, $  \mathscr{K}\llbracket \maybebluetext{ \ottnt{N_{{\mathrm{1}}}} } \rrbracket  \ottkw{id}     \mathbin{  \accentset{\mathsf{c} }{\evalto}_{\mathsf{S_1} }  ^*}     \mathscr{C}\llbracket \maybebluetext{ \ottnt{N_{{\mathrm{1}}}} } \rrbracket  $ by Lemma~\ref{lem:trans-eval}. 
    Therefore, 
    \[
       \mathscr{C}\llbracket \maybebluetext{ \ottnt{M_{{\mathrm{1}}}} } \rrbracket  =  \mathscr{K}\llbracket \maybebluetext{ \ottnt{M_{{\mathrm{1}}}} } \rrbracket  \ottkw{id}   \mathbin{  \accentset{\mathsf{e} }{\evalto}_{\mathsf{S_1} }     \accentset{\mathsf{c} }{\evalto}_{\mathsf{S_1} }  ^*}    \mathscr{K}\llbracket \maybebluetext{ \ottnt{N_{{\mathrm{1}}}} } \rrbracket  \ottkw{id}     \mathbin{  \accentset{\mathsf{c} }{\evalto}_{\mathsf{S_1} }  ^*}     \mathscr{C}\llbracket \maybebluetext{ \ottnt{N_{{\mathrm{1}}}} } \rrbracket  .
    \]
    By $ \mathscr{C}\llbracket \maybebluetext{ \ottnt{M_{{\mathrm{1}}}} } \rrbracket   \mathbin{  \accentset{\mathsf{e} }{\evalto}_{\mathsf{S_1} }     \accentset{\mathsf{c} }{\evalto}_{\mathsf{S_1} }  ^*}   \mathscr{C}\llbracket \maybebluetext{ \ottnt{N_{{\mathrm{1}}}} } \rrbracket $ and Lemma~\ref{lem:ctx-multi},
    \[
      \mathcal{E}'  [   \mathscr{C}\llbracket \maybebluetext{ \ottnt{M_{{\mathrm{1}}}} } \rrbracket   ]  \mathbin{  \accentset{\mathsf{e} }{\evalto}_{\mathsf{S_1} }     \accentset{\mathsf{c} }{\evalto}_{\mathsf{S_1} }  ^*}  \mathcal{E}'  [   \mathscr{C}\llbracket \maybebluetext{ \ottnt{N_{{\mathrm{1}}}} } \rrbracket   ].
    \]
    Thus, by eq.~\eqref{eq:sim-trans-ctx-e1},
    $ \mathscr{C}\llbracket \maybebluetext{ \mathcal{F}  [  \ottnt{M_{{\mathrm{1}}}}  ] } \rrbracket   \mathbin{  \accentset{\mathsf{e} }{\evalto}_{\mathsf{S_1} }     \accentset{\mathsf{c} }{\evalto}_{\mathsf{S_1} }  ^*}   \mathscr{C}\llbracket \maybebluetext{ \mathcal{F}  [  \ottnt{N_{{\mathrm{1}}}}  ] } \rrbracket $.

  \item[\Case{\rnp{E-CtxE} with $\mathcal{E}  \ottsym{=}  \mathcal{F}  [  \square \, \langle  \ottnt{t}  \rangle  ]$}]
    We are given
    \[
       \ottnt{M_{{\mathrm{1}}}}    \mathbin{\accentset{\mathsf{e} }{\reduces}_{\mathsf{S} } }    \ottnt{N_{{\mathrm{1}}}}  \hgap
      \ottnt{M}  \ottsym{=}  \ottsym{(}  \mathcal{F}  [  \square \, \langle  \ottnt{t}  \rangle  ]  \ottsym{)}  [  \ottnt{M_{{\mathrm{1}}}}  ] = \mathcal{F}  [  \ottnt{M_{{\mathrm{1}}}}  \langle  \ottnt{t}  \rangle  ] \hgap
      \ottnt{N}  \ottsym{=}  \ottsym{(}  \mathcal{F}  [  \square \, \langle  \ottnt{t}  \rangle  ]  \ottsym{)}  [  \ottnt{N_{{\mathrm{1}}}}  ] = \mathcal{F}  [  \ottnt{N_{{\mathrm{1}}}}  \langle  \ottnt{t}  \rangle  ]
    \]
    for some $\ottnt{M_{{\mathrm{1}}}},\ottnt{N_{{\mathrm{1}}}}$.
    By Lemma~\ref{lem:trans-ctx}~(2), there exist $\mathcal{E}'$
    such that $ \mathscr{C}\llbracket \maybebluetext{ \mathcal{F}  [  \ottnt{L}  \langle  \ottnt{t}  \rangle  ] } \rrbracket   \ottsym{=}  \mathcal{E}'  [   \mathscr{K}\llbracket \maybebluetext{ \ottnt{L} } \rrbracket   \Psi (\maybebluetext{ \ottnt{t} })    ]$ for any $\ottnt{L}$. So,
    \begin{equation}
       \mathscr{C}\llbracket \maybebluetext{ \mathcal{F}  [  \ottnt{M_{{\mathrm{1}}}}  \langle  \ottnt{t}  \rangle  ] } \rrbracket   \ottsym{=}  \mathcal{E}'  [   \mathscr{K}\llbracket \maybebluetext{ \ottnt{M_{{\mathrm{1}}}} } \rrbracket   \Psi (\maybebluetext{ \ottnt{t} })    ] \hgap
       \mathscr{C}\llbracket \maybebluetext{ \mathcal{F}  [  \ottnt{N_{{\mathrm{1}}}}  \langle  \ottnt{t}  \rangle  ] } \rrbracket   \ottsym{=}  \mathcal{E}'  [   \mathscr{K}\llbracket \maybebluetext{ \ottnt{N_{{\mathrm{1}}}} } \rrbracket   \Psi (\maybebluetext{ \ottnt{t} })    ]. \label{eq:sim-trans-ctx-e2}
    \end{equation}
    By $ \ottnt{M_{{\mathrm{1}}}}    \mathbin{\accentset{\mathsf{e} }{\reduces}_{\mathsf{S} } }    \ottnt{N_{{\mathrm{1}}}} $ and Lemma~\ref{lem:trans-reduce}~(1) with $\ottnt{K}  \ottsym{=}   \Psi (\maybebluetext{ \ottnt{t} }) $,
    \[
       \mathscr{K}\llbracket \maybebluetext{ \ottnt{M_{{\mathrm{1}}}} } \rrbracket   \Psi (\maybebluetext{ \ottnt{t} })    \mathbin{  \accentset{\mathsf{e} }{\evalto}_{\mathsf{S_1} }     \accentset{\mathsf{c} }{\evalto}_{\mathsf{S_1} }  ^*}   \mathscr{K}\llbracket \maybebluetext{ \ottnt{N_{{\mathrm{1}}}} } \rrbracket   \Psi (\maybebluetext{ \ottnt{t} })  .
    \]
    By Lemma~\ref{lem:ctx-multi},
    \[
      \mathcal{E}'  [   \mathscr{K}\llbracket \maybebluetext{ \ottnt{M_{{\mathrm{1}}}} } \rrbracket   \Psi (\maybebluetext{ \ottnt{t} })    ]  \mathbin{  \accentset{\mathsf{e} }{\evalto}_{\mathsf{S_1} }     \accentset{\mathsf{c} }{\evalto}_{\mathsf{S_1} }  ^*}  \mathcal{E}'  [   \mathscr{K}\llbracket \maybebluetext{ \ottnt{N_{{\mathrm{1}}}} } \rrbracket   \Psi (\maybebluetext{ \ottnt{t} })    ].
    \]
    Thus, by eq.~\eqref{eq:sim-trans-ctx-e2},
    $ \mathscr{C}\llbracket \maybebluetext{ \mathcal{F}  [  \ottnt{M_{{\mathrm{1}}}}  \langle  \ottnt{t}  \rangle  ] } \rrbracket   \mathbin{  \accentset{\mathsf{e} }{\evalto}_{\mathsf{S_1} }     \accentset{\mathsf{c} }{\evalto}_{\mathsf{S_1} }  ^*}   \mathscr{C}\llbracket \maybebluetext{ \mathcal{F}  [  \ottnt{N_{{\mathrm{1}}}}  \langle  \ottnt{t}  \rangle  ] } \rrbracket $.

  \item[\Case{\rnp{E-Abort} with $\mathcal{E}  \ottsym{=}  \mathcal{F}$}]
    We are given
    \[
      \ottnt{M}  \ottsym{=}  \mathcal{F}  [  \ottkw{blame} \, \ottnt{p}  ] \hgap
      \ottnt{N}  \ottsym{=}  \ottkw{blame} \, \ottnt{p}
    \]
    for some $\ottnt{p}$.
    By Lemma~\ref{lem:trans-ctx}~(1), there exists $\mathcal{E}'$ such that
    $ \mathscr{C}\llbracket \maybebluetext{ \mathcal{F}  [  \ottnt{L}  ] } \rrbracket   \ottsym{=}  \mathcal{E}'  [   \mathscr{C}\llbracket \maybebluetext{ \ottnt{L} } \rrbracket   ]$ for any $\ottnt{L}$. So,
    \[
       \mathscr{C}\llbracket \maybebluetext{ \mathcal{F}  [  \ottkw{blame} \, \ottnt{p}  ] } \rrbracket   \ottsym{=}  \mathcal{E}'  [   \mathscr{C}\llbracket \maybebluetext{ \ottkw{blame} \, \ottnt{p} } \rrbracket   ].
    \]
    By $ \mathscr{C}\llbracket \maybebluetext{ \ottkw{blame} \, \ottnt{p} } \rrbracket   \ottsym{=}   \mathscr{K}\llbracket \maybebluetext{ \ottkw{blame} \, \ottnt{p} } \rrbracket  \ottkw{id}  = \ottkw{blame} \, \ottnt{p}$,  
    \[
    \begin{array}{llll}
      \mathcal{E}'  [   \mathscr{C}\llbracket \maybebluetext{ \ottkw{blame} \, \ottnt{p} } \rrbracket   ]
      &=& \mathcal{E}'  [  \ottkw{blame} \, \ottnt{p}  ] \\
      & \accentset{\mathsf{e} }{\evalto}_{\mathsf{S_1} } & \ottkw{blame} \, \ottnt{p} &\text{by \rnp{E-Abort}}\\
      &=&  \mathscr{C}\llbracket \maybebluetext{ \ottkw{blame} \, \ottnt{p} } \rrbracket .
    \end{array}
    \]
    Thus, $  \mathscr{C}\llbracket \maybebluetext{ \mathcal{F}  [  \ottkw{blame} \, \ottnt{p}  ] } \rrbracket     \mathbin{  \accentset{\mathsf{e} }{\evalto}_{\mathsf{S_1} }  }     \mathscr{C}\llbracket \maybebluetext{ \ottkw{blame} \, \ottnt{p} } \rrbracket  $.

  \item[\Case{\rnp{E-Abort} with $\mathcal{E}  \ottsym{=}  \mathcal{F}  [  \square \, \langle  \ottnt{t}  \rangle  ]$}]
    We are given
    \[
      \ottnt{M}  \ottsym{=}  \ottsym{(}  \mathcal{F}  [  \square \, \langle  \ottnt{t}  \rangle  ]  \ottsym{)}  [  \ottkw{blame} \, \ottnt{p}  ] = \mathcal{F}  [  \ottsym{(}  \ottkw{blame} \, \ottnt{p}  \ottsym{)}  \langle  \ottnt{t}  \rangle  ] \hgap
      \ottnt{N}  \ottsym{=}  \ottkw{blame} \, \ottnt{p}.
    \]
    for some $\ottnt{p}$.
    By Lemma~\ref{lem:trans-ctx}~(2), there exist $\mathcal{E}'$
    such that $ \mathscr{C}\llbracket \maybebluetext{ \mathcal{F}  [  \ottnt{L}  \langle  \ottnt{t}  \rangle  ] } \rrbracket   \ottsym{=}  \mathcal{E}'  [   \mathscr{K}\llbracket \maybebluetext{ \ottnt{L} } \rrbracket   \Psi (\maybebluetext{ \ottnt{t} })    ]$ for any $\ottnt{L}$. So,
    \[
       \mathscr{C}\llbracket \maybebluetext{ \mathcal{F}  [  \ottsym{(}  \ottkw{blame} \, \ottnt{p}  \ottsym{)}  \langle  \ottnt{t}  \rangle  ] } \rrbracket   \ottsym{=}  \mathcal{E}'  [   \mathscr{K}\llbracket \maybebluetext{ \ottkw{blame} \, \ottnt{p} } \rrbracket   \Psi (\maybebluetext{ \ottnt{t} })    ]
    \]
    By $ \mathscr{K}\llbracket \maybebluetext{ \ottkw{blame} \, \ottnt{p} } \rrbracket   \Psi (\maybebluetext{ \ottnt{t} })    \ottsym{=}  \ottkw{blame} \, \ottnt{p}$ and $ \mathscr{C}\llbracket \maybebluetext{ \ottkw{blame} \, \ottnt{p} } \rrbracket   \ottsym{=}  \ottkw{blame} \, \ottnt{p}$,
    \[
    \begin{array}{llll}
      \mathcal{E}'  [   \mathscr{K}\llbracket \maybebluetext{ \ottkw{blame} \, \ottnt{p} } \rrbracket   \Psi (\maybebluetext{ \ottnt{t} })    ]
      &=& \mathcal{E}'  [  \ottkw{blame} \, \ottnt{p}  ] \\
      & \accentset{\mathsf{e} }{\evalto}_{\mathsf{S_1} } & \ottkw{blame} \, \ottnt{p} &\text{by \rnp{E-Abort}}\\
      &=&  \mathscr{C}\llbracket \maybebluetext{ \ottkw{blame} \, \ottnt{p} } \rrbracket .
    \end{array}
    \]
    Thus, $  \mathscr{C}\llbracket \maybebluetext{ \ottsym{(}  \mathcal{F}  [  \square \, \langle  \ottnt{t}  \rangle  ]  \ottsym{)}  [  \ottkw{blame} \, \ottnt{p}  ] } \rrbracket     \mathbin{  \accentset{\mathsf{e} }{\evalto}_{\mathsf{S_1} }  }     \mathscr{C}\llbracket \maybebluetext{ \ottkw{blame} \, \ottnt{p} } \rrbracket  $.
  \end{description}

  (2) By case analysis on the evaluation rule applied to $ \ottnt{M}    \mathbin{  \accentset{\mathsf{c} }{\evalto}_{\mathsf{S} }  }    \ottnt{N} $.
  \begin{description}
  \item[\Case{\rnp{E-CtxC}}]
    We are given
    \[
       \ottnt{M_{{\mathrm{1}}}}    \mathbin{\accentset{\mathsf{c} }{\reduces}_{\mathsf{S} } }    \ottnt{N_{{\mathrm{1}}}}  \hgap
      \ottnt{M}  \ottsym{=}  \mathcal{F}  [  \ottnt{M_{{\mathrm{1}}}}  ] \hgap
      \ottnt{N}  \ottsym{=}  \mathcal{F}  [  \ottnt{N_{{\mathrm{1}}}}  ]
    \]
    for some $\mathcal{F},\ottnt{M_{{\mathrm{1}}}},\ottnt{N_{{\mathrm{1}}}}$.
    By Lemma~\ref{lem:trans-ctx}~(1), there exists $\mathcal{E}'$ such that
    $ \mathscr{C}\llbracket \maybebluetext{ \mathcal{F}  [  \ottnt{L}  ] } \rrbracket   \ottsym{=}  \mathcal{E}'  [   \mathscr{C}\llbracket \maybebluetext{ \ottnt{L} } \rrbracket   ]$ for any $\ottnt{L}$. So,
    \begin{equation}
       \mathscr{C}\llbracket \maybebluetext{ \mathcal{F}  [  \ottnt{M_{{\mathrm{1}}}}  ] } \rrbracket   \ottsym{=}  \mathcal{E}'  [   \mathscr{C}\llbracket \maybebluetext{ \ottnt{M_{{\mathrm{1}}}} } \rrbracket   ] \hgap
       \mathscr{C}\llbracket \maybebluetext{ \mathcal{F}  [  \ottnt{N_{{\mathrm{1}}}}  ] } \rrbracket   \ottsym{=}  \mathcal{E}'  [   \mathscr{C}\llbracket \maybebluetext{ \ottnt{N_{{\mathrm{1}}}} } \rrbracket   ]. \label{eq:sim-trans-ctx-c}
    \end{equation}
    By $ \ottnt{M_{{\mathrm{1}}}}    \mathbin{\accentset{\mathsf{c} }{\reduces}_{\mathsf{S} } }    \ottnt{N_{{\mathrm{1}}}} $ and Lemma~\ref{lem:trans-reduce}~(2),
    \[
        \mathscr{C}\llbracket \maybebluetext{ \ottnt{M_{{\mathrm{1}}}} } \rrbracket     \mathbin{  \accentset{\mathsf{c} }{\evalto}_{\mathsf{S_1} }  ^+}     \mathscr{C}\llbracket \maybebluetext{ \ottnt{N_{{\mathrm{1}}}} } \rrbracket  
    \]
    By Lemma~\ref{lem:ctx-multi},
    $ \mathcal{E}'  [   \mathscr{C}\llbracket \maybebluetext{ \ottnt{M_{{\mathrm{1}}}} } \rrbracket   ]    \mathbin{  \accentset{\mathsf{c} }{\evalto}_{\mathsf{S_1} }  ^+}    \mathcal{E}'  [   \mathscr{C}\llbracket \maybebluetext{ \ottnt{N_{{\mathrm{1}}}} } \rrbracket   ] $.
    Thus, by eq.~\eqref{eq:sim-trans-ctx-c}, $  \mathscr{C}\llbracket \maybebluetext{ \mathcal{F}  [  \ottnt{M_{{\mathrm{1}}}}  ] } \rrbracket     \mathbin{  \accentset{\mathsf{c} }{\evalto}_{\mathsf{S_1} }  ^+}     \mathscr{C}\llbracket \maybebluetext{ \mathcal{F}  [  \ottnt{N_{{\mathrm{1}}}}  ] } \rrbracket  $.
    \qedhere
  \end{description}

\end{proof}

\iffull
\lemSimulation*
\else
\begin{lemma} \label{lem:trans-reduction-preserve}
  If $ \ottnt{M}    \mathbin{  \evalto_{\mathsf{S} }  }    \ottnt{N} $, then $  \mathscr{C}\llbracket \maybebluetext{ \ottnt{M} } \rrbracket     \mathbin{  \evalto_{\mathsf{S_1} }  ^+}     \mathscr{C}\llbracket \maybebluetext{ \ottnt{N} } \rrbracket  $.
\end{lemma}
\fi

\[
\xymatrix@=50pt{
    \ottnt{M} \ar@{|->}[r]_>>{\mathsf{S}} \ar[d]^{\mathscr{C} \llbracket \_  \rrbracket} & \ottnt{N} \ar[d]^{\mathscr{C} \llbracket \_  \rrbracket} \\
     \mathscr{C}\llbracket \maybebluetext{ \ottnt{M} } \rrbracket  \ar@{|.>}[r] ^>>{+}_>>{\mathsf{S}_1} &   \mathscr{C}\llbracket \maybebluetext{ \ottnt{N} } \rrbracket 
}
\]

\begin{proof}
  Immediate by Lemma~\ref{lem:trans-eval}.
\end{proof}

\iffull\thmTransSoundness*
\else
\begin{theorem}[Translation Soundness] \label{thm:trans-soundness}
  Suppose $ \Gamma    \vdash_{\mathsf{S} }    \ottnt{M}  :  \ottnt{A} $.
  \begin{enumerate}
  \item If $ \ottnt{M}    \mathbin{  \evalto_{\mathsf{S} }  ^*}    \ottnt{V} $, then $  \mathscr{C}\llbracket \maybebluetext{ \ottnt{M} } \rrbracket     \mathbin{  \evalto_{\mathsf{S_1} }  ^*}     \Psi (\maybebluetext{ \ottnt{V} })  $.
  \item If $ \ottnt{M}    \mathbin{  \evalto_{\mathsf{S} }  ^*}    \ottkw{blame} \, \ottnt{p} $, then $  \mathscr{C}\llbracket \maybebluetext{ \ottnt{M} } \rrbracket     \mathbin{  \evalto_{\mathsf{S_1} }  ^*}    \ottkw{blame} \, \ottnt{p} $.
  \item If $\ottnt{M} \,  \mathord{\Uparrow_{\mathsf{S} } } $, then $ \mathscr{C}\llbracket \maybebluetext{ \ottnt{M} } \rrbracket  \,  \mathord{\Uparrow_{\mathsf{S_1} } } $.
  \end{enumerate}
\end{theorem}
\fi
\begin{proof}
  By repeated use of Lemma~\ref{lem:trans-reduction-preserve}.
  We also use
  $ \mathscr{C}\llbracket \maybebluetext{ \ottnt{V} } \rrbracket   \ottsym{=}   \Psi (\maybebluetext{ \ottnt{V} }) $ and
  $ \mathscr{C}\llbracket \maybebluetext{ \ottkw{blame} \, \ottnt{p} } \rrbracket   \ottsym{=}   \mathscr{K}\llbracket \maybebluetext{ \ottkw{blame} \, \ottnt{p} } \rrbracket  \ottkw{id}  = \ottkw{blame} \, \ottnt{p}$.
\end{proof}

\iffull\thmTransSem*
\else
\begin{corollary}[Translation Preserves Semantcs] \label{cor:trans-correctness}  
  Suppose $  \emptyset     \vdash_{\mathsf{S} }    \ottnt{M}  :  \iota $.
  \begin{enumerate}
  \item $ \ottnt{M}    \mathbin{  \evalto_{\mathsf{S} }  ^*}    \ottnt{a} $ iff $  \mathscr{K}\llbracket \maybebluetext{ \ottnt{M} } \rrbracket   \ottkw{id} _{ \iota }      \mathbin{  \evalto_{\mathsf{S_1} }  ^*}    \ottnt{a} $.
  \item $ \ottnt{M}    \mathbin{  \evalto_{\mathsf{S} }  ^*}    \ottkw{blame} \, \ottnt{p} $ iff $  \mathscr{K}\llbracket \maybebluetext{ \ottnt{M} } \rrbracket   \ottkw{id} _{ \iota }      \mathbin{  \evalto_{\mathsf{S_1} }  ^*}    \ottkw{blame} \, \ottnt{p} $.
  \item $\ottnt{M} \,  \mathord{\Uparrow_{\mathsf{S} } } $ iff $ \mathscr{K}\llbracket \maybebluetext{ \ottnt{M} } \rrbracket   \ottkw{id} _{ \iota }   \,  \mathord{\Uparrow_{\mathsf{S_1} } } $.
  \end{enumerate}
\end{corollary}
\fi
\begin{proof}
  The left-to-right direction follows Theorem~\ref{thm:trans-soundness}.
  (Note that $ \Psi (\maybebluetext{ \ottnt{a} })   \ottsym{=}  \ottnt{a}$.)

  We prove the right-to-left direction of (1).
  We are given $  \mathscr{K}\llbracket \maybebluetext{ \ottnt{M} } \rrbracket   \ottkw{id} _{ \iota }      \mathbin{  \evalto_{\mathsf{S_1} }  ^*}    \ottnt{a} $.
  By $  \emptyset     \vdash_{\mathsf{S} }    \ottnt{M}  :  \iota $ and Corollary~\ref{cor:safety-S},  
  either of the following holds:
  \[
     \ottnt{M}    \mathbin{  \evalto_{\mathsf{S} }  ^*}    \ottnt{a}  \hgap
     \ottnt{M}    \mathbin{  \evalto_{\mathsf{S} }  ^*}    \ottkw{blame} \, \ottnt{p} \hgap
    \ottnt{M} \,  \mathord{\Uparrow_{\mathsf{S} } } 
  \]
  \begin{itemize}
  \item If $ \ottnt{M}    \mathbin{  \evalto_{\mathsf{S} }  ^*}    \ottkw{blame} \, \ottnt{p} $,
    then by Theorem~\ref{thm:trans-soundness},
    $  \mathscr{K}\llbracket \maybebluetext{ \ottnt{M} } \rrbracket   \ottkw{id} _{ \iota }      \mathbin{  \evalto_{\mathsf{S_1} }  ^*}    \ottkw{blame} \, \ottnt{p} $.
    It contradicts $  \mathscr{K}\llbracket \maybebluetext{ \ottnt{M} } \rrbracket   \ottkw{id} _{ \iota }      \mathbin{  \evalto_{\mathsf{S_1} }  ^*}    \ottnt{a} $
    by Lemma~\ref{lem:determinism-S1}.  
  \item If $\ottnt{M} \,  \mathord{\Uparrow_{\mathsf{S} } } $, then 
    by Theorem~\ref{thm:trans-soundness}, $ \mathscr{K}\llbracket \maybebluetext{ \ottnt{M} } \rrbracket  \ottkw{id}  \,  \mathord{\Uparrow_{\mathsf{S_1} } } $.
    It contradicts $  \mathscr{K}\llbracket \maybebluetext{ \ottnt{M} } \rrbracket   \ottkw{id} _{ \iota }      \mathbin{  \evalto_{\mathsf{S_1} }  ^*}    \ottnt{a} $
    by Lemma~\ref{lem:determinism-S1}.
  \end{itemize}
  Thus, $ \ottnt{M}    \mathbin{  \evalto_{\mathsf{S} }  ^*}    \ottnt{a} $.

  The right-to-left directions of (2) and (3) are similar.
\end{proof}

\section{Detailed Benchmark Results}
\label{sec:detailedbenchmark}

Figure~\ref{fig:ratios} shows the scatter plots of the running time ratios
for each benchmark program, generated by the same experiment as Figure~\ref{fig:boxplot}.
The x-axis indicates how much of type annotations in the benchmark program are given static type.

\begin{figure}[tb]
  \centering
  \begin{tabular}{cc}
    \includegraphics[width=0.5\columnwidth]{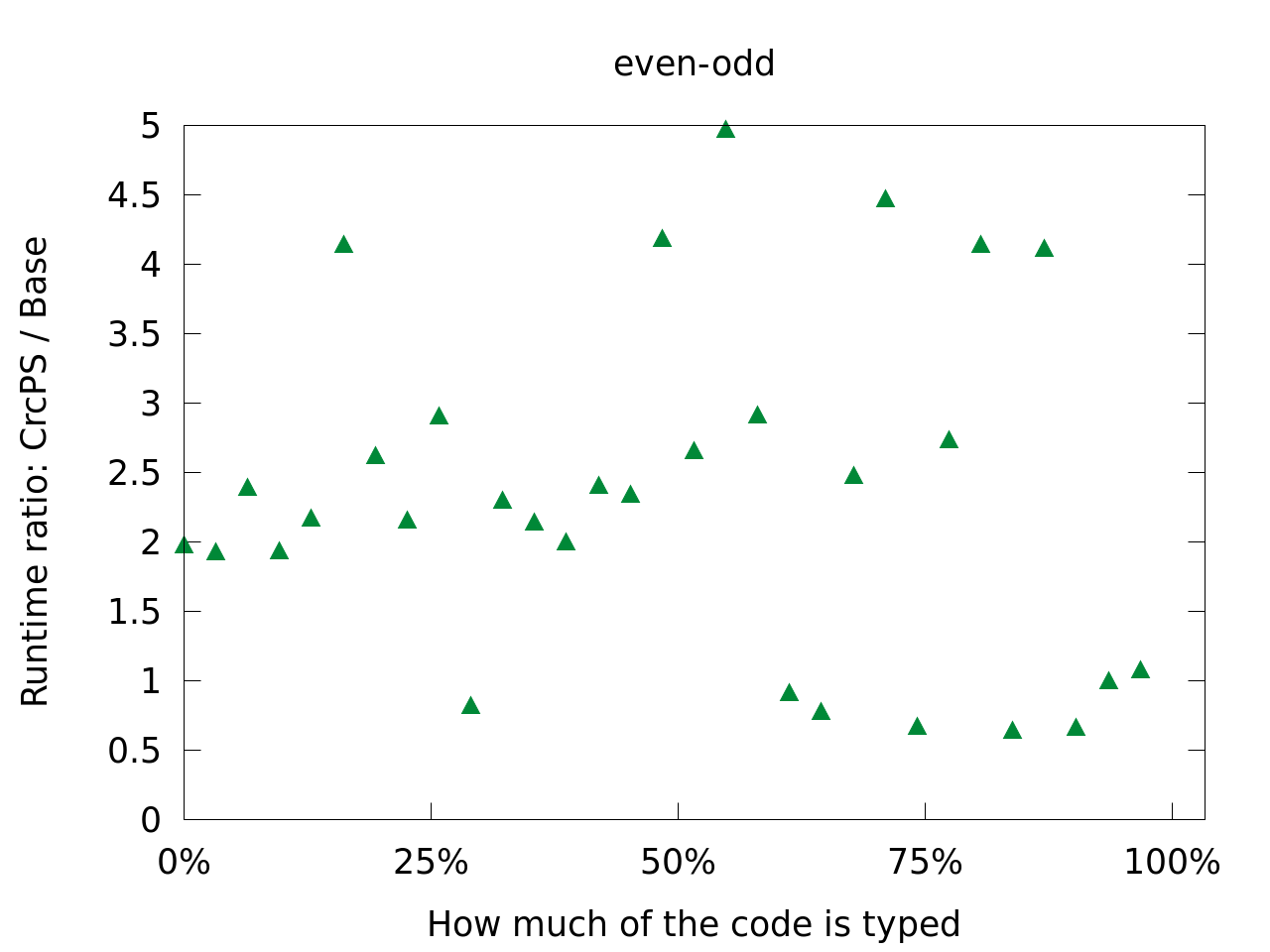} &
    \includegraphics[width=0.5\columnwidth]{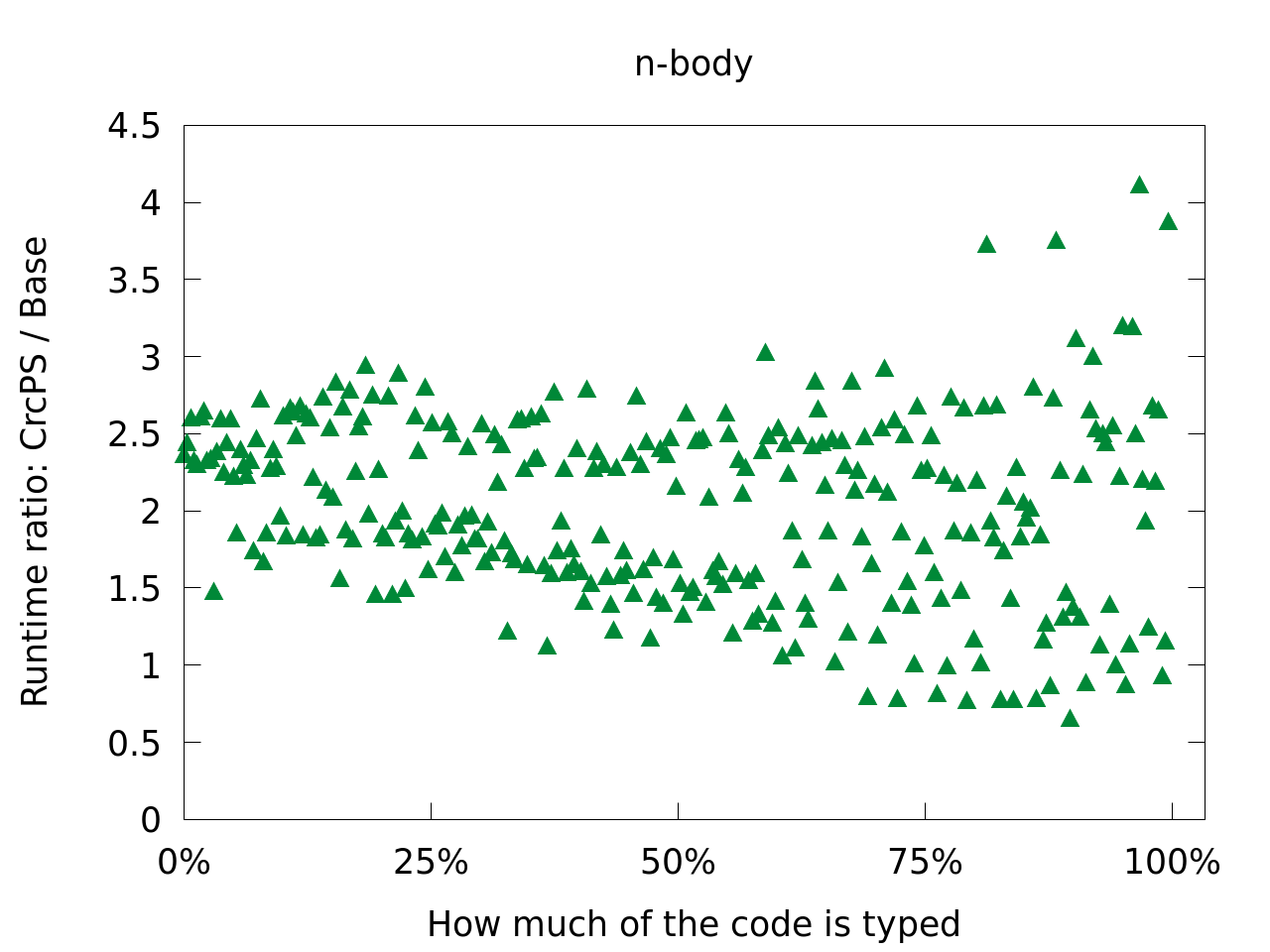} \\
    \includegraphics[width=0.5\columnwidth]{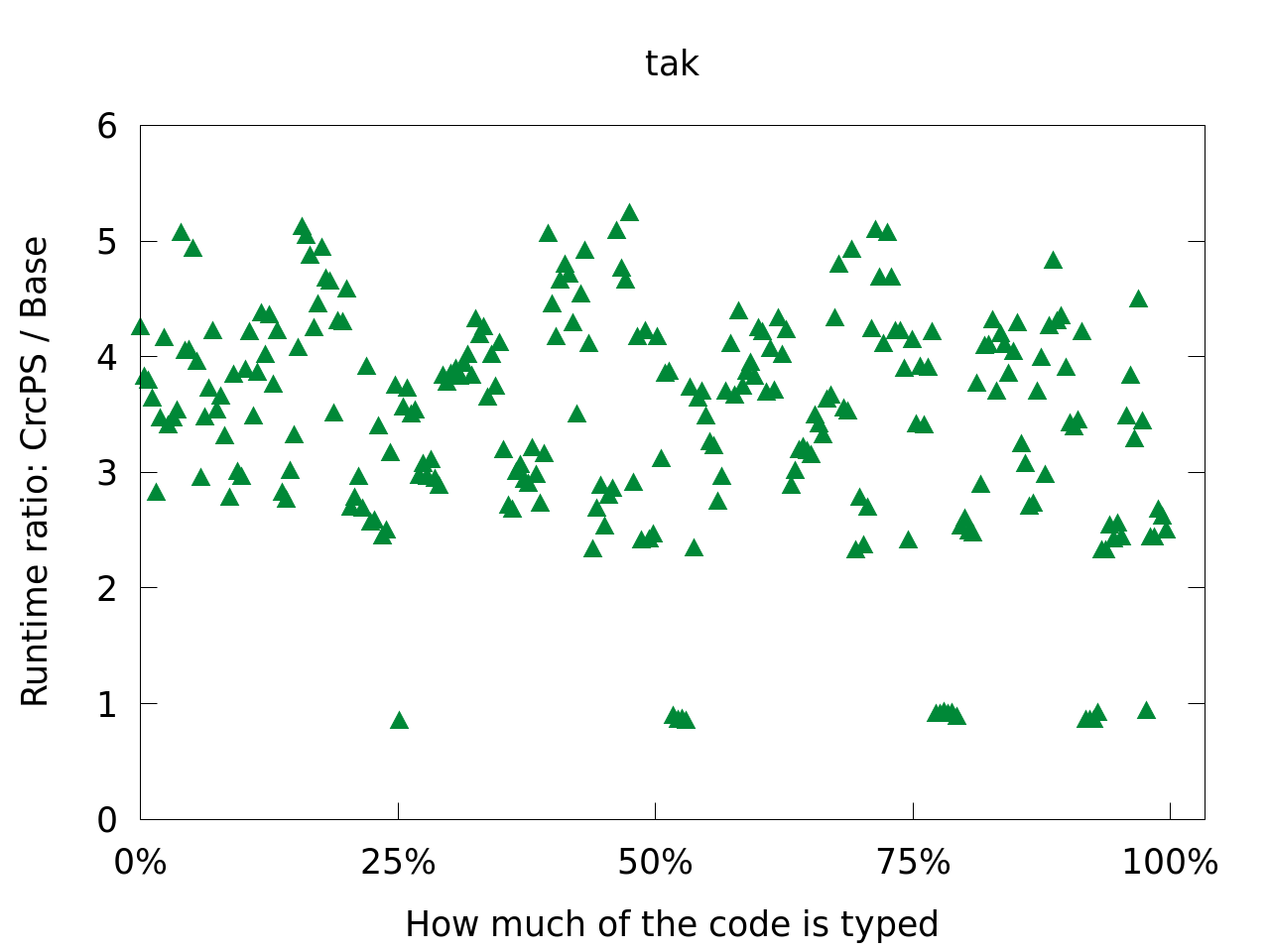} &
    \includegraphics[width=0.5\columnwidth]{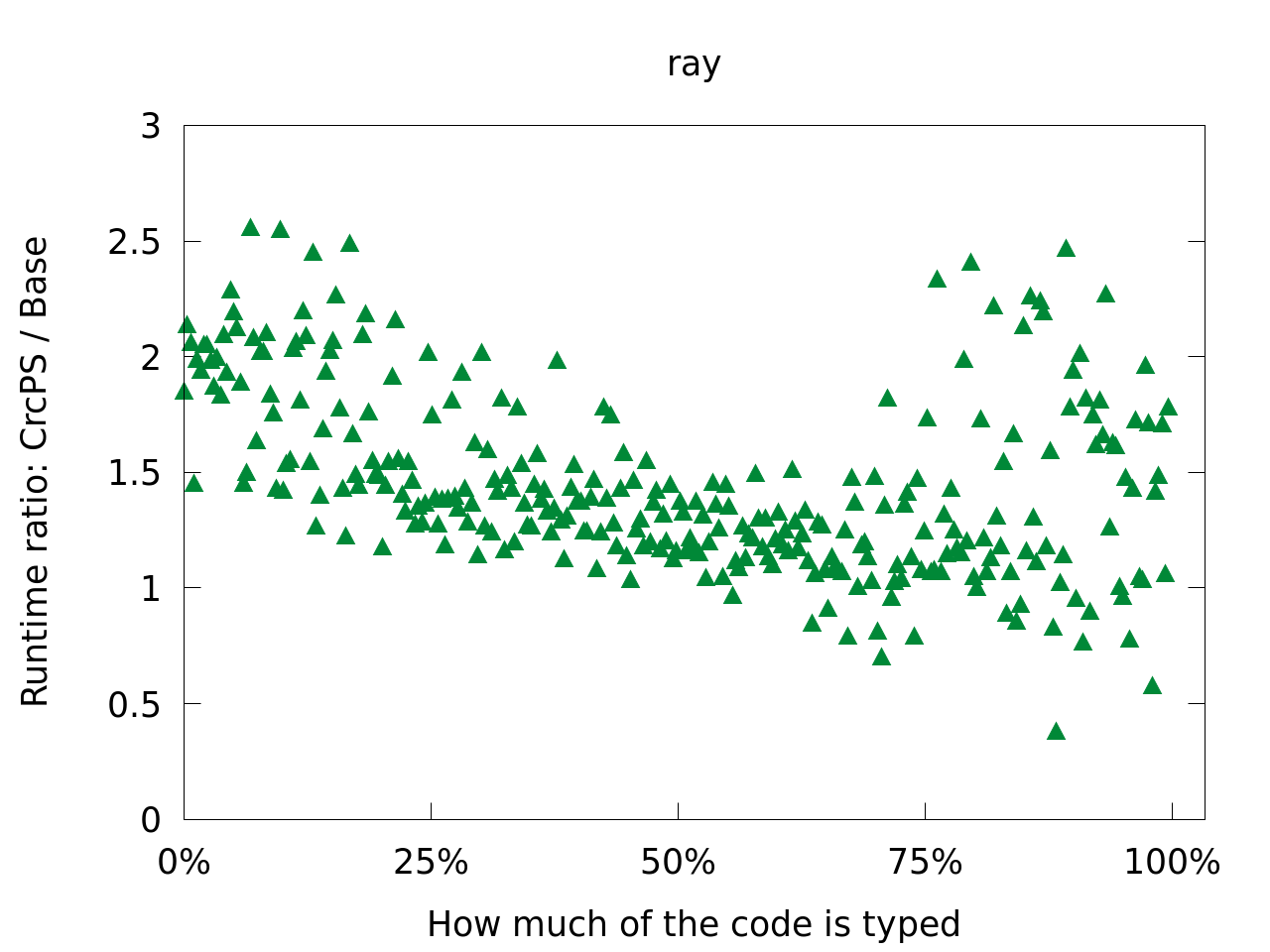} \\
    \includegraphics[width=0.5\columnwidth]{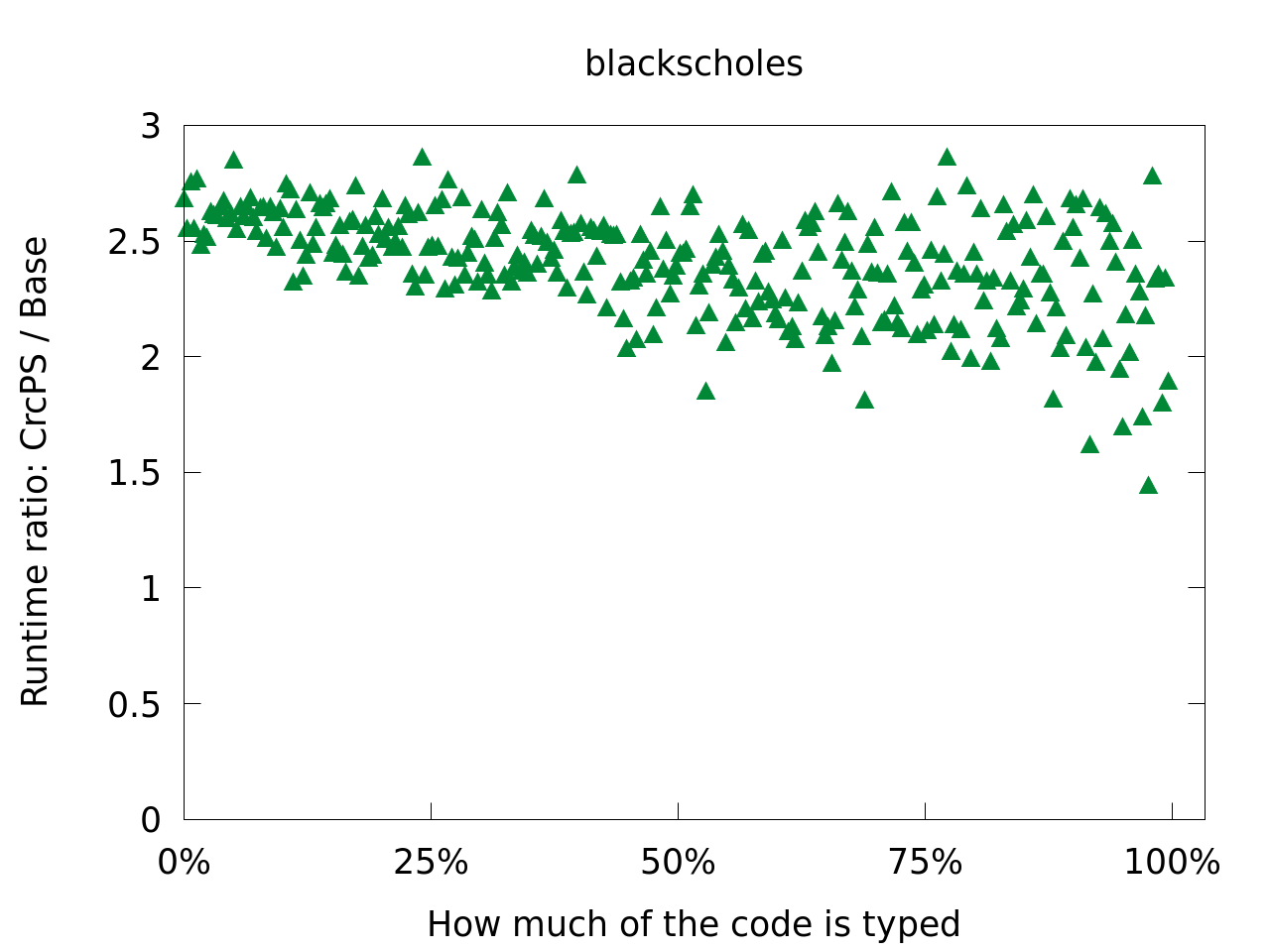} &
    \includegraphics[width=0.5\columnwidth]{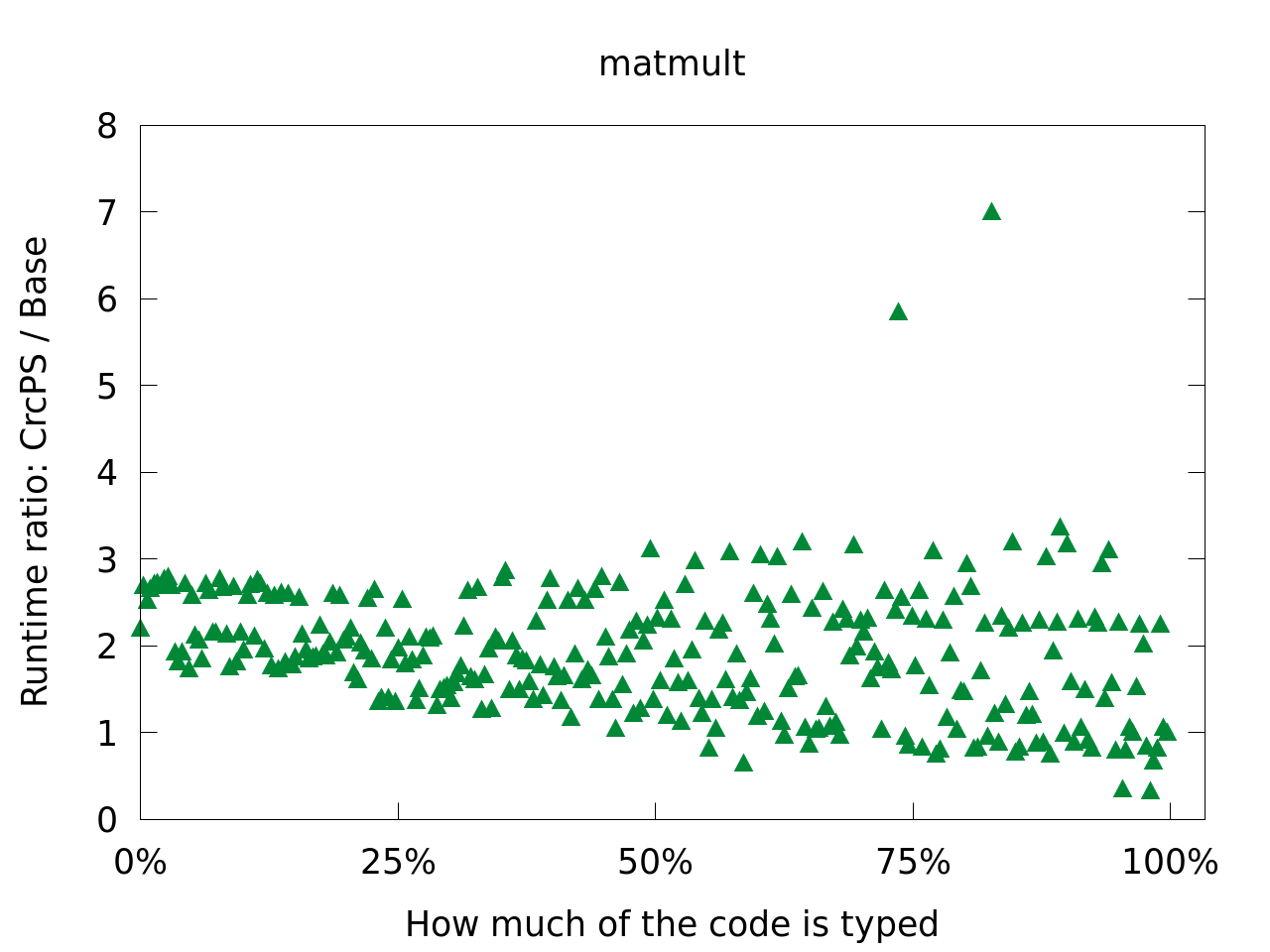} \\
    \includegraphics[width=0.5\columnwidth]{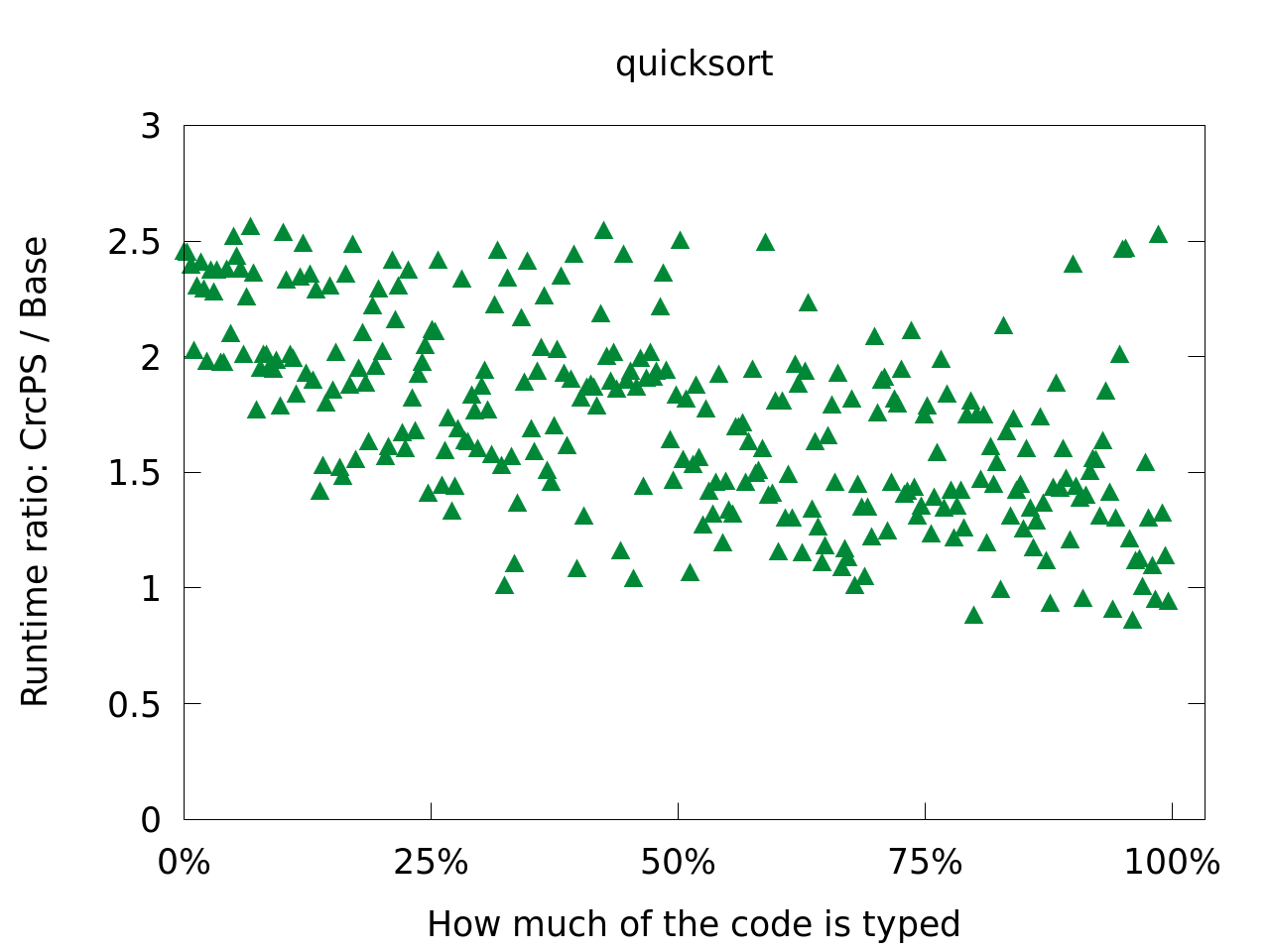} &
    \includegraphics[width=0.5\columnwidth]{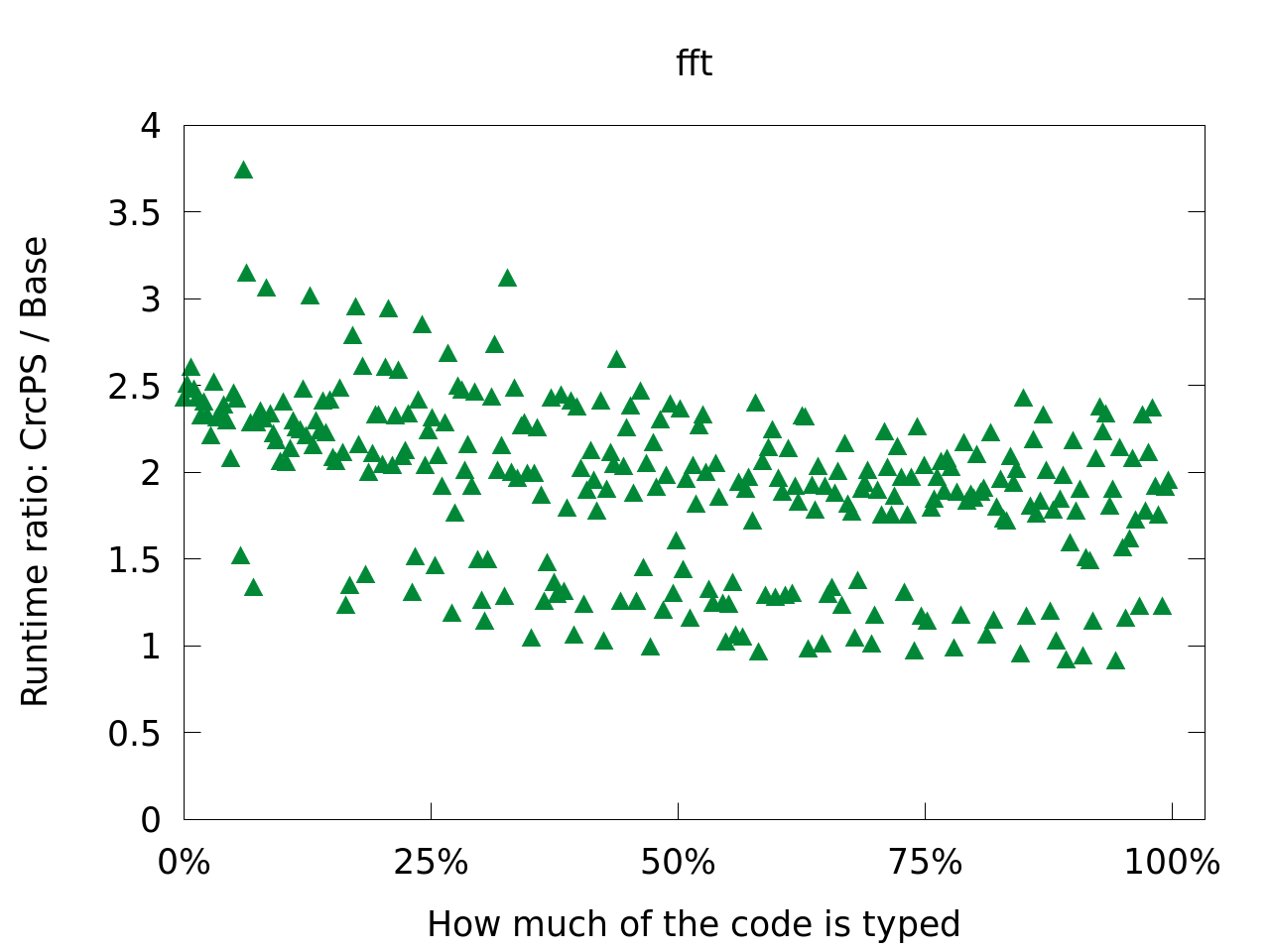}
  \end{tabular}
  \caption{Scatter plots of the running time ratios of CrcPS to Base
    across (sampled) partially typed configurations for each benchmark program.}
  \label{fig:ratios}
\end{figure}

\fi

\end{document}